\documentclass[11pt,english]{article}
\usepackage{color}
\usepackage[T1]{fontenc}
\usepackage[latin9]{inputenc}
\usepackage[a4paper]{geometry}
\usepackage{textcomp}
\usepackage{amsthm}
\usepackage{amsmath}
\usepackage{amssymb}
\usepackage{esint}
\usepackage{bbm}
\usepackage{hyperref}
\usepackage{babel}
\usepackage{amscd}
\usepackage{amsfonts}
\usepackage{amsthm}
\usepackage{mathrsfs}
\usepackage{dsfont}
\usepackage{mathtools}
\usepackage{enumerate}
\usepackage{bm}
\usepackage{bbm}
\usepackage{verbatim}
\usepackage{enumitem}
\newlist{abbrv}{itemize}{1}
\setlist[abbrv,1]{label=,labelwidth=1.2in,align=parleft,itemsep=0.1\baselineskip,leftmargin=!}

\makeatletter

\DeclareFontEncoding{LGR}{}{}
\DeclareTextSymbol{\~}{LGR}{126}

\newcommand{\tr}{\mathrm{Tr}}
\newcommand{\abs}[1]{\left| #1 \right|}

\newcommand{\scp}[2]{\big\langle #1 , #2 \big\rangle}
\newcommand{\SCP}[2]{\big\langle #1 , #2 \big\rangle}
\newcommand{\bra}[1]{\langle #1 |}
\newcommand{\ket}[1]{| #1 \rangle}

\newcommand{\norm}[1]{\left\| #1 \right\|}

\renewcommand{\Re}{\mathrm{Re}}
\renewcommand{\Im}{\mathrm{Im}}

\newcommand{\id}{\mathbbm{1}}

\newcommand{\vAcq}{\boldsymbol{\hat{A}}_{\kappa}}
\newcommand{\vAp}{\boldsymbol{\hat{A}}^{+}_{\kappa}}
\newcommand{\vAm}{\boldsymbol{\hat{A}}^{-}_{\kappa}}
\newcommand{\vAc}{\boldsymbol{A}_{\kappa}}

\newcommand{\vA}{\boldsymbol{A}}

\newcommand{\vF}{\boldsymbol{F}}
\newcommand{\vJ}{\boldsymbol{J}}

\newcommand{\vJt}{\widetilde{\boldsymbol{J}}}

\newcommand{\vAh}{\hat{\boldsymbol{A}}}

\newcommand{\vAk}{\boldsymbol{A}_{\kappa}}

\newcommand{\vAhk}{\hat{\boldsymbol{A}}_{\kappa}}

\newcommand{\vep}{\boldsymbol{\epsilon}}

\newtheorem{theorem}{Theorem}[section]
\newtheorem{lemma}[theorem]{Lemma}

\newtheorem{remark}[theorem]{Remark}
\newtheorem{definition}[theorem]{Definition}
\newtheorem{assumption}[theorem]{Assumption}
\newtheorem{proposition}[theorem]{Proposition}

\usepackage[colorinlistoftodos]{todonotes}

\usepackage{ulem}
\usepackage{yfonts}
\usepackage{stmaryrd}

\DeclareMathOperator*{\esssup}{ess\,sup}
\DeclareMathOperator\supp{supp}

\allowdisplaybreaks[1]

\begin{document}
\title{Derivation of the Maxwell--Schr\"odinger and Vlasov--Maxwell Equations from Non-Relativistic QED}

\author{Nikolai Leopold\footnote{Constructor University Bremen,  School of Science,  Campus Ring~1,  28759 Bremen, Germany,  and University of Basel, Department of Mathematics and Computer Science, Spiegelgasse~1, 4051 Basel, Switzerland.  E-mail address: {\tt nleopold@constructor.university}} }

\date{\today}

\maketitle

\begin{abstract}
\noindent
We study the spinless Pauli--Fierz Hamiltonian in a semiclassical mean-field limit of many fermions. For appropriate initial conditions, we prove, in the trace norm topology of reduced density matrices, that the many-body quantum state converges to a tensor product of a semiclassically structured Slater determinant and a coherent photon state. These evolve according to a fermionic variant of the Maxwell--Schr\"odinger equations. By combining this result with [Forum Math.  Sigma 14:e17 (2026)] through a suitable regularization of the initial data, we further show that, in the limit of large particle number, the dynamics of the Pauli--Fierz Hamiltonian can be approximately  described by the non-relativistic Vlasov--Maxwell system for extended charges.
\end{abstract}

\section{Introduction}

Depending on the physical context light can appear very differently.
It is modeled as rays in geometrical optics, as waves in classical electromagnetism, and as composed of energy quanta--first postulated  by Einstein in 1905 \cite{E1905} and later named photons--in the framework of quantum electrodynamics (QED). 
Geometrical optics is an approximation of classical electrodynamics, which in turn provides  an approximation of QED. Gaining a profound understanding of how these different descriptions of light are related is an intriguing challenge in mathematical physics.
The aim of this work is to contribute to this undertaking by establishing a link between quantum and classical electrodynamics through asymptotic analysis. Specifically, 
it will be shown that, in certain situations, the  time evolution of a mathematical model of non-relativistic QED can be approximated by two effective dynamics in which the electromagnetic field is described as a classical field satisfying Maxwell's equations.

The article is organized as follows. In the remainder of this section, we introduce the models under consideration and fix the notation. Section~\ref{section:main results} presents the main results and provides a comparison with the literature. The key components needed to prove the main results are introduced in Section~\ref{section:proof of the main results} through several lemmas. Their proofs are provided in Section~\ref{section:properties of the MS equations}--\ref{section:estimates concerning the Vlasov--Maxwell equations}.

\subsection{The Pauli--Fierz Hamiltonian}

As a mathematical model of non-relativistic quantum electrodynamics, we consider the spinless Pauli--Fierz Hamiltonian, which is obtained by the canonical quantization of the Abraham model. For further details on this procedure and the Pauli--Fierz Hamiltonian, we refer  to \cite{S2004}.
We are interested in  a system of $N$ identical fermions interacting with a quantized electromagnetic field. The state of the electrons is described by elements of the Hilbert space $L_{\rm{as}}(\mathbb{R}^{3N})$, which is the subspace of complex-valued, square integrable $N$-particle wave functions that are antisymmetric with respect to the exchange of any pair of particle coordinates. The excitations of the electromagnetic field, known as photons, are represented by the Fock space
$\mathcal{F} \coloneqq \bigoplus_{n \geq 0} \left[ L^2(\mathbb{R}^3) \otimes \mathbb{C}^2 \right]^{\otimes_{\rm{s}}^n}$, 
where the subscript ``s'' indicates symmetry under exchange of variables. The Hilbert space of a single photon, denoted by $\mathfrak{h} \coloneqq L^2(\mathbb{R}^3) \otimes \mathbb{C}^2$ in the following, consists of wave functions $f(k,\lambda)$, where $k \in \mathbb{R}^3$ is the wave vector, and $\lambda = 1,2$ denotes the photon's helicity. The inner product on this space is defined as 
\begin{align}
\scp{f}{g}_{\mathfrak{h}}
&\coloneqq \sum_{\lambda =1,2} \int_{\mathbb{R}^3} \overline{f(k,\lambda)} g(k,\lambda) \, dk .
\end{align}
Finally, the whole system is represented by the Hilbert space
\begin{align}
\mathcal{H}^{(N)}
&\coloneqq L_{as}^2(\mathbb{R}^{3N}) \otimes \mathcal{F} .
\end{align}
It's time evolution is governed by the Schr\"odinger equation
$i \hbar \partial_t \Psi_{N,t} = H_N^{\rm PF} \Psi_{N,t}$
with 
\begin{align}
H_N^{\rm PF} &= \sum_{j=1}^N \frac{1}{2m} \left( - i \hbar \nabla_j - c^{-1} e \kappa * \widehat{\vA}(x_j) \right)^2
+ \frac{e^2}{8\pi} \sum_{\substack{j,k=1\\ j\neq k}}^N  \kappa * \kappa * | \cdot |^{-1} (x_j - x_k)
+ H_f 
\end{align} 
being the spinless Pauli--Fierz Hamiltonian.
Here, $\hbar$ is the reduced Planck constant, $m >0$ is the mass of the electrons, $e$ is the coupling constant to the electromagnetic field, $c$ is the speed of light and $\omega(k) = c \abs{k}$ is the dispersion of the photons. 
$H_f  \coloneqq \sum_{\lambda =1,2} \int_{\mathbb{R}^3} \hbar\,\omega(k) a_{k,\lambda}^* a_{k,\lambda}\,dk $
denotes the energy of the electromagnetic field and 
\begin{align}
\hat{\vA}(x)
&= \frac{1}{(2 \pi)^{3/2}}  \sum_{\lambda=1,2}  \int_{\mathbb{R}^3} 
 \frac{c \, \sqrt{\hbar}}{\sqrt{2 \omega(k)}}\,\vep_{\lambda}(k) 
 \left(  e^{ikx} a_{k,\lambda}  + e^{- ik x} a^*_{k,\lambda} \right)\,dk
\end{align}
is the quantized transverse vector potential. The two real polarization vectors $\big\{ \vep_{\lambda}(k) \big\}_{\lambda = 1,2}$ implement the Coulomb gauge $\nabla\cdot \widehat{\vA}=0$ by satisfying
$\vep_{\lambda}(k) \cdot \vep_{\lambda'}(k) = \delta_{\lambda, \lambda'}$ and 
$k \cdot \vep_{\lambda}(k) = 0$.
The operators $a_{k,\lambda}$ and $a^*_{k,\lambda}$ with $k \in \mathbb{R}^3$ and $\lambda \in \{1, 2\}$ denote the usual annihilation and creation operators. They are operator valued distributions satisfying the canonical commutation relations 
\begin{align}
\label{eq:CCR}
\left[ a_{k,\lambda} , a^*_{k',\lambda'} \right] = \delta_{\lambda, \lambda'} \delta(k - k') ,
\quad 
\left[ a_{k,\lambda} , a_{k',\lambda'} \right] 
= \left[ a^*_{k,\lambda} , a^*_{k',\lambda'} \right] 
= 0,
\end{align}
where $[A,B]:=AB-BA$ is the standard commutator of the operators $A$ and $B$. The real function $\kappa$ describes the density of the electrons. In case its Fourier transform\footnote{Throughout this article we use the notation
 $\mathcal{F}[f](k) = (2 \pi)^{-3/2} \int_{\mathbb{R}^3}  e^{-ikx} f(x)\,dx$ to denote the Fourier transform of a function $f$.} satisfies
\begin{align}
\label{eq:assumption cutoff for self-adjointess of PF Hamiltonian}
\big( \abs{\cdot}^{-1} + \abs{\cdot}^{1/2} \big) \mathcal{F}[\kappa] \in L^2(\mathbb{R}^3)
\end{align}
the Pauli--Fierz Hamiltonian is self-adjoint on $\mathcal{D} \big( H_{N}^{\rm{PF}} \big) = \mathcal{D} \big( - \sum_{j=1}^N  \Delta_{x_j} + H_f \big)$ \cite{H2002,M2017}.

\subsection{Scaling and Assumptions on the Electron Density}

In the following, we consider the Pauli--Fierz Hamiltonian in the semiclassical mean-field regime by setting $e = N^{-1/2}$ and $\hbar = N^{- 1/3}$. 
If the number of photons is of order $N^{4/3}$ and most of the momenta of the electrons are of order $N^{1/3}$, this choice implies  that all terms in the Pauli--Fierz Hamiltonian are of the same order for large $N$, since the annihilation and creation operators scale with the square root of the number of photons.
For notational convenience, we additionally choose $c=1$ and $m= \frac{1}{2}$ and define $\varepsilon \coloneqq N^{- 1/3}$.
This leads to the microscopic model
\begin{align}
\label{eq:Schroedinger equation Pauli-Fierz}
i \varepsilon \partial_t \Psi_{N,t} = H_N^{\rm PF} \Psi_{N,t} ,
\end{align}
where
\begin{subequations}
\begin{align}
\label{eq:definition Pauli-Fierz Hamiltonian}
H_N^{\rm PF} &= \sum_{j=1}^N  \left( - i \varepsilon \nabla_j - \varepsilon^2 \kappa * \widehat{\vA}(x_j) \right)^2
+ \frac{1}{2 N} \sum_{\substack{i,j=1\\ j\neq i}}^N  K(x_j - x_k)
+ \varepsilon H_f  , 
\\
\label{eq:Coulomb potential with cutoff}
K(x) &= \frac{1}{4 \pi}  \kappa * \kappa * | \cdot |^{-1}(x) ,
\\ 
\hat{\vA}(x)
&=  (2 \pi)^{-3/2} \sum_{\lambda=1,2}  \int_{\mathbb{R}^3} 
\,
\frac{1}{\sqrt{2 \abs{k}}}
\vep_{\lambda}(k) 
\left(  e^{ikx} a_{k,\lambda}  + e^{- ik x} a^*_{k,\lambda} \right)\, dk ,
\\
H_f &= \sum_{\lambda =1,2} \int_{\mathbb{R}^3} \abs{k} a_{k,\lambda}^* a_{k,\lambda}\, dk .
\end{align} 
\end{subequations}

 Furthermore,  we assume that the electrons have a finite size and that their density satisfies the same conditions as in \cite{LS2023}.

\begin{assumption}
\label{assumption:cutoff function}
Let $\kappa: \mathbb{R}^3 \rightarrow \mathbb{R}$ be a real and even charge distribution such that 
\begin{align}
\label{eq:assumption on the cutoff function}
\kappa \in L^1(\mathbb{R}^3) 
\quad \text{and} \quad
\left( - \Delta \right)^{1/4} \kappa \in L^2(\mathbb{R}^3).
\end{align}
\end{assumption}
Note that \eqref{eq:assumption on the cutoff function} implies \eqref{eq:assumption cutoff for self-adjointess of PF Hamiltonian}, ensuring the self-adjointness of the Pauli--Fierz Hamiltonian. 
Assumption \ref{assumption:cutoff function} is slightly more restricted than the conditions necessary to prove the self-adjointness of the Pauli--Fierz Hamiltonian, as we require $\kappa$ to be even and summable, thereby excluding non-even electron densities and those that are not summable. 
Since $\kappa$ is real and even, we have
\begin{align}
\mathcal{F}[\kappa](k) &= \mathcal{F}[\kappa](- k)
\quad \text{and} \quad
\mathcal{F}[\kappa](k) \in \mathbb{R} 
\quad \forall k \in \mathbb{R}^3 .
\end{align}
It is worth noting that the evenness of the density is assumed primarily to simplify the computations and that Assumption \ref{assumption:cutoff function} covers Gaussian electron distributions of the form $\frac{1}{(2 \pi)^{3/2}} e^{- \frac{x^2}{2 \sigma^2}}$ with $\sigma >0$. Moreover, the estimates of our main result, Theorem~\ref{theorem:main theorem}, only require \eqref{eq:assumption cutoff for self-adjointess of PF Hamiltonian}. However, we assume \eqref{eq:assumption on the cutoff function} to rely on the well-posedness result for the Maxwell--Schr\"odinger equations from \cite{LS2023} (see Proposition~\ref{proposition:well-posedness of Maxwell-Schroedinger} below) and to prove Theorem~\ref{theorem:derivation of the Vlasov-Maxwell equations}, which shows the convergence of the Pauli-Fierz dynamics to the Vlasov--Maxwell system with extended charges.

\subsection{The (fermionic) Maxwell--Schr\"odinger equations}

We aim to study the evolution of initial states that are approximately of the product form
$\bigwedge_{j=1}^N \varphi_j \otimes W( \varepsilon^{-2} \alpha) \Omega$.
Here, $\alpha \in \mathfrak{h}$,
$\bigwedge_{j=1}^N \varphi_j$ denotes a Slater determinant of $N$ orthonormal one-particle wave functions $\varphi_1, \ldots, \varphi_N \in L^2(\mathbb{R}^3)$, $\Omega$ denotes the vacuum in $\mathcal{F}$ and $W$ is the unitary Weyl operator
\begin{align}
\label{eq:definition Weyl operator}
W(f) &\coloneqq \exp \Big( \sum_{\lambda = 1,2} \int_{\mathbb{R}^3} \Big( f(k,\lambda) a^*_{k,\lambda} - \overline{f(k,\lambda)} a_{k,\lambda} \Big) \, dk \Big)
\quad \text{with} \; f \in \mathfrak{h} ,
\end{align}
satisfying the shift property
\begin{align}
\label{eq:Weyl operators shifting property}
W^*(f) a_{k,\lambda} W(f) &= a_{k,\lambda} + f(k,\lambda) .
\end{align}

In such a state, the photons are in a coherent state, and the only correlations among the electrons arise from the antisymmetry of the electron wave function. During the time evolution, correlations emerge due to the interaction between the electrons and the electromagnetic field. Assuming the Slater determinant satisfies a certain semiclassical structure, our main result (Theorem~\ref{theorem:main theorem}) shows, however, that the emergence of correlations is weak enough that the time-evolved state, at the level of reduced density matrices and in the limit $N \rightarrow \infty$, remains approximately of product form, i.e.
\begin{align}
\label{eq:time evolved product state}
\Psi_{N,t} \approx \bigwedge_{j=1}^N \varphi_{j,t} \otimes W( \varepsilon^{-2} \alpha_t) \Omega.
\end{align} 
Slater determinants are completely characterized by their one-electron reduced density matrix $p_t = \sum_{j=1}^N \ket{\varphi_{j,t}} \bra{\varphi_{j,t}}$ and the time evolution of the quantities on the right hand side of \eqref{eq:time evolved product state} is determined by the regularized (fermionic) Maxwell--Schr\"odinger system in the Coulomb gauge
\begin{align}
\label{eq:Maxwell-Schroedinger equations}
\begin{cases}
i \varepsilon \partial_t p_t &= \left[ \left( - i \varepsilon \nabla - \kappa * \vA_{\alpha_t} \right)^2  + K * \rho_{p_t} - X_{p_t} ,  p_t \right]  ,
\\
i \partial_t \alpha_t(k,\lambda) &= \abs{k} \alpha_t(k,\lambda) - \sqrt{\frac{4 \pi^3}{\abs{k}}} \mathcal{F}[\kappa](k) \vep_{\lambda}(k) \cdot \mathcal{F}[\vJ_{p_t, \alpha_t}](k)
\end{cases}
\end{align}
with semiclassical parameter $\varepsilon = N^{-1/3}$
and initial condition $(p_t, \alpha_t) \big|_{t=0} = (\sum_{j=1}^N \ket{\varphi_{j}} \bra{\varphi_{j}}, \alpha)$. Here,  $K$ is defined as in \eqref{eq:Coulomb potential with cutoff} and
\begin{subequations}
\begin{align}
\label{eq:definition vector potential}
\vA_{\alpha}(x) &=  \frac{1}{(2 \pi )^{ 3/2}} \sum_{\lambda = 1,2} \int_{\mathbb{R}^3} \frac{1}{\sqrt{2 \abs{k}}} \vep_{\lambda}(k) \left( e^{i k x} \alpha(k,\lambda) + e^{- i k x} \overline{\alpha(k, \lambda)} \right) \, dk, 
\\
\label{eq:definition density}
\rho_{p}(x) &= N^{-1} p(x;x) ,
\\
\label{eq:definition charge current}
\vJ_{p, \alpha}(x) &= - N^{-1} \left\{ i \varepsilon \nabla, p \right\}(x;x) - 2 \rho_{p}(x) \kappa * \vA_{\alpha}(x)  ,
\\
\label{eq:definition exchange term}
X_{p}(x;y) &= N^{-1} K(x-y) p(x;y) 
\end{align}
\end{subequations}
for sufficiently regular $p \in \mathfrak{S}^1(L^2(\mathbb{R}^3))$ and $\alpha \in \mathfrak{h}$.
In \eqref{eq:definition charge current}, we use $\{A,B\}(x;y)$ to denote the kernel of the anticommutator of the operators $A$ and $B$, which is defined by $\{A,B\}:=AB+BA$. Note that \eqref{eq:Maxwell-Schroedinger equations} can be heuristically derived as effective evolution equations for many-body states of product type by observing that a quantum field acting on a coherent state behaves similarly to a classical field, and that the interaction between electrons within a Slater determinant can be approximated by its mean-field potential plus an exchange term.
In the mathematical literature,  for example in \cite{ADM2017, AMS2022,B2009, NM2005, NM2007},  the Maxwell--Schr\"odinger system is usually considered for a single point particle, i.e. $N=1$, $\varepsilon=1$, $\kappa = \delta(x)$, and $X_{p}=0$. The first equation  of \eqref{eq:Maxwell-Schroedinger equations} is typically written as the Schr\"odinger equation for the one-particle wave function, while the second equation is written in terms of the vector potential \eqref{eq:definition vector potential}, leading to Maxwell's equations in the Coulomb gauge:
\begin{equation}
\nabla \cdot \vA_{\alpha_t} = 0,\quad\quad
\left(\partial_t^2 - \Delta\right)\vA_{\alpha_t} =  - \left( 1 - \nabla \text{div} \Delta^{-1} \right) \kappa * \vJ_{p_t, \alpha_t} .
\end{equation}

\subsection{The non-relativistic Vlasov--Maxwell equations }

The Maxwell--Schr\"odinger equations \eqref{eq:Maxwell-Schroedinger equations} still depend on the number of particles. In particular, the dependence of the semiclassical parameter $\varepsilon$ on $N$ allows the state of the electrons to be approximately described by a function on the phase space $\mathbb{R}^3 \times \mathbb{R}^3$. Here, the connection between trace-class operators $\varpi \in \mathfrak{S}^1(L^2(\mathbb{R}^3))$ and phase space functions is established by the Wigner transform
\begin{align}
\mathcal{W}[\varpi](x,v) &= \left( \frac{\varepsilon}{2 \pi} \right)^3
\int_{\mathbb{R}^3} \varpi \left( x + \frac{\varepsilon y}{2} ; x - \frac{\varepsilon y}{2} \right) e^{- i v y} \, dy
\end{align}
and its inverse, the Weyl quantization
\begin{align}
\mathcal{W}^{-1}[W](x;y) = N \int_{\mathbb{R}^3} W \left( \frac{x+y}{2} , v \right) e^{i v \cdot \frac{x-y}{\varepsilon}} \, dv .
\end{align}
In \cite{LS2023} it was shown for sufficiently regular initial data that the Wigner transform of a solution of the Maxwell--Schr\"odinger equations \eqref{eq:Maxwell-Schroedinger equations} approximately satisfies in the limit $N \rightarrow \infty$ the following transport equation
\begin{align}
\label{eq: Vlasov-Maxwell different style of writing}
\begin{cases}
\partial_t f_t  &= - 2 \left( v - \kappa * \vA_{\alpha_t} \right) \cdot \nabla_x f_t + \vF_{f_t, \alpha_t} \cdot \nabla_v f_t , 
\\
i \partial_t \alpha_t(k,\lambda) &= \abs{k} \alpha_t(k,\lambda) - \sqrt{\frac{4 \pi^3}{\abs{k}}} \mathcal{F}[\kappa](k) \vep_{\lambda}(k) \cdot \mathcal{F}[\vJt_{f_t,\alpha_t}](k)
\end{cases}
\end{align}
with $A_{\alpha_t}$ being defined as in \eqref{eq:definition vector potential}, 
\begin{subequations}
\begin{align}
\label{eq:density of Vlasov equation}
\widetilde{\rho}_{f}(x) &= \int f(x,v) \, dv ,
\\
\label{eq:F-field of Vlasov equation}
\vF_{f, \alpha}(x,v) &= \nabla_x \Big[ K  * \widetilde{\rho}_{f}(x) + \left( \kappa * \vA_{\alpha}(x) \right)^2  - 2 v \cdot \kappa * \vA_{\alpha}(x) \Big] ,
\\
\label{eq:current density of Vlasov equation}
\vJt_{f,\alpha}(x) &= 2 \int \left( v - \kappa * \vA_{\alpha} \right) f(x,v) \, dv ,
\end{align}
\end{subequations}
and initial datum $( \mathcal{W}[p_0] , \alpha_0) \in L^1 \left( \mathbb{R}^3 \times \mathbb{R}^3 \right) \times \left( L^2(\mathbb{R}^3) \otimes \mathbb{C}^2 \right)$. 
As shown in \cite[Appendix A]{LS2023},  system \eqref{eq: Vlasov-Maxwell different style of writing} with $\kappa(x) = -\delta(x)$  is formally equivalent to the non-relativistic Vlasov--Maxwell system in the Coulomb gauge with $c=e=1$ and $m= \frac{1}{2}$ which was, for example, considered in \cite{A1986,D1986,DL1989,W1984,W1987}.
For this reason we refer to \eqref{eq: Vlasov-Maxwell different style of writing} as the non-relativistic Vlasov--Maxwell system for extended charges.

\subsection{Notation}

We use the letter $C$ to denote a general positive constant that may vary from one line to another. The dependence on $\kappa$ is tracked via $B_{\kappa} \coloneqq \big\| \big(  |\cdot |^{-1} + 1 \big) \mathcal{F}[\kappa] \big\|_{L^2(\mathbb{R}^3)}$ and $C_{\kappa} \coloneqq \big\| \big(  |\cdot |^{-1} + |\cdot |^{1/2} \big) \mathcal{F}[\kappa] \big\|_{L^2(\mathbb{R}^3)}$. Moreover, let $\left< x \right> = \big( 1 + \abs{x}^2 \big)^{1/2}$ with $x \in \mathbb{R}^d$ and $d \in \mathbb{N}$.
For $a \in \mathbb{R}$ and $b \geq 0$, we define the weighted $L^2$-spaces
$\dot{\mathfrak{h}}_{a} = L^2(\mathbb{R}^3, \abs{k}^{2a} dk) \otimes \mathbb{C}^2$ with norm
$\norm{\alpha}_{\dot{\mathfrak{h}}_{a}} = \Big( \sum_{\lambda=1,2} \int_{\mathbb{R}^3} 
\abs{k}^{2a} \abs{\alpha(k,\lambda)}^2dk \Big)^{1/2} $
and $\mathfrak{h}_{b} = L^2(\mathbb{R}^3, (1 + \abs{k}^2 )^b dk) \otimes \mathbb{C}^2$ with norm
$\norm{\alpha}_{\mathfrak{h}_b} = \Big( \sum_{\lambda=1,2} \int_{\mathbb{R}^3} 
\big( 1 + \abs{k}^2 \big)^b \abs{\alpha(k,\lambda)}^2 \,dk \Big)^{1/2}$. 
Note that $\mathfrak{h}_0 = \mathfrak{h}$.  For $\sigma \in \mathbb{N}_0$, $k \geq 0$ and $1 \leq p \leq \infty$ let $W_{k}^{\sigma,p} \left( \mathbb{R}^{d} \right)$ denote the Sobolev space equipped with the norm 
\begin{align}
\label{eq:definition LP-spaces}
\norm{f}_{W_{k}^{\sigma, p} ( \mathbb{R}^{d} )} 
&=
\begin{cases}
\left( \sum_{\abs{\alpha} \leq \sigma} \norm{\left< \cdot \right>^{k} D^{\alpha} f}_{L^p\left( \mathbb{R}^{d} \right)}^p \right)^{\frac{1}{p}}
\quad 1 \leq &p < \infty ,
\\
\max_{\abs{\alpha} \leq \sigma}  \norm{\left< \cdot \right>^{k} D^{\alpha} f}_{L^{\infty}\left( \mathbb{R}^{d} \right)}
\quad &p = \infty .
\end{cases}
\end{align}
In case $p=2$ we use the shorthand notation $H_k^{\sigma} \left( \mathbb{R}^{d} \right) \coloneqq W_{k}^{\sigma, 2} \left( \mathbb{R}^{d} \right)$ and if the spaces $W_{k}^{\sigma, p} ( \mathbb{R}^{d} )$ and $H_k^{\sigma} \left( \mathbb{R}^{d} \right)$ appear as subscripts we will abbreviate them by $W_{k}^{\sigma, p}$ and $H_k^{\sigma}$. Depending on the context $\| \cdot \|$ and $\scp{\cdot}{\cdot}$ will refer to the norms and scalar products of $\mathcal{H}^{(N)}$, $\mathfrak{h}$ and $L^2(\mathbb{R}^d)$.
Vectors in $\mathbb{R}^{6}$ are written as $z = (x,v)$ so that $D_z^{\alpha} = \left( \partial/ \partial x_1 \right)^{\beta_1} \left( \partial/ \partial v_1 \right)^{\gamma_1} \cdots \left( \partial/ \partial x_3 \right)^{\beta_3} \left( \partial/ \partial v_3 \right)^{\gamma_3}$ with
$\beta = (\beta_i)_{i \in \llbracket 1, 3 \rrbracket} \in \mathbb{N}_0^3$ and $\gamma = (\gamma_i)_{i \in \llbracket 1, 3 \rrbracket} \in \mathbb{N}_0^3$ 
such that $\abs{\alpha} = \sum_{i=1}^3 \beta_i + \gamma_i$. 
For two Banach spaces $A$ and $B$ we denote by $A \cap B$ the Banach space of vectors $f \in A \cap B$ with norm $\norm{f}_{A \cap B} = \norm{f}_A + \norm{f}_B$.
For a reflexive Banach space $(X, \norm{\cdot}_{X})$ and $T>0$ we denote by $L^{\infty}(0,T; X)$ the space of (equivalence classes of) strongly Lebesgue-measurable functions $\alpha: [0,T] \rightarrow X$ with the property that
$\norm{\alpha}_{L_T^{\infty} X} \coloneqq  \esssup_{t \in [0,T]} \norm{\alpha_t}_{X}$ is finite.
Let $\mathfrak{S}^{\infty} \left( L^2(\mathbb{R}^3) \right)$ denote the set of all bounded operators on $L^2 (\mathbb{R}^3)$, and let $\mathfrak{S}^1 \left( L^2(\mathbb{R}^3) \right)$ denote the set of trace-class operators on $L^2(\mathbb{R}^3)$. More generally, for $p\in[1,\infty)$, we denote by $\mathfrak{S}^p(L^2(\mathbb{R}^3))$ the $p$-Schatten space equipped with the norm $\norm{A}_{\mathfrak{S}^p}=(\tr \abs{A}^p)^\frac{1}{p}$, where $A$ is an operator, $A^*$ is its adjoint and $\abs{A}=\sqrt{A^*A}$.
For $a \geq 0$ let
\begin{align*}
\mathfrak{S}^{a,1}(L^2(\mathbb{R}^3))
&= \left\{  \varpi: \varpi \in \mathfrak{S}^{\infty}(L^2(\mathbb{R}^3)),  \varpi^* = \varpi
\; \text{and} \; \left( 1 - \Delta \right)^{a/2} \varpi \left( 1 - \Delta \right)^{a/2} \in \mathfrak{S}^{1}(L^2(\mathbb{R}^3))
\right\}
\end{align*}
with 
$
\norm{\varpi}_{\mathfrak{S}^{a,1}}
= \norm{\left( 1 - \Delta \right)^{a/2} \varpi \left( 1 - \Delta \right)^{a/2}}_{\mathfrak{S}^1}.
$
The positive cone of this space is defined as
$\mathfrak{S}_{+}^{a,1}(L^2(\mathbb{R}^3)) 
= \left\{ \varpi \in \mathfrak{S}^{a,1}(L^2(\mathbb{R}^3)) : \varpi \geq 0
\right\}$.
We, moreover,  use the notation
\begin{align}
\label{eq:definition fractional Laplacian with semiclassical weight}
D_{\varepsilon} = \left( 1 - \varepsilon^2 \Delta \right)^{1/2}
\end{align}
and
\begin{align}
\label{eq:definition Sobolev tracee norm with semiclassical weight}
\norm{\omega}_{\mathfrak{S}^{1,1}_{\varepsilon}}
&= 
\norm{D_{\varepsilon} \omega D_{\varepsilon}}_{\mathfrak{S}^1}
= \norm{\left( 1 - \varepsilon^2 \Delta \right)^{1/2} \omega \left( 1 - \varepsilon^2 \Delta \right)^{1/2}}_{\mathfrak{S}^1} .
\end{align}
to denote the fractional Laplacian with semiclassical parameter and the corresponding Sobolev trace norm, respectively.

\section{Main results}
\label{section:main results}

Throughout this article, we will rely on the following well-posedness result from \cite{LS2023}.

\begin{proposition}
\label{proposition:well-posedness of Maxwell-Schroedinger}
Let $\kappa$ satisfy Assumption~\ref{assumption:cutoff function}.
For all $(p_0, \alpha_0) \in \mathfrak{S}_{+}^{2,1} (L^2(\mathbb{R}^3)) \times \mathfrak{h}_{1/2} \cap \dot{\mathfrak{h}}_{-1/2} $ the Cauchy problem for the Maxwell--Schr\"odinger system \eqref{eq:Maxwell-Schroedinger equations} associated with $(p_0, \alpha_0)$   has a unique  $
 C \big( \mathbb{R}_{+} ; \mathfrak{S}_{+}^{2,1} ( L^2(\mathbb{R}^3) ) \big) \cap C^1 \big( \mathbb{R}_{+} ; \mathfrak{S}^{1} ( L^2(\mathbb{R}^3) )  \big)   \times  C \big( \mathbb{R}_{+} ; \mathfrak{h}_{1/2} \cap \dot{\mathfrak{h}}_{-1/2} \big) \cap C^1 \big( \mathbb{R}_{+} ; \dot{\mathfrak{h}}_{-1/2} \big) 
$ solution. 
The energy 
\begin{align}
\label{eq:Maxwell-Schroedinger with exchange term energy definition}
\mathcal{E}^{\rm{MS}}[ p_t , \alpha_t] &\coloneqq 
\tr \left( p_t \left( - i \varepsilon \nabla - \kappa * \vA_{\alpha_t} \right)^2 \right) + \frac{1}{2} \tr \left( \left( \, K * \rho_{p_t } - X_{p_t } \right) p_t  \right)
 + N \norm{\alpha_t}_{\dot{\mathfrak{h}}_{1/2}}^2
\end{align}
and mass of the system are conserved, i.e. 
\begin{align}
\label{eq:conservation of mass and energy of the MS equations}
\tr  \left( p_t \right) = \tr  \left( p_0 \right)
\quad \text{and} \quad 
\mathcal{E}^{\rm{MS}}[ p_t, \alpha_t] = \mathcal{E}^{\rm{MS}}[ p_0, \alpha_0]
\quad \text{for all} \; t \in  \mathbb{R}_{+} .
\end{align}
Under the additional assumption $\alpha_0 \in \mathfrak{h}_1$, we have $
\alpha \in   C \big( \mathbb{R}_{+} ; \mathfrak{h}_{1} \cap \dot{\mathfrak{h}}_{-1/2} \big) \cap C^1 \big( \mathbb{R}_{+} ; \mathfrak{h} \cap \dot{\mathfrak{h}}_{-1/2} \big) 
$  and
\begin{align}
\label{eq:bound for the h-1 norm of alpha}
\norm{\alpha_t}_{\mathfrak{h}_1}
&\leq \norm{\alpha_0}_{\mathfrak{h}_1}
+ C  C_{\kappa}^3 N^{-1}
\left< N^{-1}  \norm{p_0}_{\mathfrak{S}^1} \right> 
\left(  \mathcal{E}^{\rm{MS}}[p_0,\alpha_0] + C  C_{\kappa}^2 N^{-1} \norm{p_0}_{\mathfrak{S}^1}^2
\right) t .
\end{align}
\end{proposition}

\begin{proof}
The first part of the proposition, up to \eqref{eq:conservation of mass and energy of the MS equations}, is  \cite[Proposition II.1]{LS2023}. The proof of the second part is provided in Section~\ref{subsection:estimate for the h1-norm of alpha}.
\end{proof}
We define the number of photon operator as
\begin{align}
\label{eq:definition number operator}
\mathcal{N} &\coloneqq \sum_{\lambda =1,2} \int_{\mathbb{R}^3}  a_{k,\lambda}^* a_{k,\lambda}\, dk
\end{align}
and the one-electron reduced density matrix of $\Psi_N$ on $L^2(\mathbb{R}^3)$ as
\begin{align}
\gamma_{\Psi_N}^{(1,0)} &\coloneqq \tr_{2,\ldots, N} \otimes \tr_{\mathcal{F}} \left( \ket{\Psi_N} \bra{\Psi_N} \right) .
\end{align}
Moreover, for $(p,\alpha) \in \mathfrak{S}^{1,1}(L^2(\mathbb{R}^3)) \times \mathfrak{h}_{1/2}$, $\mathcal{E}^{\rm{MS}}[p,\alpha]$ as given in \eqref{eq:Maxwell-Schroedinger with exchange term energy definition}, and $\Psi_N \in \mathcal{D} \left( (H_N^{\rm PF})^{1/2} \right) \cap \mathcal{D} \left( \mathcal{N}^{1/2} \right)$ , we define  
\begin{subequations}
\begin{align}
a[\Psi_N,p] &\coloneqq  \norm{ \gamma^{(1,0)}_{\Psi_N} - N^{-1} p }_{\mathfrak{S}^1} ,
\\
b[\Psi_N,\alpha] &\coloneqq  \scp{\Psi_N}{W \big( \varepsilon^{-2} \alpha \big) \left( \varepsilon^4 \mathcal{N} + \varepsilon^2 \mathcal{N}^{1/2} \right) \, W^* \big( \varepsilon^{-2} \alpha \big) \Psi_N}_{\mathcal{H}^{(N)}},
\\
c[\Psi_N, p, \alpha] &\coloneqq
N^{-1} \Big|  \scp{\Psi_N}{H_N^{\rm{PF}} \Psi_N}_{\mathcal{H}^{(N)}}
- \mathcal{E}^{\rm{MS}}[p,\alpha] \Big| ,
\\
d[\Psi_N, p, \alpha] &\coloneqq
a[\Psi_N, p] + b[\Psi_N,\alpha]  + c[\Psi_N, p, \alpha] .
\end{align}
\end{subequations}

Our main result is the following
\begin{theorem}
\label{theorem:main theorem}
Let $\kappa$ satisfy Assumption~\ref{assumption:cutoff function}, $\widetilde{C} >0$, $N \in \mathbb{N}$, and $\varepsilon = N^{-1/3}$. Let $\alpha \in \mathfrak{h}_1 \cap \dot{\mathfrak{h}}_{-1/2}$ and $p \in \mathfrak{S}_{+}^{2,1} (L^2(\mathbb{R}^3))$ be a rank-$N$ projection such that $p$ and $q = 1 - p$ satisfy
\begin{align}
\label{eq:main result initial semiclassical structure}
\sup_{k \in \mathbb{R}^3}
\left\{ (1 + \abs{k})^{-1} \norm{q e^{ikx} p}_{\mathfrak{S}^1} 
\right\} 
+ \norm{q i \varepsilon \nabla p}_{\mathfrak{S}^1} \leq \widetilde{C} N^{2/3}
\end{align}
and
\begin{align}
\label{eq:main result condition on the energy and velocity}
\abs{N^{-1} \mathcal{E}^{\rm{MS}}[p,\alpha]} +  \norm{\alpha}_{\mathfrak{h}_1 \cap \dot{\mathfrak{h}}_{-1/2} } +  \norm{i \varepsilon \nabla p }_{\mathfrak{S}^{\infty}} \leq \widetilde{C} .
\end{align}
Let $\Psi_{N} \in \mathcal{H}^{(N)} \cap \mathcal{D} \left( (H_N^{\rm PF})^{1/2} \right) \cap \mathcal{D} \left( \mathcal{N}^{1/2} \right)$ such that $\norm{\Psi_N} = 1$. Let $(p_t, \alpha_t)$ be the unique solution of \eqref{eq:Maxwell-Schroedinger equations} with initial data $(p,\alpha)$ and let $\Psi_{N,t} = e^{- i \varepsilon^{-1} H_N^{\rm PF} t} \Psi_N$. Then, there exists a constant $C>0$ (independent of $N$ and $\kappa$) such that for any $t \geq 0$
\begin{subequations}
\begin{align}
\norm{ \sqrt{1 - \varepsilon^2 \Delta} \left( \gamma^{(1,0)}_{\Psi_{N,t}} - N^{-1} p_t \right) \sqrt{1 - \varepsilon^2 \Delta} }_{\mathfrak{S}^1} 
&\leq g(t) \sup_{j=1,2} 
\big( d[\Psi_N, p, \alpha] + \varepsilon \big)^{j/2} 
,
\\
\varepsilon^4  \big\langle\Psi_{N,t} ,W \big( \varepsilon^{-2} \alpha_t \big) \left( \mathcal{N} + H_{f} \right) \,  W^* \big( \varepsilon^{-2} \alpha_t \big) \Psi_{N,t} \big\rangle_{\mathcal{H}^{(N)}}
&\leq g(t) \big( d[\Psi_N, p, \alpha] + \varepsilon \big), 
\end{align}
where $g(t) = \exp \Big[ 
\exp \Big[  
 \exp \Big[ C \left< \big\| \big(  |\cdot |^{-1} + |\cdot |^{1/2} \big) \mathcal{F}[\kappa] \big\|_{L^2(\mathbb{R}^3)} \right>^4
 \left< t \right> \Big] \Big] \Big]$.
\end{subequations}

\end{theorem}

\begin{remark}
The operators $p$ and $q$ in Theorem~\ref{theorem:main theorem} and in the rest of the article depend on $N$, but we refrain from indicating this dependence to simplify the notation. If we consider sequences of initial data $	(\Psi_{N})_{N \in \mathbb{N}}$ and $(p_N,\alpha_N)_{N \in \mathbb{N}}$ such that the conditions of Theorem~\ref{theorem:main theorem} are satisfied for all $N \in \mathbb{N}$ and $\lim_{N \rightarrow \infty} d[\Psi_N,p_N,\alpha_N] = 0$, then for any time $t >0$,  the corresponding solutions $(p_{N,t},\alpha_{N,t})$ of \eqref{eq:Maxwell-Schroedinger equations} and $\Psi_{N,t} = e^{- i \varepsilon^{-1} H_N^{\rm PF} t} \Psi_N$ satisfy
\begin{align}
\lim_{N \rightarrow \infty} \norm{ \sqrt{1 - \varepsilon^2 \Delta} \left( \gamma^{(1,0)}_{\Psi_{N,t}} - N^{-1} p_{N,t} \right) \sqrt{1 - \varepsilon^2 \Delta} }_{\mathfrak{S}^1} &= 0 ,
\nonumber \\
\nonumber 
\lim_{N \rightarrow \infty} \varepsilon^4  \big\langle\Psi_{N,t} ,W \big( \varepsilon^{-2} \alpha_{N,t} \big) \left( \mathcal{N} + H_{f} \right) \,  W^* \big( \varepsilon^{-2} \alpha_{N,t} \big) \big) \Psi_{N,t} \big\rangle_{\mathcal{H}^{(N)}}
 &= 0 .
\end{align}
\end{remark}

\begin{remark}
If $\Psi_N= \bigwedge_{j=1}^N \varphi_j \otimes W( \varepsilon^{-2} \alpha) \Omega$ with $\{ \varphi_j \}_{j \in \{1,\ldots,N\}}$ being the orthonormal eigenfunctions of the rank-$N$ projection $p = \sum_{j=1}^N \ket{\varphi_j} 
\bra{\varphi_j}$, we have
$d[\Psi_N,p,\alpha] = \varepsilon^4 \| \abs{\cdot}^{- 1/2} \mathcal{F}[\kappa] \|_{L^2}^2$
 because, in this case, $\gamma^{(1,0)}_{\Psi_N} = N^{-1} p$, $b[\Psi_N,\alpha] = 0$, and 
$\scp{\Psi_N}{H_N^{\rm{PF}} \Psi_N}_{\mathcal{H}^{(N)}} = \mathcal{E}^{\rm{MS}}[p,\alpha] + \varepsilon \| \abs{\cdot}^{- 1/2} \mathcal{F}[\kappa] \|_{L^2}^2$.
\end{remark}

\begin{remark}
\label{remark:bound on the momentum}
Note that $\norm{i  \varepsilon \nabla p}_{\mathfrak{S}^{\infty}} = \sup \big\{ \norm{\varepsilon \nabla \psi}_{L^2(\mathbb{R}^3)} : \psi \in \text{Ran}(p),  \norm{\psi}_{L^2(\mathbb{R}^3)} =1  \big\}$. Hence,  the condition
$\norm{i  \varepsilon \nabla p}_{\mathfrak{S}^{\infty}} \leq \widetilde{C}$
from \eqref{eq:main result condition on the energy and velocity} may be interpreted as a uniform bound on the kinetic energy of normalized vectors in $\text{Ran}(p)$.   The following examples illustrate situations in which this condition is satisfied.
\begin{itemize}
\item[a)] Let $p = \sum_{j=1}^N \ket{\varphi_j} \bra{\varphi_j}$ 
and assume that the Fourier transform of each $\varphi_j$ is supported in the ball $\{ k \in \mathbb{R}^3 : \abs{k} \leq \widetilde{C} N^{1/3} \}$.Then,  every normalized $\psi \in \text{Ran}(p)$ satisfies $\id_{\abs{k} \geq \widetilde{C} N^{1/3} } \mathcal{F}[\psi](k) =0$ and $\norm{\varepsilon \nabla \psi}_{L^2(\mathbb{R}^3)} \leq \widetilde{C}$.
\item[b)] Let $h_{\varepsilon}$ be a self-adjoint operator
on $L^2(\mathbb{R}^3)$ satisfying $- \varepsilon^2 \Delta \leq h_{\varepsilon}$ and let $p = \id_{(- \infty,  \widetilde{C}^2 ]}(h_{\varepsilon})$.  Then,  for every normalized $\psi \in \text{Ran}(p)$,
$\norm{\varepsilon \nabla \psi}_{L^2(\mathbb{R}^3)} \leq \scp{\psi}{h_{\varepsilon} \psi}^{1/2} \leq \widetilde{C}$.
\end{itemize}
\end{remark}

\begin{remark}
As explained in \cite[Remark IV.2]{LP2019}, Condition \eqref{eq:main result initial semiclassical structure} is equivalent to the formulation given in \cite{BPS2014}:
\begin{align}
\sup_{k \in \mathbb{R}^3}
\left\{ (1 + \abs{k})^{-1} \norm{\left[ e^{ikx} , p \right]}_{\mathfrak{S}^1} 
\right\} 
+ \norm{\left[ i \varepsilon \nabla , p \right]}_{\mathfrak{S}^1} \leq \widetilde{C} N^{2/3} .
\end{align}
Condition~\eqref{eq:main result initial semiclassical structure} assumes that the density of the system varies on a scale of order one and that the kernel of $p(x;y)$ is concentrated close to the diagonal $x=y$, decaying sufficiently fast for distances $\abs{x-y} \gg \varepsilon$.  In \cite[Section 1]{BPS2014}, it is argued that \eqref{eq:main result initial semiclassical structure}
is a reasonable assumption for minimizers of non-interacting fermionic systems.  It is shown to hold for the following examples.
\begin{itemize}
\item[a)] In \cite[Section 1]{BPS2014} Condition~\eqref{eq:main result initial semiclassical structure} is shown to hold for the ground state of non-interacting fermions on the torus $\mathbb{T} = [- \pi , \pi]^3$, i.e., $p(x;y) = (2 \pi)^{-3} \sum_{\abs{k} \leq c N^{1/3}} e^{i k (x-y)}$ with $c >0$.  In this case, the condition $\norm{i  \varepsilon \nabla p}_{\mathfrak{S}^{\infty}} \leq \widetilde{C}$ in \eqref{eq:main result condition on the energy and velocity} follows from part a) of Remark~ \ref{remark:bound on the momentum}.  Note that the present article considers the system in the full space and not on the torus.
\item[b)] For the ground state of $N$ non-interacting fermions in a harmonic trap, the validity of \eqref{eq:main result initial semiclassical structure} follows from \cite[Theorem 3.2]{B2022} and 
\begin{align}
(1 + \abs{k})^{-1} \norm{\left[ e^{ikx} , p \right]}_{\mathfrak{S}^1} \leq \norm{\left[ x , p \right]}_{\mathfrak{S}^1} ,
\end{align}
where the latter estimate is shown by a straightforward adaptation of the proof of \cite[Corollary 1.3]{FM2020}.
Condition $\norm{i  \varepsilon \nabla p}_{\mathfrak{S}^{\infty}} \leq \widetilde{C}$ is satisfied by part b) of Remark~ \ref{remark:bound on the momentum}.
\item[c)] In  \cite{FM2020} the validity of \eqref{eq:main result initial semiclassical structure}  is proven when $p$ is the spectral projection (which need not be of rank $N$) of a Schr\"odinger operator $- \varepsilon^2 + V$ with a more general confining potential.
\end{itemize}
An example of initial data satisfying all assumptions of Theorem~\ref{theorem:main theorem} is obtained by choosing $p$ to be the ground state of $N$ non-interacting fermions in a harmonic trap, setting $\alpha = 0$, and taking $\Psi_N= \bigwedge_{j=1}^N \varphi_j \otimes  \Omega$ where $\{ \varphi_j \}_{j \in \{1,\ldots,N\}}$ are orthonormal eigenfunctions of $p$.  Note that the condition $\abs{N^{-1} \mathcal{E}^{\rm{MS}}[p,0]}  \leq \widetilde{C}$ is satisfied, since  \eqref{eq:estimate for mean-field potential},  \eqref{eq:estimate for exchange term},  the projection property of $p$,  and part b) yield
\begin{align*}
\abs{N^{-1} \mathcal{E}^{\rm{MS}}[p,0]}
&\leq N^{-1} \tr \left( - \varepsilon^2 \Delta p \right) + C B^2_{\kappa} N^{-2} \norm{p}_{\mathfrak{S}^1}^2
\\
&\leq C B^2_{\kappa} \left( 1 + \norm{i  \varepsilon \nabla p}^2_{\mathfrak{S}^{\infty}} \right)
\left( 1 + N^{-2} \norm{p}_{\mathfrak{S}^1}^2 \right)  
\\
&\leq C B^2_{\kappa} .
\end{align*}
\end{remark}

\subsection{Convegence to the Vlasov--Maxwell system with extended charges}

Combining Theorem~\ref{theorem:main theorem} with \cite[Theorem II.1]{LS2023} let us derive the Vlasov--Maxwell equations as effective evolution equations for the Pauli--Fierz dynamics. Throughout this work we will rely on the following well-posedness result from \cite{LS2023}.

\begin{proposition}
\label{proposition:global solutions VM}
Let $R>0$ and $a, b \in \mathbb{N}$ satisfying $a \geq 5$ and $b \geq 3$. For all $(f_0, \alpha_0) \in H_a^{b}(\mathbb{R}^6) \times \mathfrak{h}_b \cap \dot{\mathfrak{h}}_{-1/2}$ such that $\supp f_0 \subseteq A_R$ with $A_R = \{ (x,v) \in \mathbb{R}^6 , \abs{v} \leq R \}$ system
\eqref{eq: Vlasov-Maxwell different style of writing} has a unique
$L^{\infty} \big( \mathbb{R}_{+}; H_{a}^{b}(\mathbb{R}^6) \big) \cap C \big( \mathbb{R}_{+}; H_{a}^{b-1}(\mathbb{R}^6) \big)  \cap C^1 (\mathbb{R}_{+};  H_{a}^{b-2}(\mathbb{R}^6)) \times  C \big( \mathbb{R}_{+} ; \mathfrak{h}_{b} \cap \dot{\mathfrak{h}}_{-1/2} \big) \cap C^1 \big( \mathbb{R}_{+} ; \mathfrak{h}_{b-1} \cap \dot{\mathfrak{h}}_{-1/2} \big)$ solution.
The $L^p$--norms of the particle distribution (with $p \geq 1$) and the energy
\begin{align}
\label{eq:Vlasov-Maxwell energy definition}
\mathcal{E}^{\rm{VM}}[f, \alpha] &=  
\int_{\mathbb{R}^6} dx \, dv \, f(x,v) \left( v - \kappa * \vA_{\alpha}(x) \right)^2 
+  \frac{1}{2} \int_{\mathbb{R}^6} dx \, dv \,
f(x,v) K * \widetilde{\rho}_{f}(x) 
+  \norm{\alpha}_{\dot{\mathfrak{h}}_{1/2}}^2 
\end{align}
are conserved, i.e. $\norm{f_t}_{L^p(\mathbb{R}^6)} = \norm{f_0}_{L^p(\mathbb{R}^6)}$ and $\mathcal{E}^{\rm{VM}}[f_t, \alpha_t] =  \mathcal{E}^{\rm{VM}}[f_0, \alpha_0]$ for all $t \in \mathbb{R}_{+}$.
\end{proposition}

Since \cite[Theorem II.1]{LS2023} requires higher regularity properties on the initial data of the Vlasov--Maxwell equations as expected to hold for pure states of fermions, we introduce a regularization of the initial data of the Vlasov--Maxwell equations. More explicitly, we define for  $\sigma \in (0, 1]$
\begin{align}
\label{eq:definition regularizing Gaussian}
\mathscr{G}_{\sigma}(x,v) &\coloneqq (\sigma^2 \pi)^{-3} \exp \left[ -    \frac{x^2 + v^2}{\sigma^2} \right] 
\end{align}
and
\begin{align}
\label{eq:definition general Husimi measure}
m_{p,\sigma} &\coloneqq  \mathscr{G}_{\sigma} * \mathcal{W}[p]  .
\end{align}
Note that for the specific choice $\sigma^2 = \varepsilon$, $m_{p,\sigma}$  is usually referred to as the Husimi measure; see e.g., \cite{FLS2018}.
The fact that \cite[Theorem II.1]{LS2023} also assumes a finite support in the velocity variable requires to cut off the high velocities values by means of a smooth cutoff function. For $L \in (0, \infty)$ let $\eta_L: \mathbb{R}^3 \rightarrow \mathbb{R}$ be given by
\begin{align}
\label{eq:definition smooth cutoff function}
\eta_L  \coloneqq \id_{\abs{\cdot} \leq L +2} * \zeta  ,
\quad \text{where} \quad 
\zeta(v) &\coloneqq 
\begin{cases} I^{-1}
\exp \left[ - \frac{1}{1 - \abs{v}^2 } \right]
 \quad \text{if} \; \abs{v} <1 , \\
0 \quad \text{if} \; \abs{v} \geq 1 
\end{cases}
\end{align}
and $I =  \int_{\mathbb{R}^3}  \exp \left[ - \frac{1}{1 - \abs{u}^2 } \right] \id_{\abs{u}\leq 1} \, du$.
The function $\eta$ satisfies the properties 
\begin{subequations}
\begin{align}
\label{eq:properties of the smooth cutoff function 1}
0 \leq \eta_L &\leq 1 
\quad \text{and for} \; j \in \mathbb{N} : \;
\norm{\eta_L}_{W_0^{j,\infty}(\mathbb{R}^3)} \leq C  ,
\\
\label{eq:properties of the smooth cutoff function 2}
\eta_L(v) &= 1 \quad \text{if} \, \abs{v} \leq L+1, 
\\
\label{eq:properties of the smooth cutoff function 3}
\eta_L(v) &= 0 \quad \text{if} \, \abs{v} \geq L + 3. 
\end{align}
\end{subequations}
The regularized Wigner transform of $p$ is defined as
\begin{align}
\label{eq:definition regularized Wigner transform}
m_{p,\sigma, L}(x,v)  &\coloneqq  
\eta_L(v) \, m_{p,\sigma}(x,v)
= \eta_L(v) \, \left( \mathscr{G}_{\sigma} * \mathcal{W}[p] \right)(x,v)   .
\end{align} 
Note that if $\mathcal{W}[p]$ has  finite support in the velocity variable and $L$ is chosen sufficiently large, then $m_{p,\sigma, L}(x,v)$ is small for large values of $v$. In addition, we modify the initial data of the classical electromagnetic field by cutting off the high frequencies via the cutoff parameter $\Lambda \in [1, \infty )$. Our result is the following.

\begin{theorem}
\label{theorem:derivation of the Vlasov-Maxwell equations}
Let $\kappa$ satisfy Assumption~\ref{assumption:cutoff function}, $\widetilde{C} >0$, $R >0$, $N \in \mathbb{N}$, $\varepsilon = N^{- \frac{1}{3}}$, $\Lambda = \left( \ln(\ln(N)) \right)^{\frac{1}{2}}$, and $\sigma = \left( \ln(\ln(N)) \right)^{-1}$. Let $\alpha \in \mathfrak{h}_1 \cap \dot{\mathfrak{h}}_{-1/2}$
and $p \in \mathfrak{S}_{+}^{2,1} (L^2(\mathbb{R}^3))$ be a rank-$N$ projection such that the following conditions are satisfied. Assume that \eqref{eq:main result condition on the energy and velocity} holds, and that $p$ and $q = 1 - p$ satisfy
\begin{align}
\begin{aligned}
\label{eq:Vlasov-Maxwell result more restricted initial semiclassical structure}
\sup_{k \in \mathbb{R}^3}
\left\{  (1 + \abs{k})^{-1}  \norm{\sqrt{1 - \varepsilon^2 \Delta} q e^{ikx} p \sqrt{1 - \varepsilon^2 \Delta}}_{\mathfrak{S}^1} 
\right\} 
&\leq \widetilde{C} N^{2/3} ,
\\
\norm{\sqrt{1 - \varepsilon^2 \Delta} q i \varepsilon \nabla p \sqrt{1 - \varepsilon^2 \Delta}}_{\mathfrak{S}^1}
&\leq \widetilde{C} N^{2/3} .
\end{aligned}
\end{align}
Additionally, assume $p$ is such that its Wigner transform $\mathcal{W}[p]$ satisfies $ \norm{\mathcal{W}[p]}_{W_7^{0,2}(\mathbb{R}^6)}  \leq \widetilde{C}$ and $\norm{\chi_{\geq R}  \mathcal{W}[p]}_{W_7^{0,2}(\mathbb{R}^6)} \leq \widetilde{C} e^{- \frac{1}{2 \sigma^2}}$,  where
$\chi_{\geq R}: \mathbb{R}^6 \rightarrow \mathbb{R}$ is defined by $\chi_{\geq R}(x,v) = \id_{\abs{v} \geq R}$.
Let $\Psi_{N} \in \mathcal{H}^{(N)} \cap \mathcal{D} \left( (H_N^{\rm PF})^{1/2} \right) \cap \mathcal{D} \left( \mathcal{N}^{1/2} \right)$ such that $\norm{\Psi_N} = 1$ and let $\Psi_{N,t} = e^{- i \varepsilon^{-1} H_N^{\rm PF} t} \Psi_N$. Let $m_{p,\sigma,R}$ be defined as in \eqref{eq:definition regularized Wigner transform} and $(W_{N,t}, \alpha_t)$ be the unique solution to \eqref{eq: Vlasov-Maxwell different style of writing} with regularized initial datum $(m_{p,\sigma,R}, \id_{\abs{\cdot} \leq \Lambda} \alpha)$. 
Then, there exists a constant $C>0$ (independent of $N$ but dependent on $\kappa$) such that for any $t \geq 0$ and $N \geq  \exp \left[ \exp \left[ \exp \left[ \left< R \right> \exp \left[ \exp \left[ C \left< t \right> \right] \right] \right]  \right] \right]$
\begin{align}
&\norm{ \sqrt{1 - \varepsilon^2 \Delta} \left(
\gamma^{(1,0)}_{\Psi_{N,t}} - N^{-1}  \mathcal{W}^{-1}[W_{N,t}] \right) \sqrt{1 - \varepsilon^2 \Delta}}_{\mathfrak{S}^1}
\nonumber \\
&\quad 
+  \varepsilon^4  \big\langle\Psi_{N,t} ,W \big( \varepsilon^{-2} \alpha_t \big) \left( \mathcal{N} + H_f \right) \,  W^* \big( \varepsilon^{-2} \alpha_t \big) \Psi_{N,t} \big\rangle_{\mathcal{H}^{(N)}}
\nonumber \\
&\quad \qquad \qquad \leq 
\Big( \sup_{j=1,2} \big\{  d[\Psi_N, p, \alpha]^{j/2} \big\}
+ \left( \ln(\ln(N)) \right)^{- \frac{1}{4}}  \Big)
\exp \left[ \exp \left[  \exp \big[ C \left< t \right> \big] \right]  \right]
\nonumber \\
&\qquad \qquad \qquad + N^{- 1/6} 
 \left[  \ln (N) \right]^{ \exp \left[ \left< R \right> \exp \left[ \exp \left[ C \left< t \right> \right] \right] \right]} .
\end{align}
\end{theorem}

\begin{remark}
Note that the convergence of Maxwell--Schr\"odinger solutions to the Vlasov--Maxwell solutions with a convergence rate $\varepsilon$ is established in \cite{LS2023}. However, the result requires regularity assumptions on the initial data which are expected to hold for mixed states but not when $p$ is a rank-$N$ projection. The poor rate of convergence in Theorem~\ref{theorem:derivation of the Vlasov-Maxwell equations} results from the regularization of the initial data for the Vlasov--Maxwell equations, which is necessary to relax the assumptions on $\mathcal{W}[p]$ (see Subsection~\ref{subsection:proof of VM result}).
\end{remark}

\begin{remark}
The one-particle reduced density matrices for the ground states of non-interacting fermionic systems are expected to be of the form (or linear combinations thereof)
$p(x;y) = \varepsilon^{-3} (2 \pi)^{3/2} \varphi \left( \frac{x - y}{\varepsilon} \right) \chi \left( \frac{x+y}{2} \right)$, where $\chi$ determines the spatial density and  $\varphi$ describes the momentum distribution (see \cite{BPS2014}). The Wigner transform for states of this form is given by $\mathcal{W}[p](x,v)= \chi(x)  \mathcal{F}[\varphi](v)$. In this case, the assumptions in Theorem~\ref{theorem:derivation of the Vlasov-Maxwell equations} concerning the Wigner transform are satisfied,  provided that $\mathcal{F}[\varphi]$ has compact support and $\chi$ decays sufficiently fast.
\end{remark}

\begin{remark}
\label{remark:remark on the semiclassical structure with Sobolev trace norm}
The conditions in \eqref{eq:Vlasov-Maxwell result more restricted initial semiclassical structure} are expected to hold for electrons with a semiclassical structure and momenta of at most order $N^{1/3}$.
In the case where $p$ is of the form $p = \sum_{j=1}^N \ket{\varphi_j} \bra{\varphi_j}$ with $\varphi_j \in L^2(\mathbb{R}^3)$ such that $\id_{\abs{k} \geq \widetilde{C} N^{1/3} } \mathcal{F}[\varphi_j](k) =0$ for all $j \in \{1,2, \ldots,N \}$, \eqref{eq:main result initial semiclassical structure} implies \eqref{eq:Vlasov-Maxwell result more restricted initial semiclassical structure} with $\widetilde{C}$ replaced by $ 4 \big( 1 + \widetilde{C} \big)^3$. This is proven in Section~\ref{subsection:Proof of Remark:remark on the semiclassical structure with Sobolev trace norm}.
\end{remark}

\subsection{Comparison with the Literature}

\label{subsection:comparison with the literature}

The research domain of this work is the rigorous derivation of effective evolution equations, a field that began in the 1970s, amongst others, with seminal works \cite{BH1977,GV1979a,GV1979b, H1974,L1975,S1981} by Braun, Ginibre, Hepp, Landford, Spohn, and Velo, and has since evolved into an active area of study.
To the best of our knowledge, the present work provides the first derivations of the Maxwell--Schr\"odinger and Vlasov--Maxwell equations from the Pauli--Fierz dynamics in the fermionic semiclassical mean-field regime. The most closely related works are \cite{BPS2014,LP2019,LP2020,PP2016}. In \cite{LP2020} (see also \cite{FL2023} for a slight improvement of the result) the Maxwell--Schr\"odinger equations were derived in a mean-field limit for a large number of  bosonic particles  initially assumed to be in a Bose--Einstein condensate. Here, as the particle number $N$ increases, the coupling constant scales as $e = N^{- 1/2}$, while all other parameters remain fixed. This scenario is simpler as the semiclassical parameter remains constant, leading to a version of the Maxwell--Schr\"odinger equations that does not depend on $N$. In this case, $p$ is replaced by a function in $L^2(\mathbb{R}^3)$, which describes the evolution of the condensate. 
The work \cite{LP2019} derives a statement similar to that of the present article for the Nelson model with ultraviolet cutoff, leading to a fermionic variant of the Schr\"odinger--Klein--Gordon equations. The linear field coupling in the Nelson model simplifies the analysis and removes the need to control the kinetic energy of the particles outside the Slater determinant state. The scaling used in \cite{LP2019} and in the present work can be seen as the second-quantized analogue of the fermionic mean-field model considered in \cite{BPS2014}. The latter work was the first to use the propagation of semiclassical structure as a crucial element, which is also essential for deriving our results. In \cite{PP2016}, the Hartree-Fock equations were derived in the same setting as in \cite{BPS2014} and, additionally, in another scaling regime. The analysis in \cite{PP2016} introduced a first-quantization method, upon which the present work builds. It is therefore closely related from a methodological point of view.
The works \cite{AFH2022,K2009} are also highly relevant, as they demonstrate the emergence of classical electromagnetism by deriving the Abraham model from the Pauli--Fierz Hamiltonian in a limit where both the particles and the field approximately behave classically. Additionally, note the article \cite{S2010}, which derives the Maxwell--Schr\"odinger equations in a non-rigorous fashion by neglecting certain terms in the Pauli--Fierz Hamiltonian. 
Within the broader context of deriving effective equations from non-relativistic particle-quantum field models we would like to mention the following works on bosonic mean-field limits \cite{AF2014,AF2017,F2013,FL2023,FLLM2023,FLMP2023,LP2018,LP2020}, the strong coupling limit of the polaron \cite{FG2017,FS2014,G2017,LMRSS2021,LRSS2021,M2021}, partially classical limits \cite{BCFF2024,CF2018,CFO2023,GNV2006}, and other scaling regimes \cite{CM2024,D1979,T2002}.
Moreover, we note that the derivation of effective equations for fermionic systems with two-body interactions in the semiclassical mean-field regime began with the works \cite{NS1981,S1981}, which proved the convergence of the dynamics to the Vlasov equation. Convergence to the Hartree(-Fock) equations was initially proven for short  times and analytic potentials in \cite{EESY2004} and later extended \cite{BPS2014} to hold for all (semiclassical) times and for potentials that are at least twice differentiable. This result was further extended to relativistic dispersion relations \cite{BPS2014b}, mixed states with zero \cite{BJPSS2016} and nonvanishing pairing matrix \cite{MPS2024}, extended Fermi gases \cite{FPS2023,FPS2024},  and more singular interactions \cite{CLS2021,PRSS2017}. 
A norm approximation for the time-evolved state  of a homogeneous Fermi gas was provided in \cite{BNPSS2022}.  Derivations of the Hartree-Fock equations in other scaling regimes were achieved in \cite{BBPPT2016,BGGM2003,
BGGM2004,FK2011,PP2016,P2017}, and more recent derivations of the Vlasov equation from many-fermion quantum dynamics were obtained in \cite{CLL2021,CLL2022}. 
For results related to the derivation of the Vlasov--Maxwell equations with extended charges from the Maxwell--Schr\"odinger equations, we refer to \cite[Section II.1]{LS2023}.

\section{Proof of the main results}

\label{section:proof of the main results}

\subsection{Proof of Theorem~\ref{theorem:main theorem}}

In this subsection, we introduce a functional and state five preliminary lemmas from which the proof of Theorem~\ref{theorem:main theorem} follows easily. The proofs of these lemmas are provided in later sections.
\begin{definition}
Let $N \in \mathbb{N}$, $\Psi_{N} \in  \mathcal{H}^{(N)}$, $\alpha \in \mathfrak{h}$, and let $p \in \mathfrak{S}^1(L^2(\mathbb{R}^3))$ be a rank-$N$ projection. 
\begin{itemize}
\item 
For $m \in \{1, \ldots, N\}$,  let $p_m: L^2(\mathbb{R}^{3N}) \rightarrow L^2(\mathbb{R}^{3N})$ and $q_m: L^2(\mathbb{R}^{3N}) \rightarrow L^2(\mathbb{R}^{3N})$ be given by
\begin{align}
(p_m f)(x_1, \ldots, x_N) &\coloneqq \int_{\mathbb{R}^3} p(x_m,y)  f(x_1, \ldots,x_{m-1},y, x_{m+1}, \ldots, x_N) \, dy
\end{align}
with $f \in L^2(\mathbb{R}^{3N})$, and $q_m \coloneqq \id_{L^2(\mathbb{R}^{3N})} - p_m$. We define $\beta^{a}: \mathcal{H}^{(N)} \times  \mathfrak{S}^1(L^2(\mathbb{R}^3)) \rightarrow [0, \infty)$ by 
\begin{align}
\beta^{a} [\Psi_{N}, p ]  &\coloneqq
\scp{\Psi_{N}}{q_1 \otimes \id_{\mathcal{F}} \Psi_{N}}  .
\end{align}
In the following, we omit the tensor product with the identity and use the shorthand notations $p_m$ and $q_m$ to represent $p_m \otimes \id_{\mathcal{F}}$ and $q_m \otimes \id_{\mathcal{F}}$, respectively. If the  rank-$N$ projection depends on time, we use the notations $p_{t,m}$ and $q_{t,m}$ to indicate this time dependence.

\item
Let $W$ and $\mathcal{N}$ be given as in \eqref{eq:definition Weyl operator} and \eqref{eq:definition number operator}. We define 
$\beta^{b,1}: \mathcal{H}^{(N)} \cap \mathcal{D} \big( \mathcal{N}^{1/2} \big) \times \mathfrak{h} \rightarrow [0,\infty)$ 
and
$\beta^{b,2}: \mathcal{H}^{(N)} \cap \mathcal{D} \big( \mathcal{N}^{1/4} \big) \times \mathfrak{h} \rightarrow [0,\infty)$ 
by
\begin{subequations}
\begin{align}
\beta^{b,1} [ \Psi_{N}, \alpha ]  &\coloneqq \varepsilon^4 \scp{\Psi_N}{W \big( \varepsilon^{-2} \alpha \big) \mathcal{N} \, W^* \big( \varepsilon^{-2} \alpha \big) \Psi_N} , 
\\
\beta^{b,2} [ \Psi_{N}, \alpha ] 
&\coloneqq \varepsilon^2 \scp{\Psi_N}{W \big( \varepsilon^{-2} \alpha \big) \mathcal{N}^{1/2} \, W^* \big( \varepsilon^{-2} \alpha \big) \Psi_N} .
\end{align} 
\end{subequations}

\item Moreover, let
$\beta: \mathcal{H}^{(N)} \cap \mathcal{D} \big( \mathcal{N}^{1/2} \big) \times \mathfrak{S}^1(L^2(\mathbb{R}^3)) \times \mathfrak{h} \rightarrow [0, \infty)$ be given by
\begin{align}
\beta  [ \Psi_{N}, p , \alpha ]
&\coloneqq \beta^{a} [\Psi_{N}, p]
+  
\beta^{b,1} [\Psi_{N}, \alpha ]
+
\beta^{b,2} [ \Psi_{N}, \alpha ].
\end{align}

\end{itemize}
\end{definition}

The functional $\beta^a[\Psi_N,p]$ counts the relative number of fermions whose states belong to the kernel of the projection $p$. The functional $\beta^{b,1}[\Psi_N,\alpha]$ can be written as 
\begin{align}
\beta^{b,1} [\Psi_N, \alpha] = \sum_{\lambda =1,2} \int_{\mathbb{R}^3} \norm{\left( \varepsilon^2 a_{k,\lambda} - \alpha(k,\lambda) \right) \Psi_N}^2 \, dk .
\end{align}
It measures the relative number of excitations with respect to the coherent state $W \big( \varepsilon^{-2} \alpha \big) \Omega$. 
The functional $\beta^{b,2}$ is introduced primarily for technical reason, as it will be used in the energy estimates of Lemma~\ref{lemma:energy estimates} below.
For solutions of the Schr\"odinger equation \eqref{eq:Schroedinger equation Pauli-Fierz} and the Maxwell--Schr\"odinger equations \eqref{eq:Maxwell-Schroedinger equations}, we use the shorthand notations
\begin{align*}
\beta(t) = \beta  [ \Psi_{N,t}, p_t , \alpha_t ] ,
\;  \;
\beta^a(t) = \beta^a  [ \Psi_{N,t}, p_t ] ,
\;  \;
\beta^{b,1}(t) = \beta^{b,1}  [ \Psi_{N,t}, \alpha_t ] ,
\;  \;
\beta^{b,2}(t) = \beta^{b,2}  [ \Psi_{N,t}, \alpha_t ] .
\end{align*}
Note that for the considered initial states, $\Psi_{N} \in \mathcal{D} \left( (H_N^{\rm PF})^{1/2} \right) \cap \mathcal{D} \left( \mathcal{N}^{1/2} \right)$ and $(p, \alpha) \in \mathfrak{S}_{+}^{2,1} (L^2(\mathbb{R}^3)) \times \mathfrak{h}_{1} \cap \dot{\mathfrak{h}}_{-1/2} $ the functional is well-defined for all $t \geq 0$ due to Proposition~\ref{proposition:well-posedness of Maxwell-Schroedinger}
and the following lemma about the invariance of the many-body domain during the Pauli--Fierz time evolution.
\begin{lemma}
Let $N \in \mathbb{N}$ and let $\kappa$ satisfy $\big( \abs{\cdot}^{-1} + \abs{\cdot}^{1/2} \big) \mathcal{F}[\kappa] \in L^2(\mathbb{R}^3)$. Let $H_N^{\rm PF}$ be defined as in \eqref{eq:definition Pauli-Fierz Hamiltonian}. Then, 
$e^{- i \varepsilon^{-1} H_N^{\rm PF} t} \mathcal{D} \left( (H_N^{\rm PF})^{1/2} \right) \cap \mathcal{D} \left( \mathcal{N}^{1/2} \right) = \mathcal{D} \left( (H_N^{\rm PF})^{1/2} \right) \cap \mathcal{D} \left( \mathcal{N}^{1/2} \right)$. 
\end{lemma}

\begin{proof}
The lemma is proven in complete analogy to \cite[Lemma III.2]{FL2023}, which is obtained by slightly adapting the proof of \cite[Proposition 3.11]{AFH2022}.
\end{proof}

Within the proof of the main result, the functional $\beta(t)$ will be used to measure the distance between the time evolved many-body state and the product state \eqref{eq:time evolved product state}. This will be achieved through a Gr\"onwall estimate, which demonstrates that the value of the functional does not grow too rapidly during the time evolution. One of the key components of the Gr\"onwall estimate is the following lemma, which bounds the growth of the functional over time by its initial value and the quantities on the right-hand side. The proof is given in Section~\ref{subsection:Estimating the growth of correlations}.
\begin{lemma}
\label{lemma:estimating the time derivative of the functional}
Let $\kappa$ satisfy Assumption~\ref{assumption:cutoff function}. Let $\Psi_{N} \in \mathcal{D} \left( (H_N^{\rm PF})^{1/2} \right) \cap \mathcal{D} \left( \mathcal{N}^{1/2} \right)$ and let $(p,\alpha) \in \mathfrak{S}_{+}^{2,1} (L^2(\mathbb{R}^3)) \times \mathfrak{h}_{1} \cap \dot{\mathfrak{h}}_{-1/2}$ with $p$ being a rank-$N$ projection. Let $(p_t, \alpha_t)$ be the unique solution of \eqref{eq:Maxwell-Schroedinger equations} with initial data $(p,\alpha)$, and let $\Psi_{N,t} = e^{- i \varepsilon^{-1} H_N^{\rm PF} t} \Psi_N$.
Then, there exists a constant $C>0$ such that
\begin{align}
\label{eq:Gronwall estimate final bound}
\beta(t)  &\leq \beta(0) 
+  \int_0^t 
\bigg[   C   \left<  C_{\kappa}^2 \right> \left< \norm{\alpha_s}_{\dot{\mathfrak{h}}_{1/2}}^2 
+ \norm{i \varepsilon \nabla p_s}_{\mathfrak{S}^{\infty}}^2 
\right>
\left< \varepsilon^{-1} N^{-1} \Xi[N,p_s]  \right>
\left( \beta(s) + \varepsilon \right)
\nonumber \\
&\qquad \qquad  \qquad
+ \norm{i \varepsilon \nabla_1 q_{s,1} \Psi_{N,s}}^2
\bigg] \, ds   
\end{align}
with
\begin{align}
\label{eq:definition function for semiclassical structure}
\Xi[N,p] \coloneqq  \sup_{k \in \mathbb{R}^3}
\left\{ (1 + \abs{k})^{-1} \norm{q e^{ikx} p}_{\mathfrak{S}^1} 
\right\} 
+ \norm{q i \varepsilon \nabla p}_{\mathfrak{S}^1} .
\end{align}
\end{lemma}
The previous estimate allows us to conclude the smallness of $\beta(t)$ for times of order one, provided this is initially true, and we can ensure that $\norm{i \varepsilon \nabla_1 q_{s,1} \Psi_{N,s}}^2$ is small, and that both $\norm{i \varepsilon \nabla p_s}_{\mathfrak{S}^{\infty}}$ and $\varepsilon^{-1} N^{-1} \Xi[N,p_s]$ are of order one for all $s \in [0,t]$. Note that $\norm{\alpha_s}_{\dot{\mathfrak{h}}_{1/2}}$ is of order one because of \eqref{eq:bound for the h-1 norm of alpha}.

The quantity $\norm{i \varepsilon \nabla_1 q_{s,1} \Psi_{N,s}}^2$ measures the kinetic energy per particle of the fraction of fermions whose states belong to the kernel of the projection $p$. Both this quantity and the field energy of the relative number of excitations with respect to the coherent state can be controlled by  the functional $\beta$ and the distance between the many-body energy and the energy functional of the Maxwell--Schr\"odinger system, provided $\norm{i \varepsilon \nabla p}_{\mathfrak{S}^{\infty}}$ and $\varepsilon^{-1} N^{-1} \Xi[N,p]$ are of order one.
\begin{lemma}
\label{lemma:energy estimates}
Let $N \in \mathbb{N}$, $\Psi_N \in  \mathcal{D} \left( (H_N^{\rm PF})^{1/2} \right)$ such that $\norm{\Psi_N} = 1$, $p \in \mathfrak{S}_{+}^{2,1}(L^2(\mathbb{R}^3))$ be a rank-$N$ projection, and $\alpha \in \mathfrak{h}_1$. Moreover, let $\mathcal{E}^{\rm{MS}}$ and $\Xi$ be defined as in  \eqref{eq:Maxwell-Schroedinger with exchange term energy definition} and \eqref{eq:definition function for semiclassical structure}. 
Then, there exits a constant $C>0$ such that
\begin{align}
&\norm{i \varepsilon \nabla_1 q_1 \Psi_N}^2 + \varepsilon^4 
\scp{\Psi_N}{W \big( \varepsilon^{-2} \alpha \big) H_f W^* \big( \varepsilon^{-2} \alpha \big) \Psi_N}
\nonumber \\
&\quad \leq 
2 N^{-1} \Big|   \scp{\Psi_N}{H_N^{\rm{PF}} \Psi_N}
- \mathcal{E}^{\rm{MS}}[p,\alpha] \Big|
\nonumber \\
&\qquad 
+ C \left< C_{\kappa}^2 \right>
\left< \norm{\alpha}_{\mathfrak{h}_1} 
+ \norm{\alpha}_{\dot{\mathfrak{h}}_{1/2}}^2 + \norm{i \varepsilon \nabla p}_{\mathfrak{S}^{\infty}}^2 \right>
\left( \beta[\Psi_N,p,\alpha] + \varepsilon^2 
+ N^{-1} \Xi[N,p]
\right) .
\end{align}
\end{lemma}
The proof of Lemma~\ref{lemma:energy estimates} is given in Section~\ref{subsection:energy estimates}.
Showing that $\varepsilon^{-1} N^{-1} \Xi[N,p_t]$ remains of order one under the assumption that this is initially true, is usually referred to as the propagation of the semiclassical structure. Its importance for the derivation of fermionic mean-field equations was first noted in \cite{BPS2014}. The proof that this also holds for the model under consideration relies on the assumption that $\norm{i \varepsilon \nabla p_t }_{\mathfrak{S}^{\infty}}$ is of order one, a condition that is established by the following lemma. The proof is provided in Section~\ref{subsection:proof of the propagation of the semiclassical structure}.
\begin{lemma}
\label{lemma: semiclassical structure}
Let $(p, \alpha) \in \mathfrak{S}_{+}^{2,1} (L^2(\mathbb{R}^3)) \times \mathfrak{h}_{1} \cap \dot{\mathfrak{h}}_{-1/2} $
with $p$ being a rank-$N$ projection.
Let $(p_t, \alpha_t)$ be the unique solution of \eqref{eq:Maxwell-Schroedinger equations} with initial data $(p,\alpha)$ and let $\Xi[N,p_t]$ be defined as in \eqref{eq:definition function for semiclassical structure}. Then, there exists a constant $C>0$ such that the estimates
\begin{align}
\label{eq:momentum estimate}
\norm{i \varepsilon \nabla p_t }_{\mathfrak{S}^{\infty}}
&\leq
\left( \norm{i \varepsilon \nabla p}_{\mathfrak{S}^{\infty}}
+  C \big[ \kappa, p, \alpha \big] t \right)  \exp \left[ C \left[ \kappa, p, \alpha \right]  t \right]
\end{align}
and
\begin{align}
\label{eq: semiclassical structure}
\Xi[N,p_t] 
&\leq 
\Xi[N,p] 
\exp \left[  \left< C_{\kappa}^2 \right>  \left<
 \norm{\alpha}_{\mathfrak{h}_1}^2
+ \norm{i \varepsilon \nabla p}_{\mathfrak{S}^{\infty}}  \right>
 \exp \left[ C[ \kappa, p, \alpha] t \right] \right]  
\end{align}
hold for all $t \in \mathbb{R}_{+}$ with
$C \big[ \kappa , p, \alpha \big]
= C \left< C_{\kappa}^2 \right>
\left<   N^{-1} \mathcal{E}^{\rm{MS}}[p,\alpha] + C C_{\kappa}^2
\right>$.
\end{lemma}
The previous lemmas allow us to control the growth of the quantities $\beta(t)$, $\norm{i \varepsilon \nabla_1 q_{t,1} \Psi_{N,t}}^2$, and $\varepsilon^4 
\scp{\Psi_{N,t}}{W \big( N^{2/3} \alpha_t \big) H_f W^* \big( N^{2/3} \alpha_t \big) \Psi_{N,t}}$ during the time evolution. 
It is more common, however, to measure the distance between the electron state and the projection $p_t$  using the one-electron reduced density matrix. The relationship between this notion of proximity and the quantities mentioned above is addressed in the following lemma.

\begin{lemma}
\label{lemma:relation between trace-norm convergence and beta-a}
Let $N \in \mathbb{N}$, $\Psi_{N} \in  \mathcal{H}^{(N)}$ such that $\norm{\Psi_N} = 1$ and let $p \in \mathfrak{S}^1(L^2(\mathbb{R}^3))$ be a rank-$N$ projection. Then,
\begin{align}
\label{eq:relation between trace-norm convergence and beta-a}
2 \beta^a[\Psi_N,p] \leq \norm{\gamma^{(1,0)}_{\Psi_N} - N^{-1} p }_{\mathfrak{S}^1} \leq \sqrt{8 \beta^a[\Psi_N,p]} .
\end{align}
There exists a constant $C >0$ such that
\begin{align}
\label{eq:relation between Sobolev trace-norm convergence and beta-a plus kinetic energy of the particles outside the Slater determinant}
&\norm{\sqrt{1 - \varepsilon^2 \Delta} \left( \gamma^{(1,0)}_{\Psi_N} - N^{-1} p \right) \sqrt{1 - \varepsilon^2 \Delta}}_{\mathfrak{S}^1}
\nonumber \\
&\leq  C \left<  \norm{i \varepsilon \nabla p}_{\mathfrak{S}^{\infty}} \right>^2
\sup_{j=1,2} \left( \beta^a[\Psi_N,p] + \norm{i \varepsilon \nabla_1 q_1 \Psi_N}^2 \right)^{j/2} 
\end{align}
holds for all $\Psi_{N} \in  \mathcal{H}^{(N)} \cap \mathcal{D} \left( (H_N^{\rm PF})^{1/2} \right)$ such that $\norm{\Psi_N} = 1$ and for all rank-$N$ projections $p \in \mathfrak{S}^{1,1}(L^2(\mathbb{R}^3))$.
\end{lemma}
Inequality~\eqref{eq:relation between trace-norm convergence and beta-a} is a standard result. Its proof can be found, for example, in \cite[Section 3.1]{PP2016}. The proof of \eqref{eq:relation between Sobolev trace-norm convergence and beta-a plus kinetic energy of the particles outside the Slater determinant} is postponed to Section~\ref{section:Proof of Lemma with relation between trace norm convergence and beta-a}. Note that a bosonic variant of \eqref{eq:relation between Sobolev trace-norm convergence and beta-a plus kinetic energy of the particles outside the Slater determinant} was proven in \cite{MPP2019} and applied in \cite{LP2018}. 
Finally, we have collected all ingredients needed to prove the main result.

\begin{proof}[Proof of Theorem~\ref{theorem:main theorem}]
Combining the Lemmas \ref{lemma:estimating the time derivative of the functional} and \ref{lemma:energy estimates} with the conservation of the energies let us obtain
\begin{align}
\beta(t) - \beta(0)
&\leq  C \left< C_{\kappa}^2 \right> \int_0^t 
\left< 
\norm{\alpha_s}_{\mathfrak{h}_1} 
+ \norm{\alpha_s}_{\dot{\mathfrak{h}}_{1/2}}^2 
+ \norm{i \varepsilon \nabla p_s}_{\mathfrak{S}^{\infty}}^2 
\right>
\left<\varepsilon^{-1} N^{-1} \Xi[N,p_s] \right>
\left( \beta(s) + \varepsilon \right)
\, ds  
\nonumber \\
&\quad + 2 N^{-1} \Big|   \scp{\Psi_{N,0}}{H_N^{\rm{PF}} \Psi_{N,0}}
- \mathcal{E}^{\rm{MS}}[p_0,\alpha_0] \Big| t .
\end{align}
By means of Gr\"onwall's Lemma,  we get
\begin{align}
\beta(t)
&\leq  
\exp \bigg[ C \left< C_{\kappa}^2 \right>  \left< t \right>
\sup_{s \in [0,t]} \Big\{ \left< \norm{\alpha_s}_{\mathfrak{h}_1} 
+ \norm{\alpha_s}_{\dot{\mathfrak{h}}_{1/2}}^2 
+ \norm{i \varepsilon \nabla p_s}_{\mathfrak{S}^{\infty}}^2 
\right>
\left< \varepsilon^{-1} N^{-1} \Xi[N,p_s] \right>
\Big\}
\bigg] 
\nonumber \\
&\quad \times \left( \beta(0)
+  N^{-1} \Big| \scp{\Psi_{N,0}}{H_N^{\rm{PF}} \Psi_{N,0}}
- \mathcal{E}^{\rm{MS}}[p_0,\alpha_0] \Big|  + \varepsilon \right) .
\end{align}
Using Lemma~\ref{lemma:energy estimates} again we obtain
\begin{align}
&\beta(t)
+ \norm{i \varepsilon \nabla_1 q_{t,1} \Psi_{N,t}}^2 + \varepsilon^4 
\scp{\Psi_{N,t}}{W \big( \varepsilon^{-2} \alpha_t \big) H_f W^* \big( \varepsilon^{-2} \alpha_t \big) \Psi_{N,t}}
\nonumber \\
&\quad \leq  
\exp \bigg[ C \left< C_{\kappa}^2 \right>  \left< t \right>
\sup_{s \in [0,t]} \Big\{ \left< \norm{\alpha_s}_{\mathfrak{h}_1} 
+ \norm{\alpha_s}_{\dot{\mathfrak{h}}_{1/2}}^2 
+ \norm{i \varepsilon \nabla p_s}_{\mathfrak{S}^{\infty}}^2 
\right>
\left< \varepsilon^{-1} N^{-1} \Xi[N,p_s] \right>
\Big\}
\bigg] 
\nonumber \\
&\quad \times \left( \beta(0)
+  N^{-1} \Big| \scp{\Psi_{N,0}}{H_N^{\rm{PF}} \Psi_{N,0}}
- \mathcal{E}^{\rm{MS}}[p_0,\alpha_0] \Big|  + \varepsilon \right) .
\end{align}
Note that
\begin{align}
&\sup_{s \in [0,t]} \Big\{ \left< \norm{\alpha_s}_{\mathfrak{h}_1} 
+ \norm{\alpha_s}_{\dot{\mathfrak{h}}_{1/2}}^2 
+ \norm{i \varepsilon \nabla p_s}_{\mathfrak{S}^{\infty}}^2 
\right>
\left< \varepsilon^{-1} N^{-1} \Xi[N,p_s] \right>
\Big\}
\nonumber \\
&\quad \leq 
\left< \varepsilon^{-1} N^{-1} \Xi[N,p_0]
\right>
\exp \left[ C   \left< \norm{\alpha_0}_{\mathfrak{h}_1}^2
+ \norm{i \varepsilon \nabla p_0}_{\mathfrak{S}^{\infty}}  \right> \exp \left[
C[\kappa, p_0, \alpha_0] \left< t \right> \right]  \right]
\end{align}
with $C \big[ \kappa , p, \alpha \big]
= C \left< C_{\kappa}^2 \right>
\left<   N^{-1} \mathcal{E}^{\rm{MS}}[p,\alpha] + C C_{\kappa}^2
\right>$
because of \eqref{eq:bound for the h-1 norm of alpha}
and Lemma~\ref{lemma: semiclassical structure}.
Hence,
\begin{align}
&\beta(t)
+ \norm{i \varepsilon \nabla_1 q_{t,1} \Psi_{N,t}}^2 + \varepsilon^4 
\scp{\Psi_{N,t}}{W \big( \varepsilon^{-2} \alpha_t \big) H_f W^* \big(\varepsilon^{-2} \alpha_t \big) \Psi_{N,t}}
\nonumber \\
&\quad \leq 
\exp \left[ 
\left< \varepsilon^{-1} N^{-1} \Xi[N,p_0]  \right>
\exp \left[  C \left< \norm{\alpha_0}_{\mathfrak{h}_1}^2
+ \norm{i \varepsilon \nabla p_0}_{\mathfrak{S}^{\infty}}  \right>
 \exp \left[ C[\kappa, p_0, \alpha_0] \left< t \right> \right] \right] \right]  
\nonumber \\
&\qquad \times \left( \beta(0)
+  N^{-1} \Big| \scp{\Psi_{N,0}}{H_N^{\rm{PF}} \Psi_{N,0}}
- \mathcal{E}^{\rm{MS}}[p_0,\alpha_0] \Big|  + \varepsilon \right) .
\end{align}
Theorem~\ref{theorem:main theorem} then follows by means of Lemma~\ref{lemma:relation between trace-norm convergence and beta-a}.

\end{proof}

\subsection{Proof of Theorem~\ref{theorem:derivation of the Vlasov-Maxwell equations}}

\label{subsection:proof of VM result}

In the following, we prove Theorem~\ref{theorem:derivation of the Vlasov-Maxwell equations} by using a combination of Theorem~\ref{theorem:main theorem} and \cite[Theorem II.1]{LS2023}. The estimates in \cite{LS2023} do not track the explicit dependence on $\kappa$.   For this reason, we also do not make this dependence explicit in the estimates leading to Theorem~\ref{theorem:derivation of the Vlasov-Maxwell equations}. Throughout this section and Section~\ref{section:estimates concerning the Vlasov--Maxwell equations}, we use the letter $C$ to denote a generic constant that may depend on $\kappa$.
Moreover, note that we use the notation
\begin{align}
D_{\varepsilon} = \left( 1 - \varepsilon^2 \Delta \right)^{1/2}
\end{align}
and
\begin{align}
\norm{\omega}_{\mathfrak{S}^{1,1}_{\varepsilon}}
&= 
\norm{D_{\varepsilon} \omega D_{\varepsilon}}_{\mathfrak{S}^1}
= \norm{\left( 1 - \varepsilon^2 \Delta \right)^{1/2} \omega \left( 1 - \varepsilon^2 \Delta \right)^{1/2}}_{\mathfrak{S}^1} .
\end{align}
to denote the fractional Laplacian with semiclassical parameter and the corresponding Sobolev trace norm, respectively.
Since \cite[Theorem II.1]{LS2023} requires higher regularity properties as expected to hold for pure states of fermions,  it is essential to regularize in Theorem~\ref{theorem:derivation of the Vlasov-Maxwell equations} the initial data of the Vlasov--Maxwell system. The  comparison between the initial data obtained by the Wigner transform and its regularized version is established by the following lemma, which is proven in Section~\ref{subsection:Vlasov-Maxwell distance between initial data and its regularization}.
\begin{lemma}
\label{lemma:Vlasov-Maxwell distance between initial data and its regularization updated}
Let $R >0$, $\widetilde{C} >0$, $N \in \mathbb{N}$, $\varepsilon = N^{- 1/3}$, $\varepsilon^{1/2} \leq \sigma \leq 1$, and $\Lambda >1$. Let $\alpha \in \mathfrak{h}_1 \cap \dot{\mathfrak{h}}_{- 1/2}$,  $p \in \mathfrak{S}^{1,2}(L^2(\mathbb{R}^3))$ be a rank-$N$ projection such that \eqref{eq:Vlasov-Maxwell result more restricted initial semiclassical structure} holds and such that the Wigner transform $\mathcal{W}[p]$ of $p$ satisfies $ \norm{\mathcal{W}[p]}_{W_6^{0,2}(\mathbb{R}^6)}  \leq \widetilde{C}$ and $\norm{\chi_{\geq R}  \mathcal{W}[p]}_{W_6^{0,2}(\mathbb{R}^6)} \leq \widetilde{C} e^{- \frac{1}{2 \sigma^2}}$,  where
$\chi_{\geq R}: \mathbb{R}^6 \rightarrow \mathbb{R}$ is defined by $\chi_{\geq R}(x,v) = \id_{\abs{v} \geq R}$
. Moreover, let $m_{p,\sigma}$ and $m_{p,\sigma,R}$ be defined as in \eqref{eq:definition general Husimi measure} and \eqref{eq:definition regularized Wigner transform}. Then, there exists $C >0$ such that
\begin{subequations}
\begin{align}
\label{eq:distance between p and regularized version in Sobolev-trace norm 1}
\norm{ \mathcal{W}^{-1}[m_{p, \sigma}] - p }_{\mathfrak{S}^{1,1}_{\varepsilon}}
&\leq C N   \sigma   ,
\\
\label{eq:distance between p and regularized version in Sobolev-trace norm 2}
\norm{  \mathcal{W}^{-1}[m_{p, \sigma}] - \mathcal{W}^{-1}[m_{p, \sigma,R}]  }_{\mathfrak{S}^{1,1}_{\varepsilon}}
&\leq  C N  \left( e^{- \frac{1}{2 \sigma^2}}  + \varepsilon^{1/2} \right)   ,
\\
\label{eq:distance between alpha and its cutoff version}
\norm{\alpha - \id_{\abs{\cdot} \leq \Lambda} \alpha}_{\mathfrak{h}_{1/2} \cap \dot{\mathfrak{h}}_{- 1/2}}
= \norm{\id_{\abs{\cdot} \geq \Lambda} \alpha}_{\mathfrak{h}_{1/2} \cap \dot{\mathfrak{h}}_{- 1/2}}
&\leq C \Lambda^{- 1/2} 
\norm{\alpha}_{\mathfrak{h}_1} ,
\\
\label{eq:fermionic reduced density proeprty of the weyl quantization of m-p-sigma}
\tr \left( \mathcal{W}^{-1}[m_{p,\sigma}] \right) = N
\quad \text{and} \quad
0 \leq \mathcal{W}^{-1}[m_{p,\sigma}] &\leq 1 ,
\\
\label{eq:weyl quantization of m-p-sigma in S-2-1}
\mathcal{W}^{-1}[m_{p,\sigma} ] \in \mathfrak{S}^{2,1}(L^2(\mathbb{R}^3)) ,
\\
\label{eq:positivity of the Husimi measure}
m_{p,\sigma,R} \geq 0 \quad \text{on} \, \, \mathbb{R}^6 .
\end{align}
\end{subequations}
\end{lemma}
In addition, we will rely on the following result regarding the propagation of regularity for solutions of the Vlasov--Maxwell equations, which is shown in Section~\ref{subsection:propagation estimates Vlasov-Maxwell}.
\begin{lemma}
\label{lemma:propagation estimates Vlasov-Maxwell}
Let $R>0$ and $a, b \in \mathbb{N}$ satisfying $a \geq 5$ and $b \geq 3$. Let $(f_t,\alpha_t)$ be the unique solution of the Vlasov--Maxwell equations \eqref{eq: Vlasov-Maxwell different style of writing} with initial condition
 $(f_0, \alpha_0) \in H_a^{b}(\mathbb{R}^6) \times \mathfrak{h}_b \cap \dot{\mathfrak{h}}_{-1/2}$ such that $\supp f_0 \subseteq  \{ (x,v) \in \mathbb{R}^6 , \abs{v} \leq R \}$. For $k \in \mathbb{N}$ satisfying $k \leq b$, we have
\begin{align}
\norm{f}^2_{L_t^{\infty} H_a^{k}(\mathbb{R}^6)} + \norm{\alpha}^4_{L_t^{\infty} \mathfrak{h}_{k}}
&\leq  C e^{\left< R \right> e^{C \left< t \right>  C(t) }}
\bigg[ 1 + 
\sum_{j=0}^{k} \left(  \norm{f_0}^2_{ H_a^{k-j}(\mathbb{R}^6)} + \norm{\alpha_0}^4_{ \mathfrak{h}_{k-j} \, \cap \,  \dot{\mathfrak{h}}_{-1/2})} \right)^{3^j}
\bigg] ,
\end{align}
where $C(s) =
e^{C \left( 1 + \mathcal{E}^{\rm{VM}}[f_0, \alpha_0] 
+ C \norm{f_0}_{L^1(\mathbb{R}^6)}^2  \right) \left< s \right>}
\Big[ 1 + 
\norm{f_0}_{W_a^{0,2}(\mathbb{R}^6)}^2
+ \norm{\alpha_0}_{\mathfrak{h}_1 \, \cap \,  \dot{\mathfrak{h}}_{-1/2}}^2 \Big]$ and $C$ is a constant depending on $\kappa$.
\end{lemma}
Note that regularity estimates in the spirit of Lemma~\ref{lemma:propagation estimates Vlasov-Maxwell} for the solutions of the Vlasov equation were obtained in \cite{BPSS2016}.
Theorem~\cite[Theorem II.1]{LS2023} and the previous two Lemmas provide the following comparison between the solutions of the \eqref{eq:Maxwell-Schroedinger equations} and \eqref{eq: Vlasov-Maxwell different style of writing}.
\begin{proposition}
\label{proposition:comparison between Maxwell--Schroedinger and Vlasov--Maxwell updated}
Let $\kappa$ satisfy Assumption~\ref{assumption:cutoff function}, $\widetilde{C} > 0$, $R>0$, $N \in \mathbb{N}$,  $\varepsilon = N^{- \frac{1}{3}}$, $\Lambda = \left( \ln(\ln(N)) \right)^{\frac{1}{2}}$, and $\sigma = \left( \ln(\ln(N)) \right)^{-1}$.
Let $\alpha \in \mathfrak{h}_1 \cap \dot{\mathfrak{h}}_{-1/2}$
and let $p \in \mathfrak{S}_{+}^{2,1}(L^2(\mathbb{R}^3))$ be a rank-$N$ projection such that \eqref{eq:main result condition on the energy and velocity} and  \eqref{eq:Vlasov-Maxwell result more restricted initial semiclassical structure} hold, and such that its Wigner transform $\mathcal{W}[p]$ satisfies $ \norm{\mathcal{W}[p]}_{W_7^{0,2}(\mathbb{R}^6)}  \leq \widetilde{C}$ and $\norm{\chi_{\geq R}  \mathcal{W}[p]}_{W_7^{0,2}(\mathbb{R}^6)} \leq \widetilde{C} e^{- \frac{1}{2 \sigma^2}}$,  where
$\chi_{\geq R}: \mathbb{R}^6 \rightarrow \mathbb{R}$ is defined by $\chi_{\geq R}(x,v) = \id_{\abs{v} \geq R}$.
Let $m_{p,\sigma}$ and $m_{p,\sigma,R}$ be defined as in \eqref{eq:definition general Husimi measure} and \eqref{eq:definition regularized Wigner transform}.
Moreover, let $(\widetilde{p}_t, \widetilde{\alpha}_t)$ be the unique solution of the Maxwell--Schr\"odinger equations~\eqref{eq:Maxwell-Schroedinger equations} with regularized initial datum $(\mathcal{W}^{-1}[m_{p,\sigma}],\id_{\abs{\cdot} \leq \Lambda} \alpha)$ and let  $(W_{N,t},\alpha_t)$ be the unique solution of the Vlasov--Maxwell equations~\eqref{eq: Vlasov-Maxwell different style of writing} with regularized initial datum $(m_{p,\sigma,R}, \id_{\abs{\cdot} \leq \Lambda} \alpha)$. 
Then, there exists a constant $C > 0$ (independent of $N$ but dependent on $\kappa$) 
such that for any $t \geq 0$ and $N \geq  \exp \left[ \exp \left[ 2 \left[ 1 + \exp \left[ \left< R \right> \exp \left[ \exp \left[ C \left< t \right> \right] \right] \right]  \right] \right] \right]$
\begin{align}
& N^{-1} \norm{ \widetilde{p}_t - \mathcal{W}^{-1}[W_{N,t}] }_{\mathfrak{S}^{1,1}_{\varepsilon}} + \norm{\widetilde{\alpha}_t - \alpha_t}_{\mathfrak{h}_{1/2} \cap \, \dot{\mathfrak{h}}_{- 1/2}}
\nonumber \\
&\quad \leq 
C \left(  \ln(N)^{-1} + N^{- 1/6} 
 \left[  \ln (N) \right]^{ \exp \left[ \left< R \right> \exp \left[ \exp \left[ C \left< t \right> \right] \right] \right]}
\right) .
\end{align}
\end{proposition}
In order to combine the main theorem of our article with the previous result for regularized initial data, we need to establish a connection between solutions of the Maxwell--Schr\"odinger equation with and without regularized initial data. To this end, we rely on the following propagation result for the semiclassical structure with kinetic weights.
\begin{lemma}
\label{lemma: semiclassical structure with kinetic weights}
Let $(p, \alpha) \in \mathfrak{S}_{+}^{2,1} (L^2(\mathbb{R}^3)) \times \mathfrak{h}_{1} \cap \dot{\mathfrak{h}}_{-1/2} $
with $p$ being a rank-$N$ projection such that \eqref{eq:main result condition on the energy and velocity} holds.
Let $(p_t, \alpha_t)$ be the unique solution of \eqref{eq:Maxwell-Schroedinger equations} with initial data $(p,\alpha)$ and let $\Xi_2[N,p_t]$
be defined as 
\begin{align}
\label{eq:definition function for semiclassical structure with kinetic weights}
\Xi_2[N,p] \coloneqq  \sup_{k \in \mathbb{R}^3}
\left\{   (1 + \abs{k})^{-1} \norm{q e^{ikx} p}_{\mathfrak{S}^{1,1}_{\varepsilon}}
\right\} 
+ \norm{q i \varepsilon \nabla p}_{\mathfrak{S}^{1,1}_{\varepsilon}} , 
\end{align}
where $q  = 1 - p$. Then, there exists a constant $C>0$ (independent of $N$ but dependent on $\kappa$) such that for any $t \geq 0$
\begin{align}
\Xi_2[N,p_t] 
&\leq 
\left( \Xi_2[N,p] + N \varepsilon \right)
\exp \left[  \exp \big[ C \left< t \right> \big] \right]   .
\end{align}
\end{lemma}
This leads to the following result,  which is proved in Section~\ref{subsection:proof of proposiion about regularization for Maxwell-Schroedinger}.
\begin{proposition}
\label{proposition:Regularization result for Maxwell-Schroedinger}
Let $(p, \alpha) \in \mathfrak{S}_{+}^{2,1} (L^2(\mathbb{R}^3)) \times \mathfrak{h}_{1} \cap \dot{\mathfrak{h}}_{-1/2} $
with $p$ being a rank-$N$ projection such that \eqref{eq:main result condition on the energy and velocity} and  \eqref{eq:Vlasov-Maxwell result more restricted initial semiclassical structure} hold. Moreover,  let  $m_{p,\sigma}$ be defined as in \eqref{eq:definition general Husimi measure}.
Let $(p_t, \alpha)$ be the unique solution of the Maxwell--Schr\"odinger equations~\eqref{eq:Maxwell-Schroedinger equations} with initial datum $(p,  \alpha)$ and let $(\widetilde{p}_t, \widetilde{\alpha}_t)$ be the unique solution of the Maxwell--Schr\"odinger equations~\eqref{eq:Maxwell-Schroedinger equations} with regularized initial datum $(\mathcal{W}^{-1}[m_{p,\sigma}],\id_{\abs{\cdot} \leq \Lambda} \alpha)$.
Then, there exists a constant $C > 0$ (independent of $N$ but dependent on $\kappa$) such that for any $t \geq 0$
\begin{align}
N^{-1} \norm{ p_t - \widetilde{p}_t}_{\mathfrak{S}^{1,1}_{\varepsilon}}
+  \norm{\alpha_t - \widetilde{\alpha}_t}_{\dot{\mathfrak{h}}_{-1/2} \cap \dot{\mathfrak{h}}_{1/2}}
&\leq 
\left( \sigma + \Lambda^{- \frac{1}{2}} + \varepsilon \right)
\exp \left[ \exp \left[  \exp \big[ C \left< t \right> \big] \right]  \right] .
\end{align}
\end{proposition}
Combining the previous results with Theorem~\ref{theorem:main theorem} yields Theorem~\ref{theorem:derivation of the Vlasov-Maxwell equations}.
\begin{proof}[Proof of Theorem~\ref{theorem:derivation of the Vlasov-Maxwell equations}]
Let $\Psi_{N,t} = e^{- i \varepsilon^{-1} H_N^{\rm{PF}} t} \Psi_N$, 
let $(p^{(1)}_t, \alpha^{(1)}_t)$ be the unique solution of the Maxwell--Schr\"odinger equations~\eqref{eq:Maxwell-Schroedinger equations} with initial datum $(p,\alpha)$, let $(p^{(2)}_t, \alpha^{(2)}_t)$ be the unique solution of the Maxwell--Schr\"odinger equations~\eqref{eq:Maxwell-Schroedinger equations} with regularized initial datum $(\mathcal{W}^{-1}[ m_{p,\sigma}], \id_{\abs{\cdot} \leq \Lambda} \alpha)$, and let $(W_{N,t},\alpha_t)$ be the unique solution of the Vlasov--Maxwell equations~\eqref{eq: Vlasov-Maxwell different style of writing} with regularized initial datum $(m_{p,\sigma,R}, \id_{\abs{\cdot} \leq \Lambda} \alpha)$.
Using the triangular inequality,  \eqref{eq:Weyl operators shifting property},  and $\norm{\Psi_{N,t}} = 1$ we get
\begin{subequations}
\begin{align}
&\norm{\gamma^{(1,0)}_{\Psi_{N,t}} - N^{-1}  \mathcal{W}^{-1}[W_{N,t}]}_{\mathfrak{S}_{\varepsilon}^{1,1}}
+  \varepsilon^4  \big\langle\Psi_{N,t} ,W \big( \varepsilon^{-2} \alpha_t \big) \left( \mathcal{N} + H_f \right) \,  W^* \big( \varepsilon^{-2} \alpha_t \big) \Psi_{N,t} \big\rangle_{\mathcal{H}^{(N)}}
\nonumber \\
\label{eq:proof of VM estimates 1}
&\quad \leq \norm{\gamma^{(1,0)}_{\Psi_{N,t}} - N^{-1}  p^{(1)}_t}_{\mathfrak{S}_{\varepsilon}^{1,1}}
+ \varepsilon^4  \big\langle \Psi_{N,t} ,W \big( \varepsilon^{-2} \alpha^{(1)}_t \big) \left( \mathcal{N} + H_f \right) \,  W^* \big( \varepsilon^{-2} \alpha^{(1)}_t \big) \Psi_{N,t} \big\rangle
\\
\label{eq:proof of VM estimates 2}
&\qquad +
N^{-1}  \norm{ p^{(1)}_t - p^{(2)}_t }_{\mathfrak{S}_{\varepsilon}^{1,1}} 
+ C  \norm{\alpha^{(1)}_t - \alpha^{(2)}_t}^2_{\mathfrak{h}_{1/2}} 
\\
\label{eq:proof of VM estimates 3}
&\qquad +
N^{-1}  \norm{p^{(2)}_t - \mathcal{W}^{-1}[W_{N,t}]  }_{\mathfrak{S}_{\varepsilon}^{1,1}}
+ C   \norm{\alpha^{(2)}_t - \alpha_t}^2_{\mathfrak{h}_{1/2}} .
\end{align}
\end{subequations}
Combining the applications of Theorem~\ref{theorem:main theorem} to \eqref{eq:proof of VM estimates 1},  Proposition~\ref{proposition:Regularization result for Maxwell-Schroedinger} to \eqref{eq:proof of VM estimates 2},  and Proposition~\ref{proposition:comparison between Maxwell--Schroedinger and Vlasov--Maxwell updated} to \eqref{eq:proof of VM estimates 3} shows that there exists a constant $C > 0$ (independent of $N$ but dependent on $\kappa$) such that for any $t \geq 0$ and $N \geq  \exp \left[ \exp \left[ \exp \left[ \left< R \right> \exp \left[ \exp \left[ C \left< t \right> \right] \right] \right]  \right] \right]$
\begin{align}
&\norm{\gamma^{(1,0)}_{\Psi_{N,t}} - N^{-1}  \mathcal{W}^{-1}[W_{N,t}]}_{\mathfrak{S}_{\varepsilon}^{1,1}}
+  \varepsilon^4  \big\langle\Psi_{N,t} ,W \big( \varepsilon^{-2} \alpha_t \big) \left( \mathcal{N} + H_f \right) \,  W^* \big( \varepsilon^{-2} \alpha_t \big) \Psi_{N,t} \big\rangle_{\mathcal{H}^{(N)}}
\nonumber \\
&\quad \leq 
\Big( \sup_{j=1,2} \big\{  d[\Psi_N, p, \alpha]^{j/2} \big\}
+ \left( \ln(\ln(N)) \right)^{- \frac{1}{4}}  \Big)
\exp \left[ \exp \left[  \exp \big[ C \left< t \right> \big] \right]  \right]
\nonumber \\
&\qquad + N^{- 1/6} 
 \left[  \ln (N) \right]^{ \exp \left[ \left< R \right> \exp \left[ \exp \left[ C \left< t \right> \right] \right] \right]} .
\end{align}
\end{proof}

\section{Properties of the Maxwell--Schr\"odinger Solutions}
\label{section:properties of the MS equations}

We first collect some estimates from \cite{LS2023}, which will be useful in the following, and then prove inequality~\eqref{eq:bound for the h-1 norm of alpha} and Lemma~\ref{lemma: semiclassical structure}.

\begin{lemma}
\label{lemma:preliminary estimates for MS solution}
There exists a constant $C>0$ such that,  for any $\alpha \in \mathfrak{h}_{1/2}$ and $p \in \mathfrak{S}^{1,1}(L^2(\mathbb{R}^3))$ satisfying $0 \leq p \leq 1$, the following holds:
\begin{subequations}
\begin{align}
\label{eq:L-1 norm of the first moment of the potential in terms of kappa}
\norm{K}_{L^{\infty}} 
&\leq 
\norm{(1 + \abs{\cdot} ) \mathcal{F}[K]}_{L^1} 
\leq C \big\| (1 + \abs{\cdot}^{-1}) \mathcal{F}[\kappa] \big\|_{L^2}^2 ,
\\
\label{eq:estimate for A L-infty to h-1-2}
\norm{\kappa * \vA_{\alpha}}_{L^{\infty}} 
&\leq \big\| \abs{\cdot}^{-1} \mathcal{F}[\kappa] \big\|_{L^2} \norm{\alpha}_{\dot{\mathfrak{h}}_{1/2}} ,
\\
\label{eq:estimate vector potential W-1-infty norm}
\norm{\kappa * \vA_{\alpha}}_{W^{1,\infty}_0(\mathbb{R}^3)}
&\leq C \big\| (1 + \abs{\cdot}^{-1}) \mathcal{F}[\kappa] \big\|_{L^2} \norm{\alpha}_{\dot{\mathfrak{h}}_{1/2}} ,
\\
\label{eq:estimate for mean-field potential}
\norm{K * \rho_{p}}_{W_0^{2,\infty}(\mathbb{R}^3)} &\leq C  N^{-1} \big\| ( 1 + \abs{\cdot}^{-1}) \mathcal{F}[\kappa] \big\|_{L^2}^2  \norm{p}_{\mathfrak{S}^1} , \\
\label{eq:estimate for exchange term}
\norm{X_p}_{\mathfrak{S}^{\infty}} &\leq C N^{-1}  \big\| ( 1 + \abs{\cdot}^{-1}) \mathcal{F}[\kappa] \big\|_{L^2}^2 \norm{p}_{\mathfrak{S}^1} , \\
\label{eq:estimate for field energy}
\norm{\alpha}_{\dot{\mathfrak{h}}_{1/2}}^2
&\leq N^{-1} \mathcal{E}^{\rm{MS}}[p,\alpha] + C N^{-2}  \big\| ( 1 + \abs{\cdot}^{-1}) \mathcal{F}[\kappa] \big\|_{L^2}^2 \norm{p}_{\mathfrak{S}^1}^2 , 
\\
\label{eq:estimate for kinetic energy}
\tr \left( - \varepsilon^2 \Delta p \right)
&\leq C \left< N^{-1}  \big\| (1 + \abs{\cdot}^{-1}) \mathcal{F}[\kappa] \big\|_{L^2}^2  \norm{p}_{\mathfrak{S}^1} \right>
\nonumber \\
&\quad \times
\left(  \mathcal{E}^{\rm{MS}}[p,\alpha] + C N^{-1}  \big\| ( 1 + \abs{\cdot}^{-1}) \mathcal{F}[\kappa] \big\|_{L^2}^2 \norm{p}_{\mathfrak{S}^1}^2
\right) .
\end{align}
\end{subequations}
\end{lemma}

\begin{proof}
Inequality~\eqref{eq:L-1 norm of the first moment of the potential in terms of kappa} is a direct consequence of
\begin{align}
\label{eq:Fourier transform of the potential}
\mathcal{F}[K](k) &= \mathcal{F}[\kappa * \kappa * \abs{\cdot}^{-1}](k)
= C \mathcal{F}[\kappa](k)^2 \abs{k}^{-2} .
\end{align}
The estimates \eqref{eq:estimate for A L-infty to h-1-2} and \eqref{eq:estimate vector potential W-1-infty norm} are obtained by means of 
\begin{align}
\kappa * \vA_{\alpha}(x) &=  \sum_{\lambda = 1,2} \int_{\mathbb{R}^3}  \mathcal{F}[\kappa](k) \frac{1}{\sqrt{2 \abs{k}}} \vep_{\lambda}(k) \left( e^{i k x} \alpha(k,\lambda) + e^{- i k x} \overline{\alpha(k, \lambda)} \right) \, dk
\end{align}
and the Cauchy--Schwarz inequality. By applying \eqref{eq:Fourier transform of the potential} and Young's inequality, we obtain
\begin{align}
\norm{K * \rho_p}_{W_0^{2,\infty}(\mathbb{R}^3)}
&\leq \norm{K}_{W_0^{2,\infty}(\mathbb{R}^3)} \norm{\rho_p}_{L^1}
\leq N^{-1} \norm{\left< \cdot \right>^2 \mathcal{F}[K]}_{L^1(\mathbb{R}^3)} \norm{p}_{\mathfrak{S}^1}
\leq C  N^{-1} B_{\kappa}^2  \norm{p}_{\mathfrak{S}^1} .
\end{align}
Note that
\begin{align}
\abs{X_{p} \psi(x)}
&\leq N^{-1} \int_{\mathbb{R}^3} \abs{K(x-y)}
\abs{p(x;y)} \abs{\psi(y)} \, dy
\leq N^{-1} \norm{K}_{L^{\infty}} \norm{\psi}_{L^2} \left( \int_{\mathbb{R}^3} \abs{p(x;y)}^2 \, dy \right)^{1/2}
\end{align}
where $\psi \in L^2(\mathbb{R}^3)$. This implies that
$\norm{X_p}_{\mathfrak{S}^{\infty}}
\leq N^{-1} \norm{K}_{L^{\infty}}
\norm{p}_{\mathfrak{S}^2}^2$.
Together with $\norm{p}_{\mathfrak{S}^2}^2 \leq \norm{p}_{\mathfrak{S}^1}$ and  \eqref{eq:L-1 norm of the first moment of the potential in terms of kappa}, this proves \eqref{eq:estimate for exchange term}.
Estimate \eqref{eq:estimate for field energy} is a direct consequence of \eqref{eq:estimate for mean-field potential}, \eqref{eq:estimate for exchange term}, and the positivity of the operator $\left( - i \varepsilon \nabla - \kappa * \vA_{\alpha} \right)^2$. Using \eqref{eq:estimate vector potential W-1-infty norm}, \eqref{eq:estimate for mean-field potential}, \eqref{eq:estimate for exchange term}, \eqref{eq:estimate for field energy}, and $\left( - i \varepsilon \nabla - \kappa * \vA_{\alpha} \right)^2 \geq  - \varepsilon^2 \Delta + 2 i \varepsilon  \kappa * \vA_{\alpha} \cdot \nabla$, we obtain
\begin{align}
\tr \left( - \varepsilon^2 \Delta p \right)
&\leq  \tr \left(\left( - i \varepsilon \nabla - \kappa * \vA_{\alpha} \right)^2 p \right)
+ 2 \abs{\tr \left( i \varepsilon  \kappa \vA_{\alpha} \cdot \nabla p \right)}
\nonumber \\
&\leq  \tr \left(\left( - i \varepsilon \nabla - \kappa * \vA_{\alpha} \right)^2 p \right)
+ 2 \norm{\kappa * \vA_{\alpha}}_{L^{\infty}} \norm{i \varepsilon  \nabla \sqrt{p}}_{\mathfrak{S}^2}  \norm{\sqrt{p}}_{\mathfrak{S}^2} 
\nonumber \\
&\leq  \tr \left(\left( - i \varepsilon \nabla - \kappa * \vA_{\alpha} \right)^2 p \right)
+ \frac{1}{2} \tr \left( - \varepsilon^2 \Delta p \right) +
 C B_{\kappa}^2 \norm{\sqrt{p}}_{\mathfrak{S}^2}^2 \norm{\alpha}_{\dot{\mathfrak{h}}_{1/2}}^2 ,
\end{align} 
leading to
\begin{align}
\tr \left( - \varepsilon^2 \Delta p \right)
&\leq  2 \, \tr \left(\left( - i \varepsilon \nabla - \kappa * \vA_{\alpha} \right)^2 p \right)
+ C B_{\kappa}^2 \norm{\sqrt{p}}_{\mathfrak{S}^2}^2 \norm{\alpha}_{\dot{\mathfrak{h}}_{1/2}}^2
\nonumber \\
&\leq 2 \, \mathcal{E}^{\rm{MS}}[p,\alpha]
+ C  N^{-1} B_{\kappa}^2  \norm{p}_{\mathfrak{S}^1}^2
+
C B_{\kappa}^2 \norm{p}_{\mathfrak{S}^1} 
\left( N^{-1} \mathcal{E}^{\rm{MS}}[p,\alpha] + C N^{-2}  B_{\kappa}^2 \norm{p}_{\mathfrak{S}^1}^2
\right)
\nonumber \\
&\leq C \left( 1 + C N^{-1}  B_{\kappa}^2  \norm{p}_{\mathfrak{S}^1} 
\right)
\left(  \mathcal{E}^{\rm{MS}}[p,\alpha] + C N^{-1}  B_{\kappa}^2 \norm{p}_{\mathfrak{S}^1}^2
\right) .
\end{align} 

\end{proof}

The following lemma provides necessary estimates for the propagation of the semiclassical structure.

\begin{lemma}
There exists a constant $C>0$ such that, if $\alpha \in \mathfrak{h}_{1}$,  $p \in \mathfrak{S}^{1}(L^2(\mathbb{R}^3))$ is a rank-$N$ projection, and $q = \id_{L^2(\mathbb{R}^3)} - p$, the following estimates hold:
\begin{subequations}
\begin{align}
\label{eq:estimate for qAp}
\norm{q  \kappa * \vA_{\alpha}p}_{\mathfrak{S}^1}
&\leq  C \sup_{k \in \mathbb{R}^3}
\Big\{ \left(1 + \abs{k} \right)^{-1} \norm{q e^{ikx} p}_{\mathfrak{S}^1} 
\Big\} 
\big\| (1 + \abs{\cdot}^{-1}) \mathcal{F}[\kappa] \big\|_{L^2}
\norm{\alpha}_{\dot{\mathfrak{h}}_{1/2}} , 
\\
\label{eq:estimate for q nabla-Ap}
\norm{q \left[ i \nabla , \kappa * \vA_{\alpha} \right] p}_{\mathfrak{S}^1}
&\leq  C  \sup_{k \in \mathbb{R}^3}
\Big\{ \left(1 + \abs{k} \right)^{-1} \norm{q e^{ikx} p}_{\mathfrak{S}^1} 
\Big\} 
\big\| \abs{\cdot}^{1/2} \mathcal{F}[\kappa] \big\|_{L^2}
\norm{\alpha}_{\mathfrak{h}_{1}} ,
\end{align}
\end{subequations}
and
\begin{align}
\label{eq:semiclassical structure for exponential times magnetic gradient}
&\sup_{k \in \mathbb{R}^3} \Big\{ (1 + \abs{k})^{-1} \norm{q e^{ikx} \left(  i \varepsilon \nabla + \kappa * \vA_{\alpha} \right) p}_{\mathfrak{S}^1}  \Big\}
\nonumber \\
&\quad \leq 
C
\left(  \big\| ( 1 + \abs{\cdot}^{-1} ) \mathcal{F}[\kappa] \big\|_{L^2}  \norm{\alpha}_{\dot{\mathfrak{h}}_{1/2}} 
+\norm{i \varepsilon \nabla p}_{\mathfrak{S}^{\infty}}
\right)
\sup_{k \in \mathbb{R}^3}
\Big\{ \left(1 + \abs{k} \right)^{-1} \norm{q e^{ikx} p}_{\mathfrak{S}^1} 
\Big\} 
\nonumber \\
&\qquad 
+ \norm{q  i \varepsilon \nabla p}_{\mathfrak{S}^1} .
\end{align}
\end{lemma}

\begin{proof}
The inequalities \eqref{eq:estimate for qAp} and \eqref{eq:estimate for q nabla-Ap} are obtained by applying the Cauchy--Schwarz inequality to the expression
\begin{align}
\norm{q  \kappa * \vA_{\alpha}p}_{\mathfrak{S}^1}
&\leq \sup_{k \in \mathbb{R}^3}
\Big\{ \left(1 + \abs{k} \right)^{-1} \norm{q e^{ikx} p}_{\mathfrak{S}^1} 
\Big\} 
\int_{\mathbb{R}^3}
\abs{k}^{-1/2}  \left(1 + \abs{k} \right) \abs{\mathcal{F}[\kappa](k)} \abs{\alpha(k,\lambda)}   \, dk
\end{align}
and
\begin{align}
\norm{q \left[ i \nabla , \kappa * \vA_{\alpha} \right] p}_{\mathfrak{S}^1}
&\leq \sup_{k \in \mathbb{R}^3}
\Big\{ \left(1 + \abs{k} \right)^{-1} \norm{q e^{ikx} p}_{\mathfrak{S}^1} 
\Big\} 
\int_{\mathbb{R}^3}
\abs{k}^{1/2}  \left(1 + \abs{k} \right) \abs{\mathcal{F}[\kappa](k)} \abs{\alpha(k,\lambda)}  
\, dk.
\end{align}
Inserting $\id_{L^2(\mathbb{R}^3)} = p + q$ and using \eqref{eq:estimate for A L-infty to h-1-2} as well as  \eqref{eq:estimate for qAp}, we obtain
\begin{align}
&\norm{q e^{ikx} \left(  i \varepsilon \nabla + \kappa * \vA_{\alpha} \right) p}_{\mathfrak{S}^1}
\nonumber \\
&\quad \leq \norm{q e^{ikx} p \left(  i \varepsilon \nabla + \kappa * \vA_{\alpha} \right) p}_{\mathfrak{S}^1}
+ \norm{q e^{ikx} q \left(  i \varepsilon \nabla + \kappa * \vA_{\alpha} \right) p}_{\mathfrak{S}^1}
\nonumber \\
&\quad \leq \left( \norm{\kappa * \vA_{\alpha}}_{\mathfrak{S}^{\infty}}
+ \norm{i \varepsilon \nabla p_t}_{\mathfrak{S}^{\infty}} \right)
 \norm{q e^{ikx} p}_{\mathfrak{S}^1}
+ \norm{q  i \varepsilon \nabla p}_{\mathfrak{S}^1}
+ \norm{q  \kappa * \vA_{\alpha}  p}_{\mathfrak{S}^1}
\nonumber \\
&\quad \leq \left(  \big\|\abs{\cdot}^{-1} \mathcal{F}[\kappa] \big\|_{L^2} \norm{\alpha}_{\dot{\mathfrak{h}}_{1/2}} 
+  \norm{i \varepsilon \nabla p}_{\mathfrak{S}^{\infty}} \right)
\norm{q e^{ikx} p}_{\mathfrak{S}^1}
+ \norm{q  i \varepsilon \nabla p}_{\mathfrak{S}^1} 
\nonumber \\
&\qquad + 
C B_{\kappa} \sup_{l \in \mathbb{R}^3}
\Big\{ \left(1 + \abs{l} \right)^{-1} \norm{q e^{ilx} p}_{\mathfrak{S}^1} 
\Big\} 
\norm{\alpha}_{\dot{\mathfrak{h}}_{1/2}} 
\nonumber \\
&\quad \leq C
\left(  B_{\kappa}  \norm{\alpha}_{\dot{\mathfrak{h}}_{1/2}} 
+\norm{i \varepsilon \nabla p}_{\mathfrak{S}^{\infty}}
\right)
\bigg[ 
\norm{q e^{ikx} p}_{\mathfrak{S}^1}
+  \sup_{l \in \mathbb{R}^3}
\Big\{ \left(1 + \abs{l} \right)^{-1} \norm{q e^{ilx} p}_{\mathfrak{S}^1} 
\Big\}  \bigg]
\nonumber \\
&\qquad 
+ \norm{q  i \varepsilon \nabla p}_{\mathfrak{S}^1}  ,
\end{align}
from which \eqref{eq:semiclassical structure for exponential times magnetic gradient} follows easily.
\end{proof}

\subsection{Proofs of Inequality \eqref{eq:bound for the h-1 norm of alpha}}

\label{subsection:estimate for the h1-norm of alpha}

\begin{proof}[\unskip\nopunct]

By Duhamel's formula we have that 
\begin{align}
\label{eq:Duhamel expansion of the mode function}
\alpha_t(k,\lambda) = e^{- i \abs{k} t} \alpha_0(k,\lambda)
+ i e^{- i \abs{k} t} \int_0^t 
e^{i \abs{k} s} \sqrt{\frac{4 \pi^3}{\abs{k}}} \mathcal{F}[\kappa](k) \vep_{\lambda}(k) \mathcal{F}[\vJ_{p_s,\alpha_s}](k) \, ds
\end{align}
holds in $\dot{\mathfrak{h}}_{-1/2}$.
Note that if $\alpha \in \dot{\mathfrak{h}}_{-1/2}$ and $\beta \in \mathfrak{h}_{1}$ satisfy $\norm{\alpha - \beta}_{\dot{\mathfrak{h}}_{-1/2}} = 0$, then it implies $\alpha \in \mathfrak{h}_{1}$ and 
$\norm{\alpha - \beta}_{\mathfrak{h}_{1}} = 0$. To prove that $\alpha \in C \left( \mathbb{R}_{+};\mathfrak{h}_{1} \right) \cap C^1 \left( \mathbb{R}_{+};\mathfrak{h} \right)$, it is therefore sufficient to show the right hand side of \eqref{eq:Duhamel expansion of the mode function} is an element of $C \left( \mathbb{R}_{+};\mathfrak{h}_{1} \right) \cap C^1 \left( \mathbb{R}_{+};\mathfrak{h} \right)$. This holds for $t \mapsto e^{-i \abs{\cdot} t} \alpha_0$ because $\alpha_0 \in \mathfrak{h}_1$.
Using $ (p,\alpha) \in
 C \big( \mathbb{R}_{+} ; \mathfrak{S}_{+}^{2,1} ( L^2(\mathbb{R}^3) ) \big) \times  C \big( \mathbb{R}_{+} ; \mathfrak{h}_{1/2} \big)$ and the estimate \cite[(III.28)]{LS2023}, we have 
that $\mathcal{F}[\vJ_{p_t,\alpha_t}] \in L^{\infty}(\mathbb{R}^3)$ and 
\begin{align}
\lim_{h \rightarrow 0} \norm{\mathcal{F}[\vJ_{p_{t+h},\alpha_{t+h}}] - \mathcal{F}[\vJ_{p_t,\alpha_t}]}_{L^{\infty}(\mathbb{R}^3)} =0
\end{align}
for all $t \in \mathbb{R}_{+}$. 
This allows us to estimate
\begin{align}
& \lim_{h \rightarrow 0} \norm{\abs{\cdot}^{-1/2} \mathcal{F}[\kappa] \vep \left( \mathcal{F}[\vJ_{p_{t+h},\alpha_{t+h}}] - \mathcal{F}[\vJ_{p_t,\alpha_t}] \right)}_{\mathfrak{h}_1}
\nonumber \\
&\quad \leq C \big\| ( \abs{\cdot}^{1/2} + \abs{\cdot}^{-1/2} ) \mathcal{F}[\kappa] \big\|_{L^2(\mathbb{R}^3)}
\norm{\mathcal{F}[\vJ_{p_{t+h},\alpha_{t+h}}] - \mathcal{F}[\vJ_{p_t,\alpha_t}]}_{L^{\infty}(\mathbb{R}^3)} =0
\end{align}
and conclude that the integrand of the right-hand side of \eqref{eq:Duhamel expansion of the mode function} is a $\mathfrak{h}_1$-valued continuous function. Consequently, the integral can be defined as a Riemann integral, ensuring that 
\begin{align}
t \mapsto \int_0^t 
e^{i \abs{k} s} \sqrt{\frac{4 \pi^3}{\abs{k}}} \mathcal{F}[\kappa](k) \vep_{\lambda}(k) \mathcal{F}[\vJ_{p_s,\alpha_s}](k) \, ds
\end{align}
is a $C^1\left( \mathbb{R}_+; \mathfrak{h}_1 \right)$ function. This proves that the right-hand side of \eqref{eq:Duhamel expansion of the mode function}, and therefore that $\alpha$, belongs to $C \left( \mathbb{R}_{+};\mathfrak{h}_{1} \right) \cap C^1 \left( \mathbb{R}_{+};\mathfrak{h} \right)$. 
Using \eqref{eq:estimate for A L-infty to h-1-2}, \eqref{eq:estimate for field energy}, \eqref{eq:estimate for kinetic energy}, and
\begin{align}
\label{eq:Fourier transform of the vector current}
\mathcal{F}[ \vJ_{p_t, \alpha_t} ](k)
&= - (2 \pi)^{-3/2} N^{-1}  \tr \left( \left( \left\{ e^{- i k \cdot} , i \varepsilon \nabla \right\} + 2 e^{- i k \cdot} \vAk(\cdot,t) \right) p_t \right) 
\end{align}
we estimate 
\begin{align}
\norm{\mathcal{F}[ \vJ_{p_t, \alpha_t} ]}_{L^{\infty}(\mathbb{R}^3)}
&\leq N^{-1} \norm{i \varepsilon \nabla \sqrt{p_t}}_{\mathfrak{S}^2} \norm{\sqrt{p_t}}_{\mathfrak{S}^2}
+ N^{-1} \norm{\vAk(\cdot,t)}_{L^{\infty}} \norm{p_t}_{\mathfrak{S}^1}
\nonumber \\
&\leq N^{-1} \tr \left( - \varepsilon^2 \Delta p_t \right)
+ \norm{\alpha_t}_{\dot{\mathfrak{h}}_{1/2}}^2  + N^{-2} \big\| \abs{\cdot}^{-1} \mathcal{F}[\kappa] \big\|_{L^2}^2  \norm{p_t}_{\mathfrak{S}^1}^2
\nonumber \\
&\leq 
C  C_{\kappa}^2 N^{-1}
\left< N^{-1}  \norm{p}_{\mathfrak{S}^1} \right> 
\left(  \mathcal{E}^{\rm{MS}}[p,\alpha] + C  C_{\kappa}^2 N^{-1} \norm{p}_{\mathfrak{S}^1}^2
\right)
.
\end{align}
Together with \eqref{eq:Duhamel expansion of the mode function} and \eqref{eq:conservation of mass and energy of the MS equations}, this gives
\begin{align}
\norm{\alpha_t}_{\mathfrak{h}_1}
&\leq \norm{\alpha_0}_{\mathfrak{h}_1}
+ C \int_0^t  \big\| ( \abs{\cdot}^{1/2} + \abs{\cdot}^{-1/2} ) \mathcal{F}[\kappa] \big\|_{L^2(\mathbb{R}^3)} \norm{\mathcal{F}[\vJ_{p_s,\alpha_s}]}_{L^{\infty}(\mathbb{R}^3)} \, ds 
\nonumber \\
&\leq \norm{\alpha_0}_{\mathfrak{h}_1}
+ C  C_{\kappa}^3 N^{-1}
\left< N^{-1}  \norm{p}_{\mathfrak{S}^1} \right> 
\left(  \mathcal{E}^{\rm{MS}}[p,\alpha] + C  C_{\kappa}^2 N^{-1} \norm{p}_{\mathfrak{S}^1}^2
\right) t .
\end{align}
\end{proof}

\subsection{Proof of Lemma~\ref{lemma: semiclassical structure}}

\label{subsection:proof of the propagation of the semiclassical structure}

\begin{proof}[\unskip\nopunct]
In order to prove the estimates of the lemma in a rigorous fashion, we define a regularized version of the gradient by
\begin{align}
\label{eq:definition of the regularized gradient}
i \nabla_{\leq n} &=  \frac{i \nabla}{(1 - \Delta /n^2)^{1/2}}
\quad \text{with}
 \; n \in \mathbb{N} .
\end{align} 
Note that $i \nabla_{\leq n}$ is a bounded and symmetric operator satisfying $\big\| i \nabla_{\leq n} \big\|_{\mathfrak{S}^{\infty}} \leq n$. Moreover,
\begin{align}
\label{eq:difference between regularized and normal gradient acting on L-2 functions}
\lim_{n \rightarrow \infty} \big\| \big( i \nabla - i \nabla_{\leq n} \big) \psi \big\|_{L^2(\mathbb{R}^3)} 
&= 
\lim_{n \rightarrow \infty}
\Big( \int_{\mathbb{R}^3} \abs{k}^2 \Big( 1 - \frac{1}{(1 + \abs{k}^2/n^2 )^{1/2}} \Big)^{2} \abs{\mathcal{F}[\psi](k)}^2 \, dk \Big)^{1/2} 
= 0 
\end{align}
for $\psi \in H^1(\mathbb{R}^3)$ due to monotone convergence. 
In the following, let $\psi \in H^2(\mathbb{R}^3)$ and $(p_t, \alpha_t) \in \mathfrak{S}_{+}^{2,1} (L^2(\mathbb{R}^3)) \times \mathfrak{h}_{1} \cap \dot{\mathfrak{h}}_{-1/2} $ such that $p_t^2 = p_t$ and $\tr \left( p_t \right) = N$. Moreover, let 
\begin{align}
\label{eq:definition mean-field Hamiltonian}
H_{p,\alpha}(t) &= \left( - i \varepsilon \nabla - \kappa * \vA_{\alpha_t} \right)^2 + K * \rho_{p_t} - X_{p_t} .
\end{align}
Then,
\begin{subequations}
\begin{align}
&\lim_{n \rightarrow \infty} \norm{\left[ (i \nabla - i \nabla_{\leq n}) , (H_{p,\alpha}(t) + \varepsilon^2 \Delta) \right] \psi}_{L^2}
\nonumber \\
\label{eq: commutator between difference of regularized and normal gradient with mean field Hamiltonian 1}
&\quad \leq \lim_{n \rightarrow \infty} \norm{\kappa * \vA_{\alpha_t}  (i \nabla - i \nabla_{\leq n})  i \varepsilon \nabla \psi}_{L^2}
\\
\label{eq: commutator between difference of regularized and normal gradient with mean field Hamiltonian 2}
&\qquad + \lim_{n \rightarrow \infty}
\norm{\big( (\kappa * \vA_{\alpha_t})^2 
+ K * \rho_{p_t} - X_{p_t} \big)
(i \nabla - i \nabla_{\leq n})   \psi}_{L^2}
\\
\label{eq: commutator between difference of regularized and normal gradient with mean field Hamiltonian 3}
&\qquad + \lim_{n \rightarrow \infty}
\norm{(i \nabla - i \nabla_{\leq n})  
\big( 2 \kappa * \vA_{\alpha_t} i \varepsilon \nabla + (\kappa * \vA_{\alpha_t})^2 
+ K * \rho_{p_t} - X_{p_t} \big) \psi}_{L^2} .
\end{align}
\end{subequations}
Using $\kappa * \vA_{\alpha_t} \in W_0^{1,\infty}(\mathbb{R}^3)$, $K * \rho_{p_t} \in W_0^{2,\infty}(\mathbb{R}^3)$ and $X_{p_t} \in \mathfrak{S}^{\infty}(L^2(\mathbb{R}^3))$, which follow from \eqref{eq:estimate vector potential W-1-infty norm}, \eqref{eq:estimate for mean-field potential} and \eqref{eq:estimate for exchange term}, we obtain
\begin{align}
\eqref{eq: commutator between difference of regularized and normal gradient with mean field Hamiltonian 1}
+ \eqref{eq: commutator between difference of regularized and normal gradient with mean field Hamiltonian 2}
&\leq \left( 1 + \norm{\kappa * \vA_{\alpha_t}}_{L^{\infty}}^2 +  \norm{K * \rho_{p_t}}_{L^{\infty}}
+ \norm{X_{p_t}}_{\mathfrak{S}^{\infty}}
\right)
\nonumber \\
&\quad \times 
\lim_{n \rightarrow \infty}
\left( \norm{(i \nabla - i \nabla_{\leq n})  \psi}_{L^2}
+ \norm{ (i \nabla - i \nabla_{\leq n})  i \varepsilon \nabla \psi}_{L^2} \right)
\nonumber \\
&= 0 .
\end{align}
With the regularity properties of $\kappa * \vA_{\alpha_t}$ and $K * \rho_{p_t}$ from above and the estimate (see \cite[(II.25) and (III.26)]{LS2023}) 
\begin{align}
\norm{\left( 1 - \Delta \right)^{1/2} X_{p_s} \left(1 - \Delta \right)^{- 1/2}}_{\mathfrak{S}^{\infty}} \leq C N \norm{p_s}_{\mathfrak{S}^{1,1}}
\end{align}
it is straightforward to check that $\big( 2 \kappa * \vA_{\alpha_t} i \varepsilon \nabla + (\kappa * \vA_{\alpha_t})^2 
+ K * \rho_{p_t} - X_{p_t} \big) \psi \in H^1(\mathbb{R}^3)$, which allows us conclude that
$\eqref{eq: commutator between difference of regularized and normal gradient with mean field Hamiltonian 3} = 0
$ and
\begin{align}
\label{eq: commutator between difference of regularized and normal gradient with mean field Hamiltonian}
\lim_{n \rightarrow \infty} \norm{\left[ (i \nabla - i \nabla_{\leq n}) , (H_{p,\alpha}(t) + \varepsilon^2 \Delta) \right] \psi}_{L^2}
= 0 .
\end{align}
Since $p_t \in \mathfrak{S}_{+}^{2,1} (L^2(\mathbb{R}^3))$ is a rank-$N$ projection it
can be written as $p_t = \sum_{l=1}^N \ket{\varphi_l} \bra{\varphi_l}$ with $\varphi_l \in H^2(\mathbb{R}^3)$ and $\norm{\varphi_l}_{L^2(\mathbb{R}^3)} = 1$ for all $l \in \{1, 2, \ldots, N\}$. Hence, 
\begin{align}
\big\| \big( i \nabla - i \nabla_{\leq n} \big) p_t \big\|_{\mathfrak{S}^{\infty}}
&=
\sup_{\psi \in L^2(\mathbb{R}^3), \norm{\psi}_{L^2}= 1} \big\| \big( i \nabla - i \nabla_{\leq n} \big) p_t \psi \big\|_{L^2(\mathbb{R}^3)}
\nonumber \\
&\leq \sup_{\psi \in L^2(\mathbb{R}^3), \norm{\psi}_{L^2}= 1} \sum_{l=1}^N  \big\| \big( i \nabla - i \nabla_{\leq n} \big) \varphi_l \big\|_{L^2(\mathbb{R}^3)} \norm{\varphi_l}_{L^2(\mathbb{R}^3)}
\norm{\psi}_{L^2(\mathbb{R}^3)}
\nonumber \\
&\leq N \sup_{l \in \{1,2, \ldots, N\}} \big\| \big( i \nabla - i \nabla_{\leq n} \big) \varphi_l \big\|_{L^2(\mathbb{R}^3)} .
\end{align}
Thus if we take limit $n \rightarrow \infty$ we get $\lim_{n \rightarrow \infty} \big\| \big( i \nabla - i \nabla_{\leq n} \big) p \big\|_{\mathfrak{S}^{\infty}} = 0$ by means of \eqref{eq:difference between regularized and normal gradient acting on L-2 functions}. Together with the projection property of $p_t$ we get
\begin{align}
\label{eq:limit of the regularized gradient acting on p-t}
\lim_{n \rightarrow \infty} \big\| i  \nabla_{\leq n}  p_t \big\|_{\mathfrak{S}^{\infty}} 
&= \big\|  i  \nabla  p_t \big\|_{\mathfrak{S}^{\infty}}
\quad \text{and} \quad 
\lim_{n \rightarrow \infty} \big\| q_t i  \nabla_{\leq n}  p_t \big\|_{\mathfrak{S}^{1}}
= \big\| q_t i  \nabla  p_t \big\|_{\mathfrak{S}^{1}}
\end{align}
with $q_t = \id_{L^2(\mathbb{R}^3)} - p_t$.  Similar arguments and \eqref{eq: commutator between difference of regularized and normal gradient with mean field Hamiltonian} lead to
\begin{align}
\label{eq:limit of the commutator between the regularized gradient and the mean field Hamiltonian acting on p-t}
\lim_{n \rightarrow \infty} \norm{\left[  i \left( \nabla -  \nabla_{\leq n} \right) , H_{p,\alpha}(t)  \right] p_t}_{\mathfrak{S}^{\infty}}
&= 
0
\quad \text{and} \quad 
\lim_{n \rightarrow \infty} \norm{ q_t\left[  i \left(  \nabla -  \nabla_{\leq n} \right) , H_{p,\alpha}(t)  \right] p_t}_{\mathfrak{S}^{1}}
= 
0 .
\end{align}

\noindent
\textbf{Inequality \eqref{eq:momentum estimate}:} 
Let $(p_t,\alpha_t)$ be the unique solution of \eqref{eq:Maxwell-Schroedinger equations} as specified in Lemma~\ref{lemma: semiclassical structure}, and let $H_{p,\alpha}(t)$ be defined as in \eqref{eq:definition mean-field Hamiltonian}.
There exists a two-parameter family $\{ U_{p,\alpha}(t,s) \}$ of unitary operators on $L^2(\mathbb{R}^3)$, as discussed in \cite[Chapter V.3]{LS2023}, such that $U_{p,\alpha}(t;s) H^2(\mathbb{R}^3) \subseteq H^2(\mathbb{R}^3)$. Furthermore $\psi_s(t) = - i \varepsilon^{-1} U_{p,\alpha}(t;s) \psi$ with $\psi \in H^2(\mathbb{R}^3)$ is strongly continuous differentiable and satisfies
$\frac{d}{dt} \psi_s(t) = - i \varepsilon^{-1} H_{p,\alpha}(t) \psi_s(t)$, $\psi_s(s) = \psi$.
We use the notation $U_{p,\alpha}(t) = U_{p,\alpha}(t;0)$  and calculate
\begin{align}
\label{eq:Duhamel for operator norm of gradient p}
\frac{d}{dt} \left( U^*_{p,\alpha}(t) i \varepsilon \nabla_{\leq n} \, p_t  U_{p,\alpha}(t) \right)
&= i \varepsilon^{-1}  U^*_{p,\alpha}(t) \left[ H_{p,\alpha}(t), i \varepsilon \nabla_{\leq n} \, p_t \right] U_{p,\alpha}(t) 
\nonumber \\
&\quad - i \varepsilon^{-1}  U^*_{p,\alpha}(t) i \varepsilon \nabla_{\leq n} \left[ H_{p,\alpha}(t), p_t \right] U_{p,\alpha}(t)
\nonumber \\
&= i \varepsilon^{-1}  U^*_{p,\alpha}(t) \left[ H_{p,\alpha}(t) + \varepsilon^2 \Delta, i \varepsilon \nabla_{\leq n}  \right] p_t  U_{p,\alpha}(t) .
\end{align}
Integrating with respect to time and taking the trace norm yields
\begin{align}
\norm{i \varepsilon \nabla_{\leq n} p_t  }_{\mathfrak{S}^{\infty}}
&\leq   \norm{i \varepsilon \nabla_{\leq n} p}_{\mathfrak{S}^{\infty}}
+ \int_0^t  \norm{\left[  \nabla_{\leq n} , H_{p,\alpha}(s) + \varepsilon^2 \Delta \right] p_s}_{\mathfrak{S}^{\infty}} \, ds .
\end{align}
Together with \eqref{eq:limit of the regularized gradient acting on p-t} and 
\eqref{eq:limit of the commutator between the regularized gradient and the mean field Hamiltonian acting on p-t}, this leads to 
\begin{subequations}
\begin{align}
&\norm{ i \varepsilon \nabla p_t}_{\mathfrak{S}^{\infty}}
- \norm{ i \varepsilon \nabla p}_{\mathfrak{S}^{\infty}}
\nonumber \\
&\quad =
\lim_{n \rightarrow \infty} 
\left( \norm{ i \varepsilon \nabla_{\leq n} p_t}_{\mathfrak{S}^{\infty}}
- \norm{ i \varepsilon \nabla_{\leq n} p}_{\mathfrak{S}^{\infty}}
\right)
\nonumber \\
&\quad \leq 
\int_0^t \left( 
\norm{\left[  i \nabla  , H_{p,\alpha}(s) + \varepsilon^2 \Delta \right] p_s  }_{\mathfrak{S}^{\infty}} 
+ \lim_{n \rightarrow \infty}
\norm{\left[ i \left(  \nabla -  \nabla_{\leq n} \right) , H_{p,\alpha}(s) + \varepsilon^2 \Delta \right] p_s  }_{\mathfrak{S}^{\infty}} \right)
\, ds 
\nonumber \\
\label{eq:bound for regularized derivative with p 1}
&\quad \leq 
2 \int_0^t
\norm{\left[\nabla , \kappa * \vA_{\alpha_s} \right] \cdot i \varepsilon \nabla p_s  }_{\mathfrak{S}^{\infty}} 
\, ds 
\\
\label{eq:bound for regularized derivative with p 2}
&\qquad + 
 \int_0^t
\Big( \norm{\left[ \nabla , \kappa * \vA_{\alpha_s} \right] \cdot \kappa * \vA_{\alpha_s} p_s  }_{\mathfrak{S}^{\infty}} 
+  \norm{\kappa * \vA_{\alpha_s} \cdot\left[ \nabla , \kappa * \vA_{\alpha_s} \right]  p_s  }_{\mathfrak{S}^{\infty}} \Big) \, ds
\\
\label{eq:bound for regularized derivative with p 3}
&\qquad +   \int_0^t
\norm{\left[ \nabla , K * \rho_{p_s} \right]  p_s  }_{\mathfrak{S}^{\infty}} 
\, ds 
\\
\label{eq:bound for regularized derivative with p 4}
&\qquad + \int_0^t
\norm{\left[ \nabla , X_{p_s} \right] p_s }_{\mathfrak{S}^{\infty}} 
\, ds  .
\end{align}
\end{subequations}
Using \eqref{eq:estimate vector potential W-1-infty norm} and \eqref{eq:estimate for mean-field potential}, we estimate:
\begin{subequations}
\begin{align}
\abs{\eqref{eq:bound for regularized derivative with p 1}}
&\leq 2  \int_0^t  \norm{\kappa * \vA_{\alpha_s}}_{W^{1,\infty}_0(\mathbb{R}^3)} 
\norm{i \varepsilon \nabla p_s}_{\mathfrak{S}^{\infty}} \, ds 
\leq C C_{\kappa} \int_0^t
\norm{\alpha_s}_{\dot{\mathfrak{h}}_{1/2}} \norm{i \varepsilon \nabla p_s}_{\mathfrak{S}^{\infty}} \, ds  ,
\\
\abs{\eqref{eq:bound for regularized derivative with p 2}}
&\leq 
2  \int_0^t 
\norm{ \kappa * \vA_{\alpha_s} }_{W^{1,\infty}_0(\mathbb{R}^3)}^2  \norm{ p_s}_{\mathfrak{S}^{\infty}} \, ds 
\leq C C_{\kappa}  \int_0^t
\norm{\alpha_s}^2_{\dot{\mathfrak{h}}_{1/2}}  \, ds ,
\\
\abs{\eqref{eq:bound for regularized derivative with p 3}}
&\leq   \int_0^t  \norm{ K * \rho_{p_s} }_{W^{1,\infty}_0(\mathbb{R}^3)}
\norm{p_s}_{\mathfrak{S}^{\infty}}
 \, ds 
\leq 
C C_{\kappa}^2  t .
\end{align}
\end{subequations}
By means of 
\begin{align}
K(x-y) &= \frac{1}{(2\pi)^{3/2}} \int_{\mathbb{R}^3} e^{ik (x-y)} \mathcal{F}[K](k) \, dk
\end{align}
we get
\begin{align}
\label{eq:explicite expression for commutator between laplacian and exchange term}
\left[ i  \nabla,  X_{p_s} 
\right]
&= 
\frac{1}{(2 \pi)^{3/2} N} 
\int_{\mathbb{R}^3} \mathcal{F}[K](k) \left[ i  \nabla, e^{i k \hat{x}} p_s e^{- i k \hat{x}} \right] \, dk
\nonumber \\
&= 
\frac{1}{(2 \pi)^{3/2} N} 
\int_{\mathbb{R}^3} \mathcal{F}[K](k) e^{i k \hat{x}} \left[ i  \nabla, p_s \right] e^{- i k \hat{x}}  \, dk
\nonumber \\
&\qquad +
\frac{1}{(2 \pi)^{3/2} N} 
\int_{\mathbb{R}^3} \mathcal{F}[K](k)
\left(  \left[ i  \nabla, e^{i k \hat{x}} \right] p_s e^{- i k \hat{x}} 
+  e^{i k \hat{x}} p_s \left[ i  \nabla, e^{- i k \hat{x}}  \right]  
\right)
\, dk
\nonumber \\
&= 
\frac{1}{(2 \pi)^{3/2} N} 
\int_{\mathbb{R}^3} \mathcal{F}[K](k) e^{i k \hat{x}} \left[ i  \nabla, p_s \right] e^{- i k \hat{x}}  \, dk .
\end{align}
Using the previous equality, \eqref{eq:Fourier transform of the potential}, and $N^{-1} \varepsilon^{-1} \leq 1$, we obtain
\begin{align}
\abs{\eqref{eq:bound for regularized derivative with p 4}}
&\leq N^{-1}  \int_0^t
\int_{\mathbb{R}^3} \abs{\mathcal{F}[K](k)} 
\norm{e^{ik x} \left[ \nabla , p_s \right] e^{- ik x} p_s}_{\mathfrak{S}^{\infty}}  \, dk \, ds
\nonumber \\
&\leq N^{-1} \int_0^t
\int_{\mathbb{R}^3} \abs{\mathcal{F}[K](k)} 
\left(
\norm{\left[ \nabla , e^{- ik x} \right] p_s}_{\mathfrak{S}^{\infty}} 
+ 2 \norm{ \nabla p_s}_{\mathfrak{S}^{\infty}} 
\right)  \, dk \, ds
\nonumber \\
&\leq C N^{-1}  
\big\| (\abs{\cdot}^{-1/2} + \abs{\cdot}^{-1}) \mathcal{F}[\kappa]
 \big\|_{L^2}^2
\int_0^t
\left( \norm{p_s}_{\mathfrak{S}^{\infty}} 
+ \norm{ \nabla p_s}_{\mathfrak{S}^{\infty}}  \right) \, ds
\nonumber \\
&\leq C C_{\kappa}^2
\Big( t
+ \int_0^t \norm{i \varepsilon \nabla p_s}_{\mathfrak{S}^{\infty}}
\, ds \Big) .
\end{align}
Collecting the estimates and using 
\eqref{eq:estimate for field energy}, $\norm{p_s}_{\mathfrak{S}^1} = \tr \left( p_s \right) = N$, and \eqref{eq:conservation of mass and energy of the MS equations},  we obtain
\begin{align}
\norm{i \varepsilon \nabla p_t }_{\mathfrak{S}^{\infty}}
&\leq
\norm{i \varepsilon \nabla p}_{\mathfrak{S}^{\infty}}
+ C \left< C_{\kappa} \right>^2 \int_0^t \left< \norm{\alpha_s}_{\dot{\mathfrak{h}}_{1/2}}^2 \right>
\left< \norm{i \varepsilon \nabla p_s }_{\mathfrak{S}^{\infty}} \right> \, ds
\nonumber \\
&\leq  \norm{i \varepsilon \nabla p}_{\mathfrak{S}^{\infty}}
+ C \left< C_{\kappa} \right>^2
\left(   N^{-1} \mathcal{E}^{\rm{MS}}[p,\alpha] + C C_{\kappa}^2
\right) \left( t + 
\int_0^t
\norm{i \varepsilon \nabla p_s }_{\mathfrak{S}^{\infty}}  \, ds
\right) .
\end{align}
By applying Gr\"onwall's lemma we, finally, obtain \eqref{eq:momentum estimate}.

\noindent
\textbf{Semiclassical structure:}
Let $\nabla_{\leq n}$ with $n \in \mathbb{N}$, $H_{p,\alpha}(t)$, and $U_{p,\alpha}(t)$ be defined as above.  We calculate
\begin{align}
\frac{d}{dt} \big( U^*_{p,\alpha}(t) q_t i \varepsilon \nabla_{\leq n} p_t U_{p,\alpha}(t) \big)
&= i \varepsilon^{-1} U^*_{p,\alpha}(t)
\big[ H_{p,\alpha}(t) , q_t i \varepsilon \nabla_{\leq n} p_t \big]  U_{p,\alpha}(t)
\nonumber \\
&\quad 
- i \varepsilon^{-1} U^*_{p,\alpha}(t)
\left[ H_{p,\alpha}(t), q_t \right] i \varepsilon \nabla_{\leq n} p_t U_{p,\alpha}(t)
\nonumber \\
&\quad 
- i \varepsilon^{-1} U^*_{p,\alpha}(t)
q_t i \varepsilon \nabla_{\leq n} \left[ H_{p,\alpha}(t), p_t \right] U_{p,\alpha}(t)
\nonumber \\
&= i  U^*_{p,\alpha}(t)
q_t \big[ H_{p,\alpha}(t) + \varepsilon \Delta ,  i  \nabla_{\leq n}  \big]  p_t U_{p,\alpha}(t)   ,
\end{align} 
which leads to
\begin{align}
\norm{ q_t i \varepsilon \nabla_{\leq n} p_t}_{\mathfrak{S}^1}
- \norm{ q_0 i \varepsilon \nabla_{\leq n} p_0}_{\mathfrak{S}^1}
&\leq 
2 \int_0^t
\norm{q_s \big[ H_{p,\alpha}(s) ,  i \nabla_{\leq n}  \big] p_s  }_{\mathfrak{S}^1}
\, ds .
\end{align}
Taking the limit $n \rightarrow \infty$ and using \eqref{eq:limit of the regularized gradient acting on p-t} as well as
\eqref{eq:limit of the commutator between the regularized gradient and the mean field Hamiltonian acting on p-t} let us obtain
\begin{subequations}
\begin{align}
&\norm{ q_t i \varepsilon \nabla p_t}_{\mathfrak{S}^1}
- \norm{ q i \varepsilon \nabla p}_{\mathfrak{S}^1}
\nonumber \\
&\quad \leq 
2 \int_0^t
\norm{q_s \big[ H_{p,\alpha}(s) ,  i \nabla  \big] p_s  }_{\mathfrak{S}^1}
\, ds
\nonumber \\
\label{eq:semiclassical structure estimate for nabla 1}
&\quad \leq 
2 \int_0^t
\norm{q_s \left[ i \nabla , \kappa * \vA_{\alpha_s} \right] \cdot i \varepsilon \nabla p_s  }_{\mathfrak{S}^1}
\, ds 
\\
\label{eq:semiclassical structure estimate for nabla 2}
&\qquad +  
 \int_0^t
\Big( \norm{q_s \left[ i \nabla , \kappa * \vA_{\alpha_s} \right] \cdot \kappa * \vA_{\alpha_s} p_s  }_{\mathfrak{S}^1} 
+  \norm{q_s \kappa * \vA_{\alpha_s} \cdot\left[ i \nabla , \kappa * \vA_{\alpha_s} \right]  p_s  }_{\mathfrak{S}^1} \Big) \, ds
\\
\label{eq:semiclassical structure estimate for nabla 3}
&\qquad +  
\int_0^t \norm{q_s \left[ i \nabla , K * \rho_{p_s} \right]  p_s  }_{\mathfrak{S}^1}
\, ds 
\\
\label{eq:semiclassical structure estimate for nabla 4}
&\qquad + 
\int_0^t
\norm{q_s \left[ i  \nabla , X_{p_s} \right] p_s }_{\mathfrak{S}^{1}} 
\, ds  .
\end{align}
\end{subequations}
By means of the identity $\id_{L^2(\mathbb{R}^3)} = p_s + q_s$, \eqref{eq:estimate vector potential W-1-infty norm}, and \eqref{eq:estimate for q nabla-Ap}, we get
\begin{align}
\eqref{eq:semiclassical structure estimate for nabla 1} 
&\leq 
2  \int_0^t 
\Big( 
\norm{q_s \left[ i  \nabla , \kappa * \vA_{\alpha_s} \right] p_s}_{\mathfrak{S}^1}  \norm{ i \varepsilon \nabla p_s  }_{\mathfrak{S}^{\infty}}
+ \norm{\kappa * \vA_{\alpha_s} }_{W^{1,\infty}_0(\mathbb{R}^3)}  \norm{q_s i \varepsilon \nabla p_s  }_{\mathfrak{S}^{1}}
\Big)
\, ds 
\nonumber \\
&\leq C C_{\kappa}
\int_0^t \norm{\alpha_s}_{\mathfrak{h}_1}  \Big( 
\norm{ i \varepsilon \nabla p_s  }_{\mathfrak{S}^{\infty}} \sup_{k \in \mathbb{R}^3}
\Big\{ \left(1 + \abs{k} \right)^{-1} \norm{q_s e^{ikx} p_s}_{\mathfrak{S}^1} 
\Big\} 
+  \norm{q_s i \varepsilon \nabla p_s  }_{\mathfrak{S}^{1}}
\Big) \, ds .
\end{align}
Similarly, 
\begin{align}
\eqref{eq:semiclassical structure estimate for nabla 2} 
&\leq 
2  \int_0^t 
\norm{\kappa * \vA_{\alpha_s}}_{W^{1,\infty}_0(\mathbb{R}^3)} 
\Big( 
\norm{q_s \left[ i  \nabla , \kappa * \vA_{\alpha_s} \right] p_s}_{\mathfrak{S}^1}  
+ \norm{q_s \kappa * \vA_{\alpha_s} p_s  }_{\mathfrak{S}^{1}}
\Big)
\, ds 
\nonumber \\
&\leq C C_{\kappa}^2
\int_0^t \norm{\alpha_s}_{\mathfrak{h}_1}^2  
\sup_{k \in \mathbb{R}^3}
\Big\{ \left(1 + \abs{k} \right)^{-1} \norm{q_s e^{ikx} p_s}_{\mathfrak{S}^1} 
\Big\}  \, ds
\end{align}
by using \eqref{eq:estimate vector potential W-1-infty norm}, \eqref{eq:estimate for qAp} and \eqref{eq:estimate for q nabla-Ap}. Due to \eqref{eq:Fourier transform of the potential} and the fact that
$\norm{\mathcal{F}[\rho_{p_s}]}_{L^{\infty}}
\leq \norm{\rho_{p_s}}_{L^{1}} = 1$
we further obtain
\begin{align}
\eqref{eq:semiclassical structure estimate for nabla 3} 
&\leq 
\int_0^t \int_{\mathbb{R}^3} \abs{k}
\abs{\mathcal{F}[K * \rho_{p_s}](k)}
\norm{q_s e^{ikx} p_s}_{\mathfrak{S}^1} \, dk \, ds
\nonumber \\
&\leq 
\int_0^t \int_{\mathbb{R}^3} \abs{k}
\abs{\mathcal{F}[K](k)} \abs{\mathcal{F}[\rho_{p_s}](k)}
\norm{q_s e^{ikx} p_s}_{\mathfrak{S}^1} \, dk \, ds
\nonumber \\
&\leq C 
\int_0^t  \sup_{k \in \mathbb{R}^3}
\Big\{ \left(1 + \abs{k} \right)^{-1} \norm{q_s e^{ikx} p_s}_{\mathfrak{S}^1} 
\Big\} \, ds
\int_{\mathbb{R}^3} 
\left(1 + \abs{k} \right) \abs{k}^{-1}
\abs{\mathcal{F}[\kappa](k)}^2 \, dk 
\nonumber \\
&\leq  C C_{\kappa}^2
\int_0^t  \sup_{k \in \mathbb{R}^3}
\Big\{ \left(1 + \abs{k} \right)^{-1} \norm{q_s e^{ikx} p_s}_{\mathfrak{S}^1} 
\Big\} \, ds .
\end{align}
Next, we insert \eqref{eq:explicite expression for commutator between laplacian and exchange term} into \eqref{eq:semiclassical structure estimate for nabla 4} and use the identity $\id_{L^2(\mathbb{R}^3)} = p_s + q_s $ as well as \eqref{eq:L-1 norm of the first moment of the potential in terms of kappa}, to get
\begin{align}
\eqref{eq:semiclassical structure estimate for nabla 4}
&\leq  N^{-1} 
\int_0^t  \int_{\mathbb{R}^3} \abs{\mathcal{F}[K](k)} \norm{ q_s e^{i k x} \left[ i \nabla, p_s \right] e^{- i k x} p_s}_{\mathfrak{S}^1}  \, dk \, ds 
\nonumber \\
&\leq  N^{-1} 
\int_0^t  \int_{\mathbb{R}^3} \abs{\mathcal{F}[K](k)} \norm{ q_s e^{i k x} \left( q_s i  \nabla  p_s
- p_s i  \nabla  q_s \right) e^{- i k x} p_s}_{\mathfrak{S}^1}  \, dk \, ds 
\nonumber \\
&\leq C C_{\kappa}^2 \int_0^t   \norm{q_s i \varepsilon \nabla  p_s}_{\mathfrak{S}^1}   
\, ds  .
\end{align}
Collecting the estimates leads to
\begin{align}
\label{eq:semiclassical structure estimate for q gradient p}
\norm{ q_t i \varepsilon \nabla p_t}_{\mathfrak{S}^1}
&\leq  \norm{ q i \varepsilon \nabla p}_{\mathfrak{S}^1}
+ \int_0^t C \left< C_{\kappa} \right>^2
\left< \norm{\alpha_s}_{\mathfrak{h}_1}^2 \right> \left< \norm{i \varepsilon \nabla p_s}_{\mathfrak{S}^{\infty}}  \right> 
\nonumber \\
&\quad
+
\Big( \norm{q_s i \varepsilon \nabla  p_s}_{\mathfrak{S}^1}   
+ \sup_{k \in \mathbb{R}^3}
\Big\{ \left(1 + \abs{k} \right)^{-1} \norm{q_s e^{ikx} p_s}_{\mathfrak{S}^1} 
\Big\} \Big) \, ds  .
\end{align}
Using the notation $\boldsymbol{P}_{\alpha_t} = \left( - i \varepsilon \nabla - \kappa * \vA_{\alpha_t} \right)$, we calculate
\begin{align}
i \varepsilon \partial_t \big( q_t e^{ikx} p_t \big)
&=  \left[ H_{p,\alpha}(t), q_t \right] e^{ikx} p_t
+ q_t e^{ikx} \left[ H_{p,\alpha}(t), p_t \right]
\nonumber \\
&= \big[ H_{p,\alpha}(t) , q_t e^{ikx} p_t \big] 
- q_t \big[ H_{p,\alpha}(t) , e^{ikx} \big] p_t 
\nonumber \\
&= \big[ H_{p,\alpha}(t) , q_t e^{ikx} p_t \big] 
+ q_t \big[ X_{p_t} , e^{ikx} \big] p_t  
+
\varepsilon k \cdot q_t
\big(  \boldsymbol{P}_{\alpha_t} e^{ikx} 
+ e^{ikx} \boldsymbol{P}_{\alpha_t} \big) p_t .
\end{align}
Inserting the identity $\id_{L^2(\mathbb{R}^3)} = p_t + q_t$ gives
\begin{align}
& q_t \big(  \boldsymbol{P}_{\alpha_t} e^{ikx} 
+ e^{ikx} \boldsymbol{P}_{\alpha_t} \big) p_t
\nonumber \\
&\quad = 
q_t \boldsymbol{P}_{\alpha_t} p_t e^{ikx} p_t
+
q_t \boldsymbol{P}_{\alpha_t} q_t e^{ikx} p_t
+ q_t e^{ikx} p_t \boldsymbol{P}_{\alpha_t} p_t 
+ q_t e^{ikx} q_t \boldsymbol{P}_{\alpha_t} p_t 
\nonumber \\
&\quad = 
q_t \boldsymbol{P}_{\alpha_t} p_t e^{ikx} p_t
+
\boldsymbol{P}_{\alpha_t} q_t e^{ikx} p_t
-
p_t \boldsymbol{P}_{\alpha_t} q_t e^{ikx} p_t
+ q_t e^{ikx} p_t \boldsymbol{P}_{\alpha_t} 
- q_t e^{ikx} p_t \boldsymbol{P}_{\alpha_t} q_t 
+ q_t e^{ikx} q_t \boldsymbol{P}_{\alpha_t} p_t 
\nonumber \\
&\quad = 
\boldsymbol{P}_{\alpha_t} q_t e^{ikx} p_t
+ q_t e^{ikx} p_t \boldsymbol{P}_{\alpha_t} 
+ \big( q_t \boldsymbol{P}_{\alpha_t}  p_t - p_t \boldsymbol{P}_{\alpha_t} q_t \big) e^{ikx} p_t
+ q_t e^{ikx} \big( q_t \boldsymbol{P}_{\alpha_t}  p_t - p_t \boldsymbol{P}_{\alpha_t} q_t \big) 
\end{align}
and 
\begin{align}
i \varepsilon \partial_t \big( q_t e^{ikx} p_t \big)
&= \big( H_{p,\alpha}(t) + \varepsilon k \cdot  \boldsymbol{P}_{\alpha_t} \big)   q_t e^{ikx} p_t 
-  q_t e^{ikx} p_t \big( H_{p,\alpha}(t) - \varepsilon k \cdot \boldsymbol{P}_{\alpha_t} \big)
+ q_t \big[ X_{p_t} , e^{ikx} \big] p_t  
\nonumber \\
&\quad +
\varepsilon k \cdot \Big(
\big( q_t \boldsymbol{P}_{\alpha_t}  p_t - p_t \boldsymbol{P}_{\alpha_t} q_t \big) e^{ikx} p_t
+ q_t e^{ikx} \big( q_t \boldsymbol{P}_{\alpha_t}  p_t - p_t \boldsymbol{P}_{\alpha_t} q_t \big) 
\Big) .
\end{align}
Next, we define the unitary propagators $U_{+}$ and $U_{-}$ given by
\begin{align}
\begin{aligned}
\label{eq:propagation of q exponential p in the semiclassical structure definition of the unitaries}
i \varepsilon \partial_t U_{+}(t;s) \varphi &=  \big( H_{p,\alpha}(t) + \varepsilon k \cdot \boldsymbol{P}_{\alpha_t} \big) U_{+}(t;s) \varphi ,
\\
i \varepsilon \partial_t U_{-}(t;s) \varphi &=  \big( H_{p,\alpha}(t) - \varepsilon k \cdot \boldsymbol{P}_{\alpha_t} \big) U_{-}(t;s) \varphi
\end{aligned}
\end{align}
for all $\varphi \in H^2(\mathbb{R}^3)$ with initial conditions
$U_{+}(s;s) = U_{-}(s;s) = 1$.
This leads to
\begin{align}
\label{eq:propagation of q exponential p in the semiclassical structure}
i \varepsilon \partial_t \big( U_{+}^*(t;0) q_t e^{ikx} p_t U_{-}(t;0) \big)
&= 
U_{+}^*(t;0) q_t \big[ X_{p_t} , e^{ikx} \big] p_t   U_{-}(t;0)
\nonumber \\
&\quad +
\varepsilon k \cdot U_{+}^*(t;0) \Big(
\big( q_t \boldsymbol{P}_{\alpha_t}  p_t - p_t \boldsymbol{P}_{\alpha_t} q_t \big) e^{ikx} p_t
\nonumber \\
&\qquad \qquad \qquad   \qquad 
+ q_t e^{ikx} \big( q_t \boldsymbol{P}_{\alpha_t}  p_t - p_t \boldsymbol{P}_{\alpha_t} q_t \big) 
\Big) U_{-}(t;0)
\end{align}
Integrating in time and taking the trace norm allows us to obtain
\begin{subequations}
\begin{align}
\label{eq:semiclassical structure estimate for e-ikx 1}
\norm{q_t e^{ikx} p_t}_{\mathfrak{S}^1}
- \norm{q_0 e^{ikx} p_0}_{\mathfrak{S}^1}
&\leq 
\varepsilon^{-1} \int_0^t \norm{q_s \big[ X_{p_s} , e^{ikx} \big] p_s}_{\mathfrak{S}^1} \, ds 
\\
\label{eq:semiclassical structure estimate for e-ikx 2}
&\quad + 2 \abs{k}
\int_0^t \norm{
q_s \boldsymbol{P}_{\alpha_s}  p_s - p_s \boldsymbol{P}_{\alpha_s} q_s}_{\mathfrak{S}^1} \, ds .
\end{align}
\end{subequations}
By means of
\begin{align}
\label{eq:semiclassical structure commutator between exchange term and exponential}
\left[ X_{p_s} , e^{ikx} 
\right]
&= \frac{1}{(2 \pi)^{3/2} N} 
\int_{\mathbb{R}^3} \mathcal{F}[K](l) e^{i l \hat{x}} \left[ p_s, e^{i k \hat{x}} \right] e^{- i l \hat{x}} \, dl 
\nonumber \\
&= \frac{1}{(2 \pi)^{3/2} N} 
\int_{\mathbb{R}^3} \mathcal{F}[K](l) e^{i l \hat{x}} \left( p_s e^{i k \hat{x}} q_s - q_s e^{i k \hat{x}} p_s \right) e^{- i l \hat{x}} \, dl 
\end{align}
and \eqref{eq:L-1 norm of the first moment of the potential in terms of kappa}, we have
\begin{align}
\abs{\eqref{eq:semiclassical structure estimate for e-ikx 1} }
&\leq \varepsilon^{-1} N^{-1} 
C_{\kappa}^2
 \int_0^t   \sup_{k \in \mathbb{R}^3}
\Big\{ \left(1 + \abs{k} \right)^{-1} \norm{q_s e^{ikx} p_s}_{\mathfrak{S}^1} 
\Big\}  \, ds .
\end{align}
Through the application of \eqref{eq:estimate for qAp},
the second term is bounded by
\begin{align}
\abs{\eqref{eq:semiclassical structure estimate for e-ikx 2} }
&\leq  4 \abs{k}
\int_0^t \norm{
q_s \boldsymbol{P}_{\alpha_s}  p_s }_{\mathfrak{S}^1} \, ds
\nonumber \\
&\leq  4 \abs{k}
\int_0^t
\Big( 
\norm{ q_s i \varepsilon \nabla  p_s }_{\mathfrak{S}^1} 
+ \norm{ q_s\kappa * \vA_{\alpha_s}  p_s }_{\mathfrak{S}^1} 
\Big)
\, ds
\nonumber \\
&\leq  C  C_{\kappa} \abs{k}
\int_0^t
 \sup_{k \in \mathbb{R}^3}
\Big\{ \left(1 + \abs{k} \right)^{-1} \norm{q_s e^{ikx} p_s}_{\mathfrak{S}^1} 
\Big\}  \norm{\alpha_s}_{\dot{\mathfrak{h}}_{1/2}}
\, ds
+ 4 \abs{k}
\int_0^t
\norm{ q_s i \varepsilon \nabla  p_s }_{\mathfrak{S}^1} 
\, ds
\nonumber \\
&\leq  C \left< C_{\kappa} \right>  \abs{k} 
\int_0^t
\left< \norm{\alpha_s}_{\dot{\mathfrak{h}}_{1/2}} \right>
\bigg[
 \sup_{k \in \mathbb{R}^3}
\Big\{ \left(1 + \abs{k} \right)^{-1} \norm{q_s e^{ikx} p_s}_{\mathfrak{S}^1} 
\Big\}  
+
\norm{ q_s i \varepsilon \nabla  p_s }_{\mathfrak{S}^1}  \bigg]
\, ds .
\end{align}
In total, we get
\begin{align}
&\sup_{k \in \mathbb{R}^3}
\Big\{ \left(1 + \abs{k} \right)^{-1} \norm{q_t e^{ikx} p_t}_{\mathfrak{S}^1} 
\Big\}
- \sup_{k \in \mathbb{R}^3}
\Big\{ \left(1 + \abs{k} \right)^{-1} \norm{q_0 e^{ikx} p_0}_{\mathfrak{S}^1} 
\Big\}  
\nonumber \\
&\quad \leq 
C  \left< C_{\kappa} \right> 
\int_0^t
\left< \norm{\alpha_s}_{\dot{\mathfrak{h}}_{1/2}} \right>
\bigg[
\sup_{k \in \mathbb{R}^3}
\Big\{ \left(1 + \abs{k} \right)^{-1} \norm{q_s e^{ikx} p_s}_{\mathfrak{S}^1} 
\Big\} 
+
\norm{ q_s i \varepsilon \nabla  p_s }_{\mathfrak{S}^1}  \bigg]
\, ds .
\end{align}
Combining this estimate with \eqref{eq:semiclassical structure estimate for q gradient p} and applying Gr\"onwall's lemma yields
\begin{align}
\Xi[N,p_t] 
&\leq  \Xi[N,p] 
+ C \left< C_{\kappa} \right>^2 \int_0^t 
\left< \norm{\alpha_s}_{\mathfrak{h}_1}^2 \right> \left< \norm{i \varepsilon \nabla p_s}_{\mathfrak{S}^{\infty}}  \right> 
\Xi[N,p_s]  \, ds 
\end{align}
and
\begin{align}
\Xi[N,p_t]
&\leq \Xi[N,p] \exp \left[ C \left< C_{\kappa} \right>^2 \int_0^t 
\left< \norm{\alpha_s}_{\mathfrak{h}_1}^2 \right> \left< \norm{i \varepsilon \nabla p_s}_{\mathfrak{S}^{\infty}}  \right>  \, ds \right]  ,
\end{align}
with $\Xi[N,p]$ defined as in \eqref{eq:definition function for semiclassical structure}.
Since
\begin{align}
C \left< C_{\kappa} \right>^2 \int_0^t 
\left< \norm{\alpha_s}_{\mathfrak{h}_1}^2 \right> \left< \norm{i \varepsilon \nabla p_s}_{\mathfrak{S}^{\infty}}  \right>  \, ds
&\leq \left< C_{\kappa} \right>^2 \left<
 \norm{\alpha}_{\mathfrak{h}_1}^2
+ \norm{i \varepsilon \nabla p}_{\mathfrak{S}^{\infty}}  \right>
\exp \left[ C[\kappa, p, \alpha] t \right]
\end{align}
which holds due to \eqref{eq:bound for the h-1 norm of alpha} and \eqref{eq:momentum estimate}, this establishes \eqref{eq: semiclassical structure}.
\end{proof}

\section{Many-body estimates}

Within this section, we provide the necessary many-body estimates to prove the main result. We first present general estimates that will be needed later. Afterwards, we prove Inequality~\eqref{eq:relation between Sobolev trace-norm convergence and beta-a plus kinetic energy of the particles outside the Slater determinant}, Lemma~\ref{lemma:energy estimates}, and Lemma~\ref{lemma:energy estimates}.

\subsection{Preliminary estimates}

We begin by stating three technical lemmas, which  previously appeared, among others, in the works \cite{BBPPT2016,PP2016}.
Note that these estimates rely crucially on the antisymmetry of the many-body wave function and do not hold for bosonic systems.

\begin{lemma}\label{lemma:technical estimates using antisymmetry}
Let $N \in \mathbb{N}$ and $j \in \{1, \ldots, N\}$. Let $A_1 = A \otimes \id_{L^2(\mathbb{R}^{3(N-1)})} \otimes \id_{\mathcal{F}}$ with $A \in \mathfrak{S}^1(L^2(\mathbb{R}^3)$ and
$\Psi_{N} , \Psi_{N}' \in L^2(\mathbb{R}^{3N}) \otimes \mathcal{F}$ antisymmetric in $x_1$ and all other electron variables except $x_{l_1}, \ldots, x_{l_j}$. Then,
\begin{align}
\label{eq:technical estimates using antisymmetry 1}
\abs{\SCP{\Psi_{N}}{A_1 \Psi_{N}'}} &\leq (N-j)^{-1} \norm{A}_{\mathfrak{S}^1} \big\|\Psi_{N}\big\|  \big\|\Psi_{N}'\big\|.
\end{align}
The previous inequality implies
\begin{align}
\label{eq:technical estimates using antisymmetry 2}
\norm{A_1 \Psi_N}^2 &\leq (N-j)^{-1} \norm{A}_{\mathfrak{S}^2}^2 \norm{\Psi_N}^2 
\quad \text{for any} \; A \in \mathfrak{S}^2(L^2(\mathbb{R}^3)).
\end{align}
\end{lemma}

\begin{proof}
For the proof of \eqref{eq:technical estimates using antisymmetry 1}, we refer to \cite[Lemma V.2]{LP2019}. The second inequality is obtained by applying the first one to the operator $A_1^2$.
\end{proof}

\begin{lemma}
\label{lemma:diagonalisation of the p-p term}
Let $O: L^2({\mathbb{R}^3}) \rightarrow L^2(\mathbb{R}^3)$ be a self adjoint operator and $p$ be a rank-$N$ projection such that $\tr \left( p O p \right)$. Then, there exists orthonormal functions $\chi_1, \ldots, \chi_N$ in the range of $p$ and real values $\lambda_1, \ldots, \lambda_N$ such that
\begin{subequations}
\begin{align}
\label{eq:properties of the projection operator 1}
p &= \sum_{j=1}^N \ket{\chi_j} \bra{\chi_j}
\quad , \quad 
p O p = \sum_{j=1}^N \lambda_j \ket{\chi_j} \bra{\chi_j}
\\
\label{eq:properties of the projection operator 2}
\lambda_j &= \scp{\chi_j}{O \chi_j}
\quad \text{and} \quad
\sum_{j=1}^N \lambda_j = \tr \left( p O p \right) = \tr \left( O p \right) .
\end{align}
\end{subequations}
Let $J \subseteq \{1,2, \ldots, N \}$ and let $L^2_{J,\rm{as}}(\mathbb{R}^{3N})$ be the linear subspace of $L_{\rm{as}}^2(\mathbb{R}^{3N})$ consisting of wave functions that are antisymmetric in the variables $\{ x_j \}_{j \in J}$. For $m \in \{1, \ldots, N \}$ and $\varphi \in L^2(\mathbb{R}^3)$, let $\ket{\varphi} \bra{\varphi}_m : L^2(\mathbb{R}^{3N}) \rightarrow L^2(\mathbb{R}^{3N})$ 
be given by
\begin{align}
(p_m f)(x_1, \ldots, x_N) &\coloneqq \varphi(x_m) \int_{\mathbb{R}^3} \overline{\varphi(y)}  f(x_1, \ldots,x_{m-1},y, x_{m+1}, \ldots, x_N) \, dy
\end{align}
with $f \in L^2(\mathbb{R}^{3N})$. Define
\begin{align}
P_{J}^{\varphi} &= \sum_{m \in J} \ket{\varphi} \bra{\varphi}_m \otimes \id_{\mathcal{F}}
\quad \text{and} \quad
Q_{J}^{\varphi} = 1 - P_{J}^{\varphi} .
\end{align}
In case $J = \{1, 2, \ldots, N \}$, we use the shorthand notations  $P^{\varphi}$ and $Q^{\varphi}$. The operators
$P_{J}^{\varphi}$ and $Q_{J}^{\varphi}$ are projectors on $L_{J,\rm{as}}^2(\mathbb{R}^{3N})$, and on $L_{J,\rm{as}}^2(\mathbb{R}^{3N})$ it holds that
\begin{align}
\sum_{m \in J} p_m O_m p_m
&= \sum_{m \in J } \sum_{n=1}^N \lambda_n \ket{\chi_{n}} \bra{\chi_n}_{m}  \otimes \id_{\mathcal{F}}
= \sum_{n=1}^N \lambda_n P_{J}^{\chi_n} 
\end{align}
and 
\begin{align}
\label{eq:relation for the sum over all Q-j}
\sum_{n =1}^N Q_{J}^{\chi_n} = (N- \abs{J}) +  \sum_{m \in J} q_m .
\end{align}
\end{lemma}

\begin{proof}
The existence of orthonormal functions $\chi_1, \ldots, \chi_N$ in the range of $p$ and real values $\lambda_1, \ldots, \lambda_N$ such that $p = \sum_{j=1}^N \ket{\chi_j} \bra{\chi_j}$ and $p O p = \sum_{j=1}^N \lambda_j \ket{\chi_j} \bra{\chi_j}$ follows directly from the spectral theorem and the projection property of $p$. Moreover, note that
$\scp{\chi_j}{O \chi_j}
= \scp{\chi_j}{p O p \chi_j} = \lambda_j$
and
\begin{align}
\tr \left( O p \right)
&= \tr \left( p O p \right)
= \sum_{j=1}^N \scp{\chi_j}{O \chi_j}
= \sum_{j=1}^N \lambda_j .
\end{align}
Let $\Psi_{N,J} \in L_{J,\rm{as}}^2(\mathbb{R}^{3N})$. Due to the antisymmetry of $\Psi_{N,J}$, we have
\begin{align}
\left( P_J^{\varphi} \right)^2 \Psi_{N,J}
&= P_J^{\varphi} \Psi_{N,J}
\end{align}
and
\begin{align}
\left( Q_J^{\varphi} \right)^2 \Psi_{N,J}
&= \left( 1 - P_J^{\varphi} \right)^2 \Psi_{N,J}
= \left( 1 -  P_J^{\varphi}  \right) \Psi_{N,J}
= Q_J^{\varphi} \Psi_{N,J} .
\end{align}
Thus, $P_J^{\varphi}$ and $Q_J^{\varphi}$ are projectors on $L_{J,\rm{as}}^2(\mathbb{R}^{3N})$. The relations
\begin{align}
\sum_{m \in J} p_m O_m p_m
&= \sum_{m \in J } \sum_{n=1}^N \lambda_n \ket{\chi_{n}} \bra{\chi_n}_{m} \otimes \id_{\mathcal{F}}
= \sum_{n=1}^N \lambda_n P_{J}^{\chi_n} 
\end{align}
are directly derived from the definitions of the objects involved. The final equality of the lemma follows from
\begin{align}
\sum_{n =1}^N Q_{J}^{\chi_n} 
&= N - \sum_{n =1}^N \sum_{m \in J} \ket{\chi_n} \bra{\chi_n}_m \otimes \id_{\mathcal{F}}
= N -  \sum_{m \in J}  p_m
= (N- \abs{J} ) +  \sum_{m \in J} \left( 1 -  p_m \right)
\nonumber \\
&= (N-\abs{J}) +  \sum_{m \in J} q_m .
\end{align}
\end{proof}
The previous lemma allows us to derive the following result.
\begin{lemma}
\label{lemma:estimates of the p-p term}
Let $O \in \mathfrak{S}^{\infty}(L^2(\mathbb{R}^3))$ be a self-adjoint operator, $J \subseteq \{1,2, \ldots, N \}$, and $j \in J$. Let $L^2_{J,\rm{as}}(\mathbb{R}^{3N})$ denote the linear subspace of $L_{\rm{as}}^2(\mathbb{R}^{3N})$ consisting of wave functions that are antisymmetric with respect to the variables $\{ x_j \}_{j \in J}$. Suppose $\Psi_{N,J}, \Psi'_{N,J} \in L^2_{J,\rm{as}}(\mathbb{R}^{3N})$. Then,
\begin{align}
\label{eq:estimates of the p-p term 1}
&\Big| \scp{\Psi_{N,J}}{\Big(  \sum_{j \in J} p_j O_j p_j  - \tr \left( O p \right) \Big) \Psi'_{N,J} } \Big|
\nonumber \\
&\quad \leq 
\norm{pO p}_{\mathfrak{S}^{\infty}}
\scp{\Psi_{N,J}}{\big( (N-\abs{J}) + \abs{J} q_j \big) \Psi_{N,J} }^{1/2}
\scp{\Psi'_{N,J}}{\big( (N-\abs{J}) + \abs{J} q_j \big) \Psi'_{N,J} }^{1/2} .
\end{align}
For $\Psi_N, \Psi'_N \in L_{\rm{as}}^2(\mathbb{R}^{3N})$ we also have
\begin{align}
\label{eq:estimates of the p-p term 2}
&\Big| \scp{\Psi_{N}}{\Big( \sum_{j=2}^N p_{j}  
O_j p_{j} - \tr \big( O p \big)
\Big) q_{1} e^{ik x_1} p_{1}  \Psi'_{N}} \Big|
\nonumber \\
&\quad \leq  \norm{p O p}_{\mathfrak{S}^{\infty}}
\norm{q e^{ikx} p}_{\mathfrak{S}^1}
\left( \norm{q_1 \Psi_N} + N^{-1/2} \norm{\Psi_N} \right)
\left( \norm{q_1 \Psi'_N} + N^{-1/2} \norm{\Psi'_N} \right) . 
\end{align}
\end{lemma}

\begin{proof}
Using Lemma~\ref{lemma:diagonalisation of the p-p term} we obtain the spectral decomposition $\{ \chi_n , \lambda_{n} \}_{n \in \mathbb{N}}$ of the operator $p O p$, which gives rise to the following estimates
\begin{align}
&\abs{\scp{\Psi_{N,J}}{\Big(  \sum_{j \in J} p_j O_j p_j  - \tr \left( O p \right) \Big) \Psi'_{N,J} } }
\nonumber \\
&\quad = \abs{ \scp{\Psi_{N,J}}{\sum_{n=1}^N  \lambda_n \left( 1 - P_{J}^{\chi_n} \right) \Psi'_{N,J}} }
\nonumber \\
&\quad = \abs{ \scp{\Psi_{N,J}}{\sum_{n=1}^N  \lambda_n Q_{J}^{\chi_n}  \Psi'_{N,J}} }
\nonumber \\
&\quad \leq  \sup_{j \in \{1,\ldots, N \} }
\big\{ \abs{\lambda_j} \big\} \sum_{n=1}^N   \norm{Q_{J}^{\chi_n}  \Psi_{N,J}}
\norm{Q_{J}^{\chi_n}  \Psi'_{N,J}}
\nonumber \\
&\quad \leq
\norm{p O p}_{\mathfrak{S}^{\infty}}
\scp{\Psi_{N,J}}{\sum_{n=1}^N Q_{J}^{\chi_n} \Psi_{N,J} }^{1/2}
\scp{\Psi'_{N,J}}{\sum_{n=1}^N Q_{J}^{\chi_n} \Psi'_{N,J} }^{1/2}
\nonumber \\
&\quad \leq
\norm{p O p}_{\mathfrak{S}^{\infty}}
\scp{\Psi_{N,J}}{\big( (N-\abs{J}) + \sum_{m \in J} q_m \big) \Psi_{N,J} }^{1/2}
\scp{\Psi'_{N,J}}{\big( (N-\abs{J}) + \sum_{m \in J} q_m \big) \Psi'_{N,J} }^{1/2}
\nonumber \\
&\quad \leq
\norm{p O p}_{\mathfrak{S}^{\infty}}
\scp{\Psi_{N,J}}{\big( (N-\abs{J}) + \abs{J} q_j \big) \Psi_{N,J} }^{1/2}
\scp{\Psi'_{N,J}}{\big( (N-\abs{J}) + \abs{J} q_j \big) \Psi'_{N,J} }^{1/2} .
\end{align} 
Here, we have used that 
\begin{align}
\sup_{j \in \{1,\ldots, N \} }
\big\{ \abs{\lambda_j} \big\}
= \sup_{j \in \{1,\ldots, N \} }
\big\{ \abs{\scp{\chi_j}{O \chi_j}} \big\}
= \sup_{j \in \{1,\ldots, N \} }
\big\{ \abs{\scp{\chi_j}{p O p \chi_j}} \big\}
\leq \norm{pO p}_{\mathfrak{S}^{\infty}} .
\end{align}
Let $\widetilde{J} = \{ 2,3, \ldots , N \}$.  Using the spectral decomposition $\{ \chi_n , \lambda_{n} \}_{n \in \mathbb{N}}$ of the operator $p O p$ again we get
\begin{align}
\scp{\Psi_{N}}{\Big( \sum_{j=2}^N p_{j}  
O_j p_{j} - \tr \big( O p \big)
\Big) q_{1} e^{ik x_1} p_{1}  \Psi'_{N}}
&=
\scp{\Psi_{N}}{\sum_{n=1}^N   \lambda_n Q_{\widetilde{J}}^{\chi_n}  q_1 e^{ik x_1} p_{1}  \Psi'_{N}}  .
\end{align}
Note that $Q_{\widetilde{J}}^{\chi_n}$ is a projection on $L_{\widetilde{J},\rm{as}}^2(\mathbb{R}^{3N})$ and 
$\left[ Q_{\widetilde{J}}^{\chi_n} , q_1 e^{ik x_1} p_{1} \right] =0$. Using the singular value decomposition, we have $q_1 e^{i k x_1} p_1 = \sum_{l \in \mathbb{N}} \mu_l \ket{\phi_l} \bra{\widetilde{\phi}_l}_{1}$,
where $\mu_l \geq 0$
for all $l \in \mathbb{N}$ and $\sum_{l \in \mathbb{N}} \mu_l = \norm{q e^{ikx} p}_{\mathfrak{S}^1}$.
Combining this with
$\sup_{n \in \{1,\ldots, N \} }
\big\{ \abs{\lambda_n} \big\}
\leq \norm{pO p}_{\mathfrak{S}^{\infty}}$,   we obtain
\begin{align}
&\Big|\scp{\Psi_{N}}{\Big( \sum_{j=2}^N p_{j}  
O_j p_{j} - \tr \big( O p \big)
\Big) q_{1} e^{ik x_1} p_{1}  \Psi'_{N}} \Big|
\nonumber \\
&\quad =
\Big| \scp{\Psi_{N}}{\sum_{n=1}^N   \lambda_n Q_{\widetilde{J}}^{\chi_n}  \sum_{l \in \mathbb{N}} \mu_l \ket{\phi_l} \bra{\widetilde{\phi}_l}_{1} Q_{\widetilde{J}}^{\chi_n} \Psi'_{N}} \Big|
\nonumber \\
&\quad \leq 
\sum_{l \in \mathbb{N}} \mu_l \sum_{n=1}^N \abs{\lambda_n}
\abs{\scp{Q_{\widetilde{J}}^{\chi_n} \Psi_{N}}{ \ket{\phi_l} \bra{\widetilde{\phi}_l}_{1} Q_{\widetilde{J}}^{\chi_n} \Psi'_{N}}}
\nonumber \\
&\quad \leq 
\sum_{l \in \mathbb{N}} \mu_l \sum_{n=1}^N \abs{\lambda_n}
\norm{\bra{\phi_l}_1 Q_{\widetilde{J}}^{\chi_n} \Psi_{N}}_{L^2(\mathbb{R}^{3(N-1)} \otimes \mathcal{F} )}
\norm{\bra{\widetilde{\phi}_l}_{1} Q_{\widetilde{J}}^{\chi_n} \Psi'_{N}}_{L^2(\mathbb{R}^{3(N-1)} \otimes \mathcal{F} )}
\nonumber \\
&\quad \leq \norm{p O p}_{\mathfrak{S}^{\infty}}
\sum_{l \in \mathbb{N}} \mu_l \sum_{n=1}^N
\norm{\bra{\phi_l}_1 Q_{\widetilde{J}}^{\chi_n} \Psi_{N}}_{L^2(\mathbb{R}^{3(N-1)} \otimes \mathcal{F} )}
\norm{\bra{\widetilde{\phi}_l}_{1} Q_{\widetilde{J}}^{\chi_n} \Psi'_{N}}_{L^2(\mathbb{R}^{3(N-1)} \otimes \mathcal{F} )}
 .
\end{align}
Using Lemma~\ref{lemma:diagonalisation of the p-p term}, \eqref{eq:technical estimates using antisymmetry 1} 
and $\norm{\ket{\phi_l} \bra{\phi_l}}_{\mathfrak{S}^1} = \norm{\phi_l}_{L^2}^2 = 1$, we obtain
\begin{align}
\sum_{n=1}^N \norm{\bra{\phi_l}_1 Q_{\widetilde{J}}^{\chi_n} \Psi_{N}}^2_{L^2(\mathbb{R}^{3(N-1)} \otimes \mathcal{F} )}
&= \sum_{n=1}^N \scp{Q_{\widetilde{J}}^{\chi_n} \Psi_{N}}{\ket{\phi_l} \bra{\phi_l}_1 Q_{\widetilde{J}}^{\chi_n} \Psi_{N}}
\nonumber \\
&=  \scp{\Psi_{N}}{\ket{\phi_l} \bra{\phi_l}_1 \sum_{n=1}^N Q_{\widetilde{J}}^{\chi_n} \Psi_{N}}
\nonumber \\
&=  \scp{\Psi_{N}}{\ket{\phi_l} \bra{\phi_l}_1 \Big( 1 + \sum_{j=2}^N q_j \Big) \Psi_{N}}
\nonumber \\
&=  \scp{\Psi_{N}}{\ket{\phi_l} \bra{\phi_l}_1 \Psi_{N}}
+ (N-1) \scp{q_2 \Psi_{N}}{\ket{\phi_l} \bra{\phi_l}_1 q_2 \Psi_{N}}
\nonumber \\
&\leq  \norm{\ket{\phi_l} \bra{\phi_l}}_{\mathfrak{S}^1}
\left( N^{-1} \norm{\Psi_N}^2 + \norm{q_2 \Psi_N}^2 \right)
\nonumber \\
&\leq   N^{-1} \norm{\Psi_N}^2 + \norm{q_1 \Psi_N}^2 .
\end{align}
This and the Cauchy--Schwarz inequality lead to
\begin{align}
&\Big| \scp{\Psi_{N}}{\Big( \sum_{j=2}^N p_{j}  
O_j p_{j} - \tr \big( O p \big)
\Big) q_{1} e^{ik x_1} p_{1}  \Psi_{N}} \Big|
\nonumber \\
&\quad \leq  \norm{p O p}_{\mathfrak{S}^{\infty}}
\sum_{l \in \mathbb{N}} \mu_l
\left( \norm{q_1 \Psi_N} + N^{-1/2} \norm{\Psi_N} \right)
\left( \norm{q_1 \Psi'_N} + N^{-1/2} \norm{\Psi'_N} \right) .
\end{align}
Since $\sum_{l \in \mathbb{N}} \mu_l = \norm{q e^{ikx} p}_{\mathfrak{S}^1}$, this proves the claim.

\end{proof}

Next, we provide more details about operators in Fock space and derive estimates for the quantized vector potential, which will become important later. We introduce the usual bosonic annihilation and creation operators
\begin{align}
a(f) = \sum_{\lambda =1,2} \int_{\mathbb{R}^3} \overline{f(k,\lambda)} a_{k,\lambda}
\quad \text{and} \quad 
a^*(f) = \sum_{\lambda =1,2} \int_{\mathbb{R}^3} f(k,\lambda) a^*_{k,\lambda}
\quad \text{with} \; f \in \mathfrak{h} .
\end{align}
They satisfy the estimates (see, e.g. \cite[Section 2]{RS2009})
\begin{align}
\label{eq:standard estimate annihilation and creation operators}
\norm{a(f) \Psi_N} &\leq \norm{f}_{\mathfrak{h}} \norm{\mathcal{N}^{1/2} \Psi_N}
\quad \text{and} \quad
\norm{a^*(f) \Psi_N} \leq \norm{f}_{\mathfrak{h}} \norm{\left( \mathcal{N} + 1 \right)^{1/2} \Psi_N} .
\end{align}
Moreover, note that the number operator $\mathcal{N}$, defined in \eqref{eq:definition number operator}, satisfies
\begin{align}
\label{eq:shifting property number operator}
\left( \mathcal{N} + 1 \right)^{m} a_{k,\lambda} = a_{k,\lambda} \mathcal{N} 
\quad \text{and} \quad 
 a^*_{k,\lambda} = a^*_{k,\lambda} \left( \mathcal{N} + 1 \right)^{m}
 \quad \text{for all} \; m \in \mathbb{R} ,
\end{align}
as can be verified by the operator's action on elements of Fock space. 
In addition, we define 
\begin{align}
\label{eq:definition of G}
\boldsymbol{G}_x(k,\lambda) &=  \mathcal{F}[\kappa](k) \frac{1}{\sqrt{2 \abs{k} } } \vep_{\lambda}(k) e^{-ikx} ,
\end{align}
which allows us to write the quantized vector potential with ultraviolet cutoff as 
\begin{align}
\label{eq:definition quantized vector potential with cutoff shorthand notation}
\vAcq(x)
&\coloneqq
(\kappa * \hat{\vA})(x)
= \sum_{\lambda=1,2} \int_{\mathbb{R}^3}  
\frac{\mathcal{F}[\kappa](k)}{\sqrt{2 \abs{k}}} \vep_{\lambda}(k) \left(  e^{ikx} a_{k,\lambda}  + e^{- ik x} a^*_{k,\lambda} \right)\, dk
= a(\boldsymbol{G}_x)
+ a^*(\boldsymbol{G}_x) .
\end{align}
Here, $a(\boldsymbol{G}_x)$ is shorthand notation for $\big( a(G^1_x), a(G^2_x), a(G^3_x) \big)$.
In the following, it will be convenient to split the vector potential into its positive and negative frequency parts
\begin{align}
\label{eq:splitting of the vector potential}
\vAp(x) &\coloneqq a(\boldsymbol{G}_x) ,  
\quad 
\vAm(x) \coloneqq a^*(\boldsymbol{G}_x) .
\end{align}
From \eqref{eq:standard estimate annihilation and creation operators}, we obtain
\begin{align}
\begin{aligned}
\label{eq:Bounds for vector potential}
\sup_{x \in \mathbb{R}^3} \norm{ \vAp(x) \Psi_{N}} &\leq \big\|\abs{\cdot}^{-1/2} \mathcal{F}[\kappa] \big\|_{L^2(\mathbb{R}^3)} \norm{\mathcal{N}^{1/2} \Psi_N} ,
\\
\sup_{x \in \mathbb{R}^3} \norm{ \vAm(x) \Psi_{N}} &\leq \big\|\abs{\cdot}^{-1/2} \mathcal{F}[\kappa] \big\|_{L^2(\mathbb{R}^3)} \norm{\left( \mathcal{N} + 1\right)^{1/2} \Psi_N} .
\end{aligned}
\end{align}
Using the antisymmetry of the wave function we, moreover, obtain the following estimates.
\begin{lemma}
Let $\Psi_{N}, \Psi_{N}' \in L^2(\mathbb{R}^{3N}) \otimes \mathcal{F}$ antisymmetric in $x_1$ and all other electron variables except $x_{l_1}, \ldots, x_{l_j}$. Let $O_1 = O \otimes \id_{L^2(\mathbb{R}^{3(N-1)})} \otimes \id_{\mathcal{F}}$ with $O: L^2({\mathbb{R}^3}) \rightarrow L^2(\mathbb{R}^3)$
and let 
\begin{align}
\mathcal{K} 
&= (N-j)^{-1} \sup_{k \in \mathbb{R}^3} \Big\{ (1 + \abs{k})^{-1} \norm{q e^{ikx} O p}_{\mathfrak{S}^1}  \Big\}  \norm{(\abs{\cdot}^{1/2} + \abs{\cdot}^{-1/2}) \mathcal{F}[\kappa]  }_{L^2}
\end{align} 
For $m \in \mathbb{R}$ we have
\begin{subequations}
\begin{align}
\label{eq:estimate vector potential positive part with p and q within scalar product}
\abs{\scp{\Psi_N}{q_{1}  \vAp(x_1) O_1 p_{1} \Psi_N'} }
&\leq C \mathcal{K}  
\norm{\left( \mathcal{N} + 1 \right)^{m} \Psi_N}
\norm{\left( \mathcal{N} + 1 \right)^{-m} \mathcal{N}^{1/2} \Psi_N'}
,
\\
\label{eq:estimate vector potential negative part with p and q within scalar product}
\abs{\scp{\Psi_N}{q_{1}  \vAm(x_1) O_1 p_{1} \Psi_N'} }
&\leq C \mathcal{K} 
\norm{\left( \mathcal{N} + 1 \right)^{m} \mathcal{N}^{1/2} \Psi_N}
\norm{\left( \mathcal{N} + 1 \right)^{-m} \Psi_N'} ,
\\
\label{eq:estimate vector potential with p and q within scalar product}
\abs{\scp{\Psi_N}{q_{1}  \vAhk(x_1) O_1 p_{1} \Psi_N'} }
&\leq C \mathcal{K} 
\norm{\left( \mathcal{N} + 1 \right)^{m} \Psi_N}
\norm{\left( \mathcal{N} + 1 \right)^{-m} \left( \mathcal{N} + 1 \right)^{1/2} \Psi_N'} ,
\\
\label{eq:estimate vector potential squared with p and q within scalar product}
\abs{\scp{\Psi_N}{q_{1}  \vAhk^2(x_1) p_{1} \Psi_N'} }
&\leq C \norm{(\abs{\cdot}^{1/2} + \abs{\cdot}^{-1/2}) \mathcal{F}[\kappa] }_{L^2}^2  \norm{\left( \mathcal{N} + 1 \right)^{m} \Psi_N}
\norm{\left( \mathcal{N} + 1 \right)^{-m +1}  \Psi_N'} 
\nonumber \\
&\quad \times
(N-j)^{-1} \sup_{k \in \mathbb{R}^3} \Big\{ (1 + \abs{k})^{-1} \norm{q e^{ikx}  p}_{\mathfrak{S}^1}  \Big\}  .
\end{align}
\end{subequations}
\end{lemma}

\begin{proof}
Using \eqref{eq:shifting property number operator}, \eqref{eq:technical estimates using antisymmetry 1}
and the Cauchy--Schwarz inequality let us obtain
\begin{align}
&\abs{\scp{\Psi_N}{q_{1}  \vAp(x_1) O_1 p_{1} \Psi_N'} }
\nonumber \\
&\quad \leq \sum_{\lambda = 1,2} \int_{\mathbb{R}^3} \abs{k}^{-1/2} \abs{\mathcal{F}[\kappa](k)}
\abs{\scp{\left( \mathcal{N} + 2 \right)^{m} \Psi_N}{q_{1} e^{ik x_1}  O_1 p_{1} a_{k,\lambda}  \left( \mathcal{N} + 1 \right)^{-m}  \Psi_N'}} \, dk
\nonumber \\
&\quad \leq (N-j)^{-1}  \sum_{\lambda = 1,2} \int_{\mathbb{R}^3} \abs{k}^{-1/2} \abs{\mathcal{F}[\kappa](k)}
\norm{q e^{ik \cdot}  O p}_{\mathfrak{S}^1}
\nonumber \\
&\qquad \qquad \times
\norm{\left( \mathcal{N} + 2 \right)^{m} \Psi_N} \norm{ a_{k,\lambda}  \left( \mathcal{N} + 1 \right)^{-m} \Psi_N'} \, dk
\nonumber \\
&\leq C (N-j)^{-1} 
\sup_{k \in \mathbb{R}^3} \Big\{ (1 + \abs{k})^{-1} \norm{q e^{ikx} O p}_{\mathfrak{S}^1}  \Big\}
\sum_{\lambda = 1,2} \int_{\mathbb{R}^3} \abs{k}^{-1/2} (1 + \abs{k}) \abs{\mathcal{F}[\kappa](k)}
\nonumber \\
&\qquad \qquad \times
\norm{\left( \mathcal{N} + 2 \right)^{m} \Psi_N} \norm{ a_{k,\lambda}  \left( \mathcal{N} + 1 \right)^{-m}  \Psi_N'} \, dk
\nonumber \\
&\leq C\mathcal{K}
\norm{\left( \mathcal{N} + 1 \right)^{m} \Psi_N} \norm{  \left( \mathcal{N} + 1 \right)^{-m} \mathcal{N}^{1/2} \Psi_N'} \, dk .
\end{align}
Inequality \eqref{eq:estimate vector potential negative part with p and q within scalar product} is obtained by similar means and \eqref{eq:estimate vector potential with p and q within scalar product} is a direct consequence of $\vAk = \vAp + \vAm$, \eqref{eq:estimate vector potential positive part with p and q within scalar product} and \eqref{eq:estimate vector potential negative part with p and q within scalar product} with $m \mapsto m - \frac{1}{2}$.
Using \eqref{eq:CCR} we get
\begin{align}
\vAh^2(x_1)
&=  (\vAp(x_1))^2  + (\vAm(x_1))^2 + \vAm(x_1) \vAp(x_1) + \vAp(x_1) \vAm(x_1) 
\nonumber \\
&=  (\vAp(x_1))^2  + (\vAm(x_1))^2 + 2 \vAm(x_1) \vAp(x_1) + \big\| \abs{\cdot}^{-1/2} \mathcal{F}[\kappa] \big\|_{L^2}^2 .
\end{align}
Together with $q_1p_1 = 0$ this leads to
\begin{align}
\label{eq:estimate for vector potential squared between q and p proof}
\abs{\scp{\Psi_N}{q_{1}  \vAhk^2(x_1) p_{1} \Psi_N'} }
&\leq \abs{\scp{\Psi_N}{q_{1}  (\vAp(x_1))^2 p_{1} \Psi_N'} }
+ \abs{\scp{\Psi_N}{q_{1}  (\vAm(x_1))^2 p_{1} \Psi_N'} }
\nonumber \\
&\quad
+ 2 \abs{\scp{\Psi_N}{q_{1}  \vAm(x_1) \vAp(x_1) p_{1} \Psi_N'} } .
\end{align}
By means of \eqref{eq:shifting property number operator}, \eqref{eq:technical estimates using antisymmetry 1},
and the Cauchy--Schwarz inequality we get
\begin{align}
&\abs{\scp{\Psi_N}{q_{1}  (\vAp(x_1))^2 p_{1} \Psi_N'}}
\nonumber \\
&\quad \leq
\scp{\left( \mathcal{N} + 3 \right)^{m} \Psi_N}{q_{1}  (\vAp(x_1))^2 p_{1} \left( \mathcal{N} + 1 \right)^{-m}  \Psi_N'}
\nonumber \\
&\quad \leq \sum_{\lambda, \lambda' = 1,2} \int_{\mathbb{R}^3 \times \mathbb{R}^3} \abs{k}^{-1/2} \abs{\mathcal{F}[\kappa](k)}
\abs{k'}^{-1/2} \abs{\mathcal{F}[\kappa](k')}
\nonumber \\
&\qquad \qquad \times
\abs{\scp{\left( \mathcal{N} + 3 \right)^{m} \Psi_N}{q_{1} e^{ik x_1} e^{ik' x_1}  p_{1} a_{k,\lambda} a_{k',\lambda'}  \left( \mathcal{N} + 1 \right)^{-m}  \Psi_N'}} \, dk \, dk'
\nonumber \\
&\leq (N-j)^{-1}
 \sum_{\lambda, \lambda' = 1,2} \int_{\mathbb{R}^3 \times \mathbb{R}^3} \abs{k}^{-1/2} \abs{\mathcal{F}[\kappa](k)}
\abs{k'}^{-1/2} \abs{\mathcal{F}[\kappa](k')}
\norm{q e^{ik \cdot} e^{ik' \cdot}  p}_{\mathfrak{S}^1}
\nonumber \\
&\qquad \qquad \times
\norm{\left( \mathcal{N} + 3 \right)^{m} \Psi_N} \norm{ a_{k,\lambda} a_{k',\lambda'}  \left( \mathcal{N} + 1 \right)^{-m}  \Psi_N'} \, dk \, dk'
\nonumber \\
&\leq (N-j)^{-1}
\sup_{k, k' \in \mathbb{R}^3} \Big\{ (1 + \abs{k})^{-1} (1 + \abs{k'})^{-1} \norm{q e^{ikx } e^{ik'x} p}_{\mathfrak{S}^1}  \Big\}
\nonumber \\
&\qquad \qquad \times
 \sum_{\lambda, \lambda' = 1,2} \int_{\mathbb{R}^3 \times \mathbb{R}^3} \abs{k}^{-1/2}  (1 + \abs{k}) \abs{\mathcal{F}[\kappa](k)}
\abs{k'}^{-1/2} 
 (1 + \abs{k'}) \abs{\mathcal{F}[\kappa](k')}
\nonumber \\
&\qquad \qquad \times
\norm{\left( \mathcal{N} + 3 \right)^{m} \Psi_N} \norm{ a_{k,\lambda} a_{k',\lambda'}  \left( \mathcal{N} + 1 \right)^{-m}  \Psi_N'} \, dk \, dk'
\nonumber \\
&\leq C (N-j)^{-1}
\sup_{k, k' \in \mathbb{R}^3} \Big\{ (1 + \abs{k})^{-1} (1 + \abs{k'})^{-1} \norm{q e^{ikx } e^{ik'x} p}_{\mathfrak{S}^1}  \Big\} \norm{(\abs{\cdot}^{1/2} + \abs{\cdot}^{-1/2}) \mathcal{F}[\kappa] }_{L^2}^2 
\nonumber \\
&\qquad \qquad \times 
\norm{\left( \mathcal{N} + 1 \right)^{m} \Psi_N} \norm{  \left( \mathcal{N} + 1 \right)^{-m}  \left( \mathcal{N} + 1 \right) \Psi_N'} .
\end{align}
Inserting the identity $\id_{L^2(\mathbb{R}^3)} = p + q$ gives
$\norm{q e^{ikx } e^{ik'x} p}_{\mathfrak{S}^1}
\leq \norm{q e^{ikx } p}_{\mathfrak{S}^1}
+ \norm{q e^{ik'x} p}_{\mathfrak{S}^1}$,
which leads to
\begin{align}
\abs{\scp{\Psi_N}{q_{1}  (\vAp(x_1))^2 p_{1} \Psi_N'}}
&\leq C  \norm{(\abs{\cdot}^{1/2} + \abs{\cdot}^{-1/2}) \mathcal{F}[\kappa] }_{L^2}^2 
\norm{\left( \mathcal{N} + 1 \right)^{m} \Psi_N} \norm{  \left( \mathcal{N} + 1 \right)^{-m +1}  \Psi_N'} 
\nonumber \\
&\qquad \qquad \times 
(N-j)^{-1}
\sup_{k \in \mathbb{R}^3} \Big\{ (1 + \abs{k})^{-1}  \norm{q e^{ikx } p}_{\mathfrak{S}^1}  \Big\} .
\end{align}
The two remaining terms on the right-hand side of \label{eq:estimate for vector potential squared between q and p proof} are estimated by almost the same means. The main difference is that the annihilation and creation operators must be distributed differently in the scalar product and $m$ chosen accordingly. In total this shows \eqref{eq:estimate vector potential with p and q within scalar product}.
\end{proof}

\subsection{Proof of inequality~\eqref{eq:relation between Sobolev trace-norm convergence and beta-a plus kinetic energy of the particles outside the Slater determinant}}

\label{section:Proof of Lemma with relation between trace norm convergence and beta-a}

\begin{proof}[\unskip\nopunct]

Note that 
\begin{align}
&\norm{\sqrt{1 - \varepsilon^2 \Delta} \left( \gamma^{(1,0)}_{\Psi_N} - N^{-1} p \right) \sqrt{1 - \varepsilon^2 \Delta}}_{\mathfrak{S}^1}
\nonumber \\
&\quad = \sup_{\substack{A \in \mathfrak{S}^{\infty}(L^2(\mathbb{R}^3)) \\ \norm{A}_{\mathfrak{S}^{\infty} = 1}}}
\abs{ \tr \left(  A \sqrt{1 - \varepsilon^2 \Delta} \left( \gamma^{(1,0)}_{\Psi_N} - N^{-1} p \right) \sqrt{1 - \varepsilon^2 \Delta} \right)
}
\end{align}
holds because the space of bounded operators is the dual space of trace-class operators. Using the notation
$A_1 = A \otimes \id_{L^2(\mathbb{R}^{3(N-1)})} \otimes \id_{\mathcal{F}}$, the identity $1 = p_1 + q_1$, and $\norm{\Psi_N} = 1$, we obtain
\begin{subequations}
\begin{align}
&\tr \left(  A \sqrt{1 - \varepsilon^2 \Delta} \left( \gamma^{(1,0)}_{\Psi_N} - N^{-1} p \right) \sqrt{1 - \varepsilon^2 \Delta} \right)
\nonumber \\
\label{eq:Sobolev trace norm estimate 1}
&\quad =
\tr \left(  A \sqrt{1 - \varepsilon^2 \Delta} \left(  q \gamma^{(1,0)}_{\Psi_N} q
+ q \gamma^{(1,0)}_{\Psi_N} p
+ p \gamma^{(1,0)}_{\Psi_N} q
 \right) \sqrt{1 - \varepsilon^2 \Delta} \right)
\\
\label{eq:Sobolev trace norm estimate 2}
&\qquad +
\scp{\Psi_N}{ p_1 \sqrt{1 - \varepsilon^2 \Delta_1} A_1 \sqrt{1 - \varepsilon^2 \Delta_1} p_1 \Psi_N } 
\\
\label{eq:Sobolev trace norm estimate 3}
&\qquad - N^{-1}
\scp{\Psi_N}{\tr \left( \sqrt{1 - \varepsilon^2 \Delta} A \sqrt{1 - \varepsilon^2 \Delta} p \right) \Psi_N } .
\end{align}
\end{subequations}
By applying the Cauchy--Schwarz inequality and using the projection property of $p$, we obtain
\begin{align}
\abs{\eqref{eq:Sobolev trace norm estimate 1}}
&\leq \norm{\sqrt{1 - \varepsilon^2 \Delta_1} q_1 \Psi_N}^2
+ 2 \norm{\sqrt{1 - \varepsilon^2 \Delta} p}_{\mathfrak{S}^{\infty}}  \norm{\sqrt{1 - \varepsilon^2 \Delta_1} q_1 \Psi_N}
\nonumber \\
&\leq C \left( 1 + \norm{i \varepsilon \nabla p}_{\mathfrak{S}^{\infty}} \right)
\sup_{j=1,2} \left( \norm{q_1 \Psi_N}^2 + \norm{i \varepsilon \nabla_1 q_1 \Psi_N}^2 \right)^{j/2} .
\end{align}
If we use the antisymmetry of the many-body wave function and apply Lemma~\ref{lemma:estimates of the p-p term}
we get
\begin{align}
\abs{\eqref{eq:Sobolev trace norm estimate 2} + \eqref{eq:Sobolev trace norm estimate 3} }
&\leq \norm{p \sqrt{1 - \varepsilon^2 \Delta} A \sqrt{1 - \varepsilon^2 \Delta} p}_{\mathfrak{S}^{\infty}} \norm{q_1 \Psi_N}^2
\leq  C \Big( 1 + \norm{i \varepsilon \nabla p}_{\mathfrak{S}^{\infty}}^2 \Big) \norm{q_1 \Psi_N}^2.
\end{align}
Altogether this proves the claim.
\end{proof}

\subsection{Proof of Lemma~\ref{lemma:energy estimates}}
\label{subsection:energy estimates}

\begin{proof}[\unskip\nopunct]
Throughout the proof, we will use the shorthand notations
$\vAk = \kappa * \vA_{\alpha}$, $\xi_N = W^*(\varepsilon^{-2} \alpha) \Psi_N$, and \eqref{eq:definition quantized vector potential with cutoff shorthand notation}. By the antisymmetry of the many-body wave function, we have
\begin{subequations}
\begin{align}
&N^{-1} \big( \scp{\Psi_N}{H_N^{\rm{PF}} \Psi_N}
- \mathcal{E}^{\rm{MS}}[p,\alpha] 
\big)
\nonumber \\
\label{eq:energy estimates 1}
&\quad = 
\scp{\Psi_N}{\big( - i \varepsilon \nabla_1 - \varepsilon^2 \vAhk(x_1) \big)^2 \Psi_N} 
- N^{-1} \tr \big( p \left( - i \varepsilon \nabla - \kappa * \vA_{\alpha} \right)^2 \big)
\\
\label{eq:energy estimates 2}
&\qquad + \varepsilon^4 
\scp{\Psi_N}{H_f \Psi_N}
- \norm{\alpha}^2_{
\dot{\mathfrak{h}}_{1/2}} 
\\
\label{eq:energy estimates 3}
&\qquad +
\frac{1}{2N} \left( (N-1)  \scp{\Psi_N}{K(x_1 - x_2) \Psi_N} - 
\tr \left( \left( K * \rho_{p} - X_{p} \right) p \right) \right) .
\end{align}
\end{subequations}
Next, each term will be treated separately.\\

\noindent
\textbf{The term \eqref{eq:energy estimates 1}:} 
We start with
\begin{subequations}
\begin{align}
&\scp{\Psi_N}{q_1 \big( - i \varepsilon \nabla_1 - \varepsilon^2 \vAhk(x_1) \big)^2 q_1 \Psi_N} 
-
\eqref{eq:energy estimates 1}
\nonumber \\
&\quad \leq 
2
\abs{\scp{\Psi_N}{q_1 \big( - i \varepsilon \nabla_1 - \varepsilon^2 \vAhk(x_1) \big)^2 p_1 \Psi_N} }
\nonumber \\
&\qquad + 
\abs{\scp{\Psi_N}{p_1 \big( - i \varepsilon \nabla_1 - \varepsilon^2 \vAhk(x_1) \big)^2 p_1 \Psi_N} 
- N^{-1} \tr \big( p \left( - i \varepsilon \nabla - \kappa * \vA_{\alpha} \right)^2 \big) }
\nonumber \\
\label{eq:energy estimates 1 a}
&\quad \leq 
2
\abs{\scp{\Psi_N}{q_1 (- \varepsilon^2 \Delta_1) p_1 \Psi_N} }
\\
\label{eq:energy estimates 1 b}
&\qquad + 
4 \varepsilon^2
\abs{\scp{\Psi_N}{q_1  \vAhk(x_1) \cdot i \varepsilon \nabla_1 p_1 \Psi_N} }
\\
\label{eq:energy estimates 1 c}
&\qquad + 
2 \varepsilon^4
\abs{\scp{\Psi_N}{q_1  \vAhk^2(x_1)  p_1 \Psi_N} }
\\
\label{eq:energy estimates 1 d}
&\qquad + 
\abs{\scp{\Psi_N}{p_1 \big( - i \varepsilon \nabla_1 - \varepsilon^2 \vAhk(x_1) \big)^2 p_1 \Psi_N} 
- N^{-1} \tr \big( p \left( - i \varepsilon \nabla - \kappa * \vA_{\alpha} \right)^2 \big) } .
\end{align}
\end{subequations}
By means of $\id_{\mathcal{H}^{(N)}} = p_1 + q_1$,  
$\norm{q i \varepsilon \nabla p}_{\mathfrak{S}^2}^2
= \tr \left(  q i \varepsilon \nabla p^2  i \varepsilon \nabla q \right)
\leq \norm{i \varepsilon \nabla p}_{\mathfrak{S}^{\infty}} \norm{q i \varepsilon \nabla p}_{\mathfrak{S}^1}$, \eqref{eq:technical estimates using antisymmetry 1}, and \eqref{eq:technical estimates using antisymmetry 2} 
we estimate
\begin{align}
\abs{\eqref{eq:energy estimates 1 a} }
&\leq 
2
\abs{\scp{\Psi_N}{q_1 i \varepsilon \nabla_1 q_1 i \varepsilon \nabla_1 p_1 \Psi_N} }
+ 2
\abs{\scp{\Psi_N}{q_1 i \varepsilon \nabla_1 p_1 i \varepsilon \nabla_1 p_1 \Psi_N} }
\nonumber \\
&\leq  
2 \norm{ i \varepsilon \nabla_1 q_1 \Psi_N} \norm{ q_1 i \varepsilon \nabla_1 p_1 \Psi_N}
+ 2
\abs{\scp{\Psi_N}{q_1 i \varepsilon \nabla_1 p_1 i \varepsilon \nabla_1 p_1 \Psi_N} }
\nonumber \\
&\leq \frac{1}{4} \norm{ i \varepsilon \nabla_1 q_1 \Psi_N}^2
+ C N^{-1} \norm{q i \varepsilon \nabla p}_{\mathfrak{S}^2}^2
+ C N^{-1} \norm{q i \varepsilon \nabla p i \varepsilon \nabla p}_{\mathfrak{S}^1}
\nonumber \\
&\leq \frac{1}{4} \norm{ i \varepsilon \nabla_1 q_1 \Psi_N}^2
+ C N^{-1} \Big[ \norm{q i \varepsilon \nabla p}_{\mathfrak{S}^2}^2
+ \norm{i \varepsilon \nabla p}_{\mathfrak{S}^{\infty}}
\norm{q i \varepsilon \nabla p}_{\mathfrak{S}^1}
\Big] 
\nonumber \\
&\leq \frac{1}{4} \norm{ i \varepsilon \nabla_1 q_1 \Psi_N}^2
+ C N^{-1}  \norm{i \varepsilon \nabla p}_{\mathfrak{S}^{\infty}}
\norm{q i \varepsilon \nabla p}_{\mathfrak{S}^1} .
\end{align}
Note that 
\begin{align}
\label{eq:action of coherent states on quantized vector potential}
W^*(N^{2/3} \alpha) \varepsilon^2 \vAhk(x_1) W(N^{2/3} \alpha) = \varepsilon^2 \vAhk(x_1) + \vAk(x_1)
\end{align}
due to \eqref{eq:Weyl operators shifting property}.
Moreover, 
\begin{align}
\norm{q e^{ikx} i \varepsilon \nabla p}_{\mathfrak{S}^{1}}
&\leq \norm{ i \varepsilon \nabla p}_{\mathfrak{S}^{\infty}}
\norm{ p}_{\mathfrak{S}^{1}}
\leq \norm{ i \varepsilon \nabla p}_{\mathfrak{S}^{\infty}}
N
\end{align}
for a rank-$N$ projection $p$.  Together with the unitarity of the Weyl operators,  \eqref{eq:estimate vector potential with p and q within scalar product}, and \eqref{eq:technical estimates using antisymmetry 1}, this allows us to obtain
\begin{align}
\abs{\eqref{eq:energy estimates 1 b}}
&\leq 
4 \varepsilon^2
\abs{\scp{ \xi_N}{q_1  \vA(x_1) \cdot i \varepsilon \nabla_1 p_1  \xi_N} }
+ 4
\abs{\scp{\xi_N}{q_1  \vAk(x_1) \cdot i \varepsilon \nabla_1 p_1 \xi_N} }
\nonumber \\
&\leq C  C_{\kappa} N^{-1} \norm{q e^{ikx} i \varepsilon \nabla p}_{\mathfrak{S}^{1}}
\varepsilon^2 \norm{ \left( \mathcal{N} + 1 \right)^{1/4} \xi_N}^2
+  N^{-1} \norm{q  \vAk \cdot i \varepsilon \nabla p}_{\mathfrak{S}^1}
\nonumber \\
&\leq C  C_{\kappa}  \norm{ i \varepsilon \nabla p}_{\mathfrak{S}^{\infty}}
\varepsilon^2 \norm{ \left( \mathcal{N} + 1 \right)^{1/4} \xi_N}^2
+  N^{-1} \norm{q  \vAk \cdot i \varepsilon \nabla p}_{\mathfrak{S}^1}
\end{align}
By means of $1 = p+q$, \eqref{eq:estimate for A L-infty to h-1-2} and \eqref{eq:estimate for qAp}, we have
\begin{align}
\norm{q  \vAk \cdot i \varepsilon \nabla p}_{\mathfrak{S}^1}
&\leq \norm{q  \vAk q \cdot i \varepsilon \nabla p}_{\mathfrak{S}^1}
+ \norm{q  \vAk p \cdot i \varepsilon \nabla p}_{\mathfrak{S}^1}
\nonumber \\
&\leq  C B_{\kappa}
\norm{\alpha}_{\dot{\mathfrak{h}}_{1/2}}
\bigg[ \norm{q i \varepsilon \nabla p}_{\mathfrak{S}^1} 
+  \norm{i \varepsilon \nabla p}_{\mathfrak{S}^{\infty}}
\sup_{k \in \mathbb{R}^3}
\Big\{ \left(1 + \abs{k} \right)^{-1} \norm{q e^{ikx} p}_{\mathfrak{S}^1} 
\Big\} \bigg] ,
\end{align}
leading to
\begin{align}
\abs{\eqref{eq:energy estimates 1 b}}
&\leq 
C C_{\kappa} \left( 1 + \norm{i \varepsilon \nabla p}_{\mathfrak{S}^{\infty}} \right) 
\bigg[ 
\norm{\alpha}_{\dot{\mathfrak{h}}_{1/2}} 
N^{-1}
\sup_{k \in \mathbb{R}^3}
\Big\{ \left(1 + \abs{k} \right)^{-1} \norm{q e^{ikx} p}_{\mathfrak{S}^1} 
+  \norm{q i \varepsilon \nabla p}_{\mathfrak{S}^1} 
\Big\}
\nonumber \\
&\quad \qquad \qquad \qquad \qquad \qquad +  \beta^{b,2}[\Psi_N,\alpha] + \varepsilon^2
\bigg] .
\end{align}
Similarly, we get
\begin{align}
\abs{\eqref{eq:energy estimates 1 c}}
&\leq 2 
\abs{\scp{\xi_N}{q_1 \big( \varepsilon^2 \vAhk(x_1)   + \vAk(x_1) \big)^2  p_1 \xi_N} }
\nonumber \\
&\leq 4 \varepsilon^4
\abs{\scp{\xi_N}{q_1  \vAhk^2(x_1)  p_1 \xi_N} }
+ 4 \abs{\scp{\xi_N}{q_1  \vAk^2(x_1)  p_1 \xi_N} }
\nonumber \\
&\leq C \big\|\abs{\cdot}^{-1/2} \mathcal{F}[\kappa] \big\|_{L^2(\mathbb{R}^3)}^2 \varepsilon^4 \norm{\left( \mathcal{N} + 1 \right)^{1/2} \xi_N}^2 
+ C N^{-1} \norm{q_1  \vAk^2(x_1)  p_1}_{\mathfrak{S}^1}
\nonumber \\
&\leq C B_{\kappa}^2
\bigg[ \beta^{b,1}[\Psi_N,\alpha] + \varepsilon^4
+ \norm{\alpha}_{\dot{\mathfrak{h}}_{1/2}}^2 N^{-1}
\sup_{k \in \mathbb{R}^3}
\Big\{ \left(1 + \abs{k} \right)^{-1} \norm{q e^{ikx} p}_{\mathfrak{S}^1} 
\Big\} 
\bigg] 
\end{align}
by means of \eqref{eq:action of coherent states on quantized vector potential}, \eqref{eq:Bounds for vector potential}, \eqref{eq:technical estimates using antisymmetry 1}, $1 = p + q$, \eqref{eq:estimate for A L-infty to h-1-2} and \eqref{eq:estimate for qAp}.
Using 
\begin{align}
\big( - i \varepsilon \nabla_1 - \varepsilon^2 \vAhk(x_1) \big)^2
&= \big( - i \varepsilon \nabla_1 -  \vAk(x_1) \big)^2
+ 2    W(\varepsilon^{-2} \alpha)  \varepsilon^2 \vAhk(x_1) \cdot i \varepsilon \nabla_1
W^*(\varepsilon^{-2} \alpha)
\nonumber \\
&\quad +
W(\varepsilon^{-2} \alpha) 
\Big( \varepsilon^4 \vAhk^2(x_1) 
+ 2 \varepsilon^2 \vAhk(x_1) \cdot \vAk(x_1) 
\Big)
W^*(\varepsilon^{-2} \alpha) ,
\end{align}
\eqref{eq:Bounds for vector potential}
and \eqref{eq:estimate for A L-infty to h-1-2} we continue with
\begin{align}
\abs{\eqref{eq:energy estimates 1 d}}
&\leq 
\abs{\scp{\xi_N}{p_1 \Big( \varepsilon^4 \vAhk^2(x_1) 
+ 2 \varepsilon^2 \vAhk(x_1) \cdot \vAk(x_1) 
\Big) p_1 \xi_N}}
\nonumber \\
&\quad + 2 \varepsilon^2
\abs{\scp{\xi_N}{p_1 i \vAhk(x_1) \cdot i \varepsilon \nabla_1 p_1 \xi_N}}
\nonumber \\
&\quad + 
\abs{\scp{\Psi_N}{p_1 \big( - i \varepsilon \nabla_1 - \varepsilon^2 \vAk(x_1) \big)^2 p_1 \Psi_N} 
- N^{-1} \tr \big( p \left( - i \varepsilon \nabla - \kappa * \vA_{\alpha} \right)^2 \big) }
\nonumber \\
&\leq C
\big\| \abs{\cdot}^{-1/2} \mathcal{F}[\kappa] \big\|_{L^2}
\bigg[
\big( \norm{\vAk}_{\mathfrak{S}^{\infty}} + \norm{i \varepsilon \nabla p}_{\mathfrak{S}^{\infty}}
\big)
\varepsilon^2 \norm{\left( \mathcal{N} +1 \right)^{1/4} \xi_N}^2
\nonumber \\
&\qquad \qquad \qquad \qquad 
+ \big\| \abs{\cdot}^{-1/2} \mathcal{F}[\kappa] \big\|_{L^2}
\varepsilon^4 \norm{\left( \mathcal{N} +1 \right)^{1/2} \xi_N}^2
\bigg]
\nonumber \\
&\quad + 
\abs{\scp{\Psi_N}{p_1 \big( - i \varepsilon \nabla_1 - \varepsilon^2 \vAk(x_1) \big)^2 p_1 \Psi_N} 
- N^{-1} \tr \big( p \left( - i \varepsilon \nabla - \kappa * \vA_{\alpha} \right)^2 \big) } 
\nonumber \\
&\leq C
\left( 1 + \big\| (  \abs{\cdot}^{-1/2 } + \abs{\cdot}^{-1} ) \mathcal{F}[\kappa] \big\|_{L^2}^2 \right)
\big(  1 + \norm{\alpha}_{\dot{\mathfrak{h}}_{1/2}} + \norm{i \varepsilon \nabla p}_{\mathfrak{S}^{\infty}}
\big)
\nonumber \\
&\qquad \times
\left( \beta^{b,1}[\Psi_N,\alpha]
+ \beta^{b,2}[\Psi_N,\alpha] 
+ \varepsilon^2 \right)
\nonumber \\
&\quad + 
\abs{\scp{\Psi_N}{p_1 \big( - i \varepsilon \nabla_1 - \varepsilon^2 \vAk(x_1) \big)^2 p_1 \Psi_N} 
- N^{-1} \tr \big( p \left( - i \varepsilon \nabla - \kappa * \vA_{\alpha} \right)^2 \big) }  .
\end{align}
By the self-adjointness of $\big( - i \varepsilon \nabla - \varepsilon^2 \vAk \big)^2$ and Lemma~\ref{lemma:diagonalisation of the p-p term} we conclude the existence of orthonormal $\chi_{1}, \ldots \chi_{N}$ in the range of $p$ and real $\gamma_{1}, \ldots, \gamma_{N}$ such that $p = \sum_{j=1}^N \ket{\chi_{j}} \bra{\chi_{j}}$,
\begin{align}
\sup_{j \in \{1,\ldots, N \}} \abs{\gamma_j} 
&\leq 
\sup_{j \in \{1,\ldots, N \}}
\abs{\scp{\chi_{j}}{p \big( - i \varepsilon \nabla - \varepsilon^2 \vAk \big)^2 p \chi_{j}} }
\leq 
2 \left( \norm{i \varepsilon \nabla p}^2_{\mathfrak{S}^{\infty}}
+ \norm{\vAk}^2_{\mathfrak{S}^{\infty}} \right) ,
\nonumber \\
\sum_{j=1}^N \gamma_{j} Q^{\chi_{j}} 
&= -
\sum_{i=1}^N p_{i}  \big( - i \varepsilon \nabla_i - \varepsilon^2 \vAk(x_i) \big)^2 p_{i}
- \tr \left( \big( - i \varepsilon \nabla - \varepsilon^2 \vAk \big)^2 p \right) .
\end{align}
Together with the antisymmetry of the many-body wave function, \eqref{eq:relation for the sum over all Q-j}, and \eqref{eq:estimate for A L-infty to h-1-2} we obtain
\begin{align}
& \abs{\scp{\Psi_N}{p_1 \big( - i \varepsilon \nabla_1 - \varepsilon^2 \vAk(x_1) \big)^2 p_1 \Psi_N} 
- N^{-1} \tr \big( p \left( - i \varepsilon \nabla - \kappa * \vA_{\alpha} \right)^2 \big) }
\nonumber \\
&\quad \leq  N^{-1} \sum_{j=1}^N \abs{\gamma_j}
\abs{\scp{\Psi_N}{  Q^{\chi_j}  \Psi_N}  }
\nonumber \\
&\quad \leq 2  
\left( \norm{i \varepsilon \nabla p}^2_{\mathfrak{S}^{\infty}}
+ \norm{\vAk}^2_{\mathfrak{S}^{\infty}} \right) 
\scp{\Psi_N}{N^{-1} \sum_{j=1}^N Q^{\chi_j}  \Psi_N}
\nonumber \\
&\quad \leq  C  
\left( \norm{i \varepsilon \nabla p}^2_{\mathfrak{S}^{\infty}}
+ \big\| \abs{\cdot}^{-1} \mathcal{F}[\kappa] \big\|_{L^2}^2 \norm{\alpha}_{\dot{\mathfrak{h}}_{1/2}}^2 \right) 
\beta^a[\Psi_N,p] .
\end{align}
In total,
\begin{align}
\abs{\eqref{eq:energy estimates 1 d}}
&\leq C
\left( 1 + \big\| (  \abs{\cdot}^{-1/2 } + \abs{\cdot}^{-1} ) \mathcal{F}[\kappa] \big\|_{L^2}^2 \right)
\big(  1 + \norm{\alpha}_{\dot{\mathfrak{h}}_{1/2}}^2 + \norm{i \varepsilon \nabla p}_{\mathfrak{S}^{\infty}}^2
\big)
\nonumber \\
&\qquad \times
\left( 
\beta^{a}[\Psi_N,p]
+ \beta^{b,1}[\Psi_N,\alpha]
+ \beta^{b,2}[\Psi_N,\alpha] 
+ \varepsilon^2 \right)
\end{align}
and
\begin{align}
&\abs{ \scp{\Psi_N}{q_1 \big( - i \varepsilon \nabla_1 - \varepsilon^2 \vAhk(x_1) \big)^2 q_1 \Psi_N} 
-
\eqref{eq:energy estimates 1} }
\nonumber \\
&\quad \leq 
\frac{1}{4} \norm{ i \varepsilon \nabla_1 q_1 \Psi_N}^2
+ C
\left( 1 + \big\| (  \abs{\cdot}^{-1/2 } + \abs{\cdot}^{-1} ) \mathcal{F}[\kappa] \big\|_{L^2}^2 \right)
\big(  1 + \norm{\alpha}_{\dot{\mathfrak{h}}_{1/2}}^2 + \norm{i \varepsilon \nabla p}_{\mathfrak{S}^{\infty}}^2 \big)
\nonumber \\
&\qquad  \times 
\left[ \beta[\Psi_N,p,\alpha] + \varepsilon^2 
+ N^{-1}
\sup_{k \in \mathbb{R}^3}
\Big\{ \left(1 + \abs{k} \right)^{-1} \norm{q e^{ikx} p}_{\mathfrak{S}^1} 
+  \norm{q i \varepsilon \nabla p}_{\mathfrak{S}^1} 
\Big\}
\right] .
\end{align}
By \eqref{eq:action of coherent states on quantized vector potential}, \eqref{eq:Bounds for vector potential},  and \eqref{eq:estimate for A L-infty to h-1-2} we have
\begin{align}
&\abs{\norm{i \varepsilon \nabla_1 q_1 \Psi_N}^2  
- \scp{\Psi_N}{q_1 \big( - i \varepsilon \nabla_1 - \varepsilon^2 \vAhk(x_1) \big)^2 q_1 \Psi_N} }
\nonumber \\
&\quad = 
\abs{\norm{i \varepsilon \nabla_1 q_1 \Psi_N}^2  
- \scp{\xi_N}{q_1 \big( - i \varepsilon \nabla_1 - \varepsilon^2 \vAhk(x_1) - \vAk(x_1) \big)^2 q_1 \xi_N} }
\nonumber \\
&\quad \leq 
\frac{1}{4} \norm{i \varepsilon \nabla_1 q_1 \Psi_N}^2 + 
C \varepsilon^4 \norm{\vAhk(x_1) q_1 \xi_N}^2
+ 
C\norm{\vAk(x_1) q_1 \xi_N}^2
\nonumber \\
&\quad \leq 
\frac{1}{4} \norm{i \varepsilon \nabla_1 q_1 \Psi_N}^2 + 
C  \big\| \abs{\cdot}^{-1/2} \mathcal{F}[\kappa] \big\|_{L^2}^2
\varepsilon^4 \norm{\left( \mathcal{N} + 1 \right)^{1/2} \xi_N}^2
+ C \norm{\vAk}^2_{\mathfrak{S}^{\infty}} \norm{q_1 \Psi_N}^2
\nonumber \\
&\quad \leq 
\frac{1}{4} \norm{i \varepsilon \nabla_1 q_1 \Psi_N}^2 + 
C  \big\| (  \abs{\cdot}^{-1/2 } + \abs{\cdot}^{-1} ) \mathcal{F}[\kappa] \big\|_{L^2}^2 
\left<  \norm{\alpha}_{\dot{\mathfrak{h}}_{1/2}}^2 
\right>
\left( \beta^a[\Psi_N,p] + \beta^{b,1}[\Psi_N,\alpha] + \varepsilon^4 \right)
\end{align}
and therefore
\begin{align}
\label{eq:energy estimates estimate for the magentic laplacian term}
&\abs{\norm{i \varepsilon \nabla_1 q_1 \Psi_N}^2   -
\eqref{eq:energy estimates 1} }
\nonumber \\
&\quad \leq
\abs{\norm{i \varepsilon \nabla_1 q_1 \Psi_N}^2  
- \scp{\Psi_N}{q_1 \big( - i \varepsilon \nabla_1 - \varepsilon^2 \vAhk(x_1) \big)^2 q_1 \Psi_N} }
\nonumber \\
&\qquad +
\abs{ \scp{\Psi_N}{q_1 \big( - i \varepsilon \nabla_1 - \varepsilon^2 \vAhk(x_1) \big)^2 q_1 \Psi_N} 
-
\eqref{eq:energy estimates 1} }
\nonumber \\
&\quad \leq 
\frac{1}{2} \norm{ i \varepsilon \nabla_1 q_1 \Psi_N}^2
+ C
\left( 1 + \big\| (  \abs{\cdot}^{-1/2 } + \abs{\cdot}^{-1} ) \mathcal{F}[\kappa] \big\|_{L^2}^2 \right)
\big(  1 + \norm{\alpha}_{\dot{\mathfrak{h}}_{1/2}}^2 + \norm{i \varepsilon \nabla p}_{\mathfrak{S}^{\infty}}^2 \big)
\nonumber \\
&\qquad  \times 
\left[ \beta[\Psi_N,p,\alpha] + \varepsilon^2 
+ N^{-1}
\sup_{k \in \mathbb{R}^3}
\Big\{ \left(1 + \abs{k} \right)^{-1} \norm{q e^{ikx} p}_{\mathfrak{S}^1} 
+  \norm{q i \varepsilon \nabla p}_{\mathfrak{S}^1} 
\Big\}
\right] .
\end{align}

\noindent
\textbf{The term \eqref{eq:energy estimates 2}:}
Using the fact that 
$\varepsilon^4 W^*(\varepsilon^{-2} \alpha) H_f  W(\varepsilon^{-2} \alpha)
= \norm{\alpha}^2_{
\dot{\mathfrak{h}}_{1/2}}  
+ \varepsilon^4  H_f
+ \varepsilon^2 \big( a^*( \abs{\cdot} \alpha) + a ( \abs{\cdot} \alpha) \big)$ due to  \eqref{eq:Weyl operators shifting property},
together with \eqref{eq:standard estimate annihilation and creation operators}, we obtain
\begin{align}
\label{eq:energy estimates estimate for the field energy}
\abs{\varepsilon^4 
\scp{\xi_N}{H_f \xi_N}  - \eqref{eq:energy estimates 2} }
&=  \varepsilon^2 \abs{\scp{\xi_N}{a^*( \abs{\cdot} \alpha) + a ( \abs{\cdot} \alpha) \xi_N}}
\leq C \norm{\alpha}_{\mathfrak{h}_1} 
\beta^{b,2}[\Psi_N,\alpha] .
\end{align}

\noindent
\textbf{The term \eqref{eq:energy estimates 3}:}
Note that 
\begin{subequations}
\begin{align}
\label{eq:energy estimates direct interaction 1}
\abs{\eqref{eq:energy estimates 3} }
&\leq 
\abs{\scp{\Psi_N}{q_1 K(x_1 - x_2) q_1 \Psi_N}}
\\
\label{eq:energy estimates direct interaction 2}
&\quad + 2
\abs{\scp{\Psi_N}{q_1 K(x_1 - x_2) p_1 \Psi_N}}
\\
\label{eq:energy estimates direct interaction 3}
&\quad +
\abs{\scp{\Psi_N}{q_2 p_1 K(x_1 - x_2) p_1 q_2 \Psi_N}}
\\
\label{eq:energy estimates direct interaction 4}
&\quad + 2
\abs{\scp{\Psi_N}{q_2 p_1 K(x_1 - x_2) p_1 p_2 \Psi_N}}
\\
\label{eq:energy estimates direct interaction 5}
&\quad + \frac{1}{2N}
\abs{(N-1) \scp{\Psi_N}{p_2 p_1 K(x_1 - x_2) p_1 p_2 \Psi_N} - 
\tr \left( K * \rho_{p} p \right)}
\\
\label{eq:energy estimates direct interaction 6}
&\quad +
\frac{1}{2N} \abs{ \tr \left( X_{p} \, p \right)} .
\end{align}
\end{subequations}
By means of \eqref{eq:L-1 norm of the first moment of the potential in terms of kappa} and the antisymmetry of the many-body wave function we get
\begin{align}
\eqref{eq:energy estimates direct interaction 1}
+ \abs{\eqref{eq:energy estimates direct interaction 3}}
&\leq 2 \norm{K}_{L^{\infty}} \norm{q_1 \Psi_N}^2
\leq C \big\| (1 + \abs{\cdot}^{-1}) \mathcal{F}[\kappa] \big\|_{L^2}^2
\beta^{a}[\Psi_N, p] .
\end{align}
Using, in addition, the Fourier decomposition of the potential
and \eqref{eq:technical estimates using antisymmetry 1} gives
\begin{align}
\abs{\eqref{eq:energy estimates direct interaction 2}}
&\leq  \int_{\mathbb{R}^3} \abs{\mathcal{F}[K](k)} 
\abs{\scp{\Psi_N}{q_1 e^{ik x_1} p_1 e^{-ikx_2} \Psi_N}} \, dk
\nonumber \\
&\leq (N-1)^{-1} \big\| (1 + \abs{\cdot} ) \mathcal{F}[K](k) \big\|_{L^1}
\sup_{k \in \mathbb{R}^3} \Big\{ (1 + \abs{k})^{-1} \norm{q e^{ikx}  p}_{\mathfrak{S}^1}  \Big\}
\nonumber \\
&\leq C B_{\kappa}^2
N^{-1} \sup_{k \in \mathbb{R}^3} \Big\{ (1 + \abs{k})^{-1} \norm{q e^{ikx}  p}_{\mathfrak{S}^1}  \Big\}  
\end{align} 
and 
\begin{align}
\abs{\eqref{eq:energy estimates direct interaction 4}}
&\leq C B_{\kappa}^2
N^{-1} \sup_{k \in \mathbb{R}^3} \Big\{ (1 + \abs{k})^{-1} \norm{q e^{ikx}  p}_{\mathfrak{S}^1}  \Big\} . 
\end{align}
Note that
\begin{align}
&(N-1) \scp{\Psi_N}{p_2 p_1 K(x_1 - x_2) p_1 p_2 \Psi_N}
-  \tr \left( K * \rho_{p} \, p \right)  
\nonumber \\
&\quad =  (2 \pi)^{-3/2} \int_{\mathbb{R}^3} \mathcal{F}[K](k)
\scp{\Psi_N}{\Big( \sum_{j=2}^N p_j e^{ ik x_j} p_j - \tr \left( e^{ik \cdot} p \right) \Big) p_1 e^{- ik x_1} p_1 \Psi_N} \, dk
\nonumber \\
&\qquad +  (2 \pi)^{-3/2}  N^{-1} \int_{\mathbb{R}^3} \mathcal{F}[K](k)
\tr \left( e^{ik \cdot} p \right) 
\scp{\Psi_N}{\Big( \sum_{j=1}^N p_j e^{- ik x_j} p_j 
-  \tr \left( e^{-ik \cdot} p \right) \Big)
\Psi_N} \, dk . 
\end{align}
Thus if we use \eqref{eq:estimates of the p-p term 1} and \eqref{eq:estimates of the p-p term 2} with $O = e^{i k \cdot}$, \eqref{eq:L-1 norm of the first moment of the potential in terms of kappa}, $\norm{p e^{ik\cdot} p}_{\mathfrak{S}^{\infty}} \leq 1$, and $\norm{q e^{- i k \cdot} p}_{\mathfrak{S}^1} \leq N$, we get
\begin{align}
\abs{\eqref{eq:energy estimates direct interaction 5} }
&\leq C B_{\kappa}^2 \left( \beta^{a}[\Psi_N, p] + N^{-1} \right) .
\end{align}
Using the definition of the exchange term, \eqref{eq:L-1 norm of the first moment of the potential in terms of kappa}, and  $\norm{p}^2_{\mathfrak{S}^2} = \norm{p}_{\mathfrak{S}^1} = n$ let us obtain
\begin{align}
\eqref{eq:energy estimates direct interaction 6}
&\leq \frac{1}{2 N^2} \int_{\mathbb{R}^6}
\abs{K(x-y) } \abs{p(x;y)}^2  \, dx \, dy 
\leq  \frac{1}{2 N^2} \norm{K}_{L^{\infty}}
\norm{p}^2_{\mathfrak{S}^2}
\leq C B_{\kappa}^2 N^{-1}  .
\end{align}
In total, this gives
\begin{align}
\label{eq:energy estimates estimate for the direct interaction}
\abs{\eqref{eq:energy estimates 3} }
&\leq 
C \big\| (1 + \abs{\cdot}^{-1}) \mathcal{F}[\kappa] \big\|_{L^2}^2
\Big( \beta^{a}[\Psi_N, p] + N^{-1} 
+ N^{-1} \sup_{k \in \mathbb{R}^3} \Big\{ (1 + \abs{k})^{-1} \norm{q e^{ikx}  p}_{\mathfrak{S}^1}  \Big\}
\Big) .
\end{align}
Finally, we use 
$N^{-1} \big( \scp{\Psi_N}{H_N^{\rm{PF}} \Psi_N}
- \mathcal{E}^{\rm{MS}}[p,\alpha] 
\big) = \eqref{eq:energy estimates 1}
+ \eqref{eq:energy estimates 2} +\eqref{eq:energy estimates 3}$
as well as the estimates \eqref{eq:energy estimates estimate for the magentic laplacian term}, \eqref{eq:energy estimates estimate for the field energy} and \eqref{eq:energy estimates estimate for the direct interaction} to obtain
\begin{align}
&\norm{i \varepsilon \nabla_1 q_1 \Psi_N}^2 + \varepsilon^4 
\scp{\xi_N}{H_f \xi_N}
\nonumber \\
&\quad \leq 
N^{-1} \Big| \scp{\Psi_N}{H_N^{\rm{PF}} \Psi_N}
- \mathcal{E}^{\rm{MS}}[p,\alpha] \Big|
+
\abs{\norm{i \varepsilon \nabla_1 q_1 \Psi_N}^2   -
\eqref{eq:energy estimates 1} }
\nonumber \\
&\qquad 
+ \abs{\varepsilon^4 
\scp{\xi_N}{H_f \xi_N}  - \eqref{eq:energy estimates 2} }
+ \abs{\eqref{eq:energy estimates 3}}
\nonumber \\
&\quad \leq 
N^{-1} \Big|  \scp{\Psi_N}{H_N^{\rm{PF}} \Psi_N}
- \mathcal{E}^{\rm{MS}}[p,\alpha] \Big|
+  \frac{1}{2} \norm{ i \varepsilon \nabla_1 q_1 \Psi_N}^2
\nonumber \\
&\qquad 
+ C
\left( 1 + \big\| (  \abs{\cdot}^{-1/2 } + \abs{\cdot}^{-1} ) \mathcal{F}[\kappa] \big\|_{L^2}^2 \right)
\big(  1 
+ \norm{\alpha}_{\mathfrak{h}_1} 
+ \norm{\alpha}_{\dot{\mathfrak{h}}_{1/2}}^2 + \norm{i \varepsilon \nabla p}_{\mathfrak{S}^{\infty}}^2 \big)
\nonumber \\
&\qquad \quad  \times 
\left[ \beta[\Psi_N,p,\alpha] + \varepsilon^2 
+ N^{-1}
\sup_{k \in \mathbb{R}^3}
\Big\{ \left(1 + \abs{k} \right)^{-1} \norm{q e^{ikx} p}_{\mathfrak{S}^1} 
+  \norm{q i \varepsilon \nabla p}_{\mathfrak{S}^1} 
\Big\}
\right] .
\end{align}
Subtracting half of the left-hand side and multiplying by two then shows the claim.
\end{proof}

\subsection{Estimating the growth of correlations}

\label{subsection:Estimating the growth of correlations}

Within this section we provide estimates that allow us to control the growth of $\beta^a(t)$, $\beta^{b,1}(t)$, and $\beta^{b,2}(t)$ during the time evolution. Combining these results proves Lemma~\ref{lemma:estimating the time derivative of the functional}.

\subsubsection{Estimates for $\beta^a$}

\begin{lemma}
\label{lemma:estimating the time derivative of beta-a}
Let $\kappa$ satisfy Assumption~\ref{assumption:cutoff function}. Let $\Psi_{N} \in \mathcal{D} \left( (H_N^{\rm PF})^{1/2} \right) \cap \mathcal{D} \left( \mathcal{N}^{1/2} \right)$ and let $(p,\alpha) \in \mathfrak{S}_{+}^{2,1} (L^2(\mathbb{R}^3)) \times \mathfrak{h}_{1} \cap \dot{\mathfrak{h}}_{-1/2}$ with $p$ being a rank-$N$ projection. Let $(p_t, \alpha_t)$ be the unique solution of \eqref{eq:Maxwell-Schroedinger equations} with initial data $(p,\alpha)$, and let $\Psi_{N,t} = e^{- i \varepsilon^{-1} H_N^{\rm PF} t} \Psi_N$.
Then, there exists a constant $C>0$ such that
\begin{align}
\label{eq:estimate growth of beta-a in time}
\beta^a(t) 
&\leq \beta^a(0) +  \int_{0}^t
C \left<  C_{\kappa}^2 \right>
\left< \norm{\alpha_s}_{\dot{\mathfrak{h}}_{1/2}} \right>
\Big< \norm{\varepsilon i \nabla p_s}_{\mathfrak{S}^{\infty}}
\Big>
\left( \beta(s) + \varepsilon^2 \right) 
\nonumber \\
&\qquad \qquad  \times
\varepsilon^{-1} N^{-1} 
\Big[ 1 
+ \sup_{k \in \mathbb{R}^3} \Big\{ (1 + \abs{k})^{-1} \norm{q_s e^{ikx}  p_s}_{\mathfrak{S}^1}  \Big\}  + \norm{q_s \varepsilon i \nabla p_s}_{\mathfrak{S}^1} 
\Big] 
\; ds  .
\end{align}
\end{lemma}

\begin{proof}[Proof of Lemma~\ref{lemma:estimating the time derivative of beta-a}]
Throughout the proof we use the notations $\xi_{N,t} = W^*(\varepsilon^{-2} \alpha_t) \Psi_{N,t}$ and $\vAc(x,t) = (\kappa * \vA_{\alpha_t})(x)$.
Using 
$\frac{d}{dt} q_{t} = - i \varepsilon^{-1} \left[ H_{p,\alpha}(t) , q_{t} \right]$ with $H_{p,\alpha}(t)$ being defined as in \eqref{eq:definition mean-field Hamiltonian}
and the antisymmetry of the many-body wave function, we get
\begin{align}
\frac{d}{dt} \beta^{a}(t) 
&= i \varepsilon^{-1}
\scp{\Psi_{N,t}}{\left[ \left( H_N^{\rm{PF}} - H_{p,\alpha}(t)   \otimes \id_{L^2(\mathbb{R}^{3(N-1)})} \otimes \id_{\mathcal{F}} \right) , q_{t,1} \right] \Psi_{N,t}}
\nonumber \\
&= - 4 \Re \,
\scp{\Psi_{N,t}}{\left( \varepsilon^2 \vAhk(x_1) - \vAk(x_1,t) \right) \cdot \nabla_1 q_{t,1}  \Psi_{N,t}}
\nonumber \\
&\quad - 2 \varepsilon^{-1} \Im \,
\scp{\Psi_{N,t}}{\left( \varepsilon^4 \vAhk^2(x_1) - \vAk^2(x_1,t) \right)  q_{t,1} \Psi_{N,t}}
\nonumber \\
&\quad - 2 \varepsilon^{-1}  \Im \,
\scp{\Psi_{N,t}}{ \left( (N-1)/N  K(x_1 - x_2) -  K * \rho_{p_t}(x_1) \right)   q_{t,1}  \Psi_{N,t}}
\nonumber \\
&\quad - i
\varepsilon^{-1}
\scp{\Psi_{N,t}}{\left[ X_{p_t,1} , q_{t,1} \right] \Psi_{N,t}} .
\end{align}
Inserting the identity $\id_{\mathcal{H}^{(N)}} = p_{t,1} + q_{t,1}$ and 
\begin{align}
\Re \,
\scp{\Psi_{N,t}}{q_{t,1} \left( \varepsilon^2 \vAhk(x_1) - \vAk(x_1,t) \right) \cdot \nabla_1 q_{t,1}  \Psi_{N,t}}
&= 0 , 
\nonumber \\
\Im \,
\scp{\Psi_{N,t}}{q_{t,1} \left( \varepsilon^4 \vAhk^2(x_1) - \vAk^2(x_1,t) \right)  q_{t,1} \Psi_{N,t}}
&= 0 , 
\nonumber \\
\Im \,
\scp{\Psi_{N,t}}{ q_{t,1}  \left( (N-1)/N  K(x_1 - x_2) -  K * \rho_{p_t}(x_1)   \right) q_{t,1}  \Psi_{N,t}}
&= 0  
\end{align}
let us obtain  
\begin{subequations}
\begin{align}
\label{eq:time derivative beta-a-1 A}
\frac{d}{dt} \beta^{a,1}(t)
&= - 4 \Re \,
\scp{\Psi_{N,t}}{p_{t,1} \left( \varepsilon^2 \vAhk(x_1) - \vAk(x_1,t) \right) \cdot \nabla_1 q_{t,1}  \Psi_{N,t}}
\\
\label{eq:time derivative beta-a-1 B}
&\quad - 2 \varepsilon^{-1} \Im \,
\scp{\Psi_{N,t}}{p_{t,1} \left( \varepsilon^4 \vAhk^2(x_1) - \vAk^2(x_1,t) \right)  q_{t,1} \Psi_{N,t}}
\\
\label{eq:time derivative beta-a-1 C}
&\quad - 2 \varepsilon^{-1}  \Im \,
\scp{\Psi_{N,t}}{p_{t,1} \left( (N-1)/N  K(x_1 - x_2) -  K * \rho_{p_t}(x_1) \right) q_{t,1}  \Psi_{N,t}}
\\
\label{eq:time derivative beta-a-1 D}
&\quad - i
\varepsilon^{-1}
\scp{\Psi_{N,t}}{p_{t,1} \left[ X_{p_t,1} , q_{t,1} \right] \Psi_{N,t}}  .
\end{align}
\end{subequations}

In the following we estimate each term separately.

\noindent
\textbf{The term \eqref{eq:time derivative beta-a-1 A}:}
By \eqref{eq:Weyl operators shifting property} and \eqref{eq:estimate vector potential with p and q within scalar product} with $O_1 = \varepsilon i \nabla_1$ and $m= 1/4$ we get
\begin{align}
\abs{\eqref{eq:time derivative beta-a-1 A}}
&= -4 \varepsilon \abs{
\scp{\xi_{N,t}}{p_{t,1}  \vAhk(x_1) \cdot \varepsilon \nabla_1 q_{t,1}  \xi_{N,t}} }
\nonumber \\
&=   4  \varepsilon \abs{\scp{ q_{t,1} \vAhk(x_1) \cdot \varepsilon i \nabla_1 p_{t,1} \xi_{N,t}}{ q_{t,1} \xi_{N,t}} }
\nonumber \\
&\leq 
C \varepsilon^{-1} N^{-1} \sup_{k \in \mathbb{R}^3} \Big\{ (1 + \abs{k})^{-1} \norm{q_t e^{ikx} \varepsilon i \nabla p_t}_{\mathfrak{S}^1}  \Big\} 
\nonumber \\
&\qquad \times
\big\| \big(\abs{\cdot}^{1/2} + \abs{\cdot}^{-1/2} \big) \mathcal{F}[\kappa] \big\|_{L^2} 
\left( \beta^{b,2}(t) + \varepsilon^2 \right) .
\end{align}
By means of $\id_{L^2(\mathbb{R}^3)}  = p_t + q_t$ we get 
$
\norm{q_t e^{ikx} \varepsilon i \nabla p_t}_{\mathfrak{S}^1}
\leq 
\norm{\varepsilon i \nabla p_t}_{\mathfrak{S}^{\infty}}
\norm{q_t e^{ikx} p_t}_{\mathfrak{S}^1}
+ \norm{q_t \varepsilon i \nabla p_t}_{\mathfrak{S}^1}$, 
which leads to
\begin{align}
\abs{\eqref{eq:time derivative beta-a-1 A}}
&\leq 
C \varepsilon^{-1} N^{-1} \Big[
\norm{q_t \varepsilon i \nabla p_t}_{\mathfrak{S}^1}
+ \norm{\varepsilon i \nabla p_t}_{\mathfrak{S}^{\infty}} \sup_{k \in \mathbb{R}^3} \Big\{ (1 + \abs{k})^{-1} \norm{q_t e^{ikx}  p_t}_{\mathfrak{S}^1}  \Big\} 
\Big]
\nonumber \\
&\quad \times
\big\| \big(\abs{\cdot}^{1/2} + \abs{\cdot}^{-1/2} \big) \mathcal{F}[\kappa] \big\|_{L^2} 
\left( \beta^{b,2}(t) + \varepsilon^2 \right) .
\end{align}

\noindent
\textbf{The term \eqref{eq:time derivative beta-a-1 B}:}
Using \eqref{eq:Weyl operators shifting property} again we obtain
\begin{align}
\label{eq:rewriting of the quadratic difference of the vector potentials by means of the Weyl operators}
\varepsilon^4 \vAhk^2(x_1) - \vAk^2(x_1,t)
&= W(\varepsilon^{-2} \alpha_t)  \Big( \varepsilon^4 \vAhk^2(x_1) 
+ 2 \vAk(x_1,t) \varepsilon^2 \vAhk(x_1) \Big) W^*(\varepsilon^{-2} \alpha_t)  .
\end{align}
Thus if we apply \eqref{eq:estimate vector potential squared with p and q within scalar product} with $m=1/2$ as well as 
\eqref{eq:estimate vector potential with p and q within scalar product} with $O = \vAk(\cdot,t)$ and $m=1/4$ we get
\begin{align}
\abs{\eqref{eq:time derivative beta-a-1 B}}
&\leq  2 \varepsilon^{3} \abs{
\scp{\xi_{N,t}}{p_{t,1}  \vAhk^2(x_1)   q_{t,1} \xi_{N,t}}}  + 4 \varepsilon 
\abs{ \scp{\xi_{N,t}}{p_{t,1}  \vAk(x_1,t)  \vAhk(x_1)   q_{t,1} \xi_{N,t}} }
\nonumber \\
&\leq 
C \varepsilon^{-1} N^{-1} \sup_{k \in \mathbb{R}^3} \Big\{ (1 + \abs{k})^{-1} \Big( \norm{q_t e^{ikx}  p_t}_{\mathfrak{S}^1}
+ \norm{q_t e^{ikx} \vAk(\cdot,t) p_t}_{\mathfrak{S}^1} \Big)
\Big\}
\nonumber \\
&\qquad \times
\bigg[
\varepsilon^4 \norm{\left( \mathcal{N} + 1 \right)^{1/2} \xi_{N,t}}^2
+ \varepsilon^2 \norm{\left( \mathcal{N} + 1 \right)^{1/4} \xi_{N,t}}^2
\bigg]  .
\end{align}
Due to $\id_{L^2(\mathbb{R}^3)} = p_t + q_t$, \eqref{eq:estimate for A L-infty to h-1-2} and \eqref{eq:estimate for qAp} we have
\begin{align}
\norm{q_t e^{ikx} \vAk(\cdot,t) p_t}_{\mathfrak{S}^1}
&\leq  \norm{q_t e^{ikx} p_t}_{\mathfrak{S}^1} \norm{\vAk(\cdot,t)}_{\mathfrak{S}^{\infty}}
+ \norm{q_t \vAk(\cdot,t) p_t}_{\mathfrak{S}^1}
\nonumber \\
&\leq B_{\kappa} \norm{\alpha_t}_{\dot{\mathfrak{h}}_{1/2}}
\bigg[
\norm{q_t e^{ikx}  p_t}_{\mathfrak{S}^1}
+
\sup_{l \in \mathbb{R}^3} \Big\{ (1 + \abs{l})^{-1}  \norm{q_t e^{ilx}  p_t}_{\mathfrak{S}^1}
\Big\}
\bigg] 
\end{align}
and
\begin{align}
\abs{\eqref{eq:time derivative beta-a-1 B}}
&\leq 
C B_{\kappa} \varepsilon^{-1} N^{-1} \sup_{k \in \mathbb{R}^3} \Big\{ (1 + \abs{k})^{-1} \norm{q_t e^{ikx}  p_t}_{\mathfrak{S}^1}
\Big\} \norm{\alpha_t}_{\dot{\mathfrak{h}}_{1/2}}
\left( \beta^{b,1}(t) + \beta^{b,2}(t) \right) .
\end{align}

\noindent
\textbf{The term \eqref{eq:time derivative beta-a-1 C}:}
The next term is estimated in analogy to \cite[Section 9]{PP2016}. Inserting the identity $\id_{\mathcal{H}^{(N)}} = p_{t,2} + q_{t,2}$ and using that
$\Im \scp{\Psi_{N,t}}{q_{t,1} p_{t,2} K(x_1 - x_2) p_{t,1} q_{t,2} \Psi_{N,t}}  = 0$
holds to the antisymmetry under exchange of $x_1$ and $x_2$ let us obtain
\begin{subequations}
\begin{align}
\eqref{eq:time derivative beta-a-1 C}
&=  2 \varepsilon^{-1}  \Im \,
\scp{\Psi_{N,t}}{q_{t,1} \left( (N-1)/N  K(x_1 - x_2) -  K * \rho_{p_t}(x_1) \right) p_{t,1}  \Psi_{N,t}}
\nonumber \\
\label{eq:time derivative beta-a-1 C 1}
&=  2 \varepsilon^{-1}  \Im \,
\scp{\Psi_{N,t}}{q_{t,1} \left( (N-1)/N  p_{t,2} K(x_1 - x_2) p_{t,2} -  K * \rho_{p_t}(x_1) \right) p_{t,1}  \Psi_{N,t}}
\\
\label{eq:time derivative beta-a-1 C 3}
&\quad +  2 \varepsilon^{-1} (N-1)/N  \Im \,
\scp{\Psi_{N,t}}{q_{t,1} q_{t,2}  K(x_1 - x_2)  p_{t,1} q_{t,2}  \Psi_{N,t}}
\\
\label{eq:time derivative beta-a-1 C 2}
&\quad +  2 \varepsilon^{-1} (N-1)/N  \Im \,
\scp{\Psi_{N,t}}{q_{t,1} q_{t,2}  K(x_1 - x_2)  p_{t,1} p_{t,2}  \Psi_{N,t}} .
\end{align}
\end{subequations}
Note that $K(x_1 - x_2) =
(2 \pi)^{-3/2} \int_{\mathbb{R}^3} \mathcal{F}[K](k) e^{ik(x_1 - x_2)}  \, dk$ and the antisymmetry of the many-body state  leads to
\begin{align}
\eqref{eq:time derivative beta-a-1 C 1}
&=  2 \varepsilon^{-1} N^{-1} (2 \pi)^{-3/2} \int_{\mathbb{R}^3}
\mathcal{F}[K](k)
\nonumber \\
&\quad \times
\bigg[
 \Im \,
\scp{\Psi_{N,t}}{\Big( \sum_{j=2}^N p_{t,j}  
\cos(k x_j) p_{t,j} - \tr \big( \cos(k \cdot) p_t \big)
\Big) q_{t,1} e^{ik x_1} p_{t,1}  \Psi_{N,t}}
\nonumber \\
&\qquad 
-  \Re \,
\scp{\Psi_{N,t}}{\Big( \sum_{j=2}^N p_{t,j}  
\sin(k x_j) p_{t,j} - \tr \big( \sin(k \cdot) p_t \big)
\Big) q_{t,1} e^{ik x_1} p_{t,1}  \Psi_{N,t}}
\bigg] .
\end{align}
Since $\norm{\cos(k \cdot)}_{\mathfrak{S}^{\infty}} \leq 1$ and $\norm{\sin(k \cdot)}_{\mathfrak{S}^{\infty}} \leq 1$ we obtain 
\begin{align}
\abs{\eqref{eq:time derivative beta-a-1 C 1} }
&\leq  C \varepsilon^{-1} N^{-1}  \int_{\mathbb{R}^3}
\abs{\mathcal{F}[K](k)} \norm{q_t e^{ikx} p}_{\mathfrak{S}^1}
\left( \norm{q_{t,1} \Psi_{N,t}}^2 + N^{-1} \right) \, dk
\nonumber \\
&\leq 
C \varepsilon^{-1} N^{-1} \sup_{k \in \mathbb{R}^3}
\Big\{
(1 + \abs{k})^{-1}
\norm{q_t e^{ikx} p}_{\mathfrak{S}^1}
\Big\} 
\norm{(1 + \abs{\cdot} ) \mathcal{F}[K]}_{L^1}
\left( \beta^a(t) + N^{-1} \right) 
\nonumber \\
&\leq 
C B_{\kappa}^2 \varepsilon^{-1} N^{-1} \sup_{k \in \mathbb{R}^3}
\Big\{
(1 + \abs{k})^{-1}
\norm{q_t e^{ikx} p}_{\mathfrak{S}^1}
\Big\} 
\left( \beta^a(t) + N^{-1} \right)  
\end{align}
by means of \eqref{eq:estimates of the p-p term 2} and \eqref{eq:L-1 norm of the first moment of the potential in terms of kappa}.
Using the Fourier decomposition of the potential, Lemma~\ref{lemma:technical estimates using antisymmetry} with $A_1 = p_{t,1} e^{i k x_1} p_{t,1}$ and \eqref{eq:L-1 norm of the first moment of the potential in terms of kappa}, we get 
\begin{align}
\eqref{eq:time derivative beta-a-1 C 3}
&\leq 
\varepsilon^{-1}  \int_{\mathbb{R}^3} \abs{\mathcal{F}[K](k)} \abs{\scp{q_{t,2} \Psi_{N,t}}{  q_{t,1} e^{ik x_1}  p_{t,1}
e^{- i k x_2} q_{t,2}  \Psi_{N,t}}}
\, dk
\nonumber \\
&\leq 
\varepsilon^{-1} (N-1)^{-1} \int_{\mathbb{R}^3} \abs{\mathcal{F}[K](k)}
\norm{q_{t,1} e^{ik x_1}  p_{t,1}}_{\mathfrak{S}^1}
\norm{q_{t,2} \Psi_{N,t}}^2 \, dk
\nonumber \\
&\leq 
C  B_{\kappa}^2 \varepsilon^{-1} N^{-1}
\sup_{k \in \mathbb{R}^3}
\Big\{ (1 + \abs{k})^{-1}  \norm{q_t e^{ik x} p_t}_{\mathfrak{S}^1}  \Big\}
\beta^{a}(t) .
\end{align}
Before estimating \eqref{eq:time derivative beta-a-1 C 2}, we first show that
\begin{align}
\label{eq:semiclassical structure and p-p-q-q term}
\norm{q_{1} q_{2}  K(x_1 - x_2) p_{1} p_{2} \Psi_{N} }
&\leq C B_{\kappa}^2 (N - (j+1))^{-1} 
\sup_{k \in \mathbb{R}^3}
\Big\{ (1 + \abs{k})^{-1}  \norm{q e^{ik \cdot} p}_{\mathfrak{S}^1}  \Big\} \norm{\Psi_N}
\end{align}
holds for all $\Psi_{N} \in L^2(\mathbb{R}^{3N}) \otimes \mathcal{F}$ which are antisymmetric in $x_1$, $x_2$ and all other electron variables except $x_{l_1}, \ldots, x_{l_j}$. Using the Fourier decomposition of $K$ and applying Lemma~\ref{lemma:technical estimates using antisymmetry} twice, we get
\begin{align}
\norm{q_{1} q_{2}  K(x_1 - x_2) p_{1} p_{2} \Psi_{N} }
&\leq (2 \pi)^{-3/2}
\int_{\mathbb{R}^3} \abs{\mathcal{F}[K](k)}
\norm{q_1 e^{ik x_1} p_1
q_2 e^{-ik x_2} p_2 \Psi_N} \, dk 
\nonumber \\
&\leq  (N - (j+1))^{-1}
\int_{\mathbb{R}^3} \abs{\mathcal{F}[K](k)}
\norm{q e^{ik \cdot} p}_{\mathfrak{S}^2}^2 
\norm{\Psi_N} \, dk .
\end{align}
Due to the boundedness of $p$, $q$ and $e^{ik \cdot}$ we have
$\norm{q e^{ik \cdot} p}_{\mathfrak{S}^2}^2 
= \tr \left( q e^{ik \cdot} p^2  e^{-ik \cdot} q \right) 
\leq \norm{q e^{ik \cdot} p}_{\mathfrak{S}^1}$.
Together with \eqref{eq:L-1 norm of the first moment of the potential in terms of kappa} this shows \eqref{eq:semiclassical structure and p-p-q-q term}.
This and the antisymmetry of the many-body wave function allows us to obtain
\begin{align}
\eqref{eq:time derivative beta-a-1 C 2}
&\leq 2 \varepsilon^{-1} N^{-1}  \abs{
\scp{\Psi_{N,t}}{q_{t,1} \sum_{m=2}^N q_{t,m}  K(x_1 - x_m) p_{t,1} p_{t,m} \Psi_{N,t}} }
\nonumber \\
&\leq 2 \varepsilon^{-1} N^{-1} \norm{q_{t,1} \Psi_{N,t}}
\big\| q_{t,1} \sum_{m=2}^N q_{t,m}  K(x_1 - x_m) p_{t,1} p_{t,m} \Psi_{N,t} \big\|
\nonumber \\
&\leq 2 \varepsilon^{-1} N^{-1} \norm{q_{t,1} \Psi_{N,t}}
\bigg[ N
\big\| q_{t,1} q_{t,2}  K(x_1 - x_2) p_{t,1} p_{t,2} \Psi_{N,t} \big\|^2
\nonumber \\
&\qquad \quad +
N^2 \scp{ q_{t,1} q_{t,2} K(x_1 - x_2) p_{t,1} p_{t,2}   \Psi_{N,t}}{ q_{t,1} q_{t,3} K(x_1 - x_3) p_{t,1}  p_{t,3}  \Psi_{N,t}}
\bigg]^{1/2}
\nonumber \\
&\leq C \varepsilon^{-1}  \norm{q_{t,1} \Psi_{N,t}}
\bigg[ \norm{ q_{t,1} q_{t,2} K(x_1 - x_2) p_{t,1} p_{t,2}  q_{t,3} \Psi_{N,t}}
\nonumber \\
&\qquad \qquad \qquad \qquad \quad
+ N^{-1/2}
\big\| q_{t,1} q_{t,2}  K(x_1 - x_2) p_{t,1} p_{t,2} \Psi_{N,t} \big\|
\bigg]
\nonumber \\
&\leq C B_{\kappa}^2 \varepsilon^{-1} N^{-1}
\sup_{k \in \mathbb{R}^3}
\Big\{ (1 + \abs{k})^{-1}  \norm{q_t e^{ik x} p_t}_{\mathfrak{S}^1} \Big\}
 \norm{q_{t,1} \Psi_{N,t}}
\Big( \norm{ q_{t,1} \Psi_{N,t}}
+ N^{-1/2}
\Big)
\nonumber \\
&\leq C B_{\kappa}^2 \varepsilon^{-1} N^{-1}
\sup_{k \in \mathbb{R}^3}
\Big\{ (1 + \abs{k})^{-1}  \norm{q_t e^{ik x} p_t}_{\mathfrak{S}^1}  \Big\}
\left( \beta^a(t) + N^{-1} \right).
\end{align}
In total, this gives
\begin{align}
\abs{\eqref{eq:time derivative beta-a-1 C}}
&\leq  C B_{\kappa}^2  \varepsilon^{-1} N^{-1}  \sup_{k \in \mathbb{R}^3} \Big\{ (1 + \abs{k})^{-1} 
\norm{q_t e^{ik x} p_t}_{\mathfrak{S}^1} \Big\}
\left( \beta^{a}(t) + N^{-1} \right) .
\end{align}

\noindent
\textbf{The term \eqref{eq:time derivative beta-a-1 D}:}
Using
\begin{align}
X_{p_t,1} \, p_{t,1}
&= N^{-1} (2 \pi)^{-3/2}
\int_{\mathbb{R}^3} \mathcal{F}[K](k) e^{ik x_1} p_{t,1} e^{- i k x_1} p_{t,1} \, dk,
\end{align}
\eqref{eq:L-1 norm of the first moment of the potential in terms of kappa}, and $\varepsilon^{-1} N^{-1} = N^{- 2/3}$, we obtain
\begin{align}
\abs{\eqref{eq:time derivative beta-a-1 D}}
&\leq 
\varepsilon^{-1}
\abs{\scp{\Psi_{N,t}}{q_{t,1}  X_{p_t,1}  p_{t,1}  \Psi_{N,t}}} 
\nonumber \\
&\leq \varepsilon^{-1}
N^{-1} 
\int_{\mathbb{R}^3} \abs{\mathcal{F}[K](k)} \abs{\scp{\Psi_{N,t}}{q_{t,1} e^{ik x_1} p_{t,1} e^{- i k x_1} p_{t,1} \Psi_{N,t}}} \, dk
\nonumber \\
&\leq \varepsilon^{-1} N^{-1} \big\| \mathcal{F}[K] \big\|_{L^1} \norm{q_{t,1} \Psi_{N,t}}
\nonumber \\
&\leq C  B_{\kappa}^2 
\big( \beta^{a}(t) + N^{-1} \big) .
\end{align}
Collecting the estimates and integrating in time finally leads to
\eqref{eq:estimate growth of beta-a in time}.
\end{proof}

\subsubsection{Estimates for $\beta^{b,1}$}

\begin{lemma}
\label{lemma:estimate for beta-b-1}
Let $\kappa$ satisfy Assumption~\ref{assumption:cutoff function}. Let $\Psi_{N} \in \mathcal{D} \left( (H_N^{\rm PF})^{1/2} \right) \cap \mathcal{D} \left( \mathcal{N}^{1/2} \right)$ and let $(p,\alpha) \in \mathfrak{S}_{+}^{2,1} (L^2(\mathbb{R}^3)) \times \mathfrak{h}_{1} \cap \dot{\mathfrak{h}}_{-1/2}$ with $p$ being a rank-$N$ projection. Let $(p_t, \alpha_t)$ be the unique solution of \eqref{eq:Maxwell-Schroedinger equations} with initial data $(p,\alpha)$, and let $\Psi_{N,t} = e^{- i \varepsilon^{-1} H_N^{\rm PF} t} \Psi_N$.
Then, there exists a constant $C>0$ such that
\begin{align}
\label{eq:estimate for beta-b-1}
\beta^{b,1}(t)
&\leq  \beta^{b,1}(0)
+ \int_0^t  \bigg[ \norm{i \varepsilon \nabla_1 q_{s,1} \Psi_{N,s}}^2
+
 C \left( 1 + \big\| ( 1 + \abs{\cdot}^{-1} ) \mathcal{F}[\kappa] \big\|^2_{L^2} \right) 
\nonumber \\
&\qquad  \qquad  \qquad  \qquad  \times 
\Big<
\norm{ i \varepsilon \nabla p_s}_{\mathfrak{S}^{\infty}}  
+ \norm{\alpha_s}_{\dot{\mathfrak{h}}_{1/2}}^2  \Big> \left( \beta^a(s) + \beta^{b,1}(s) + \varepsilon^4 \right)
\bigg] \, ds .
\end{align}
\end{lemma}

\begin{proof}
We again use the notations $\xi_{N,t} = W^*(\varepsilon^{-2} \alpha_t) \Psi_{N,t}$ and $\vAc(x,t) = (\kappa * \vA_{\alpha_t})(x)$.
Since $\mathcal{N}$ is not a bounded operator, special care must be taken when computing the time derivatives of $\beta^{b,1}(t)$ and $\beta^{b,2}(t)$. Following the approach of \cite[Section III.3]{FL2023}, we define, for $\delta >0$, the bounded operator
$\mathcal{N}_{\delta}
= \mathcal{N} e^{- \delta \mathcal{N}}$. If $\Psi_N \in \mathcal{D} \left( \mathcal{N}^{1/2} \right)$ and $f \in \mathfrak{h}$, then the following holds
\begin{subequations}
\begin{align}
\label{eq:commutator of regularized number operator 1}
\lim_{\delta \rightarrow 0}
\norm{\left( \left[ \mathcal{N}_{\delta}, a(f) \right] + a(f) \right) \Psi_N} &= 0 ,
\\
\label{eq:commutator of regularized number operator 2}
\lim_{\delta \rightarrow 0}
\norm{\left( \left[ \mathcal{N}_{\delta}, a^*(f) \right] - a^*(f) \right) \Psi_N} &= 0 ,
\\
\label{eq:commutator of regularized number operator 3}
\lim_{\delta \rightarrow 0}
\norm{\left( \left[ \mathcal{N}_{\delta}^{1/2}, a(f) \right] + \left( \left( \mathcal{N} + 1 \right)^{1/2} - \mathcal{N}^{1/2} \right) a(f) \right) \Psi_N} &= 0 ,
\\
\label{eq:commutator of regularized number operator 4}
\lim_{\delta \rightarrow 0}
\norm{\left( \left[ \mathcal{N}_{\delta}^{1/2}, a^*(f) \right] + \left( \left( \mathcal{N} -1 \right)^{1/2} - \mathcal{N}^{1/2} \right) a^*(f) \right) \Psi_N} &= 0 .
\end{align}
\end{subequations}
Note that \eqref{eq:shifting property number operator} implies
$\left[ \mathcal{N}_{\delta}, a(f) \right] + a(f) = \mathcal{N} \big( 1 - e^{- \delta} \big) e^{- \delta \mathcal{N}} a(f) 
- e^{- \delta (\mathcal{N} + 1)} a(f)$.
By \eqref{eq:standard estimate annihilation and creation operators}, the spectral calculus of $\mathcal{N}$, and the monotone convergence theorem, we obtain \eqref{eq:commutator of regularized number operator 1} because
\begin{align}
&\lim_{\delta \rightarrow 0} \norm{\left( \left[ \mathcal{N}_{\delta}, a(f) \right] + a(f) \right) \Psi_N}^2
\nonumber \\
&\quad \leq 2 \norm{f}_{\mathfrak{h}} \lim_{\delta \rightarrow 0} \int_0^{\infty}
\left( \lambda^3 \big( 1 - e^{- \delta} \big)^2 e^{- 2 \delta \lambda} 
+ \big( 1 - e^{- \delta (\lambda +1)} \big)^2 \lambda \right) \scp{\Psi_N}{d E(\lambda) \Psi_N} 
= 0 .
\end{align}
The remaining relations are shown by similar means. Using the unitarity of the Weyl operators, \eqref{eq:Weyl operators shifting property} and the fact $W^*(\varepsilon^{-2} \alpha_t)$ with $\alpha \in C^1 \big( \mathbb{R}_{+} ; \mathfrak{h} \big)$  is strongly differentiable in $t$ from $\mathcal{D} \left( \mathcal{N}^{1/2} \right)$ to $\mathcal{H}^{(N)}$ (see \cite[Lemma 3.1]{GV1979a}) with
\begin{align}
\frac{d}{dt} W^*(\varepsilon^{-2} \alpha_t) = 
\Big( \varepsilon^{-2} \big( a(\dot{\alpha}_t) - a^*(\dot{\alpha}_t) \big) - \varepsilon^{-4} i \Im \, \scp{\alpha_t}{\dot{\alpha}_t} \Big) W^*(\varepsilon^{-2} \alpha_t)  
\end{align}
we compute
\begin{align}
\label{eq:time derivative of the fluctuation vector}
\frac{d}{dt} \xi_{N,t} 
&= \bigg( \frac{d}{dt}  W^*(\varepsilon^{-2} \alpha_t) \bigg) \Psi_{N,t} 
- i \varepsilon^{-1} W^*(\varepsilon^{-2} \alpha_t) H^{\rm{PF}}_N W(\varepsilon^{-2} \alpha_t) \xi_{N,t}
= - i \varepsilon^{-1} \mathcal{G}(t) \xi_{N,t} ,
\end{align}
where
\begin{align}
\label{eq:generator for fluctuation dynamics which cancels out the coeherent state}
\mathcal{G}(t)
&= 
\sum_{j=1}^N  \left(  i \varepsilon \nabla_j + \varepsilon^2 \vAhk(x_j) + \vAk(x_j,t) \right)^2
+ \frac{1}{2 N} \sum_{\substack{i,j=1\\ j\neq i}}^N  K(x_j - x_k)
+ \varepsilon H_f 
\nonumber \\
&\quad 
+ \varepsilon^{-1} \big( a(\abs{\cdot} \alpha_t - i \dot{\alpha}_t) + a^*(\abs{\cdot} \alpha_t - i \dot{\alpha}_t) \big)
+ \varepsilon^{-3} \norm{\alpha_t}^2_{\dot{\mathfrak{h}}_{1/2}}
+ \varepsilon^{-3} \Im \, \scp{\alpha_t}{\dot{\alpha}_t}_{\mathfrak{h}} .
\end{align}
Together with the anti-symmetry of the many-body wave function and $\left[ \mathcal{N}_{\delta}, H_f \right] = 0$, this gives
\begin{align}
\frac{d}{dt}
\varepsilon^4 \scp{\xi_{N,t}}{\mathcal{N}_{\delta} \xi_{N,t}}
&= i \varepsilon^3  \scp{\xi_{N,t}}{\left[ \mathcal{G}(t), \mathcal{N}_{\delta} \right] \xi_{N,t}}
\nonumber \\
&=   2 \varepsilon^2 \Im  \scp{\xi_{N,t}}{\left[ \mathcal{N}_{\delta} ,  \vAhk(x_1) \right] \left(  i \varepsilon \nabla_1 + \varepsilon^2 \vAhk(x_1) + \vAk(x_1,t) \right)  \xi_{N,t}}
\nonumber \\
&\quad +  \varepsilon^2  \Im 
\scp{\xi_{N,t}}{\left[ \mathcal{N}_{\delta} , \big( a(\abs{\cdot} \alpha_t - i \dot{\alpha}_t) + a^*(\abs{\cdot} \alpha_t - i \dot{\alpha}_t) \big)  \right] \xi_{N,t}} .
\end{align}
Using notation \eqref{eq:splitting of the vector potential} and \eqref{eq:commutator of regularized number operator 1} as well \eqref{eq:commutator of regularized number operator 2}, we obtain
\begin{align}
&\lim_{\delta \rightarrow 0} \frac{d}{dt}
\varepsilon^4 \scp{\xi_{N,t}}{\mathcal{N}_{\delta} \xi_{N,t}}
\nonumber \\
&\quad =  2 \varepsilon^2 \Im \, 
\scp{\xi_{N,t}}{  \big( \vAm(x_1) - \vAp(x_1) \big)   \left(  i \varepsilon \nabla_1 + \varepsilon^2 \vAhk(x_1) + \vAk(x_1,t) \right)  \xi_{N,t}}
\nonumber \\
&\qquad + 
\varepsilon^2  \Im
\scp{\xi_{N,t}}{\big( a^*(\abs{\cdot} \alpha_t - i \dot{\alpha}_t)
- a(\abs{\cdot} \alpha_t - i \dot{\alpha}_t) \big) \xi_{N,t}} 
\end{align}
By means of $\Im \, \scp{\xi_{N,t}}{\left[ \vAp(x_1) , \vAhk(x_1) \right] \xi_{N,t} } = 0$, $\Im(- \overline{z}) = \Im (z)$
and \eqref{eq:Maxwell-Schroedinger equations}, we get
\begin{subequations}
\begin{align}
&\lim_{\delta \rightarrow 0} \frac{d}{dt}
\varepsilon^4 \scp{\xi_{N,t}}{\mathcal{N}_{\delta} \xi_{N,t}}
\nonumber \\
\label{eq:time derivative beta-b A}
&\quad =  4 \varepsilon^2 \Im \, 
\scp{\xi_{N,t}}{  \vAm(x_1) \cdot \left(  i \varepsilon \nabla_1 + \varepsilon^2 \vAhk(x_1) + \vAk(x_1,t) \right)  \xi_{N,t}}
\\
\label{eq:time derivative beta-b B}
&\qquad 
+  2  \, \Im \, \sum_{\lambda=1,2} \int_{\mathbb{R}^3} dk \, 
\sqrt{\frac{4 \pi^3}{\abs{k}}} \mathcal{F}[\kappa](k) \vep_{\lambda}(k) \cdot \mathcal{F}[\vJ_{p_t, \alpha_t}](k)
\scp{\varepsilon^2 a_{k,\lambda} \xi_{N,t}}{  \xi_{N,t}} .
\end{align}
\end{subequations}
Inserting the identity $1 = p_{t,1} + q_{t,1}$ and $\left[ \vAm(x_1) \cdot , \left(  i \varepsilon \nabla_1  + \vAk(x_1,t) \right) \right] = 0$ lead to
\begin{subequations}
\begin{align}
\eqref{eq:time derivative beta-b A}
&=    4 \varepsilon^2 \Im \, 
\scp{\xi_{N,t}}{  \vAm(x_1) \cdot \left(  i \varepsilon \nabla_1 + \varepsilon^2 \vAhk(x_1) + \vAk(x_1,t) \right)  \xi_{N,t}}
\nonumber \\
\label{eq:time derivative beta-b A1}
&=   4 \varepsilon^4 \Im \, 
\scp{\xi_{N,t}}{  \vAm(x_1)  
\cdot \vAhk(x_1)  \xi_{N,t}}
\\
\label{eq:time derivative beta-b A2}
&\quad + 4 \varepsilon^2 \Im \, 
\scp{\xi_{N,t}}{ q_{t,1} \vAm(x_1) \cdot \left(  i \varepsilon \nabla_1  + \vAk(x_1,t) \right) p_{t,1} \xi_{N,t}}
\\
\label{eq:time derivative beta-b A3}
&\quad +  4 \varepsilon^2 \Im \, 
\scp{\xi_{N,t}}{ p_{t,1} \left(  i \varepsilon \nabla_1  + \vAk(x_1,t) \right) \cdot \vAm(x_1)  q_{t,1} \xi_{N,t}}
\\
\label{eq:time derivative beta-b A4}
&\quad +  4 \varepsilon^2 \Im \, 
\scp{\xi_{N,t}}{ q_{t,1} \vAm(x_1) \cdot \left(  i \varepsilon \nabla_1  + \vAk(x_1,t) \right) q_{t,1} \xi_{N,t}}
\\
\label{eq:time derivative beta-b A5}
&\quad +  4 \varepsilon^2 \Im \, 
\scp{\xi_{N,t}}{ p_{t,1} \vAm(x_1) \cdot \left(  i \varepsilon \nabla_1  + \vAk(x_1,t) \right) p_{t,1} \xi_{N,t}} .
\end{align}
\end{subequations}
Due to \eqref{eq:Bounds for vector potential} and \eqref{eq:estimate for A L-infty to h-1-2} we have
\begin{align}
\abs{\eqref{eq:time derivative beta-b A1}}
&\leq C \varepsilon^4 \big\| \abs{\cdot}^{-1/2} \mathcal{F}[\kappa] \big\|_{L^2}^2
\norm{\left( \mathcal{N} + 1 \right)^{1/2} \xi_{N,t}}^2
\nonumber \\
&\leq   C \big\| \abs{\cdot}^{-1/2} \mathcal{F}[\kappa] \big\|_{L^2}^2 \left( \beta^{b,1}(t) + \varepsilon^4 \right) ,
\end{align}
\begin{align}
\abs{\eqref{eq:time derivative beta-b A2} }
&\leq 4 \varepsilon^2
\norm{q_{t,1} \xi_{N,t}} 
\norm{ q_{t,1} \vAm(x_1) \cdot \left(  i \varepsilon \nabla_1  + \vAk(x_1,t) \right) p_{t,1} \xi_{N,t}}
\nonumber \\
&\leq C \varepsilon^2 \big\| \abs{\cdot}^{-1/2} \mathcal{F}[\kappa] \big\|_{L^2}
\norm{q_{t,1} \xi_{N,t}}
\norm{\left( \mathcal{N} + 1 \right)^{1/2}  \left(  i \varepsilon \nabla_1  + \vAk(x_1,t) \right) p_{t,1} \xi_{N,t}}
\nonumber \\
&\leq C \varepsilon^2 \big\| \abs{\cdot}^{-1/2} \mathcal{F}[\kappa] \big\|_{L^2} 
\norm{ \left(  i \varepsilon \nabla  + \vAk(\cdot,t) \right) p_{t}}_{\mathfrak{S}^{\infty}}
\norm{q_{t,1} \xi_{N,t}}
\norm{\left( \mathcal{N} + 1 \right)^{1/2} \xi_{N,t}}
\nonumber \\
&\leq C \big\| \abs{\cdot}^{-1/2} \mathcal{F}[\kappa] \big\|_{L^2} 
\norm{ \left(  i \varepsilon \nabla  + \vAk(\cdot,t) \right) p_{t}}_{\mathfrak{S}^{\infty}}
\sqrt{\beta^a(t)}
\sqrt{\beta^{b,1}(t) + \varepsilon^4} 
\nonumber \\
&\leq  C \big\| \abs{\cdot}^{-1/2} \mathcal{F}[\kappa] \big\|_{L^2} \left<   \big\| \abs{\cdot}^{-1} \mathcal{F}[\kappa] \big\|_{L^2} \right>
\left( \norm{  i \varepsilon \nabla  p_t}_{\mathfrak{S}^{\infty}}
+ \norm{\alpha_t}_{\dot{\mathfrak{h}}_{1/2}} \right)
\sqrt{\beta^a(t)}
\sqrt{\beta^{b,1}(t) + \varepsilon^4}  
\nonumber \\
&\leq  C \left< B_{\kappa} \right>^2  
\left( \norm{  i \varepsilon \nabla  p_t}_{\mathfrak{S}^{\infty}}
+ \norm{\alpha_t}_{\dot{\mathfrak{h}}_{1/2}} \right)
\left( \beta^a(t) + \beta^{b,1}(t) + \varepsilon^4 \right) ,
\end{align}
\begin{align}
\abs{\eqref{eq:time derivative beta-b A3} }
&\leq  4 \varepsilon^2 \abs{ 
\scp{\vAp(x_1) \cdot \left(  i \varepsilon \nabla_1  + \vAk(x_1,t) \right) p_{t,1} \xi_{N,t}}{  q_{t,1} \xi_{N,t}} }
\nonumber \\
&\leq C \varepsilon^2
\| \abs{\cdot}^{-1/2} \mathcal{F}[\kappa] \big\|_{L^2}
\norm{\mathcal{N}^{1/2} \left(  i \varepsilon \nabla_1  + \vAk(x_1,t) \right) p_{t,1} \xi_{N,t}}
\norm{q_{t,1} \xi_{N,t}}
\nonumber \\
&\leq C \varepsilon^2  \big\| \abs{\cdot}^{-1/2} \mathcal{F}[\kappa] \big\|_{L^2} 
\norm{ \left(  i \varepsilon \nabla  + \vAk(\cdot,t) \right) p_{t}}_{\mathfrak{S}^{\infty}}
\norm{\mathcal{N}^{1/2} \xi_{N,t}}  \norm{q_{t,1} \xi_{N,t}}
\nonumber \\
&\leq  C  \left< B_{\kappa} \right>^2
\left( \norm{  i \varepsilon \nabla  p_t}_{\mathfrak{S}^{\infty}}
+ \norm{\alpha_t}_{\dot{\mathfrak{h}}_{1/2}} \right)
\left( \beta^a(t) + \beta^{b,1}(t) \right) ,
\end{align}
and
\begin{align}
\abs{\eqref{eq:time derivative beta-b A4} }
&\leq  4 \varepsilon^2 \abs{
\scp{ \vAp(x_1) q_{t,1} \xi_{N,t}}{ \left(  i \varepsilon \nabla_1  + \vAk(x_1,t) \right) q_{t,1} \xi_{N,t}} }
\nonumber \\
&\leq C \varepsilon^2  \big\|\abs{\cdot}^{-1/2} \mathcal{F}[\kappa] \big\|_{L^2(\mathbb{R}^3)}
\norm{\mathcal{N}^{1/2} \xi_{N,t} }
\Big[ \norm{\vAk(\cdot,t)}_{\mathfrak{S}^{\infty}}
\norm{q_{t,1} \Psi_{N,t} }
+ \norm{i \varepsilon \nabla_1 q_{t,1} \Psi_{N,t}} \Big]
\nonumber \\
&\leq C \left( 1 + \big\| \abs{\cdot}^{-1/2} \mathcal{F}[\kappa] \big\|_{L^2(\mathbb{R}^3)}^2 
 +  \norm{\vAk(\cdot,t)}_{\mathfrak{S}^{\infty}}^2 \right)
\left( \beta^{a}(t) + \beta^{b,1}(t) \right)
+ \norm{i \varepsilon \nabla_1 q_{t,1} \Psi_{N,t}}^2 
\nonumber \\
&\leq C \left( B_{\kappa} \right>^2
\left( 1 +  \norm{\alpha_t}_{\dot{\mathfrak{h}}_{1/2}}^2 \right)
\left( \beta^{a}(t) + \beta^{b,1}(t) \right)
+ \norm{i \varepsilon \nabla_1 q_{t,1} \Psi_{N,t}}^2 .
\end{align}
Next, we use the antisymmetry of the many-body wave function and \eqref{eq:definition of G} to write \eqref{eq:time derivative beta-b A5} as 
\begin{align}
\eqref{eq:time derivative beta-b A5} 
&= 4 \varepsilon^2 N^{-1} \Im \,  \sum_{\lambda=1,2} \int_{\mathbb{R}^3} 
\scp{a_{k,\lambda} \xi_{N,t}}{ \sum_{j=1}^N p_{t,j} \Re[\boldsymbol{G}_{x_j}](k,\lambda) \cdot \left(  i \varepsilon \nabla_j  + \vAk(x_j,t) \right) p_{t,j} \xi_{N,t}} \, dk 
\nonumber \\
&\quad + 4 \varepsilon^2 N^{-1} \Re \,  \sum_{\lambda=1,2} \int_{\mathbb{R}^3} 
\scp{a_{k,\lambda} \xi_{N,t}}{\sum_{j=1}^N p_{t,j} \Im[\boldsymbol{G}_{x_j}](k,\lambda) \cdot \left(  i \varepsilon \nabla_j  + \vAk(x_j,t) \right) p_{t,j} \xi_{N,t}} \, dk  .
\end{align}
Due to \eqref{eq:Fourier transform of the vector current}
and $\left[ \boldsymbol{G}_{\cdot}(k,\lambda) \cdot , i \varepsilon \nabla \right] = 0$, we have
\begin{align}
\eqref{eq:time derivative beta-b B} 
&=
- 2 \varepsilon^2 N^{-1} \, \Im \, \sum_{\lambda=1,2} \int_{\mathbb{R}^3} 
\tr \left( \left( \left\{ \boldsymbol{G}_{\cdot}(k, \lambda) \cdot , i \varepsilon \nabla \right\} + 2 \boldsymbol{G}_{\cdot}(k, \lambda) \cdot \vAk(\cdot,t) \right) p_t \right) 
\nonumber \\
&\qquad \qquad \qquad \times
\scp{a_{k,\lambda} \xi_{N,t}}{  \xi_{N,t}}
\, dk
\nonumber \\
&=
- 4 \varepsilon^2 N^{-1} \, \Im \, \sum_{\lambda=1,2} \int_{\mathbb{R}^3} 
\tr \big( \Re[\boldsymbol{G}_{\cdot}](k, \lambda) \cdot \left(  i \varepsilon \nabla  + \vAk(\cdot,t) \right) p_t \big) 
\scp{a_{k,\lambda} \xi_{N,t}}{  \xi_{N,t}}  \, dk  
\nonumber \\
&\quad 
- 4 \varepsilon^2 N^{-1} \, \Re \, \sum_{\lambda=1,2} \int_{\mathbb{R}^3} 
\tr \big( \Im[\boldsymbol{G}_{\cdot}](k, \lambda) \cdot \left(  i \varepsilon \nabla  + \vAk(\cdot,t) \right) p_t \big) 
\scp{a_{k,\lambda} \xi_{N,t}}{  \xi_{N,t}}  \, dk  .
\end{align}
Application of \eqref{eq:estimates of the p-p term 1} with $J = \{ 1, 2, \ldots, N \}$,   \eqref{eq:estimate for A L-infty to h-1-2} and $\norm{\norm{ \Im[\boldsymbol{G}_{\cdot}](k, \lambda) }_{\mathfrak{S}^{\infty}}} \leq \abs{k}^{-1/2} \abs{\mathcal{F}[\kappa](k)}$ as well as $\norm{\norm{ \Re[\boldsymbol{G}_{\cdot}](k, \lambda) }_{\mathfrak{S}^{\infty}}} \leq \abs{k}^{-1/2} \abs{\mathcal{F}[\kappa](k)}$
then gives
\begin{align}
&\abs{\eqref{eq:time derivative beta-b A5}
+ \eqref{eq:time derivative beta-b B}}
\nonumber \\
&\quad \leq C \varepsilon^2 \norm{\left(  i \varepsilon \nabla  + \vAk(\cdot,t) \right) p_t}_{\mathfrak{S}^{\infty}}
\norm{q_{t,1} \Psi_{N,t}}
\sum_{\lambda = 1,2} \int_{\mathbb{R}^3}
\abs{k}^{-1/2} \abs{\mathcal{F}[\kappa](k)}
\norm{a_{k,\lambda} \xi_{N,t}}
\nonumber \\
&\quad \leq C  \big\| \abs{\cdot}^{-1/2} \mathcal{F}[\kappa] \big\|_{L^2}
\norm{\left(  i \varepsilon \nabla  + \vAk(\cdot,t) \right) p_t}_{\mathfrak{S}^{\infty}}
\sqrt{\beta^a(t)}
\sqrt{\beta^{b,1}(t)} 
\nonumber \\
&\quad \leq C  \left< B_{\kappa} \right>^2
\left( \norm{  i \varepsilon \nabla  p_t}_{\mathfrak{S}^{\infty}}
+ \norm{\alpha_t}_{\dot{\mathfrak{h}}_{1/2}} \right)
\left(  \beta^a(t) + \beta^{b,1}(t) \right)   .
\end{align}
In total, we obtain
\begin{align}
\abs{\lim_{\delta \rightarrow 0} \frac{d}{dt}
\varepsilon^4 \scp{\xi_{N,t}}{\mathcal{N}_{\delta} \xi_{N,t}} }
&\leq
\norm{i \varepsilon \nabla_1 q_{t,1} \Psi_{N,t}}^2
+ C \left< B_{\kappa} \right>^2
\left( 1 +
\norm{ i \varepsilon \nabla p_t}_{\mathfrak{S}^{\infty}}  
+ \norm{\alpha_t}_{\dot{\mathfrak{h}}_{1/2}}^2  \right)
\nonumber \\
&\qquad \times
\left( \beta^a + \beta^{b,1}(t) + \varepsilon^4 \right) .
\end{align}
Note that one can show by similar estimates as above that for $\delta >0$ sufficiently small and $T>0$ 
 the function $[0,T] \rightarrow \mathbb{R}$, $t \mapsto \frac{d}{dt}
\varepsilon^4 \scp{\xi_{N,t}}{\mathcal{N}_{\delta} \xi_{N,t}}$ has an integrable majorant uniform in~$\delta$.
Moreover, note that
\begin{align}
\lim_{\delta \rightarrow 0} 
\varepsilon^4 \scp{\xi_{N,t}}{\mathcal{N}_{\delta} \xi_{N,t}}
&=
\lim_{\delta \rightarrow 0}  \int_{0}^{\infty} \lambda e^{-  \delta \lambda} \scp{\xi_{N,t}}{ d E(\lambda) \xi_{N,t}} = \varepsilon^4 \scp{\xi_{N,t}}{\mathcal{N} \xi_{N,t}} = \beta^{b,1}(t)
\end{align}
because of the spectral theorem and monotone convergence theorem. Inequality \eqref{eq:estimate for beta-b-1} then follows by Duhamel's formula and the dominated convergence theorem.

\end{proof}

\subsubsection{Estimates for $\beta^{b,2}$}

\label{subsubsection:Estimates for beta-b-2}

\begin{lemma}
Let $\kappa$ satisfy Assumption~\ref{assumption:cutoff function}. Let $\Psi_{N} \in \mathcal{D} \left( (H_N^{\rm PF})^{1/2} \right) \cap \mathcal{D} \left( \mathcal{N}^{1/2} \right)$ and let $(p,\alpha) \in \mathfrak{S}_{+}^{2,1} (L^2(\mathbb{R}^3)) \times \mathfrak{h}_{1} \cap \dot{\mathfrak{h}}_{-1/2}$ with $p$ being a rank-$N$ projection. Let $(p_t, \alpha_t)$ be the unique solution of \eqref{eq:Maxwell-Schroedinger equations} with initial data $(p,\alpha)$, and let $\Psi_{N,t} = e^{- i \varepsilon^{-1} H_N^{\rm PF} t} \Psi_N$.
Then, there exists a constant $C>0$ such that
\begin{align}
\beta^{b,2}(t)
&\leq  \beta^{b,2}(0)
+ \int_0^t 
\bigg[ \norm{i \varepsilon \nabla_1 q_{s,1} \Psi_{N,s}}^2
+  C   \left<  C_{\kappa}^2 \right>  \Big< \norm{\alpha_s}_{\dot{\mathfrak{h}}_{1/2}}^2
+ \norm{i \varepsilon \nabla p_s}_{\mathfrak{S}^{\infty}}^2 \Big>
\nonumber \\
&\qquad  
\times
\Big(  \beta(s) + \varepsilon^2 
+ N^{-1} \sup_{k \in \mathbb{R}^3} \Big\{ (1 + \abs{k})^{-1} \norm{q_s e^{ikx}  p_s}_{\mathfrak{S}^1} + \norm{q_s \varepsilon i \nabla p_s}_{\mathfrak{S}^1}  \Big\} 
\Big)
\bigg] \, ds  .
\end{align}
\end{lemma}

\begin{proof}
Let $\xi_{N,t}$, $\vAc(x,t)$, and $\mathcal{N}_{\delta}$ be defined as in the proof of Lemma~\ref{lemma:estimate for beta-b-1}.
By means of \eqref{eq:time derivative of the fluctuation vector} and the symmetry of the many-body wave function we compute 
\begin{align}
\frac{d}{dt}
\varepsilon^2 \scp{\xi_{N,t}}{\mathcal{N}_{\delta}^{1/2} \xi_{N,t}}
&= i \varepsilon \scp{\xi_{N,t}}{\left[ \mathcal{G}(t), \mathcal{N}_{\delta}^{1/2} \right] \xi_{N,t}}
\nonumber \\
&=  2 \Im 
\scp{\xi_{N,t}}{ \left[ \mathcal{N}_{\delta}^{1/2} ,  \vAhk(x_1)  \right] \cdot \left(  i \varepsilon \nabla_1 + \varepsilon^2 \vAhk(x_1) + \vAk(x_1,t) \right)  \xi_{N,t}} 
\nonumber \\
&\quad + \Im
\scp{\xi_{N,t}}{ \left[  \mathcal{N}_{\delta}^{1/2} ,
\big( a(\abs{\cdot} \alpha_t - i \dot{\alpha}_t) + a^*(\abs{\cdot} \alpha_t - i \dot{\alpha}_t) \big)  \right]  \xi_{N,t}} .
\end{align}
Using \eqref{eq:commutator of regularized number operator 3},  \eqref{eq:commutator of regularized number operator 4} and $\left[ \big( \left( \mathcal{N} + 1 \right)^{1/2} - \mathcal{N}^{1/2} \big) \vAp(x_1) \cdot , \left( i \varepsilon \nabla_1 + \vAc(x_1,t) \right) \right] = 0$ we get
\begin{subequations}
\begin{align}
&\lim_{\delta \rightarrow 0} \frac{d}{dt}
\varepsilon^2 \scp{\xi_{N,t}}{\mathcal{N}_{\delta}^{1/2} \xi_{N,t}}
\nonumber \\
\label{eq:time derivative beta-b-2 A}
&\quad = - 2 \varepsilon^2 \Im
\scp{\xi_{N,t}}{\big( \left( \mathcal{N} + 1 \right)^{1/2} - \mathcal{N}^{1/2} \big)  \vAp(x_1) \cdot  \vAhk(x_1)  \xi_{N,t}} 
\\
\label{eq:time derivative beta-b-2 B}
&\qquad + 2 \varepsilon^2
\Im \scp{\xi_{N,t}}{\big( \mathcal{N}^{1/2} - \left( \mathcal{N} - 1 \right)^{1/2} \big)  \vAm(x_1) \cdot  \vAhk(x_1)  \xi_{N,t}} 
\\
\label{eq:time derivative beta-b-2 C}
&\qquad + 4 \Im  
\scp{\xi_{N,t}}{\big( \mathcal{N}^{1/2} -  \left( \mathcal{N} - 1 \right)^{1/2} \big)  \vAm(x_1) \cdot \left(  i \varepsilon \nabla_1 + \vAk(x_1,t) \right)  \xi_{N,t}}  
\\
\label{eq:time derivative beta-b-2 D}
&\qquad  +  2 \Im   
\scp{\xi_{N,t}}{
\big( \mathcal{N}^{1/2} - \left( \mathcal{N} - 1 \right)^{1/2} \big) a^*(\abs{\cdot} \alpha_t - i \dot{\alpha}_t) \xi_{N,t}} . 
\end{align}
\end{subequations}
In the following we will repeatedly use $\mathcal{N}^{1/2} -  \left( \mathcal{N} - 1 \right)^{1/2} \leq \frac{1}{2} \, \mathcal{N}^{-1/2}$ which holds due to the spectral theorem and $(x+1)^{1/2} - x^{1/2} \leq  \frac{1}{2} x^{-1/2}$ with $x \geq 0$.

\noindent
\textbf{The terms \eqref{eq:time derivative beta-b-2 A} and \eqref{eq:time derivative beta-b-2 B}:}
By \eqref{eq:shifting property number operator} and 
\eqref{eq:Bounds for vector potential} we obtain
\begin{align}
\abs{\eqref{eq:time derivative beta-b-2 A}}
&\leq 
2 \varepsilon^2
\abs{\scp{\left( \mathcal{N}+ 3 \right)^{1/4}\xi_{N,t}}{\big( \left( \mathcal{N} + 1 \right)^{1/2} - \mathcal{N}^{1/2} \big)  \vAp(x_1) \cdot  \vAp(x_1) \left( \mathcal{N} + 1 \right)^{-1/4} \xi_{N,t}} }
\nonumber \\
&\quad +
2 \varepsilon^2
\abs{\scp{\left(\mathcal{N}+ 1 \right)^{1/4} \xi_{N,t}}{\big( \left( \mathcal{N} + 1 \right)^{1/2} - \mathcal{N}^{1/2} \big)  \vAp(x_1) \cdot  \vAm(x_1) \left( \mathcal{N} + 1 \right)^{-1/4} \xi_{N,t}} }
\nonumber \\
&\leq C \varepsilon^2 
\big\| \abs{\cdot}^{-1/2} \mathcal{F}[\kappa] \big\|^2_{L^2}
\norm{\left( \mathcal{N} + 1 \right)^{3/4} \big( \left( \mathcal{N} + 1 \right)^{1/2} - \mathcal{N}^{1/2} \big) \xi_{N,t} }
\norm{\left( \mathcal{N} + 1 \right)^{1/4} \xi_{N,t}}
\nonumber \\
&\leq C \varepsilon^2 
\big\| \abs{\cdot}^{-1/2} \mathcal{F}[\kappa] \big\|^2_{L^2}
\norm{\left( \mathcal{N} + 1 \right)^{1/4} \xi_{N,t} }^2
\nonumber \\
&\leq C 
\big\| \abs{\cdot}^{-1/2} \mathcal{F}[\kappa] \big\|^2_{L^2}  \left( \beta^{b,2}(t) + \varepsilon^2 \right)
\end{align}
and
\begin{align}
\abs{\eqref{eq:time derivative beta-b-2 B}}
&\leq 
2 \varepsilon^2
\abs{\scp{\left( \mathcal{N}+ 1 \right)^{1/4}\xi_{N,t}}{\big( \mathcal{N}^{1/2} - \left( \mathcal{N} - 1 \right)^{1/2} \big)   \vAm(x_1) \cdot  \vAp(x_1) \left( \mathcal{N} + 1 \right)^{-1/4} \xi_{N,t}} }
\nonumber \\
&\quad +
2 \varepsilon^2
\abs{\scp{\left(\mathcal{N}+ 1 \right)^{1/4} \xi_{N,t}}{\big( \mathcal{N}^{1/2} - \left( \mathcal{N} - 1 \right)^{1/2} \big)   \vAm(x_1) \cdot  \vAm(x_1) \left( \mathcal{N} + 3 \right)^{-1/4} \xi_{N,t}} }
\nonumber \\
&\leq C \varepsilon^2 
\big\| \abs{\cdot}^{-1/2} \mathcal{F}[\kappa] \big\|^2_{L^2}
\norm{\left( \mathcal{N} + 1 \right)^{3/4} \big( \mathcal{N}^{1/2} - \left( \mathcal{N} - 1 \right)^{1/2} \big) \id_{\mathcal{N} \geq 1}  \xi_{N,t} }
\norm{\left( \mathcal{N} + 1 \right)^{1/4} \xi_{N,t}} 
\nonumber \\
&\leq C \varepsilon^2 
\big\| \abs{\cdot}^{-1/2} \mathcal{F}[\kappa] \big\|^2_{L^2}
\norm{ \xi_{N,t} }
\norm{\left( \mathcal{N} + 1 \right)^{1/4} \xi_{N,t}} 
\nonumber \\
&\quad + C \varepsilon^2 
\big\| \abs{\cdot}^{-1/2} \mathcal{F}[\kappa] \big\|^2_{L^2}
\norm{\left( \mathcal{N} + 1 \right)^{3/4} \left( \mathcal{N} - 1 \right)^{-1/2} \id_{\mathcal{N} \geq 2}  \xi_{N,t} }
\norm{\left( \mathcal{N} + 1 \right)^{1/4} \xi_{N,t}} 
\nonumber \\
&\leq C \varepsilon^2 
\big\| \abs{\cdot}^{-1/2} \mathcal{F}[\kappa] \big\|^2_{L^2}
\norm{\left( \mathcal{N} + 1 \right)^{1/4} \xi_{N,t}}^2
\nonumber \\ 
&\leq 
C  \big\| \abs{\cdot}^{-1/2} \mathcal{F}[\kappa] \big\|^2_{L^2}
\left( \beta^{b,2}(t) + \varepsilon^2 \right) .
\end{align}

\noindent
\textbf{The terms \eqref{eq:time derivative beta-b-2 C} and \eqref{eq:time derivative beta-b-2 D}:}
Inserting $\id_{\mathcal{H}^{(N)}} = p_{t,1} + q_{t,1}$ we write \eqref{eq:time derivative beta-b-2 C} as 
\begin{subequations}
\begin{align}
\label{eq:time derivative beta-b-2 C 1}
\eqref{eq:time derivative beta-b-2 C}
&=  4 \Im  
\scp{\xi_{N,t}}{\big( \mathcal{N}^{1/2} - \big( \left( \mathcal{N} - 1 \right)^{1/2} \big) q_{t,1}  \vAm(x_1) \cdot \left(  i \varepsilon \nabla_1 + \vAk(x_1,t) \right)  p_{t,1} \xi_{N,t}}  
\\
\label{eq:time derivative beta-b-2 C 2}
&\quad +  4 \Im  
\scp{\xi_{N,t}}{\big( \mathcal{N}^{1/2} - \big( \left( \mathcal{N} - 1 \right)^{1/2} \big) p_{t,1} \vAm(x_1) \cdot \left(  i \varepsilon \nabla_1 + \vAk(x_1,t) \right) q_{t,1} \xi_{N,t}}  
\\
\label{eq:time derivative beta-b-2 C 3}
&\quad +  4 \Im  
\scp{\xi_{N,t}}{\big( \mathcal{N}^{1/2} - \big( \left( \mathcal{N} - 1 \right)^{1/2} \big) q_{t,1} \vAm(x_1) \cdot \left(  i \varepsilon \nabla_1 + \vAk(x_1,t) \right) q_{t,1} \xi_{N,t}}  
\\
\label{eq:time derivative beta-b-2 C 4}
&\quad +  4 \Im  
\scp{\xi_{N,t}}{\big( \mathcal{N}^{1/2} - \big( \left( \mathcal{N} - 1 \right)^{1/2} \big) p_{t,1} \vAm(x_1) \cdot \left(  i \varepsilon \nabla_1 + \vAk(x_1,t) \right)  p_{t,1} \xi_{N,t}}  .
\end{align}
\end{subequations}
By means of \eqref{eq:estimate vector potential negative part with p and q within scalar product} with $m=0$ and $O =  i \varepsilon \nabla + \vAk(\cdot,t) $ we get
\begin{align}
\abs{\eqref{eq:time derivative beta-b-2 C 1} }
&\leq 
C N^{-1} \sup_{k \in \mathbb{R}^3} \Big\{ (1 + \abs{k})^{-1} \norm{q e^{ikx} \left(  i \varepsilon \nabla + \vAk(\cdot,t) \right) p}_{\mathfrak{S}^1}  \Big\}
\norm{(\abs{\cdot}^{1/2} + \abs{\cdot}^{-1/2}) \mathcal{F}[\kappa] }_{L^2} 
\nonumber \\
&\qquad \times
\norm{\mathcal{N}^{1/2} \big( \mathcal{N}^{1/2} - \big( \left( \mathcal{N} - 1 \right)^{1/2} \big) \xi_{N,t}} \norm{\xi_{N,t}}
\nonumber \\
&\leq 
C N^{-1} \sup_{k \in \mathbb{R}^3} \Big\{ (1 + \abs{k})^{-1} \norm{q e^{ikx} \left(  i \varepsilon \nabla + \vAk(\cdot,t) \right) p}_{\mathfrak{S}^1}  \Big\}
\norm{(\abs{\cdot}^{1/2} + \abs{\cdot}^{-1/2}) \mathcal{F}[\kappa] }_{L^2} .
\end{align}
Similarly, 
\begin{align}
\abs{\eqref{eq:time derivative beta-b-2 C 2} }
&\leq 
4  \abs{\scp{ \xi_{N,t}}{ q_{t,1} \vAp(x_1) \cdot \left(  i \varepsilon \nabla_1 + \vAk(x_1,t) \right) p_{t,1} \big( \mathcal{N}^{1/2} - \big( \left( \mathcal{N} - 1 \right)^{1/2} \big) \xi_{N,t}}}
\nonumber \\
&\leq 
C N^{-1} \sup_{k \in \mathbb{R}^3} \Big\{ (1 + \abs{k})^{-1} \norm{q_t e^{ikx} \left(  i \varepsilon \nabla + \vAk(\cdot,t) \right) p_t}_{\mathfrak{S}^1}  \Big\}
\norm{(\abs{\cdot}^{1/2} + \abs{\cdot}^{-1/2}) \mathcal{F}[\kappa] }_{L^2} 
\nonumber \\
&\qquad \times
\norm{\mathcal{N}^{1/2} \big( \mathcal{N}^{1/2} - \big( \left( \mathcal{N} - 1 \right)^{1/2} \big) \xi_{N,t}} \norm{\xi_{N,t}}
\nonumber \\
&\leq 
C N^{-1} \sup_{k \in \mathbb{R}^3} \Big\{ (1 + \abs{k})^{-1} \norm{q_t e^{ikx} \left(  i \varepsilon \nabla + \vAk(\cdot,t) \right) p_t}_{\mathfrak{S}^1}  \Big\}
\norm{(\abs{\cdot}^{1/2} + \abs{\cdot}^{-1/2}) \mathcal{F}[\kappa] }_{L^2}  
\end{align}
by means of 
$\left[ \vAm(x_1) \cdot , \left(  i \varepsilon \nabla_1 + \vAk(x_1,t) \right) \right] = 0$
and \eqref{eq:estimate vector potential positive part with p and q within scalar product}.
Using the Cauchy--Schwarz inequality, \eqref{eq:Bounds for vector potential} and \eqref{eq:estimate for A L-infty to h-1-2}  we get
\begin{align}
\abs{\eqref{eq:time derivative beta-b-2 C 3} }
&\leq 
4  \norm{\vAp(x_1) \big( \mathcal{N}^{1/2} - \big( \left( \mathcal{N} - 1 \right)^{1/2} \big) q_{t,1} \xi_{N,t}}
\norm{\left(  i \varepsilon \nabla_1 + \vAk(x_1,t) \right) q_{t,1} \xi_{N,t}}  
\nonumber \\
&\leq 
4 \big\|\abs{\cdot}^{-1/2} \mathcal{F}[\kappa] \big\|_{L^2(\mathbb{R}^3)}
\norm{\mathcal{N}^{1/2} \big( \mathcal{N}^{1/2} - \big( \left( \mathcal{N} - 1 \right)^{1/2} \big) q_{t,1} \xi_{N,t}}
\nonumber \\
&\quad \times
\norm{\left(  i \varepsilon \nabla_1 + \vAk(x_1,t) \right) q_{t,1} \xi_{N,t}}  
\nonumber \\
&\leq 
4 \big\|\abs{\cdot}^{-1/2} \mathcal{F}[\kappa] \big\|_{L^2(\mathbb{R}^3)}
\norm{q_{t,1} \xi_{N,t}}
\Big[ \norm{\vAk(\cdot,t)}_{L^{\infty}} \norm{q_{t,1} \xi_{N,t}}
+ \norm{i \varepsilon \nabla_1  q_{t,1} \xi_{N,t}}
\Big]
\nonumber \\
&\leq C \left( 1 + \big\|\abs{\cdot}^{-1/2} \mathcal{F}[\kappa] \big\|^2_{L^2(\mathbb{R}^3)} 
 + \norm{\vAk(\cdot,t)}_{L^{\infty}}^2 \right) \norm{q_{t,1} \Psi_{N,t}}^2
+ \norm{i \varepsilon \nabla_1 q_{t,1} \Psi_{N,t}}^2 
\nonumber \\
&\leq C \left( 1 + \big\| ( \abs{\cdot}^{-1/2} + \abs{\cdot}^{-1} ) \mathcal{F}[\kappa] \big\|^2_{L^2(\mathbb{R}^3)} \right)  \beta^a(t)
+ \norm{i \varepsilon \nabla_1 q_{t,1} \Psi_{N,t}}^2 .
\end{align}
Next, we estimate the remaining terms  \eqref{eq:time derivative beta-b-2 C 4} and \eqref{eq:time derivative beta-b-2 D}.
Using \eqref{eq:Maxwell-Schroedinger equations}, \eqref{eq:definition of G}, \eqref{eq:Fourier transform of the vector current}, $\left[ \boldsymbol{G}_{\cdot}(k,\lambda), i \varepsilon \nabla \right] = 0$ and the shorthand notations
$\Upsilon_{N,t} = \big( \mathcal{N}^{1/2} -  \left( \mathcal{N} - 1 \right)^{1/2} \big)\xi_{N,t}$ 
and $O = i \varepsilon \nabla  + \vAk(\cdot,t)$, we obtain
\begin{align}
\eqref{eq:time derivative beta-b-2 D}
&= 2  \, \Im \, \sum_{\lambda=1,2} \int_{\mathbb{R}^3}  
\sqrt{\frac{4 \pi^3}{\abs{k}}} \mathcal{F}[\kappa](k) \vep_{\lambda}(k) \cdot \mathcal{F}[\vJ_{p_t, \alpha_t}](k)
\scp{ a_{k,\lambda}  \Upsilon_{N,t}}{  \xi_{N,t}}
\, dk
\nonumber \\
&= 
- 4  N^{-1} \, \Im \, \sum_{\lambda=1,2} \int_{\mathbb{R}^3} 
\tr \big( \boldsymbol{G}_{\cdot}(k, \lambda) \cdot \left(  i \varepsilon \nabla  + \vAk(\cdot,t) \right) p_t \big) 
\scp{ a_{k,\lambda}  \Upsilon_{N,t}}{  \xi_{N,t}}  \, dk  
\end{align}
and 
\begin{align}
&\eqref{eq:time derivative beta-b-2 C 4} + \eqref{eq:time derivative beta-b-2 D}
\nonumber \\
&\quad = 
4 \Im  \sum_{\lambda = 1,2} \int_{\mathbb{R}^3}
\scp{a_{k,\lambda} \Upsilon_{N,t}}{p_{t,1} \boldsymbol{G}_{x_1}(k,\lambda) \cdot O_1  p_{t,1} \xi_{N,t}} 
\, dk
\nonumber \\
&\qquad 
- 4  N^{-1} \, \Im \, \sum_{\lambda=1,2} \int_{\mathbb{R}^3} 
\tr \big( \boldsymbol{G}_{\cdot}(k, \lambda) \cdot O p_t \big) 
\scp{ a_{k,\lambda}  \Upsilon_{N,t}}{  \xi_{N,t}}  \, dk 
\nonumber \\
&\quad = 
4  N^{-1} \, \Im \, \sum_{\lambda=1,2} \int_{\mathbb{R}^3}
\scp{ a_{k,\lambda}  \Upsilon_{N,t}}{\Big( \sum_{j=1}^N p_{t,j} \boldsymbol{G}_{x_j}(k,\lambda) \cdot O_j  p_{t,j}
-  \tr \big( \boldsymbol{G}_{\cdot}(k, \lambda) \cdot O p_t \big) \Big)  \xi_{N,t}}  \, dk .
\end{align}
If we do the splitting $\boldsymbol{G}_{\cdot}(k,\lambda) = \Re[\boldsymbol{G}_{\cdot}](k, \lambda) + i \Im[\boldsymbol{G}_{\cdot}](k, \lambda)$ and apply \eqref{eq:estimates of the p-p term 1} to the terms containing the real and imaginary part of $\boldsymbol{G}_{\cdot}$ we get
\begin{align}
\abs{\eqref{eq:time derivative beta-b-2 C 4} + \eqref{eq:time derivative beta-b-2 D} }
&\leq C \sum_{\lambda = 1, 2} \int_{\mathbb{R}^3}
\left( \norm{\Re[\boldsymbol{G}_{\cdot}](k,\lambda)}_{\mathfrak{S}^{\infty}}
+ \norm{\Im[\boldsymbol{G}_{\cdot}](k,\lambda)}_{\mathfrak{S}^{\infty}} \right)
\norm{i \varepsilon \nabla  + \vAk(\cdot,t) p_t}_{\mathfrak{S}^{\infty}}
\nonumber \\
&\qquad \qquad \times 
\norm{q_{t,1} a_k \Upsilon_{N,t}} \norm{q_{t,1} \xi_{N,t}} \, dk
\nonumber \\
&\leq C \sum_{\lambda = 1, 2} \int_{\mathbb{R}^3}
\abs{k}^{-1/2} \abs{\mathcal{F}[\kappa](k)}
\norm{i \varepsilon \nabla  + \vAk(\cdot,t) p_t}_{\mathfrak{S}^{\infty}}
\nonumber \\
&\qquad \qquad \times 
\norm{q_{t,1} \mathcal{N}^{1/2} \big( \mathcal{N}^{1/2} - \big( \left( \mathcal{N} - 1 \right)^{1/2} \big)\xi_{N,t}} \norm{q_{t,1} \xi_{N,t}}
\, dk
\nonumber \\
&\leq C 
 \big\|\abs{\cdot}^{-1/2} \mathcal{F}[\kappa] \big\|_{L^2(\mathbb{R}^3)}
\left( \norm{i \varepsilon \nabla p_t}_{\mathfrak{S}^{\infty}}
+ \norm{\vAk(\cdot,t) }_{\mathfrak{S}^{\infty}}
\right)
\norm{q_{t,1} \xi_{N,t}}^2 .
\end{align}
Due to \eqref{eq:estimate for A L-infty to h-1-2} we have
\begin{align}
\abs{\eqref{eq:time derivative beta-b-2 C 4} + \eqref{eq:time derivative beta-b-2 D} }
&\leq C  \left(  \big\| ( \abs{\cdot}^{-1/2} + \abs{\cdot}^{-1} ) \mathcal{F}[\kappa] \big\|^2_{L^2(\mathbb{R}^3)} 
\left( 1 + \norm{\alpha_t}_{\dot{\mathfrak{h}}_{1/2}}^2 \right)
+ \norm{i \varepsilon \nabla p_t}_{\mathfrak{S}^{\infty}}^2 \right)  \beta^a(t) .
\end{align}
Collecting the estimates and \eqref{eq:semiclassical structure for exponential times magnetic gradient} then leads to
\begin{align}
&\abs{\eqref{eq:time derivative beta-b-2 C} + \eqref{eq:time derivative beta-b-2 D}}
\nonumber \\
&\quad \leq \abs{\eqref{eq:time derivative beta-b-2 C 1}}
+ \abs{\eqref{eq:time derivative beta-b-2 C 2}}
+ \abs{\eqref{eq:time derivative beta-b-2 C 3}} 
+ \abs{\eqref{eq:time derivative beta-b-2 C 4} + \eqref{eq:time derivative beta-b-2 D} }
\nonumber \\
&\quad \leq 
 \norm{i \varepsilon \nabla_1 q_{t,1} \Psi_{N,t}}^2
+ C   \left( 1 +  \big\| ( \abs{\cdot}^{-1/2} + \abs{\cdot}^{-1} ) \mathcal{F}[\kappa] \big\|^2_{L^2(\mathbb{R}^3)} 
\left< \norm{\alpha_t}_{\dot{\mathfrak{h}}_{1/2}}^2 \right>
+ \norm{i \varepsilon \nabla p_t}_{\mathfrak{S}^{\infty}}^2 
\right)  \beta^a(t)
\nonumber \\
&\qquad +
C N^{-1} \sup_{k \in \mathbb{R}^3} \Big\{ (1 + \abs{k})^{-1} \norm{q_t e^{ikx} \left(  i \varepsilon \nabla + \vAk(\cdot,t) \right) p_t}_{\mathfrak{S}^1}  \Big\}
\norm{(\abs{\cdot}^{1/2} + \abs{\cdot}^{-1/2}) \mathcal{F}[\kappa] }_{L^2} 
\nonumber \\
&\quad \leq 
 \norm{i \varepsilon \nabla_1 q_{t,1} \Psi_{N,t}}^2
+ C   \left( 1 +  \big\| ( \abs{\cdot}^{-1/2} + \abs{\cdot}^{-1} ) \mathcal{F}[\kappa] \big\|^2_{L^2(\mathbb{R}^3)}  \left< \norm{\alpha_t}_{\dot{\mathfrak{h}}_{1/2}}^2 \right>
+ \norm{i \varepsilon \nabla p_t}_{\mathfrak{S}^{\infty}}^2 
\right)  
\nonumber \\
&\qquad \qquad \times 
\left(  \beta^a(t) 
+ N^{-1} \sup_{k \in \mathbb{R}^3} \Big\{ (1 + \abs{k})^{-1} \norm{q_t e^{ikx}  p_t}_{\mathfrak{S}^1} + \norm{q_t \varepsilon i \nabla p_t}_{\mathfrak{S}^1}  \Big\} 
\right)
\end{align}
and
\begin{align}
&\abs{\lim_{\delta \rightarrow 0} \frac{d}{dt}
\varepsilon^2 \scp{\xi_{N,t}}{\mathcal{N}_{\delta}^{1/2} \xi_{N,t}}}
\nonumber \\
&\quad \leq \abs{\eqref{eq:time derivative beta-b-2 A}}
+ \abs{\eqref{eq:time derivative beta-b-2 B}}
+ \abs{\eqref{eq:time derivative beta-b-2 C}}
+ \abs{\eqref{eq:time derivative beta-b-2 D}}
\nonumber \\
&\quad \leq 
 \norm{i \varepsilon \nabla_1 q_{t,1} \Psi_{N,t}}^2
+ C   \left( 1 +  \big\| ( \abs{\cdot}^{-1/2} + \abs{\cdot}^{-1} ) \mathcal{F}[\kappa] \big\|^2_{L^2(\mathbb{R}^3)}  \left( 1 + \norm{\alpha_t}_{\dot{\mathfrak{h}}_{1/2}}^2 \right)
+ \norm{i \varepsilon \nabla p_t}_{\mathfrak{S}^{\infty}}^2 
\right)  
\nonumber \\
&\qquad \qquad \times 
\left(  \beta^a(t) + \beta^{b,2}(t) + \varepsilon^2
+ N^{-1} \sup_{k \in \mathbb{R}^3} \Big\{ (1 + \abs{k})^{-1} \norm{q_t e^{ikx}  p_t}_{\mathfrak{S}^1} + \norm{q_t \varepsilon i \nabla p_t}_{\mathfrak{S}^1}  \Big\} 
\right) .
\end{align}
The Lemma then follows by Duhamel's formula, 
$\lim_{\delta \rightarrow 0} 
\varepsilon^2 \scp{\xi_{N,t}}{\mathcal{N}_{\delta} \xi_{N,t}} = \beta^{b,2}(t)$, the fact that for $\delta >0$ sufficiently small and $T>0$ 
the function $[0,T] \rightarrow \mathbb{R}$, $t \mapsto \frac{d}{dt}
\varepsilon^2 \scp{\xi_{N,t}}{\mathcal{N}^{1/2}_{\delta} \xi_{N,t}}$ has an integrable majorant uniform in $\delta$ and the dominated convergence theorem.
\end{proof}

\section{Estimates concerning the Vlasov--Maxwell equations}

\label{section:estimates concerning the Vlasov--Maxwell equations}

In the following, we will prove Remark~\ref{remark:remark on the semiclassical structure with Sobolev trace norm},  Lemma~\ref{lemma:Vlasov-Maxwell distance between initial data and its regularization updated}, Lemma~\ref{lemma:propagation estimates Vlasov-Maxwell},  Proposition~\ref{proposition:comparison between Maxwell--Schroedinger and Vlasov--Maxwell updated}, Lemma~\ref{lemma: semiclassical structure with kinetic weights},  and Proposition~\ref{proposition:Regularization result for Maxwell-Schroedinger}.
Recall that, in this section, we refrain from tracking the dependence on $\kappa$ and use the letter $C$ to denote a generic constant that depends on the choice of $\kappa$. Moreover, note that we are using the definitions \eqref{eq:definition fractional Laplacian with semiclassical weight} and \eqref{eq:definition Sobolev tracee norm with semiclassical weight}.
In the proofs, we will rely on the following estimates.
\begin{lemma}
\label{lemma:derivation Vlasov--Maxwell auxiliary estimates}
Let $j, k \in \mathbb{N}$, $l \in \mathbb{Z}$, $\sigma \in (0,1]$, $L \in (0, \infty)$, and $\Lambda \in [1, \infty)$. Moreover let, $\mathcal{E}^{\rm{VM}}$, $\mathscr{G}_{\sigma}$, and $\eta_L$ be defined as in  
\eqref{eq:Vlasov-Maxwell energy definition} and \eqref{eq:definition regularizing Gaussian}, and
\eqref{eq:definition smooth cutoff function}. Then,
\begin{subequations}
\begin{align}
\label{eq:estimate for regularized Wigner 1}
\sup \big\{
\norm{\eta_L \mathscr{G}_{\sigma} * f}_{W_k^{j,2}(\mathbb{R}^6)} ,
\norm{\mathscr{G}_{\sigma} * f}_{W_k^{j,2}(\mathbb{R}^6)} 
\big\}
&\leq C \sigma^{-j}  \norm{f}_{W_k^{0,2}(\mathbb{R}^6)} ,
\\
\label{eq:estimate for regularized Wigner 2}
\sup
\big\{
\norm{\eta_L \mathscr{G}_{\sigma} * f}_{W_2^{0,1}(\mathbb{R}^6)} ,
\norm{\mathscr{G}_{\sigma} * f}_{W_2^{0,1}(\mathbb{R}^6)}
\big\}
&\leq C \norm{f}_{W_2^{0,1}(\mathbb{R}^6)}
\leq C \norm{f}_{W_{6}^{0,2}(\mathbb{R}^6)}  ,
\\
\label{eq:estimate for alpha with cutoff}
\norm{\id_{\abs{\cdot} \leq \Lambda} \alpha}_{\dot{\mathfrak{h}}_l} &\leq 
\norm{ \alpha}_{\dot{\mathfrak{h}}_l} ,
\\
\label{eq:estimate for VM-energy with regularized data}
\abs{\mathcal{E}^{\rm{VM}}[\eta_L \mathscr{G}_{\sigma} * f , \id_{\abs{\cdot} \leq \Lambda} \alpha]} &\leq 
C \left( 1 + \norm{\alpha}_{\dot{\mathfrak{h}}_{1/2}}^2 \right)
\left( 1 + \norm{f}^2_{W_{6}^{0,2}(\mathbb{R}^6)} \right) 
\end{align}
\end{subequations}
for $f$ and $\alpha$ chosen appropriately such that the norms on the right-hand side are finite. 
Moreover, for $R >0$,  $\widetilde{C} >0$, and $f \in W_k^{0,2}(\mathbb{R}^6)$ satisfying
$\norm{\chi_{\geq R}  f}_{W_k^{0,2}(\mathbb{R}^6)} \leq \widetilde{C} e^{- \frac{1}{2 \sigma^2}}$,  where
$\chi_{\geq R}: \mathbb{R}^6 \rightarrow \mathbb{R}$ is defined by $\chi_{\geq R}(x,v) = \id_{\abs{v} \geq R}$, 
 we have
\begin{align}
\label{eq:estimate Husimi with cutoff of the high velocities}
\norm{( 1 - \eta_{R}) \mathscr{G}_{\sigma} * f}_{W_k^{0,2}(\mathbb{R}^6)}
&\leq C e^{- \frac{1}{2 \sigma^2}} 
 \left(  \norm{f}_{W_k^{0,2}(\mathbb{R}^6)} + 1 \right) .
\end{align}
\end{lemma}

\begin{proof}[Proof of Lemma~\ref{lemma:derivation Vlasov--Maxwell auxiliary estimates}]
Note that 
$\big| \big( \left< \cdot \right>^{k} f*g \big)(z) \big|
\leq C \big( \big( \left< \cdot \right>^{k} \abs{f} \big) * \big( \left< \cdot \right>^{k} \abs{g} \big) \big)(z)$
for $z \in \mathbb{R}^6$, $k \in \mathbb{N}$ and $f,g: \mathbb{R}^6 \rightarrow \mathbb{R}$. Together with Young's inequality this gives
\begin{align}
\norm{\mathscr{G}_{\sigma} * f}_{W_k^{j,2}}^2 
&\leq  \sum_{\abs{\alpha} \leq j} \norm{\left( \left< \cdot \right>^k \abs{D_z^{\alpha} \mathscr{G}_{\sigma}} \right) * \left( \left< \cdot \right>^k \abs{f} \right)}_{L^2(\mathbb{R}^6)}^2
\nonumber \\
&\leq  \sum_{\abs{\alpha} \leq j} \norm{\left< \cdot \right>^k D^{\alpha} \mathscr{G}_{\sigma} }_{L^1(\mathbb{R}^6)}^2
\norm{ \left< \cdot \right>^k f }_{L^2(\mathbb{R}^6)}^2
\nonumber \\
&\leq C
\norm{\mathscr{G}_{\sigma}}_{W_k^{j,1}(\mathbb{R}^6)}^2
\norm{f}_{W_k^{0,2}(\mathbb{R}^6)}^2 .
\end{align}
Since
\begin{align}
\partial_{x_j} \mathscr{G}_{\sigma}(x,v) = - \left( \sigma^2 \pi \right)^{-3}  \frac{2 x_j}{\sigma^2} \exp \left[ -    \frac{x^2 + v^2}{\sigma^2} \right]
\end{align}
we get
$\norm{\mathscr{G}_{\sigma}}_{W_k^{j,1}(\mathbb{R}^6)} \leq C \sigma^{-j} \left< \sigma \right>^k$ 
and therefore
$\norm{\mathscr{G}_{\sigma} * f}_{W_k^{j,2}}
\leq C \sigma^{-j}   \norm{f}_{W_k^{0,2}(\mathbb{R}^6)}$. By means of \eqref{eq:properties of the smooth cutoff function 1} we obtain
\begin{align}
\norm{\eta_L \mathscr{G}_{\sigma} * f}_{W_k^{j,2}} &\leq \norm{\eta_L}_{W_0^{j,\infty}(\mathbb{R}^3)} \norm{\mathscr{G}_{\sigma} * f}_{W_k^{j,2}}
\leq C \norm{\mathscr{G}_{\sigma} * f}_{W_k^{j,2}} .
\end{align}
In total, this shows \eqref{eq:estimate for regularized Wigner 1}. Using \eqref{eq:estimate for regularized Wigner 1} and similar estimates we get
\begin{align}
\norm{\eta_L \mathscr{G}_{\sigma} * f}_{W_2^{0,1}(\mathbb{R}^6)}
&\leq 
\norm{\mathscr{G}_{\sigma} * f}_{W_2^{0,1}(\mathbb{R}^6)}
\leq 
 \norm{\mathscr{G}_{\sigma} }_{W_2^{0,1}(\mathbb{R}^6)}
\norm{f}_{W_2^{0,1}(\mathbb{R}^6)}
\leq C 
\norm{f}_{W^{0,1}_2(\mathbb{R}^6)} .
\end{align}
Due to the Cauchy--Schwarz inequality we, moreover, have
\begin{align}
\norm{f}_{W^{0,1}_2(\mathbb{R}^6)}
&\leq \norm{\left< \cdot \right>^{-4}}_{L^2(\mathbb{R}^6)} \norm{\left< \cdot \right>^{6} f}_{L^2(\mathbb{R}^6)}
\leq C \norm{f}_{W^{0,2}_6(\mathbb{R}^6)} ,
\end{align}
which proves \eqref{eq:estimate for regularized Wigner 2}. Inequality \eqref{eq:estimate for alpha with cutoff} follows from the pointwise estimates $\id_{\abs{k} \leq \Lambda} \abs{\alpha(k,\lambda)} \leq \abs{\alpha(k,\lambda)}$ for $k \in \mathbb{R}^3$ and $\lambda \in \{1,2\}$. Using the estimate \cite[(III.35c)]{LS2023}, i.e.
\begin{align}
\mathcal{E}^{\rm{VM}}[f,\alpha]
&\leq  2 \int_{\mathbb{R}^6} f(x,v) v^2 dx \, dv + C \norm{f}_{L^1(\mathbb{R}^6)}^2 
+ \norm{\alpha}_{\dot{\mathfrak{h}}_{1/2}}^2
\left( C \norm{f}_{L^1(\mathbb{R}^6)} + 1 \right) ,
\end{align}
in combination with \eqref{eq:estimate for regularized Wigner 2} and \eqref{eq:estimate for alpha with cutoff}
leads to \eqref{eq:estimate for VM-energy with regularized data}. 
For the proof of \eqref{eq:estimate Husimi with cutoff of the high velocities}, recall that, for any $R >0$, we define $\chi_{\geq R}: \mathbb{R}^6 \rightarrow \mathbb{R}$ by $\chi_{\geq R}(x,v) = \id_{\abs{v} \geq R}$.  Next,  we define
\begin{align}
f_{\leq R}(x,v) = \id_{\abs{v} \leq R} \, f(x,v)
\quad \text{and} \quad 
f_{\geq R}(x,v) = \id_{\abs{v} \geq R} \, f(x,v)
\end{align}
such that $f = f_{\leq R} + f_{\geq R}$.
Note that $\supp f_{\leq R} \subseteq \{ (x,v) \in \mathbb{R}^6, \abs{v} \leq R \}$ and \eqref{eq:properties of the smooth cutoff function 2} implies
\begin{align}
(1 - \eta_{R})(v) (\mathscr{G}_{\sigma} * f_{\leq R})(x,v)
&= \id_{\abs{v} \geq R+1} (1 - \eta_R)(v) \int_{\mathbb{R}^6}
\mathscr{G}_{\sigma}(x-y,v-u) f_{\leq R}(y,u) \id_{\abs{u} \leq R} \, dy \, du
\nonumber \\
&= (1 - \eta_{R})(v) \big( ( \chi_{\geq 1} \mathscr{G}_{\sigma} ) * f_{\leq R} \big)(x,v)
\end{align}
because $\abs{v-u} \geq 1$ if $\abs{v} \geq R+1$ and $\abs{u} \leq R$. Using similar estimates as in the proof of \eqref{eq:estimate for regularized Wigner 1}, \eqref{eq:properties of the smooth cutoff function 1},
$\mathscr{G}_{\sigma} \chi_{\geq 1} \leq C e^{- \frac{1}{2 \sigma^2}} \mathscr{G}_{\sqrt{2} \sigma}$
, and $\big\| \mathscr{G}_{\sigma} \big\|_{W_k^{0,1}(\mathbb{R}^6)} \leq C \left< \sigma \right>^k$ we get
\begin{align}
\norm{( 1 - \eta_{R}) \mathscr{G}_{\sigma} * f_{\leq R}}_{W_k^{0,2}(\mathbb{R}^6)}
&\leq \norm{( \chi_{\geq 1} \mathscr{G}_{\sigma} ) * f_{\leq R} }_{W_k^{0,2}(\mathbb{R}^6)}
\nonumber \\
&\leq C \norm{\chi_{\geq 1} \mathscr{G}_{\sigma}  }_{W_k^{0,1}(\mathbb{R}^6)}
\norm{ f_{\leq R} }_{W_k^{0,2}(\mathbb{R}^6)}
\nonumber \\
&\leq C e^{- \frac{1}{2 \sigma^2}} \big\| \mathscr{G}_{\sqrt{2} \sigma} \big\|_{W_k^{0,1}(\mathbb{R}^6)}
\norm{ f_{\leq R} }_{W_k^{0,2}(\mathbb{R}^6)}
\nonumber \\
&\leq C e^{- \frac{1}{2 \sigma^2}} 
\norm{ f_{\leq R} }_{W_k^{0,2}(\mathbb{R}^6)} 
\nonumber \\
&\leq C e^{- \frac{1}{2 \sigma^2}} 
\norm{ f }_{W_k^{0,2}(\mathbb{R}^6)} .
\end{align}
By means of   $\big\| \mathscr{G}_{\sigma} \big\|_{W_k^{0,1}(\mathbb{R}^6)} \leq C \left< \sigma \right>^k$ we, moreover,  have
\begin{align}
\norm{( 1 - \eta_{R}) \mathscr{G}_{\sigma} * f_{\geq R}}_{W_k^{0,2}(\mathbb{R}^6)}
&\leq
\norm{ \mathscr{G}_{\sigma} * f_{\geq R}}_{W_k^{0,2}(\mathbb{R}^6)}
\nonumber \\
&\leq C \norm{\mathscr{G}_{\sigma}  }_{W_k^{0,1}(\mathbb{R}^6)}
\norm{ f_{\geq R} }_{W_k^{0,2}(\mathbb{R}^6)}
\nonumber \\
&\leq C  \norm{ f_{\geq R}}_{W_k^{0,2}(\mathbb{R}^6)} .
\end{align}
Combining both estimates,  using $f_{\geq R} = \chi_{\geq R}  f$, and applying the assumptions of Lemma~\ref{lemma:derivation Vlasov--Maxwell auxiliary estimates},  proves \eqref{eq:estimate Husimi with cutoff of the high velocities}.
\end{proof}

\subsection{Proof of Remark~\ref{remark:remark on the semiclassical structure with Sobolev trace norm}}

\label{subsection:Proof of Remark:remark on the semiclassical structure with Sobolev trace norm}

\begin{proof}[\unskip\nopunct]

For $\nu >0$ we define the spectral projections $\Pi_{\leq \nu}: L^2(\mathbb{R}^3) \rightarrow L^2(\mathbb{R}^3)$ and  $\Pi_{\geq \nu}: L^2(\mathbb{R}^3) \rightarrow L^2(\mathbb{R}^3)$ by 
$\left( \Pi_{\leq \nu} \psi \right) = \mathcal{F}^{-1} \left[ \id_{\abs{\cdot} \leq \nu} \mathcal{F} \left[ \psi \right] \right]$ and $\Pi_{\geq \nu} = 1 - \Pi_{\leq \nu}$. Using the properties of the Fourier transform one easily verifies
\begin{align}
\left[ \nabla, \Pi_{\leq \nu} \right] = 0 ,
\quad 
e^{ikx} \Pi_{\leq \nu}
= \Pi_{\leq \nu + \abs{k}}  e^{ikx} \Pi_{\leq \nu},
\quad 
\norm{\sqrt{1 - \varepsilon^2 \Delta} \Pi_{\leq \nu}}_{\mathfrak{S}^{\infty}}
\leq \sqrt{1 + \varepsilon^2 \nu^2} .
\end{align} 
For $p = \sum_{j=1}^N \ket{\varphi_j} \bra{\varphi_j}$ with $\varphi_j \in L^2(\mathbb{R}^3)$ such that $\id_{\abs{k} \geq \widetilde{C} N^{1/3} } \mathcal{F}[\varphi_j](k) =0$ for all $j \in \{1,2, \ldots,N \}$ and $\nu = \widetilde{C} N^{1/3}$, we have $p = \Pi_{\leq \nu} p = p \Pi_{\leq \nu} $,
$\left[ \Pi_{\leq \nu}, p \right] = 0$ and 
$\left[ \Pi_{\leq \nu}, q \right] = 0$. This implies
\begin{align}
\norm{\sqrt{1 - \varepsilon^2 \Delta} q e^{ikx} p \sqrt{1 - \varepsilon^2 \Delta}}_{\mathfrak{S}^1} 
&= \norm{\sqrt{1 - \varepsilon^2 \Delta} \Pi_{\leq \nu + \abs{k}} q e^{ikx} p \Pi_{\leq \nu} \sqrt{1 - \varepsilon^2 \Delta}}_{\mathfrak{S}^1} 
\nonumber \\
&\leq 
\norm{\sqrt{1 - \varepsilon^2 \Delta} \Pi_{\leq \nu + \abs{k}}}_{\mathfrak{S}^{\infty}}
\norm{\sqrt{1 - \varepsilon^2 \Delta} \Pi_{\leq \nu }}_{\mathfrak{S}^{\infty}}
\norm{q e^{ikx} p}_{\mathfrak{S}^1}
\nonumber \\
&\leq \sqrt{1 + \varepsilon^2 (\nu + \abs{k})^2} \sqrt{1 + \varepsilon^2 \nu^2} \norm{q e^{ikx} p}_{\mathfrak{S}^1}
\nonumber \\
&\leq 2 \left( 1 + \varepsilon^2 \nu^2 \right) \norm{q e^{ikx} p}_{\mathfrak{S}^1}
+ 2 \varepsilon \abs{k} \sqrt{1 + \varepsilon^2 \nu^2} \norm{p}_{\mathfrak{S}^1}
\nonumber \\
&\leq 4 \big( 1 + \widetilde{C} \big)^3 \left( 1 + \abs{k} \right) N^{2/3}
\end{align}
and
\begin{align}
\norm{\sqrt{1 - \varepsilon^2 \Delta} q i \varepsilon \nabla p \sqrt{1 - \varepsilon^2 \Delta}}_{\mathfrak{S}^1}
&=  \norm{\sqrt{1 - \varepsilon^2 \Delta} \Pi_{\leq \nu} q i \varepsilon \nabla p \Pi_{\leq \nu}  \sqrt{1 - \varepsilon^2 \Delta}}_{\mathfrak{S}^1}
\nonumber \\
&\leq \norm{\sqrt{1 - \varepsilon^2 \Delta} \Pi_{\leq \nu}}_{\mathfrak{S}^{\infty}}^2 
\norm{q i \varepsilon \nabla p}_{\mathfrak{S}^1}
\nonumber \\
&\leq 
\left( 1 + \varepsilon^2 \nu^2 \right)  \norm{q i \varepsilon \nabla p}_{\mathfrak{S}^1} 
\nonumber \\
&\leq \big( 1 + \widetilde{C} \big)^3 N^{2/3} ,
\end{align}
which proves the claim.
\end{proof}

\subsection{Proof of Lemma \ref{lemma:Vlasov-Maxwell distance between initial data and its regularization updated}}

\label{subsection:Vlasov-Maxwell distance between initial data and its regularization}

\begin{proof} 
Note that 
\begin{align}
\mathcal{W}^{-1}[m_{p,\sigma}](x;y)
&=  \int_{\mathbb{R}^{6}} 
\mathscr{G}_{\sigma}(z,\xi)  
\left( e^{i \xi \cdot \frac{\hat{x}}{\varepsilon}}  p e^{- i \xi \cdot \frac{\hat{x}}{\varepsilon}} \right) \left( x - z ; y - z  \right)
\, dz \, d \xi   .
\end{align}
Using $\int_{\mathbb{R}^6} \mathscr{G}_{\sigma}(z, \xi) \, dz \, d \xi = 1$ and $p =  e^{i \xi \cdot \frac{\hat{x}}{\varepsilon}} \left[  e^{- i \xi \cdot \frac{\hat{x}}{\varepsilon}} , p \right] + e^{i \xi \cdot \frac{\hat{x}}{\varepsilon}}  p e^{- i \xi \cdot \frac{\hat{x}}{\varepsilon}}$
we get 
\begin{align}
p(x;y)
&= \int_{\mathbb{R}^6} \mathscr{G}_{\sigma}(z, \xi) \left[ \left( e^{i \xi \cdot \frac{\hat{x}}{\varepsilon}} \left[  e^{- i \xi \cdot \frac{\hat{x}}{\varepsilon}} , p \right]  \right) (x;y)
+ p(x;y) e^{ i  \xi \cdot \frac{x-y}{\varepsilon}}
\right] \, dz \, d \xi ,
\end{align}
which leads to
\begin{align}
& D_{\varepsilon}
\left( \mathcal{W}^{-1}[m_{p,\sigma}] - p \right) D_{\varepsilon}
\nonumber \\
&\quad =  - \int_{\mathbb{R}^6} \mathscr{G}_{\sigma}(z, \xi)  \left( D_{\varepsilon} e^{i \xi \cdot \frac{\hat{x}}{\varepsilon}} \left[  e^{- i \xi \cdot \frac{\hat{x}}{\varepsilon}} , p \right]  D_{\varepsilon} \right) \, dz \, d \xi
\nonumber \\
&\qquad 
+ \int_{\mathbb{R}^6} \mathscr{G}_{\sigma}(z, \xi) \Big( D_{\varepsilon} e^{i \xi \cdot \frac{\hat{x}}{\varepsilon}} \big( p(\cdot -z; \cdot -z) -  p \big) e^{-i \xi \cdot \frac{\hat{x}}{\varepsilon}}D_{\varepsilon} \Big)  \, dz \, d \xi .
\end{align}
Taking the trace norm then gives
\begin{subequations}
\begin{align}
\label{eq:comparison between regularized and normal p in Sobolev trace norm estimate 1}
\norm{ \mathcal{W}^{-1}[m_{p,\sigma}] - p}_{S^{1,1}_{\varepsilon}}
&\leq   \int_{\mathbb{R}^6} \mathscr{G}_{\sigma}(z, \xi)  \norm{ e^{i \xi \cdot \frac{\hat{x}}{\varepsilon}} \left[  e^{- i \xi \cdot \frac{\hat{x}}{\varepsilon}} , p \right] }_{S^{1,1}_{\varepsilon}} \, dz \, d \xi
\\
\label{eq:comparison between regularized and normal p in Sobolev trace norm estimate 2}
&\qquad 
+ \int_{\mathbb{R}^6} \mathscr{G}_{\sigma}(z, \xi) \norm{  e^{i \xi \cdot \frac{\hat{x}}{\varepsilon}} \big( p(\cdot -z; \cdot -z) -  p \big) e^{-i \xi \cdot \frac{\hat{x}}{\varepsilon}} }_{S^{1,1}_{\varepsilon}}  \, dz \, d \xi .
\end{align}
\end{subequations}
Inserting the identities $1 = p + q$ and $1 =D_{\varepsilon}^{-1} D_{\varepsilon} $ let us estimate
\begin{align}
\norm{e^{i \xi \cdot \frac{\hat{x}}{\varepsilon}} \left[  e^{- i \xi \cdot \frac{\hat{x}}{\varepsilon}} , p \right] }_{\mathfrak{S}^{1,1}_{\varepsilon}}
&\leq 
\norm{ D_{\varepsilon} e^{i \xi \cdot \frac{\hat{x}}{\varepsilon}} D_{\varepsilon}^{-1} }_{\mathfrak{S}^{\infty}} 
\Big(
\norm{q  e^{- i \xi \cdot \frac{\hat{x}}{\varepsilon}}  p}_{\mathfrak{S}^{1,1}_{\varepsilon}}
+  \norm{p  e^{- i \xi \cdot \frac{\hat{x}}{\varepsilon}}  q}_{\mathfrak{S}^{1,1}_{\varepsilon}}
\Big) .
\end{align}
Since
\begin{align}
\norm{q  e^{i \xi \cdot \frac{\hat{x}}{\varepsilon}}  p }_{\mathfrak{S}^{1,1}_{\varepsilon}}
&\leq C N \varepsilon \left( 1 + \varepsilon^{-1} \abs{\xi} \right) 
\end{align}
by assumption and 
\begin{align}
\label{eq:comparison between regularized and normal p auxilliary estimate}
\norm{ D_{\varepsilon} e^{i \xi \cdot \frac{\hat{x}}{\varepsilon}} D_{\varepsilon}^{-1} }_{\mathfrak{S}^{\infty}} 
&\leq C \left( 1 + \abs{\xi} \right) 
\end{align}
we get
\begin{align}
\eqref{eq:comparison between regularized and normal p in Sobolev trace norm estimate 1}
&\leq C N   \int_{\mathbb{R}^6} \mathscr{G}_{\sigma}(z, \xi) \left( \varepsilon + (1 + \varepsilon) \abs{\xi} + \abs{\xi}^2 \right)   \, dz \, d \xi
\leq C N \left( \varepsilon +  \sigma + \sigma^2 \right)  .
\end{align}
By Duhamel's formula we obtain
\begin{align}
D_{\varepsilon}
\big( p(\cdot -z; \cdot -z) -  p \big) D_{\varepsilon}
&=
D_{\varepsilon}
\big( e^{- z \cdot \nabla} p e^{ z \cdot \nabla}  -  p \big) D_{\varepsilon}
= 
- z \cdot 
\int_0^1 
e^{- \lambda z \cdot \nabla}
D_{\varepsilon} \left[ \nabla, p \right]   D_{\varepsilon}   e^{ \lambda z \cdot \nabla} \, d \lambda .
\end{align}
Together with \eqref{eq:comparison between regularized and normal p auxilliary estimate} and \eqref{eq:Vlasov-Maxwell result more restricted initial semiclassical structure} we obtain
\begin{align}
\eqref{eq:comparison between regularized and normal p in Sobolev trace norm estimate 2}
&\leq
\int_{\mathbb{R}^6} \mathscr{G}_{\sigma}(z, \xi) \norm{ D_{\varepsilon} e^{i \xi \cdot \frac{\hat{x}}{\varepsilon}}D_{\varepsilon}^{-1} }_{\mathfrak{S}^{\infty}}^2
\norm{ D_{\varepsilon}
\big( p(\cdot -z; \cdot -z) -  p \big) D_{\varepsilon} }_{\mathfrak{S}^1}  \, dz \, d \xi
\nonumber \\
&\quad \leq C
\int_{\mathbb{R}^6} \mathscr{G}_{\sigma}(z, \xi) \left( 1 + \abs{\xi} \right)^2 
\abs{z} \int_0^1
\norm{e^{ - \lambda z \cdot \nabla}
D_{\varepsilon} \left[ \nabla, p \right]   D_{\varepsilon}   e^{ \lambda z \cdot \nabla} }_{\mathfrak{S}^1} \, d \lambda  \, dz \, d \xi 
\nonumber \\
&\quad \leq C
\int_{\mathbb{R}^6} \mathscr{G}_{\sigma}(z, \xi) \left( 1 + \abs{\xi} \right)^2 
\abs{z} 
\norm{ \left[ \nabla, p \right]    }_{\mathfrak{S}^{1,1}_{\varepsilon}}   \, dz \, d \xi 
\nonumber \\
&\quad \leq C
\int_{\mathbb{R}^6} \mathscr{G}_{\sigma}(z, \xi) \left( 1 + \abs{\xi} \right)^2 
\abs{z} 
\norm{  q \nabla p  }_{\mathfrak{S}^{1,1}_{\varepsilon}}   \, dz \, d \xi 
\nonumber \\
&\quad \leq C N
\int_{\mathbb{R}^6} \mathscr{G}_{\sigma}(z, \xi) \left( 1 + \abs{\xi} \right)^2 
\abs{z}  \, dz \, d \xi 
\nonumber \\
&\quad \leq C N \left( 1 + \sigma^2 \right) \sigma .
\end{align}
In total, this shows \eqref{eq:distance between p and regularized version in Sobolev-trace norm 1}. 
By means of (see \cite[Appendix B]{LS2023} for its proof)
\begin{align}
\norm{
\mathcal{W}^{-1}[(1 - \eta_R) m_{p,\sigma}]  }_{\mathfrak{S}^{1,1}_{\varepsilon}}
&\leq C N \sum_{j=0}^6 \varepsilon^j \norm{(1 - \eta_R) m_{p,\sigma}}_{W_6^{j,2}(\mathbb{R}^6)} ,
\end{align}
\eqref{eq:estimate for regularized Wigner 1}, \eqref{eq:estimate Husimi with cutoff of the high velocities}, and the assumption $\varepsilon^{1/2} \leq \sigma \leq 1$ we get
\begin{align}
\norm{ \left( \mathcal{W}^{-1}[m_{p, \sigma}] - \mathcal{W}^{-1}[m_{p, \sigma,R}] \right)  }_{\mathfrak{S}^{1,1}_{\varepsilon}}
&\leq C N  \left( \norm{\mathcal{W}[p]}_{W_6^{0,2}(\mathbb{R}^6)} + 1 \right) \Big( e^{- \frac{1}{2 \sigma^2}} 
+ \sum_{j=1}^6 \varepsilon^j \sigma^{-j} \Big) 
\nonumber \\
&\leq C N  \left( e^{- \frac{1}{2 \sigma^2}}  + \varepsilon^{1/2} \right).
\end{align}
Inequality \eqref{eq:distance between alpha and its cutoff version} follows from the estimate
$\id_{\abs{k} \geq \Lambda}  \leq \Lambda^{- 1/2} \abs{k}^{1/2}$. 
The properties in \eqref{eq:fermionic reduced density proeprty of the weyl quantization of m-p-sigma} follow by an argument completely analogous to that in \cite[(4.3)]{BPSS2016}.
Using 
\begin{align}
\mathcal{W}^{-1}[m_{p,\sigma}]
&= \int_{\mathbb{R}^6}   \mathscr{G}_{\sigma}(z,\xi) 
e^{i \xi \cdot \frac{\hat{x}}{\varepsilon}}
e^{- z \cdot \nabla} p e^{z \cdot \nabla}  e^{- i \xi \cdot \frac{\hat{x}}{\varepsilon}} 
\end{align}
and $\norm{ \left(1 -  \Delta \right) e^{i \xi \cdot \frac{\hat{x}}{\varepsilon}} \left( 1 - \Delta \right)^{-1} }_{\mathfrak{S}^{\infty}} \leq C \big( 1 + \varepsilon^{-2} \abs{\xi}^2 \big)$
we estimate
\begin{align}
\norm{\mathcal{W}^{-1}[m_{p,\sigma}]}_{\mathfrak{S}^{2,1}}
&\leq  \int_{\mathbb{R}^6}   \mathscr{G}_{\sigma}(z,\xi) \norm{\left( 1 - \Delta \right) e^{i \xi \cdot \frac{\hat{x}}{\varepsilon}}
e^{- z \cdot \nabla} p e^{z \cdot \nabla}  e^{- i \xi \cdot \frac{\hat{x}}{\varepsilon}}  \left( 1 - \Delta \right)}_{\mathfrak{S}^1}
\nonumber \\
&\leq  \int_{\mathbb{R}^6}   \mathscr{G}_{\sigma}(z,\xi)
\norm{ \left(1 - \ \Delta \right) e^{i \xi \cdot \frac{\hat{x}}{\varepsilon}} \left( 1 -  \Delta \right)^{-1} }_{\mathfrak{S}^{\infty}}^2  
\norm{\left( 1 - \Delta \right) 
e^{- z \cdot \nabla} p e^{z \cdot \nabla}   \left( 1 - \Delta \right)}_{\mathfrak{S}^1}
\nonumber \\
&= \int_{\mathbb{R}^6}   \mathscr{G}_{\sigma}(z,\xi)
\norm{ \left(1 -  \Delta \right) e^{i \xi \cdot \frac{\hat{x}}{\varepsilon}} \left( 1 -  \Delta \right)^{-1} }_{\mathfrak{S}^{\infty}}^2  
\norm{\left( 1 - \Delta \right) 
 p  \left( 1 - \Delta \right)}_{\mathfrak{S}^1}
 \nonumber \\
&\leq C \norm{p}_{\mathfrak{S}^{2,1}}  \int_{\mathbb{R}^6}   \mathscr{G}_{\sigma}(z,\xi)  \left( 1 + \varepsilon^{-2} \abs{\xi}^2 \right)  .
\end{align}
Since  $p \in \mathfrak{S}^{2,1}(L^2(\mathbb{R}^3)) $ by assumption and the Gaussian has finite moments, we obtain $\mathcal{W}^{-1}[m_{p,\sigma}] \in \mathfrak{S}^{2,1}(L^2(\mathbb{R}^3)) $,  which proves \eqref{eq:weyl quantization of m-p-sigma in S-2-1}.
The non-negativity of $m_{p,\sigma,R}$ is a consequence of \eqref{eq:properties of the smooth cutoff function 1} and the non-negativity of $m_{p,\sigma} \geq 0$. The latter is proven in direct analogy to \cite{C1975}, and the proof is included here primarily for completeness. Integrating with respect to $u$ and changing the variables $q = z + \frac{\varepsilon y}{2}$, $r = z - \frac{\varepsilon y}{2}$ leads to
\begin{align}
m_{p,\sigma}(x,v) &= \left( \frac{\varepsilon}{2 \pi} \right)^3  \int_{\mathbb{R}^9} \mathscr{G}_{\sigma}(x-z,v-u) p \left(z + \frac{\varepsilon y}{2} ; z - \frac{\varepsilon y}{2}  \right) e^{- i u y} \, dy \,  dz \, du
\nonumber \\
&= \frac{1}{8 \pi^{\frac{9}{2}} \sigma^3} \int_{\mathbb{R}^6}
p(q;r) 
\exp \left[ - \frac{(2 x - (q+r))^2}{4 \sigma^2} - \frac{(q-r)^2 \sigma^2}{4 \varepsilon^2} - \frac{i v (q-r)}{\varepsilon} \right] \, dq \, dr \, .
\end{align}
If we define the function 
$f(q) = \exp \left[ - q^2 \left( \frac{1}{4 \sigma^2} + \frac{\sigma^2}{4 \varepsilon^2} \right) + q \left( \frac{x}{\sigma^2} + \frac{i v}{\varepsilon} \right) \right]$ and use that there exists a set of orthonormal $L^2(\mathbb{R}^3)$-functions $\{ \varphi_j \}_{j=1}^N$ such that $p(q;r) = \sum_{j=1}^N \varphi_j(q) \overline{\varphi_j(r)}$, we can write the right-hand side  as
\begin{align}
m_{p,\sigma}(x,v) &=
\frac{1}{8 \pi^{\frac{9}{2}} \sigma^3} e^{- \frac{x^2}{\sigma^2}} \sum_{j=1}^N
\int_{\mathbb{R}^6}
\varphi_j(q) f(q) \overline{\varphi_j(r)} \, \overline{f(r)}  e^{\gamma qr} \, dq \, dr 
\quad \text{with} \; \gamma = \frac{\sigma^2}{2 \varepsilon^2} - \frac{1}{2 \sigma^2} .
\end{align}
Note that $\gamma \geq 0$ because $\sigma \geq \varepsilon^{1/2}$. Writing the exponential in its series expansion gives
\begin{align}
m_{p,\sigma}(x,v) &=
\frac{1}{8 \pi^{9/2} \sigma^3} e^{- \frac{x^2}{\sigma^2}} \sum_{j=1}^N
\sum_{n \in \mathbb{N}} \frac{\gamma^n}{n!} \abs{\int_{\mathbb{R}^3} \varphi_j(q) f(q) q^n \, dq}^2 \geq 0,
\end{align}
which proves the claim.
\end{proof}

\subsection{Proof of Lemma~\ref{lemma:propagation estimates Vlasov-Maxwell}}

\label{subsection:propagation estimates Vlasov-Maxwell}

Within this subsection we will prove Lemma~\ref{lemma:propagation estimates Vlasov-Maxwell}. In doing thus we will rely on the following estimates from \cite{LS2023}.
\begin{lemma}
\label{lemma:propagation of regularity auxilliary lemma}
Let $R>0$ and $a, b, c \in \mathbb{N}$ satisfying $5 \leq a$, $3 \leq b$, and $c \leq b$. Let $(f_t,\alpha_t)$ be the unique solution of the Vlasov--Maxwell equations \eqref{eq: Vlasov-Maxwell different style of writing} with initial condition
 $(f_0, \alpha_0) \in H_a^{b}(\mathbb{R}^6) \times \mathfrak{h}_b \cap \dot{\mathfrak{h}}_{-1/2}$ such that $\supp f_0 \subseteq  \{ (x,v) \in \mathbb{R}^6 , \abs{v} \leq R \}$. For $t \in \mathbb{R}_{+}$ we have $\supp f_t \subseteq \{ (x,v) \in \mathbb{R}^6 , \abs{v} \leq R(t) \}$, where
\begin{align}
\label{eq:estimate for the support of f}
R(t) \leq C \left< R \right> \exp \left( C \int_0^t \left( \norm{f_s}_{W_4^{0,2}(\mathbb{R}^6)} + \norm{\alpha_s}_{\mathfrak{h}}^2 \right) \, ds \right) .
\end{align}
Moreover,
\begin{subequations}
\begin{align}
\label{eq:propagation of gradient moments estimate for W-a-0-2}
\norm{f_t}_{W_a^{0,2}(\mathbb{R}^6)}^2
&\leq e^{C \left( 1 + \mathcal{E}^{\rm{VM}}[f_0, \alpha_0] 
+ C \norm{f_0}_{L^1(\mathbb{R}^6)}^2  \right) \left< t \right>}
\norm{f_0}_{W_a^{0,2}(\mathbb{R}^6)}^2 , 
\\
\label{eq:propagation of gradient moments estimate for alpha-t in h-1}
\norm{\alpha_t}_{\mathfrak{h}_1 \, \cap \,  \dot{\mathfrak{h}}_{-1/2}}
&\leq 
\norm{\alpha_0}_{\mathfrak{h}_1 \, \cap \,  \dot{\mathfrak{h}}_{-1/2}} + C t \left( 1 + \mathcal{E}^{\rm{VM}}[f_0, \alpha_0] 
+  C \norm{f_0}_{L^1(\mathbb{R}^6)}^2  \right)
,
\\
\label{eq:propagation of gradient moments estimate for alpha-t in h-b}
\norm{\alpha_t}_{\mathfrak{h}_{c} \, \cap \,  \dot{\mathfrak{h}}_{-1/2}}
&\leq \norm{\alpha_0}_{\mathfrak{h}_{c} \, \cap \,  \dot{\mathfrak{h}}_{-1/2}}
+ C t  \left( 1 + \norm{\alpha}_{L_t^{\infty}\mathfrak{h}_{c-2}} \right) \norm{f}_{L_t^{\infty} H_5^{c-1}(\mathbb{R}^6)}
\end{align}
with $\norm{\alpha}_{L_t^{\infty}\mathfrak{h}_{c-2}} = \norm{\alpha}_{L_t^{\infty}\mathfrak{h}}$ in case $c \in \{1, 2 \}$.
\end{subequations}
\end{lemma}
\begin{proof}[Proof of Lemma~\ref{lemma:propagation of regularity auxilliary lemma}]
Estimate \eqref{eq:estimate for the support of f} follows from the inequality \cite[(VI.4)]{LS2023} and the fact that $f_t$ can be written as  $f_t(x,v) = f_0 \big( X_t^{-1}(x,v), V_t^{-1}(x,v) \big)$ with the characteristics $(X_t,V_t)$ being defined as in \cite[(VI.1)]{LS2023}.
The remaining inequalities are proven in \cite[Section VI.4]{LS2023}.
\end{proof}

\begin{proof}[Proof of Lemma~\ref{lemma:propagation estimates Vlasov-Maxwell}]
We define the transport operator
\begin{align}
T_{f,\alpha} &= 2 \left( v - \kappa * \vA_{\alpha} \right) \cdot \nabla_x - \vF_{f,\alpha} \cdot \nabla_v 
\end{align}
with $\vF_{f,\alpha}$ being defined as in \eqref{eq:F-field of Vlasov equation}.
In analogy to \cite[Section VI.2]{LS2023} we get 
\begin{subequations}
\begin{align}
&\norm{\left( 1 - \Delta_z \right)^{k/2} \left< \cdot \right>^a f_t}^2_{L^2(\mathbb{R}^6)}
- \norm{\left( 1 - \Delta_z \right)^{k/2} \left< \cdot \right>^a f_0}^2_{L^2(\mathbb{R}^6)}
\nonumber \\
&\quad = 
- 2 \int_0^t \scp{\left( 1 - \Delta_z \right)^{k/2} \left< \cdot \right>^a f_s}{\left( 1 - \Delta_z \right)^{k/2} \left[ \left< \cdot \right>^a , T_{f_s,\alpha_s} \right] [f_s]} \, ds 
\nonumber \\
&\qquad  
- 2 \int_0^t \scp{\left( 1 - \Delta_z \right)^{k/2} \left< \cdot \right>^a f_s}{ \left[ \left( 1 - \Delta_z \right)^{k/2}  , T_{f_s,\alpha_s} \right] [ \left< \cdot \right>^a f_s]} \, ds 
\nonumber \\
&\quad \leq 2 \int_0^t \norm{\left( 1 - \Delta_z \right)^{k/2} \left< \cdot \right>^a f_s}^2_{L^2(\mathbb{R}^6)} \, ds
\\
\label{eq:propagation estimates for the Wigner transform 1}
&\qquad +  
\int_0^t \norm{\left( 1 - \Delta_z \right)^{k/2} \left[ \left< \cdot \right>^a , T_{f_s,\alpha_s} \right] [f_s]}^2_{L^2(\mathbb{R}^6)} \, ds
\\
\label{eq:propagation estimates for the Wigner transform 2}
&\qquad +  
\int_0^t \norm{ \left[ \left( 1 - \Delta_z \right)^{k/2}  , T_{f_s,\alpha_s} \right] [ \left< \cdot \right>^a f_s]}^2_{L^2(\mathbb{R}^6)} \, ds
\end{align}
\end{subequations}
by Duhamel's formula. Using 
\begin{align}
\nabla_x \left< z \right>^a = a  \left< z \right>^{a-2} x
\quad \text{and} \quad 
\nabla_v \left< z \right>^a = a  \left< z \right>^{a-2} v 
\end{align}
we obtain
\begin{align}
\left[ T_{f_s,\alpha_s} , \left< z \right>^a \right] &= a \left< z \right>^{a-2} \left( 2 (v - \kappa * \vA_{\alpha_s}) \cdot x - \vF_{f_s,\alpha_s} \cdot v \right) .
\end{align}
If we combine these with the inequalities $\norm{\kappa * \vA_{\alpha}}_{W_0^{\sigma,\infty}(\mathbb{R}^3)} \leq C \norm{\alpha}_{\mathfrak{h}_{\sigma -1}}$ and $\norm{K * \widetilde{\rho}_f}_{W_0^{\sigma,\infty}(\mathbb{R}^3)} \leq C \norm{f}_{H_4^{\sigma -3}(\mathbb{R}^6)}$ (see \cite[(III.2b) and (III.8b)]{LS2023}) we get
\begin{align}
\eqref{eq:propagation estimates for the Wigner transform 1}
&\leq  C \sum_{j=0}^k \int_0^t 
\left( \norm{K * \widetilde{\rho}_{f_s}}_{W_0^{j+1,\infty}(\mathbb{R}^6)} 
+ \norm{\kappa * \vA_{\alpha_s}}_{W_0^{j+1,\infty}(\mathbb{R}^6)}^2 \right)^2  \norm{f_s}_{H_a^{k-j}(\mathbb{R}^6)}^2 \, ds 
\nonumber \\
&\leq 
C \sum_{j=0}^k
\int_0^t 
\left( 
\norm{f_s}^2_{H_a^{j-2}(\mathbb{R}^6)} 
+ \norm{\alpha_s}^4_{\mathfrak{h}_{j}} \right)  \norm{f_s}^2_{H_a^{k-j}(\mathbb{R}^6)}  \, ds .
\end{align}
Using 
\begin{align}
\Big[  D_{z}  , T_{f_s, \alpha_s} \Big] 
&= 2 \left[ D_{z} , v \right] \cdot \nabla_x
- 2 \left[ D_{z} , \kappa * \vA_{\alpha_t}  \right] \cdot \nabla_x
+  2 \sum_{i=1}^3 \nabla \kappa * \vA_{\alpha_s}^i  \left[   D_{z} , v^i  \right] \cdot \nabla_v
\nonumber \\
&\quad + 2 \sum_{i=1}^3 \left[   D_{z} , \nabla \kappa * \vA_{\alpha_s}^i \right]  v^i  \cdot \nabla_v
- \left[ D_{z} , \left( \nabla  K * \widetilde{\rho}_{f_s} + \nabla (\kappa * \vA_{\alpha_s})^2 \right) \right] \cdot \nabla_v
\end{align}
and \eqref{eq:estimate for the support of f},
we estimate
\begin{align}
&\norm{ \left[ \left( 1 - \Delta_z \right)^{k/2}  , T_{f_s,\alpha_s} \right] [ \left< \cdot \right>^a f_s]}_{L^2(\mathbb{R}^6)}
\nonumber \\
&\quad \leq C \sum_{j=0}^k
\left( \norm{K * \widetilde{\rho}_{f_s}}_{W_0^{j+1,\infty}(\mathbb{R}^6)} 
+ \norm{\kappa * \vA_{\alpha_s}}_{W_0^{j+1,\infty}(\mathbb{R}^6)}^2 \right)  \norm{\left< v \right> (\nabla_x + \nabla_v) \left< \cdot \right>^a  f_s}_{H^{k-j-1}(\mathbb{R}^6)} 
\nonumber \\
&\quad \leq C \sum_{j=0}^k
\left( \norm{f_s}_{H_a^{j-2}(\mathbb{R}^6)} 
+ \norm{\alpha_s}^2_{\mathfrak{h}_j} \right)  \norm{\left< v \right> (\nabla_x + \nabla_v) \left< \cdot \right>^a  f_s}_{H^{k-j-1}(\mathbb{R}^6)} .
\nonumber \\
&\quad \leq C 
\left< R \right> \exp \left( C \int_0^s \left( \norm{f_u}_{W_4^{0,2}(\mathbb{R}^6)} + \norm{\alpha_u}_{\mathfrak{h}}^2 \right) \, du \right)
\nonumber \\
&\qquad \times
\sum_{j=0}^k
\left( \norm{f_s}_{H_a^{j-2}(\mathbb{R}^6)} 
+ \norm{\alpha_s}^2_{\mathfrak{h}_j} \right)  \norm{ f_s}_{H_a^{k-j}(\mathbb{R}^6)} ,
\end{align}
leading to 
\begin{align}
\eqref{eq:propagation estimates for the Wigner transform 2}
&\leq 
C \left< R \right> 
\int_0^t 
\exp \left( C \int_0^s \left( \norm{f_u}_{W_4^{0,2}(\mathbb{R}^6)} + \norm{\alpha_u}_{\mathfrak{h}}^2 \right) \, du \right)
\nonumber \\
&\qquad \times
\sum_{j=0}^k
\left( 
\norm{f_s}^2_{H_a^{j-2}(\mathbb{R}^6)} 
+ \norm{\alpha_s}^4_{\mathfrak{h}_{j}} \right)  \norm{f_s}^2_{H_a^{k-j}(\mathbb{R}^6)}  \, ds .
\end{align}
Collecting the estimates and using
\begin{align}
&\sum_{j=0}^k
\left( 
\norm{f_s}^2_{H_a^{j-2}(\mathbb{R}^6)} 
+ \norm{\alpha_s}^4_{\mathfrak{h}_{j}} \right)  \norm{f_s}^2_{H_a^{k-j}(\mathbb{R}^6)} 
\nonumber \\
&\quad \leq C 
\left( 
\norm{f_s}^2_{W_a^{0,2}(\mathbb{R}^6)} 
+ \norm{\alpha_s}^4_{\mathfrak{h}_1} \right)  \norm{f_s}^2_{H_a^{k}(\mathbb{R}^6)} 
+ C \sum_{j=2}^k
\left( 
\norm{f_s}^2_{H_a^{j-2}(\mathbb{R}^6)} 
+ \norm{\alpha_s}^4_{\mathfrak{h}_{j}} \right)  \norm{f_s}^2_{H_a^{k-j}(\mathbb{R}^6)}  
\end{align}
we obtain
\begin{align}
\norm{f_t}^2_{H_a^{k}(\mathbb{R}^6)}  
&\leq \norm{f_0}^2_{H_a^{k}(\mathbb{R}^6)}  
+ C \left< R \right> \int_0^t
e^{C \left< s \right> \widetilde{C}(s) }
\norm{f_s}^2_{H_a^{k}(\mathbb{R}^6)}  \, ds 
\nonumber \\
&\quad +
C \left< R \right> \int_0^t
e^{C \left< s \right> \widetilde{C}(s) }
\sum_{j=2}^{k}
\left( 
\norm{f_s}^2_{H_a^{j-2}(\mathbb{R}^6)} 
+ \norm{\alpha_s}^4_{\mathfrak{h}_{j}} \right)  \norm{f_s}^2_{H_a^{k-j}(\mathbb{R}^6)} 
\, ds 
\end{align}
with
$\widetilde{C}(s)
= \norm{f}_{L_s^{\infty} W_a^{0,2}(\mathbb{R}^6)} + \norm{\alpha}_{L_s^{\infty} \mathfrak{h}_1}^2$.
By Gr\"onwall's lemma we get
\begin{align}
\norm{f_t}^2_{H_a^{k}(\mathbb{R}^6)}  
&\leq C e^{\left< R \right> e^{C \left< t \right> \left< \widetilde{C}(t) \right>}}
\bigg[ \norm{f_0}^2_{H_a^{k}(\mathbb{R}^6)} 
\nonumber \\
&\qquad \qquad 
+ \int_0^t
\sum_{j=2}^k
\left( 
\norm{f_s}^2_{H_a^{j-2}(\mathbb{R}^6)} 
+ \norm{\alpha_s}^4_{\mathfrak{h}_{j}} \right)  \norm{f_s}^2_{H_a^{k-j}(\mathbb{R}^6)} 
\, ds 
\bigg]
\end{align}
and
\begin{align}
\norm{f}^2_{L_t^{\infty} H_a^{k}(\mathbb{R}^6)} 
&\leq C e^{\left< R \right> e^{C \left< t \right> \left< \widetilde{C}(t) \right>}}
\bigg[ \norm{f_0}^2_{H_a^{k}(\mathbb{R}^6)}  
+ \left( 
\norm{f}^2_{L_t^{\infty} H_a^{k-2,2}(\mathbb{R}^6)} 
+ \norm{\alpha}^4_{L_t^{\infty} \mathfrak{h}_{k}}
\right) \norm{f}^2_{L_t^{\infty} W_a^{0,2}(\mathbb{R}^6)} 
\nonumber \\
&\qquad \qquad 
+
\sum_{j=2}^{k-1}
\left( 
\norm{f}^2_{L_t^{\infty} H_a^{j-2}(\mathbb{R}^6)} 
+ \norm{\alpha}^4_{L_t^{\infty} \mathfrak{h}_{j}} \right)  \norm{f}^2_{L_t^{\infty} H_a^{k-j}(\mathbb{R}^6)} 
\bigg] 
\nonumber \\
&\leq C e^{\left< R \right> e^{C \left< t \right> \left< \widetilde{C}(t) \right>}}
\bigg[ 1 + \norm{f_0}^2_{H_a^{k}(\mathbb{R}^6)}  
+ \norm{\alpha}^4_{L_t^{\infty} \mathfrak{h}_{k}}
\nonumber \\
&\qquad \qquad \qquad \quad 
+ 
\norm{f}^4_{L_t^{\infty} H_a^{k-2}(\mathbb{R}^6)} 
+ \norm{f}^2_{L_t^{\infty} H_a^{k-2}(\mathbb{R}^6)}   \norm{\alpha}^4_{L_t^{\infty} \mathfrak{h}_{k-1}}
\bigg] .
\end{align}
Together with \eqref{eq:propagation of gradient moments estimate for alpha-t in h-b} this leads to
\begin{align}
&\norm{f}^2_{L_t^{\infty} H_a^{k}(\mathbb{R}^6)} + \norm{\alpha}^4_{L_t^{\infty} \mathfrak{h}_{k}}
\nonumber \\
&\quad \leq C e^{\left< R \right> e^{C \left< t \right> \left< \widetilde{C}(t) \right>}}
\bigg[ 1 + \norm{f_0}^2_{H_a^{k}(\mathbb{R}^6)}  
+ \norm{\alpha_0}^4_{ \mathfrak{h}_{k} \, \cap \,  \dot{\mathfrak{h}}_{-1/2}}
+ \left( 1 + \norm{\alpha}_{L_t^{\infty} \mathfrak{h}_{k-2}}^4 \right) \norm{f}^4_{L_t^{\infty} H_5^{k-1}(\mathbb{R}^6)}
\nonumber \\
&\qquad \qquad  \qquad \qquad \quad 
+ 
\norm{f}^4_{L_t^{\infty} H_a^{k-2}(\mathbb{R}^6)} 
+ \norm{f}^2_{L_t^{\infty} H_a^{k-2}(\mathbb{R}^6)}   \norm{\alpha}^4_{L_t^{\infty} \mathfrak{h}_{k-1}}
\bigg] 
\nonumber \\
&\quad \leq C e^{\left< R \right> e^{C \left< t \right> \left< \widetilde{C}(t) \right>}}
\bigg[ 1 + \norm{f_0}^2_{H_a^{k}(\mathbb{R}^6)}  
+ \norm{\alpha_0}^4_{ \mathfrak{h}_{k} \, \cap \,  \dot{\mathfrak{h}}_{-1/2}}
+ \left(  \norm{f}^2_{L_t^{\infty} H_a^{k-1}(\mathbb{R}^6)} + \norm{\alpha}^4_{L_t^{\infty}  \mathfrak{h}_{k-1} } \right)^3
\bigg] .
\end{align}
Iterating the inequality gives
\begin{align}
\norm{f}^2_{L_t^{\infty} H_a^{k}(\mathbb{R}^6)} + \norm{\alpha}^4_{L_t^{\infty} \mathfrak{h}_{k} }
&\leq  C e^{\left< R \right> e^{C \left< t \right> \left< \widetilde{C}(t) \right>}}
\bigg[ 1 + 
\sum_{j=0}^{k} \left(  \norm{f_0}^2_{ H_a^{k-j}(\mathbb{R}^6)} + \norm{\alpha_0}^4_{ \mathfrak{h}_{k-j} \, \cap \,  \dot{\mathfrak{h}}_{-1/2}} \right)^{3^j}
\bigg] .
\end{align}
Note that due to \eqref{eq:propagation of gradient moments estimate for W-a-0-2} 
and \eqref{eq:propagation of gradient moments estimate for alpha-t in h-1} we have
\begin{align}
\widetilde{C}(s)
&\leq 
e^{C \left( 1 + \mathcal{E}^{\rm{VM}}[f_0, \alpha_0] 
+ C \norm{f_0}_{L^1(\mathbb{R}^6)}^2  \right) \left< s \right>}
\Big[ 1 + 
\norm{f_0}_{W_a^{0,2}(\mathbb{R}^6)}^2
+ \norm{\alpha_0}_{\mathfrak{h}_1 \, \cap \,  \dot{\mathfrak{h}}_{-1/2}}^2 \Big] .
\end{align}
In total, this proves the claim.
\end{proof}

\subsection{Proof of Proposition~\ref{proposition:comparison between Maxwell--Schroedinger and Vlasov--Maxwell updated}}

\begin{proof}
In the following we  apply  \cite[Theorem II.1]{LS2023} to obtain an estimate of the form
\begin{align}
\label{eq:comparison between Maxwell--Schroedinger and Vlasov--Maxwell auxilliary estimate new}
& N^{-1} \norm{ \widetilde{p}_t - \mathcal{W}^{-1}[W_{N,t}]}_{\mathfrak{S}^{1,1}_{\varepsilon}} + \norm{\widetilde{\alpha}_t - \alpha_t}_{\mathfrak{h}_{1/2} \cap \, \dot{\mathfrak{h}}_{- 1/2}}
\nonumber \\
&\leq
\left[ N^{-1} \norm{ \mathcal{W}^{-1}[m_{p,\sigma}] - \mathcal{W}^{-1}[m_{p,\sigma,R}] }_{\mathfrak{S}^{1,1}_{\varepsilon}}  
+ \varepsilon \widetilde{K}(t) \right]
e^{\mathcal{K}(t)} ,
\end{align}
where
\begin{align}
\widetilde{\mathcal{K}}(t) &\leq 
\int_0^t ds \,
\sum_{k=0}^{6} \varepsilon^k \norm{W_{N,s}}_{W_{7}^{k+2,2}}
\bigg( 1 +   \sum_{k=0}^3 \varepsilon^{2 + k} \norm{W_{N,s}}_{W_{4}^{k+1,2}} 
+ \sum_{k=0}^4 \varepsilon^k  \norm{\alpha_s}_{\mathfrak{h}_{k+1}}^2
\bigg) ,
\\
\mathcal{K}(t) &= \mathcal{K} \left( 1 + N^{-1} \mathcal{E}^{\rm{MS}}[\mathcal{W}^{-1}[m_{p,\sigma}],  \id_{\abs{\cdot} \leq \Lambda} \alpha] + \mathcal{E}^{\rm{VM}}[m_{p,\sigma,R}, \id_{\abs{\cdot} \leq \Lambda} \alpha] + 
\norm{m_{p,\sigma,R}}_{L^1(\mathbb{R}^6)} \right)^2
\nonumber \\
&\quad \times
\bigg( \left< t \right> + \int_0^t ds \, \sum_{k=0}^{7} \varepsilon^k \norm{W_{N,s}}_{W_{7}^{k+1,2}} \bigg)  ,
\end{align}
and $\mathcal{K}$ is a numerical constant which depends on the specific choice of the cutoff function $\kappa$.  Note that all assumptions of 
\cite[Theorem II.1]{LS2023} are satisfied,  since $\tr \left( \mathcal{W}^{-1}[m_{p,\sigma}] \right) = N$,  $0 \leq \mathcal{W}^{-1}[m_{p,\sigma}] \leq 1$, $\mathcal{W}^{-1}[m_{p,\sigma} ] \in \mathfrak{S}^{2,1}(L^2(\mathbb{R}^3))$,  and $m_{p,\sigma,R} \geq 0$ by Lemma~\ref{lemma:Vlasov-Maxwell distance between initial data and its regularization updated}.  Moreover, we have $m_{p,\sigma,R} \in W_2^{0,1}(\mathbb{R}^6)$ and $m_{p,\sigma,R} \in W_7^{8,2}(\mathbb{R}^6)$ by Lemma~\ref{lemma:derivation Vlasov--Maxwell auxiliary estimates}. Finally, $\norm{\id_{\abs{\cdot} \leq \Lambda} \alpha}_{\mathfrak{h}_8 \cap \, \dot{\mathfrak{h}}_{- 1/2}} \leq C (1 + \Lambda^7) \norm{\alpha}_{\mathfrak{h}_1 \cap \, \dot{\mathfrak{h}}_{- 1/2}}$ and $\supp m_{p,\sigma,R} \subset \{(x,v) \in \mathbb{R}^6, \abs{v} \leq R +3 \}$ by the definition of $m_{p,\sigma,R}$.
By means of Lemma~\ref{lemma:Vlasov-Maxwell distance between initial data and its regularization updated} we have
\begin{align}
N^{-1} \norm{\mathcal{W}^{-1}[m_{p,\sigma}] - \mathcal{W}^{-1}[m_{p,\sigma,R}] }_{\mathfrak{S}^{1,1}_{\varepsilon}}  
&\leq C    \left( e^{- \frac{1}{2 \sigma^2}}  + \varepsilon^{1/2} \right). 
\end{align}
Using Lemma \ref{lemma:propagation estimates Vlasov-Maxwell},  Lemma \ref{lemma:derivation Vlasov--Maxwell auxiliary estimates},  and \eqref{eq:estimate for VM-energy with regularized data} gives
\begin{align}
&\norm{W_{N,\cdot}}^2_{L_t^{\infty} W_7^{k,2}(\mathbb{R}^6)} + \norm{\alpha_{\cdot}}^4_{L_t^{\infty} \mathfrak{h}_{k}}
\nonumber \\
&\quad \leq  C e^{\left< R \right> e^{C \left< t \right>  C(t) }}
\bigg[ 1 + 
\sum_{j=0}^{k} \left(  \norm{m_{p,\sigma,R}}^2_{ W_7^{k-j,2}(\mathbb{R}^6)} + \norm{\id_{\abs{\cdot} \leq \Lambda} \alpha}^4_{ \mathfrak{h}_{k-j} \cap \, \dot{\mathfrak{h}}_{- 1/2} } \right)^{3^j}
\bigg] 
\nonumber \\
&\quad \leq  C e^{\left< R \right> e^{C \left< t \right>  C(t) }}
\bigg[ 1 + 
\sum_{j=0}^{k} \left(  \sigma^{- 2 (k - j)} \norm{\mathcal{W}[p]}^2_{ W_7^{0,2}(\mathbb{R}^6)} + \Lambda^{4(k-j)} \norm{ \alpha}^4_{ \mathfrak{h} \cap \, \dot{\mathfrak{h}}_{- 1/2} } \right)^{3^j}
\bigg] 
\end{align}
with 
$C(s) \leq e^{C \left( 1 + \norm{\alpha}_{\mathfrak{h}_1 \, \cap \,  \dot{\mathfrak{h}}_{-1/2}}^2 \right)
\left( 1 + \norm{\mathcal{W}[p]}^2_{W_{7}^{0,2}(\mathbb{R}^6)} \right) \left< s \right>}$.
For the choices
\begin{align}
\Lambda = \sigma^{- \frac{1}{2}}
\quad \text{and} \quad
\sigma^{-1} =  \ln (\ln(N))
\quad \text{with} \quad  N \geq 3
\end{align}
this shows the existence of a constant $C$ which depends on $\kappa$, $\norm{\mathcal{W}[p]}_{W_7^{0,2}(\mathbb{R}^6)}$, $\norm{\alpha}_{\mathfrak{h}_1 \, \cap \,  \dot{\mathfrak{h}}_{-1/2}}$,  such that 
\begin{align}
\norm{W_{N,\cdot}}^2_{L_t^{\infty} W_7^{k,2}(\mathbb{R}^6)} + \norm{\alpha_{\cdot}}^4_{L_t^{\infty} \mathfrak{h}_{k}}
&\leq \exp \left[ \left< R \right> \exp \left[ \exp \left[C \left< t \right> \right] \right] \right] 
\bigg[ 
1 +
\sum_{j=0}^{k} \left(  \sigma^{- 2 (k - j)}  + \Lambda^{4(k-j)}  \right)^{3^j}
\bigg] 
\nonumber \\
&\leq \exp \left[ \left< R \right> \exp \left[ \exp \left[C \left< t \right> \right] \right] \right] 
\bigg[ 
1 + 
\sum_{j=0}^{k} (\ln(\ln(N)))^{2(k-j) 3^j}
\bigg]  .
\end{align}
For $k=1 $ this gives
\begin{align}
\norm{W_{N,\cdot}}_{L_t^{\infty} W_7^{1,2}(\mathbb{R}^6)}
&\leq \exp \left[ \left< R \right> \exp \left[ \exp \left[ C \left< t \right> \right] \right] \right] 
\bigg[ 
1 + 
\sum_{j=0}^{1} (\ln(\ln(N)))^{2(1-j)3^j}
\bigg]^{1/2}
\nonumber \\
&\leq \exp \left[ \left< R \right> \exp \left[ \exp \left[ C \left< t \right> \right] \right] \right] 
\bigg[ 
1 +  (\ln(\ln(N)))^{2}
\bigg]^{1/2}
\nonumber \\
&\leq \exp \left[ \left< R \right> \exp \left[ \exp \left[ C \left< t \right> \right] \right] \right] 
\left(
1 +  \ln(\ln(N)) \right) .
\end{align}
Note that for arbitrary but fixed $a>0$ and $0 < \delta < 1$ there exists a constant $C$ (depending on $a$ and $\delta$) such that 
\begin{align}
(\ln (\ln (N)))^a \leq C N^{\delta}
\quad \text{for every $N >e$.}
\end{align}
For any $k \leq 8$ ,  $N > 3$, and any $\delta >0$ fixed but arbitrary this gives the estimate 
\begin{align}
\norm{W_{N,\cdot}}^2_{L_t^{\infty} W_7^{k,2}(\mathbb{R}^6)} + \norm{\alpha_{\cdot}}^4_{L_t^{\infty} \mathfrak{h}_{k}}
&\leq 
\exp \left[ \left< R \right> \exp \left[ \exp \left[ C \left< t \right> \right] \right] \right] N^{\delta}.
\end{align}
These estimates lead to 
\begin{align}
\varepsilon \widetilde{\mathcal{K}}(t)
&\leq \varepsilon^{1/2} \exp \left[ \left< R \right> \exp \left[ \exp \left[ C \left< t \right> \right] \right] \right]
\end{align}
and 
\begin{align}
\mathcal{K}(t)
&\leq  \left( 1 + N^{-1} \mathcal{E}^{\rm{MS}}[\mathcal{W}^{-1}[m_{p,\sigma}],  \id_{\abs{\cdot} \leq \Lambda} \alpha] + \mathcal{E}^{\rm{VM}}[m_{p,\sigma,R}, \id_{\abs{\cdot} \leq \Lambda} \alpha] + 
\norm{m_{p,\sigma,R}}_{L^1(\mathbb{R}^6)} \right)^2
\nonumber \\
&\quad \times 
\exp \left[ \left< R \right> \exp \left[ \exp \left[ C \left< t \right> \right] \right] \right]
\left(
1 +  \ln(\ln(N)) \right) .
\end{align}
Note the inequalities (see \cite[(III.29b) and (III.29c)]{LS2023} for their proof)
\begin{align}
\tr \left( - \varepsilon^2 \Delta \omega \right)
&\leq C \left(  \mathcal{E}^{\rm{MS}}[\omega,   \alpha]
+ C N^{-1} \norm{\omega}_{\mathfrak{S}^1}^2  \right)
\left( 1 +  N^{-1} \norm{\omega}_{\mathfrak{S}^1} \right)
\\
\mathcal{E}^{\rm{MS}}[\omega,   \alpha]
&\leq \tr \left( - \varepsilon^2 \Delta \omega \right) + C N^{-1} \norm{\omega}_{\mathfrak{S}^1}^2 
+ 
\norm{\alpha}_{\dot{\mathfrak{h}}_{1/2}}^2
\left( C \norm{\omega}_{\mathfrak{S}^1} + N \right)  .
\end{align}
Using \eqref{eq:distance between p and regularized version in Sobolev-trace norm 1}
together with \eqref{eq:main result condition on the energy and velocity},  $\tr \left( p\right) = N$,  and  \eqref{eq:fermionic reduced density proeprty of the weyl quantization of m-p-sigma},
we obtain
\begin{align}
\label{eq:derivation of VM bound for the energy of regularized initial data}
\mathcal{E}^{\rm{MS}}[\mathcal{W}^{-1}[m_{p,\sigma}],  \id_{\abs{\cdot} \leq \Lambda} \alpha]
&\leq \tr \left( - \varepsilon^2 \Delta \mathcal{W}^{-1}[m_{p,\sigma}] \right) + C N^{-1} \norm{\mathcal{W}^{-1}[m_{p,\sigma}]}_{\mathfrak{S}^1}^2 
\nonumber \\
&\quad + 
\norm{\id_{\abs{\cdot} \leq \Lambda} \alpha}_{\dot{\mathfrak{h}}_{1/2}}^2
\left( C \norm{\mathcal{W}^{-1}[m_{p,\sigma}]}_{\mathfrak{S}^1} + N \right) 
\nonumber \\
&\leq C N + \tr \left( - \varepsilon^2 \Delta p \right)
+ \norm{ \mathcal{W}^{-1}[m_{p, \sigma}] - p }_{\mathfrak{S}^{1,1}_{\varepsilon}}
\nonumber \\
&\leq C N .
\end{align}
By means of
\eqref{eq:estimate for VM-energy with regularized data} we get
\begin{align}
\abs{\mathcal{E}^{\rm{VM}}[m_{p,\sigma,R} , \id_{\abs{\cdot} \leq \Lambda} \alpha]} &\leq 
C \left( 1 + \norm{\alpha}_{\dot{\mathfrak{h}}_{1/2}}^2 \right)
\left( 1 + \norm{\mathcal{W}[p]}^2_{W_{6}^{0,2}(\mathbb{R}^6)} \right) \leq C.
\end{align}
Together with \eqref{eq:estimate for regularized Wigner 2} this shows the  existence of a constant $C$ which depends on $\kappa$, $\norm{\mathcal{W}[p]}_{W_7^{0,2}(\mathbb{R}^6)}$, $\norm{\alpha}_{\mathfrak{h}_1 \, \cap \,  \dot{\mathfrak{h}}_{-1/2}}$, and $N^{-1} \mathcal{E}^{\rm{MS}}[p, \alpha]$ such that 
\begin{align}
\mathcal{K}(t)
&\leq  \left(
1 +  \ln(\ln(N)) \right) 
\exp \left[ \left< R \right> \exp \left[ \exp \left[ C \left< t \right> \right] \right] \right].
\end{align}
Collecting the estimates let us obtain for $N > 16$
\begin{align}
& N^{-1} \norm{ \widetilde{p}_t - \mathcal{W}^{-1}[W_{N,t}]}_{\mathfrak{S}^{1,1}_{\varepsilon}} + \norm{\widetilde{\alpha}_t - \alpha_t}_{\mathfrak{h}_{1/2} \cap \, \dot{\mathfrak{h}}_{- 1/2}}
\nonumber \\
&\leq  C  \left[   e^{- \frac{1}{2 \sigma^2}}  + \varepsilon^{1/2}
+ \varepsilon \widetilde{K}(t) \right]
e^{\mathcal{K}(t)} 
\nonumber \\
&\leq  C  \left[   e^{- \frac{1}{2} (\ln (\ln(N)))^2}  + \varepsilon^{1/2}
 \right]
\exp \left[ \left( 1 + \ln(\ln(N)) \right)  \exp \left[ \left< R \right> \exp \left[ \exp \left[ C \left< t \right> \right] \right] \right]  \right]
\nonumber \\
&\leq  C  \left[   e^{- \frac{1}{2} (\ln (\ln(N)))^2}  + \varepsilon^{1/2}
 \right]
\exp \left[ \ln(\ln(N)) \exp \left[ \left< R \right> \exp \left[ \exp \left[ C \left< t \right> \right] \right] \right]   \right]
\nonumber \\
&= C \exp \left[  - \frac{ \ln(\ln(N))}{2} \left\{  \ln(\ln(N)) - 2  \exp \left[ \left< R \right> \exp \left[ \exp \left[ C \left< t \right> \right] \right] \right]  \right\}  \right] 
\nonumber \\
&\qquad 
+  C \varepsilon^{1/2} \exp \left[  \ln(\ln(N)) \exp \left[ \left< R \right> \exp \left[ \exp \left[ C \left< t \right> \right] \right] \right]   \right] .
\end{align}
For all $N \in \mathbb{N}$ such that 
\begin{align}
&\ln(\ln(N)) - 2  \exp \left[ \left< R \right> \exp \left[ \exp \left[ C \left< t \right> \right] \right] \right] \geq 2
\nonumber \\
&\quad \Longleftrightarrow
N \geq  \exp \left[ \exp \left[ 2 \left[ 1 + \exp \left[ \left< R \right> \exp \left[ \exp \left[ C \left< t \right> \right] \right] \right]  \right] \right] \right] 
\end{align}
we consequently get 
\begin{align}
& N^{-1} \norm{ \widetilde{p}_t - \mathcal{W}^{-1}[W_{N,t}]}_{\mathfrak{S}^{1,1}_{\varepsilon}} + \norm{\widetilde{\alpha}_t - \alpha_t}_{\mathfrak{h}_{1/2} \cap \, \dot{\mathfrak{h}}_{- 1/2}}
\nonumber \\
&\quad \leq  C \left(  \ln(N)^{-1} + N^{- 1/6} 
 \left[  \ln (N) \right]^{ \exp \left[ \left< R \right> \exp \left[ \exp \left[ C \left< t \right> \right] \right] \right]}
\right) .
\end{align}
\end{proof}

\subsection{Proof of Lemma~\ref{lemma: semiclassical structure with kinetic weights}}

\begin{proof}
Let $\nabla_{\leq n}$ be defined as in \eqref{eq:definition of the regularized gradient} and the regularized version of the operator $D_{\varepsilon} = \left( 1 - \varepsilon^2 \Delta \right)^{1/2}$   defined by
\begin{align}
\label{eq:regularized fractional Laplacian with semiclassical parameter}
D_{\leq n} = \frac{\left( 1 - \varepsilon^2 \Delta \right)^{1/2}}{\left( 1 - \varepsilon^2 \Delta / n^2 \right)^{1/2}}
\quad \text{with} \quad n \in \mathbb{N} .
\end{align}
Moreover, let $H_{p,\alpha}(t)$ and $U_{p,\alpha}(t)$ be defined as in  \eqref{eq:definition mean-field Hamiltonian} and \eqref{eq:Duhamel for operator norm of gradient p}, respectively.  Then,
\begin{align}
\frac{d}{dt} \big( U^*_{p,\alpha}(t) D_{\leq n} q_t i \varepsilon \nabla_{\leq n} p_t D_{\leq n} U_{p,\alpha}(t) \big)
&= i \varepsilon^{-1} U^*_{p,\alpha}(t)
\big[ H_{p,\alpha}(t) ,  D_{\leq n} \big]  q_t i \varepsilon \nabla_{\leq n} p_t D_{\leq n}  U_{p,\alpha}(t)
\nonumber \\
&\quad 
+ i \varepsilon^{-1} U^*_{p,\alpha}(t)
 D_{\leq n} q_t i \varepsilon \nabla_{\leq n} p_t \big[ H_{p,\alpha}(t) ,  D_{\leq n} \big]  U_{p,\alpha}(t)
\nonumber \\
&\quad + i  U^*_{p,\alpha}(t) D_{\leq n}
q_t \big[ H_{p,\alpha}(t) + \varepsilon^2 \Delta ,  i  \nabla_{\leq n}  \big]  p_t D_{\leq n} U_{p,\alpha}(t) .
\end{align}
Using Duhamel's formula and integrating in time gives
\begin{align}
&\norm{D_{\leq n} q_t i \varepsilon \nabla_{\leq n} p_t D_{\leq n}}_{\mathfrak{S}^1} -
\norm{D_{\leq n} q_0 i \varepsilon \nabla_{\leq n} p_0 D_{\leq n}}_{\mathfrak{S}^1}
\nonumber \\
&\quad \leq \varepsilon^{-1} \int_0^t 
\left( \norm{\big[ H_{p,\alpha}(s) ,  D_{\leq n} \big]  q_s i \varepsilon \nabla_{\leq n} p_s D_{\leq n}}_{\mathfrak{S}^1}
+
\norm{ D_{\leq n} q_s i \varepsilon \nabla_{\leq n} p_s \big[ H_{p,\alpha}(s) ,  D_{\leq n} \big]}_{\mathfrak{S}^1} \right) \, ds
\nonumber \\
&\qquad +   \int_0^t 
\norm{ D_{\leq n}
q_s \big[ H_{p,\alpha}(s) + \varepsilon^2 \Delta ,  i  \nabla_{\leq n}  \big]  p_s D_{\leq n}}_{\mathfrak{S}^1}  \, ds
\nonumber \\
&\quad \leq 2  \varepsilon^{-1} \int_0^t 
\norm{\big[ H_{p,\alpha}(s)  + \varepsilon^2 \Delta ,  D_{\leq n} \big] D_{\varepsilon}^{-1}}_{\mathfrak{S}^{\infty}}
\norm{ D_{\varepsilon}  q_s i \varepsilon \nabla_{\leq n} p_s D_{\varepsilon}}_{\mathfrak{S}^1}
 \, ds
\nonumber \\
&\qquad +   \int_0^t 
\norm{ D_{\leq n}
q_s \big[ H_{p,\alpha}(s) + \varepsilon^2 \Delta ,  i  \nabla_{\leq n}  \big]  p_s D_{\leq n}}_{\mathfrak{S}^1}  \, ds .
\end{align}
Taking the limit $n \rightarrow \infty$  leads to
\begin{subequations}
\begin{align}
&\norm{ q_t i \varepsilon \nabla p_t }_{\mathfrak{S}_{\varepsilon}^{1,1}} -
\norm{q_0 i \varepsilon \nabla p_0 }_{\mathfrak{S}_{\varepsilon}^{1,1}}
\nonumber \\
\label{eq:improved semiclassical structure gradient term a}
&\quad \leq 2  \varepsilon^{-1} \int_0^t 
\norm{ q_s i \varepsilon \nabla p_s }_{\mathfrak{S}_{\varepsilon}^{1,1}}
\norm{\big[ H_{p,\alpha}(s)  + \varepsilon^2 \Delta ,  D_{\varepsilon} \big] D_{\varepsilon}^{-1}}_{\mathfrak{S}^{\infty}} \, ds 
\\
\label{eq:improved semiclassical structure gradient term b1}
&\qquad +  2  \int_0^t 
\lim_{n \rightarrow \infty} \norm{ 
q_s \big[ \kappa * \vA_{\alpha_s} ,  i  \nabla_{\leq n}  \big] \cdot i \varepsilon \nabla p_s}_{\mathfrak{S}_{\varepsilon}^{1,1}}  \, ds 
\\
\label{eq:improved semiclassical structure gradient term b}
&\qquad +   \int_0^t 
\norm{
q_s \big[ H_{p,\alpha}(s) + \varepsilon^2 \Delta - 2  \kappa * \vA_{\alpha_s} \cdot i \varepsilon \nabla  ,  i  \nabla  \big]  p_s }_{\mathfrak{S}_{\varepsilon}^{1,1}}  \, ds  .
\end{align}
\end{subequations}
In the following we estimate each term separately.
\paragraph{The term \eqref{eq:improved semiclassical structure gradient term a}}
By means of (see \cite[Lemma~III.5]{LS2023} for a proof)
\begin{align}
\norm{\left[ V , D_{\varepsilon} \right] }_{\mathfrak{S}^{\infty}} &\leq 2 \varepsilon \int_{\mathbb{R}^3} \abs{k} \abs{\mathcal{F}[V](k)} \, dk 
\end{align}
for $V: \mathbb{R}^3 \rightarrow \mathbb{C}^3$ regular enough,  the estimate
\begin{align}
\label{eq:regularized MS estimate for the Fourier transform of the density}
\norm{\mathcal{F}[\rho_{p_s}]}_{\mathfrak{S}^{\infty}}
&\leq N^{-1} \norm{p_s}_{\mathfrak{S}^1},
\end{align} 
and  \eqref{eq:Fourier transform of the potential}
we obtain
\begin{align}
\label{eq:estimate for S-infty norm of commutator with regularized fractional Laplacian}
\norm{
 \left[   \kappa * \vA_{ \alpha_s}   , D_{\varepsilon} \right] }_{\mathfrak{S}^{\infty}}
&\leq C  \varepsilon \sum_{\lambda = 1,2} \int_{\mathbb{R}^3} \abs{k}^{1/2} \abs{\mathcal{F}[\kappa](k)} \abs{\alpha_s(k,\lambda)} \, dk
\leq C \varepsilon \norm{\mathcal{F}[\kappa]}_{L^2} 
\norm{ \alpha_s}_{\dot{\mathfrak{h}}_{1/2}} ,
\nonumber \\
\norm{ \left[  K * \rho_{p_s}  ,  D_{\varepsilon} \right] }_{\mathfrak{S}^{\infty}}
&\leq C  \varepsilon N^{-1} \norm{p_s}_{\mathfrak{S}^1}
\big\|  \abs{\cdot}^{-1/2}  \mathcal{F}[\kappa] \big\|^2_{L^2}  .
\end{align}
Due to
\begin{align}
X_{p_s}
&= 
\frac{1}{(2 \pi)^{3/2} N} 
\int_{\mathbb{R}^3} \mathcal{F}[K](k) e^{i k \hat{x}} p_s  e^{- i k \hat{x}}  \, dk ,
\end{align}
the estimate
\begin{align}
\norm{D_{\varepsilon} e^{i k \hat{x}} p_s}_{\mathfrak{S}^{\infty}}
&\leq C \left( \norm{i \varepsilon \nabla p_s}_{\mathfrak{S}^{\infty}}
+ (1 + \varepsilon \abs{k}) \norm{p_s}_{\mathfrak{S}^{\infty}} \right) ,
\end{align}
and \eqref{eq:Fourier transform of the potential},  we get
\begin{align}
\norm{  \left[   X_{p_s}   ,  D_{\varepsilon}  \right] D_{\varepsilon}^{-1} }_{\mathfrak{S}^{\infty}} 
&\leq 2  \norm{ D_{\varepsilon}   X_{p_s}   }_{\mathfrak{S}^{\infty}} 
\nonumber \\
&\leq N^{-1} \int_{\mathbb{R}^3} \abs{\mathcal{F}[K](k)} \norm{ D_{\varepsilon} e^{i k \hat{x}} p_s}_{\mathfrak{S}^{\infty}} \, dk
\nonumber \\
&\leq C N^{-1} \left( 1  + \norm{i \varepsilon \nabla p_s}_{\mathfrak{S}^{\infty}} 
+  \norm{p_s}_{\mathfrak{S}^{\infty}}
\right)
 \int_{\mathbb{R}^3} \abs{\mathcal{F}[\kappa](k)}^2 \abs{k}^{-2} \left( 1 + \abs{k} \right)
 \, dk
\nonumber \\
&\leq C N^{-1} 
\left( 1  + \norm{i \varepsilon \nabla p_s}_{\mathfrak{S}^{\infty}}
+  \norm{p_s}_{\mathfrak{S}^{\infty}} \right) .
\end{align}
Together with \eqref{eq:estimate for A L-infty to h-1-2},
$ \norm{p_s}_{\mathfrak{S}^1} = N$,  and $ \norm{p_s}_{\mathfrak{S}^{\infty}}  \leq 1$ this leads to
\begin{align}
\label{eq:semiclassical structure with kinetic weights commuator between hamiltonian and regularized fractional Laplacian}
&\norm{ \left[  H_{p,\alpha}(s)  + \varepsilon^2 \Delta   ,  D_{\varepsilon} \right] D_{\varepsilon}^{-1} }_{\mathfrak{S}^{\infty}} 
\nonumber \\
&\quad \leq 2  
\left( 1 + \norm{  \kappa * \vA_{ \alpha_s} }_{\mathfrak{S}^{\infty}} \right)
 \norm{  \left[   \kappa * \vA_{ \alpha_s}   ,  D_{\varepsilon} \right]   }_{\mathfrak{S}^{\infty}} 
+ \norm{ \left[  K * \rho_{p_s}  ,  D_{\varepsilon}  \right] }_{\mathfrak{S}^{\infty}}
+ \norm{  \left[   X_{p_s}   ,  D_{\varepsilon}  \right] D_{\varepsilon}^{-1} }_{\mathfrak{S}^{\infty}} 
\nonumber \\
&\quad \leq  C \varepsilon 
 \left( 1 + \norm{ \alpha_s}_{\dot{\mathfrak{h}}_{1/2}}^2
+  \norm{i \varepsilon \nabla p_s}_{\mathfrak{S}^{\infty}}  \right)  
\end{align}
and 
\begin{align}
\eqref{eq:improved semiclassical structure gradient term a}
&\leq C  \int_0^t 
\left( 1 + \norm{ \alpha_s}_{\dot{\mathfrak{h}}_{1/2}}^2
+  \norm{i \varepsilon \nabla p_s}_{\mathfrak{S}^{\infty}}  \right)
\norm{ q_s i \varepsilon \nabla p_s }_{\mathfrak{S}_{\varepsilon}^{1,1}}  \,  ds .
\end{align}
\paragraph{The term~\eqref{eq:improved semiclassical structure gradient term b1}}
We insert the identity $1 = p_s + q_s$ and estimate
\begin{align}
&\norm{ 
q_s \big[ \kappa * \vA_{\alpha_s} ,  i  \nabla_{\leq n}  \big] \cdot i \varepsilon \nabla p_s}_{\mathfrak{S}_{\varepsilon}^{1,1}}
\nonumber \\
&\quad \leq  2 \norm{ 
q_s \big[ \kappa * \vA_{\alpha_s} ,  i  \nabla_{\leq n}  \big] \cdot q_s i \varepsilon \nabla p_s }_{\mathfrak{S}_{\varepsilon}^{1,1}}
+ 2 \norm{ 
q_s \big[ \kappa * \vA_{\alpha_s} ,  i  \nabla_{\leq n}  \big] \cdot p_s i \varepsilon \nabla p_s }_{\mathfrak{S}_{\varepsilon}^{1,1}}
\nonumber \\
&\quad \leq  2 \norm{ D_{\varepsilon}
q_s \big[ \kappa * \vA_{\alpha_s} ,  i  \nabla_{\leq n}  \big] D_{\varepsilon}^{-1} }_{\mathfrak{S}^{\infty}} 
\norm{q_s i \varepsilon \nabla p_s }_{\mathfrak{S}_{\varepsilon}^{1,1}}
\nonumber \\
&\qquad + 2 \norm{ q_s \big[ \kappa * \vA_{\alpha_s} ,  i  \nabla_{\leq n}  \big] p_s }_{\mathfrak{S}_{\varepsilon}^{1,1}}
\norm{D_{\varepsilon}^{-1} i \varepsilon \nabla p_s D_{\varepsilon}}_{\mathfrak{S}^{\infty}} 
\nonumber \\
&\quad \leq C \left( 1 + \norm{i \varepsilon \nabla p_s}_{\mathfrak{S}^{\infty}} \right)
\norm{ D_{\varepsilon} \big[ \kappa * \vA_{\alpha_s} ,  i  \nabla_{\leq n}  \big] D_{\varepsilon}^{-1} }_{\mathfrak{S}^{\infty}} 
\norm{q_s i \varepsilon \nabla p_s }_{\mathfrak{S}_{\varepsilon}^{1,1}}
\nonumber \\
&\qquad
+ C \left( 1 + \norm{i \varepsilon \nabla p_s}_{\mathfrak{S}^{\infty}} \right) \norm{ q_s \big[ \kappa * \vA_{\alpha_s} ,  i  \nabla_{\leq n}  \big] p_s }_{\mathfrak{S}_{\varepsilon}^{1,1}}  .
\end{align}
Using $\norm{\left[ i \nabla_{\leq n} ,  e^{ikx} \right]}_{\mathfrak{S}^{\infty}} \leq 3 \abs{k}$  and $\norm{(1 + \varepsilon^2 (- i \nabla + k)^2)^{1/2} D_{\varepsilon}^{-1}} \leq  C \left( 1 + \varepsilon \abs{k} \right)$ we get
\begin{align}
&\norm{ D_{\varepsilon} \left[ i \nabla_{\leq n} ,  \kappa * \vA_{\alpha_s} \right] D_{\varepsilon}^{-1}}_{\mathfrak{S}^{\infty}}
\nonumber \\
&\quad \leq \int_{\mathbb{R}^3} \abs{\mathcal{F}[\kappa * \vA_{\alpha_s}](k)}
\norm{D_{\varepsilon}  \left[ i \nabla_{\leq n} ,  e^{ikx} \right] D_{\varepsilon}^{-1}  }_{\mathfrak{S}^{\infty}} \, dk
\nonumber \\
&\quad = \int_{\mathbb{R}^3} \abs{\mathcal{F}[\kappa * \vA_{\alpha_s}](k)}
\norm{  \left[ i \nabla_{\leq n} ,  e^{ikx} \right] (1 + \varepsilon^2 (- i \nabla + k)^2)^{1/2} D_{\varepsilon}^{-1}  }_{\mathfrak{S}^{\infty}} \, dk
\nonumber \\
&\quad \leq \int_{\mathbb{R}^3} \abs{\mathcal{F}[\kappa * \vA_{\alpha_s}](k)}
\norm{  \left[ i \nabla_{\leq n} ,  e^{ikx} \right] }_{\mathfrak{S}^{\infty}} 
\norm{ (1 + \varepsilon^2 (- i \nabla + k)^2)^{1/2} D_{\varepsilon}^{-1}  }_{\mathfrak{S}^{\infty}} \, dk
\nonumber \\
&\quad \leq 
C \sum_{\lambda = 1,2}  \int_{\mathbb{R}^3} \abs{k}^{1/2} (1 + \abs{k}) \abs{\mathcal{F}[\kappa](k)} \abs{\alpha_s(k,\lambda)}  \,  dk
\nonumber \\
&\quad \leq C \big\| \abs{\cdot}^{1/2} \mathcal{F}[\kappa] \big\|_{L^2} \norm{\alpha_s}_{\mathfrak{h}_1} .
\end{align}
Moreover,
\begin{align}
&\lim_{n \rightarrow \infty} \norm{ q_s \big[ \kappa * \vA_{\alpha_s} ,  i  \nabla_{\leq n}  \big] p_s }_{\mathfrak{S}_{\varepsilon}^{1,1}}
\nonumber \\
&\quad =  \norm{ q_s \big[   \kappa * \vA_{\alpha_s} ,   i  \nabla \big] p_s }_{\mathfrak{S}_{\varepsilon}^{1,1}}
\nonumber \\
&\quad \leq C \sum_{\lambda = 1,2} \int_{\mathbb{R}^3}
\abs{\mathcal{F}[\kappa](k)} \abs{k}^{1/2} \abs{\alpha_s(k,\lambda)} 
\norm{ q_s e^{ikx} p_s }_{\mathfrak{S}_{\varepsilon}^{1,1}}
\nonumber \\
&\quad \leq C \big\| \abs{\cdot}^{1/2} \mathcal{F}[\kappa] \big\|_{L^2} \norm{\alpha_s}_{\mathfrak{h}_1}
\sup_{k \in \mathbb{R}^3}
\left\{ (1 + \abs{k})^{-1}  \norm{ q_s e^{ikx} p_s }_{\mathfrak{S}_{\varepsilon}^{1,1}}  \right\}.
\end{align}
In total,  this leads to
\begin{align}
\eqref{eq:improved semiclassical structure gradient term b1}
&\leq C \int_0^t   \left( 1 + \norm{i \varepsilon \nabla p_s}_{\mathfrak{S}^{\infty}} \right)  \norm{\alpha_s}_{\mathfrak{h}_1}  \Xi_2[N,p_s]
\, ds  .
\end{align}
\paragraph{The term \eqref{eq:improved semiclassical structure gradient term b}}
Note that
\begin{subequations}
\begin{align}
\label{eq:improved semiclassical structure gradient term b a}
\norm{
q_s \big[ H_{p,\alpha}(s) + \varepsilon^2 \Delta - 2  \kappa * \vA_{\alpha_s} \cdot i \varepsilon \nabla  ,  i  \nabla  \big]  p_s }_{\mathfrak{S}_{\varepsilon}^{1,1}}
&\leq  \norm{ q_s \big[ \left(  \kappa * \vA_{\alpha_s} \right)^2 ,  i  \nabla  \big]  p_s }_{\mathfrak{S}_{\varepsilon}^{1,1}}
\\
\label{eq:improved semiclassical structure gradient term b b}
&\quad + \norm{q_s \big[ K * \rho_{p_s} ,  i  \nabla  \big]  p_s }_{\mathfrak{S}_{\varepsilon}^{1,1}}
\\
\label{eq:improved semiclassical structure gradient term b c}
&\quad + \norm{q_s \big[ X_{p_s} ,  i  \nabla  \big]  p_s }_{\mathfrak{S}_{\varepsilon}^{1,1}} .
\end{align}
\end{subequations}
Inserting the identity $1 = p_s + q_s$ leads to
\begin{align}
\eqref{eq:improved semiclassical structure gradient term b a}
&\leq  \norm{ q_s \kappa * \vA_{\alpha_s} \big[ \kappa * \vA_{\alpha_s}  ,  i  \nabla  \big]  p_s }_{\mathfrak{S}_{\varepsilon}^{1,1}}
+
\norm{ q_s \big[  \kappa * \vA_{\alpha_s} ,  i  \nabla  \big] \kappa * \vA_{\alpha_s}  p_s }_{\mathfrak{S}_{\varepsilon}^{1,1}}
\nonumber \\
&\leq \norm{ q_s \kappa * \vA_{\alpha_s} p_s \big[ \kappa * \vA_{\alpha_s}  ,  i  \nabla  \big]  p_s }_{\mathfrak{S}_{\varepsilon}^{1,1}}
+ \norm{ q_s \kappa * \vA_{\alpha_s} q_s \big[ \kappa * \vA_{\alpha_s}  ,  i  \nabla  \big]  p_s }_{\mathfrak{S}_{\varepsilon}^{1,1}}
\nonumber \\
&\quad +
\norm{ q_s \big[  \kappa * \vA_{\alpha_s} ,  i  \nabla  \big] p_s \kappa * \vA_{\alpha_s}  p_s }_{\mathfrak{S}_{\varepsilon}^{1,1}}
+ 
\norm{ q_s \big[  \kappa * \vA_{\alpha_s} ,  i  \nabla  \big] q_s \kappa * \vA_{\alpha_s}  p_s }_{\mathfrak{S}_{\varepsilon}^{1,1}}
\nonumber \\
&\leq 
\norm{D_{\varepsilon} p_s \big[ \kappa * \vA_{\alpha_s}  ,  i  \nabla  \big] }_{\mathfrak{S}^{\infty}}
\norm{ q_s   \kappa * \vA_{\alpha_s}    p_s }_{\mathfrak{S}_{\varepsilon}^{1,1}}
\nonumber \\
&\quad +
\norm{D_{\varepsilon} q_s \kappa * \vA_{\alpha_s} D_{\varepsilon}^{-1}}_{\mathfrak{S}^{\infty}} 
 \norm{q_s \big[ \kappa * \vA_{\alpha_s}  ,  i  \nabla  \big]  p_s }_{\mathfrak{S}_{\varepsilon}^{1,1}}
\nonumber \\
&\quad +
\norm{D_{\varepsilon} p_s \kappa * \vA_{\alpha_s}  }_{\mathfrak{S}^{\infty}}
\norm{ q_s \big[  \kappa * \vA_{\alpha_s} ,  i  \nabla  \big] p_s}_{\mathfrak{S}_{\varepsilon}^{1,1}}
\nonumber \\
&\quad +
\norm{D_{\varepsilon} q_s \big[  \kappa * \vA_{\alpha_s} ,  i  \nabla  \big] D_{\varepsilon}^{-1}}_{\mathfrak{S}^{\infty}}  
\norm{q_s \kappa * \vA_{\alpha_s}  p_s }_{\mathfrak{S}_{\varepsilon}^{1,1}}
\nonumber \\
&\leq C \left( 1 + \norm{i \varepsilon \nabla p_s}_{\mathfrak{S}^{\infty}} \right) \norm{\kappa * \vA_{\alpha_s}}_{W_0^{2,\infty}(\mathbb{R}^3)}
\Big( \norm{ q_s   \kappa * \vA_{\alpha_s}    p_s }_{\mathfrak{S}_{\varepsilon}^{1,1}}
\nonumber \\
&\qquad  \qquad   \qquad  \qquad  \qquad   \qquad  \qquad  \qquad  \qquad  
+ \norm{ q_s \big[  \kappa * \vA_{\alpha_s} ,  i  \nabla  \big] p_s}_{\mathfrak{S}_{\varepsilon}^{1,1}}
\Big) .
\end{align}
Since
\begin{align}
\norm{\kappa * \vA_{\alpha_s}}_{W_0^{2,\infty}(\mathbb{R}^3)}
&\leq C \sum_{\lambda=1,2} \int_{\mathbb{R}^3} \abs{k}^{-1/2} (1 + \abs{k}^2) \abs{\mathcal{F}[\kappa](k)} \abs{\alpha_s(k,\lambda)} \, dk
\leq  C  \norm{\alpha_s}_{\mathfrak{h}_1}  ,
\end{align}
\begin{align}
\label{eq:semiclassical structure with kinetic weights estimate for q A p}
\norm{ q_s   \kappa * \vA_{\alpha_s}    p_s }_{\mathfrak{S}_{\varepsilon}^{1,1}}
&\leq C \sum_{\lambda=1,2} \int_{\mathbb{R}^3} \abs{k}^{-1/2}  \abs{\mathcal{F}[\kappa](k)} \abs{\alpha_s(k,\lambda)} 
\norm{ q_s  e^{ikx}  p_s }_{\mathfrak{S}_{\varepsilon}^{1,1}}
\, dk
\nonumber \\
&\leq C  \norm{\alpha_s}_{\dot{\mathfrak{h}}_{1/2}}  \Xi_2[N,p_s] ,
\end{align}
and
\begin{align}
\norm{ q_s \big[  \kappa * \vA_{\alpha_s} ,  i  \nabla  \big] p_s}_{\mathfrak{S}_{\varepsilon}^{1,1}}
&\leq  C \sup_{k \in \mathbb{R}^3}
\left\{  (1 + \abs{k})^{-1}  \norm{ q_s  e^{ikx}  p_s }_{\mathfrak{S}_{\varepsilon}^{1,1}} \right\}
\nonumber \\
&\quad \times 
\sum_{\lambda=1,2} \int_{\mathbb{R}^3} \abs{k}^{1/2}  (1 + \abs{k}) \abs{\mathcal{F}[\kappa](k)} \abs{\alpha_s(k,\lambda)} 
\, dk
\nonumber \\
&\leq C \norm{\alpha_s}_{\mathfrak{h}_1}  \Xi_2[N,p_s]
\end{align}
we get
\begin{align}
\eqref{eq:improved semiclassical structure gradient term b a}
&\leq  C   \left( 1 + \norm{i \varepsilon \nabla p_s}_{\mathfrak{S}^{\infty}} \right)  \norm{\alpha_s}_{\mathfrak{h}_1}^2  \Xi_2[N,p_s]  .
\end{align}
Using  \eqref{eq:Fourier transform of the potential}  and
$\norm{\mathcal{F}[\rho_{p_s}]}_{L^{\infty}} \leq \norm{\rho_{p_s}} = N^{-1} \tr \left( p_s \right) = 1$ we obtain
\begin{align}
\eqref{eq:improved semiclassical structure gradient term b b}
&\leq C \norm{\mathcal{F}[\rho_{p_s}]}_{L^{\infty}} 
\sup_{k \in \mathbb{R}^3}
\left\{  (1 + \abs{k})^{-1}  \norm{ q_s  e^{ikx}  p_s }_{\mathfrak{S}_{\varepsilon}^{1,1}}  \right\}
\int_{\mathbb{R}^3} \abs{k}^{-1}  (1 + \abs{k}) \abs{\mathcal{F}[\kappa](k)}^2
\nonumber \\
&\leq C \,  \Xi_2[N,p_s].
\end{align}
By means of \eqref{eq:explicite expression for commutator between laplacian and exchange term},  $\left[ i \nabla , p_s \right] = q_s i \nabla p_s - p_s i \nabla q_s$,  and \eqref{eq:Fourier transform of the potential}
we estimate 
\begin{align}
\eqref{eq:improved semiclassical structure gradient term b c}
&\leq 
C N^{-1}
\int_{\mathbb{R}^3} \abs{\mathcal{F}[K](k)}
\norm{ q_s e^{i k \hat{x}} \left[ i  \nabla, p_s \right] e^{- i k \hat{x}} p_s }_{S_{\varepsilon}^{1,1}}  \, dk 
\nonumber \\
&\leq 
C N^{-1}  \left( 1 + \norm{i \varepsilon \nabla p_s}_{\mathfrak{S}^{\infty}}^2 \right)
\norm{q_s i \nabla p_s}_{\mathfrak{S}_{\varepsilon}^{1,1}} \int_{\mathbb{R}^3}  \abs{\mathcal{F}[K](k)} ( 1 + \abs{k}^2) \, dk 
\nonumber \\
&\leq 
C N^{-1}  \left( 1 + \norm{i \varepsilon \nabla p_s}_{\mathfrak{S}^{\infty}}^2 \right)
\norm{q_s i \nabla p_s}_{\mathfrak{S}_{\varepsilon}^{1,1}} \int_{\mathbb{R}^3}  \abs{\mathcal{F}[\kappa](k)}^2 ( \abs{k}^{-2} + 1) \, dk 
\nonumber \\
&\leq 
C N^{-1}   \left( 1 + \norm{i \varepsilon \nabla p_s}_{\mathfrak{S}^{\infty}}^2 \right)
\norm{q_s i \nabla p_s}_{\mathfrak{S}_{\varepsilon}^{1,1}}  .
\end{align}
Collecting the estimates we,  finally,   gives
\begin{align}
\eqref{eq:improved semiclassical structure gradient term b}
&\leq  
C  
 \int_0^t
\left( 1 + \norm{i \varepsilon \nabla p_s}_{\mathfrak{S}^{\infty}}^2 \right)  \norm{\alpha_s}_{\mathfrak{h}_1}^2
\Xi_2[N,p_s] \,  ds
\end{align}
and 
\begin{align}
\norm{ q_t i \varepsilon \nabla p_t }_{\mathfrak{S}_{\varepsilon}^{1,1}} -
\norm{q_0 i \varepsilon \nabla p_0 }_{\mathfrak{S}_{\varepsilon}^{1,1}}
&\leq 
C  \int_0^t
\left( 1 + \norm{i \varepsilon \nabla p_s}_{\mathfrak{S}^{\infty}}^2 \right)   \left( 1 + \norm{\alpha_s}^2_{\mathfrak{h}_1}  \right)
\Xi_2[N,p_s] \,  ds .
\end{align}
Next,  we provide a suitable estimate for the quantity $(1 + \abs{k})^{-1}  \norm{ q_s  e^{ikx}  p_s }_{\mathfrak{S}_{\varepsilon}^{1,1}}$. By the same argument as in the derivation of \eqref{eq:propagation of q exponential p in the semiclassical structure} we get
\begin{align}
&i \varepsilon \partial_t \big( U_{+}^*(t;0)D_{\leq n}  q_t e^{ikx} p_t  D_{\leq n}  U_{-}(t;0) \big)
\nonumber \\
&\quad = 
U_{+}^*(t;0) D_{\leq n}  q_t \big[ X_{p_t} , e^{ikx} \big] p_t D_{\leq n}   U_{-}(t;0)
\nonumber \\
&\qquad +
\varepsilon k \cdot U_{+}^*(t;0) D_{\leq n}  \Big(
\big( q_t \boldsymbol{P}_{\alpha_t}  p_t - p_t \boldsymbol{P}_{\alpha_t} q_t \big) e^{ikx} p_t
+ q_t e^{ikx} \big( q_t \boldsymbol{P}_{\alpha_t}  p_t - p_t \boldsymbol{P}_{\alpha_t} q_t \big) 
\Big) D_{\leq n}  U_{-}(t;0)
\nonumber \\
&\qquad -  U_{+}^*(t;0) \left[ H_{p,\alpha}(t) + \varepsilon k \cdot \boldsymbol{P}_{\alpha_t}  ,  D_{\leq n} \right]  q_t e^{ikx} p_t  D_{\leq n}  U_{-}(t;0)
\nonumber \\
&\qquad -  U_{+}^*(t;0) D_{\leq n} q_t e^{ikx} p_t  \left[ H_{p,\alpha}(t) - \varepsilon k \cdot \boldsymbol{P}_{\alpha_t}  ,  D_{\leq n} \right]   U_{-}(t;0)  .
\end{align}
Here,  $\boldsymbol{P}_{\alpha_t} = \left( - i \varepsilon \nabla - \kappa * \vA_{\alpha_t} \right)$ and the operators $H_{p,\alpha}(t)$,  $U_{+}$, and $U_{-}$ are defined as in 
\eqref{eq:definition mean-field Hamiltonian} and \eqref{eq:propagation of q exponential p in the semiclassical structure definition of the unitaries}.  Using $\left[ \nabla ,  D_{\leq n} \right] = 0$,  Duhamel's formula,  taking the trace norm and passing to the limit $n \rightarrow \infty$,  we obtain
\begin{subequations}
\begin{align}
&\norm{q_t e^{ikx}  p_t}_{\mathfrak{S}_{\varepsilon}^{1,1}}
-  \norm{q_0 e^{ikx}  p_0}_{\mathfrak{S}_{\varepsilon}^{1,1}}
\nonumber \\
&\quad \leq 
\varepsilon^{-1} \int_0^t  \norm{ q_s \big[ X_{p_s} , e^{ikx} \big] p_s }_{\mathfrak{S}_{\varepsilon}^{1,1}} \,   ds 
\nonumber \\
&\qquad + \abs{k}
\int_0^t 
\left(  \norm{ \big( q_s \boldsymbol{P}_{\alpha_s}  p_s - p_s \boldsymbol{P}_{\alpha_s} q_s \big) e^{ikx} p_s}_{\mathfrak{S}_{\varepsilon}^{1,1}}
+ \norm{q_s e^{ikx} \big( q_s \boldsymbol{P}_{\alpha_s}  p_s - p_s \boldsymbol{P}_{\alpha_s} q_s \big)  }_{\mathfrak{S}_{\varepsilon}^{1,1}}
\right) \,   ds 
\nonumber \\
&\qquad + \varepsilon^{-1}
\int_0^t 
\left( 
\norm{\left[ H_{p,\alpha}(s) + \varepsilon^2 \Delta - \varepsilon k \cdot \kappa * \vA_{\alpha_s}  ,  D_{\varepsilon} \right]  q_s e^{ikx} p_s D_{\varepsilon} }_{\mathfrak{S}^{1}}
\right) \,   ds 
\nonumber \\
&\qquad + \varepsilon^{-1}
\int_0^t 
\left( 
\norm{D_{\varepsilon} q_s e^{ikx} p_s \left[ H_{p,\alpha}(s) + \varepsilon^2 \Delta + \varepsilon k \cdot \kappa * \vA_{\alpha_s}  ,  D_{\varepsilon} \right] }_{\mathfrak{S}^{1}}
\right)  \,   ds 
\nonumber \\
\label{eq:semiclassical structure with kinetic weight exponential term a}
&\quad \leq 
\varepsilon^{-1} \int_0^t  \norm{ q_s \big[ X_{p_s} , e^{ikx} \big] p_s }_{\mathfrak{S}_{\varepsilon}^{1,1}} \,   ds 
\\
\label{eq:semiclassical structure with kinetic weight exponential term b}
&\qquad + C \abs{k}
\int_0^t 
\left(  \norm{ D_{\varepsilon} p_s e^{ikx} D_{\varepsilon}^{-1} }_{\mathfrak{S}^{\infty}}
+  \norm{ D_{\varepsilon} q_s e^{ikx} D_{\varepsilon}^{-1} }_{\mathfrak{S}^{\infty}}
\right) 
\norm{q_s \boldsymbol{P}_{\alpha_s} p_s}_{\mathfrak{S}_{\varepsilon}^{1,1}}
\,   ds 
\\
\label{eq:semiclassical structure with kinetic weight exponential term c}
&\qquad + 
\int_0^t 
\left(  \varepsilon^{-1}
\norm{\left[ H_{p,\alpha}(s) + \varepsilon^2 \Delta   ,  D_{\varepsilon} \right] D_{\varepsilon}^{-1}}_{\mathfrak{S}^{\infty}}
+ \abs{k}
\norm{\left[   \kappa * \vA_{\alpha_s} ,  D_{\varepsilon} \right]}_{\mathfrak{S}^{\infty}}
\right) 
\norm{ q_s e^{ikx} p_s  }_{\mathfrak{S}_{\varepsilon}^{1,1}}  \, ds .
\end{align}
\end{subequations}
In the following we estimate each term separately.
Using \eqref{eq:semiclassical structure commutator between exchange term and exponential} and \eqref{eq:Fourier transform of the potential} we get
\begin{align}
\eqref{eq:semiclassical structure with kinetic weight exponential term a}
&\leq C \varepsilon^{-1} N^{-1} \int_0^t
\left( 1  + 
\norm{i \varepsilon \nabla p_s}_{\mathfrak{S}^{\infty}}^2 \right) \int_{\mathbb{R}^3} 
\abs{\mathcal{F}[K](l)} \norm{D_{\varepsilon} e^{ilx} D_{\varepsilon}^{-1}}_{\mathfrak{S}^{\infty}}^2 \norm{q_s e^{ikx} p_s}_{\mathfrak{S}_{\varepsilon}^{1,1}} \, dl \,  ds
\nonumber \\
&\leq C \varepsilon^{-1} N^{-1} \int_0^t \left( 1  + 
\norm{i \varepsilon \nabla p_s}_{\mathfrak{S}^{\infty}}^2 \right)  \int_{\mathbb{R}^3}
\abs{\mathcal{F}[\kappa](l)}^2 \abs{l}^{-2} \left( 1 + \abs{l} \right)^2  \norm{q_s e^{ikx} p_s}_{\mathfrak{S}_{\varepsilon}^{1,1}} \, dl \,  ds
\nonumber \\
&\leq C  \varepsilon^{-1} N^{-1} \int_0^t  
\left( 1  + 
\norm{i \varepsilon \nabla p_s}_{\mathfrak{S}^{\infty}}^2 \right) 
\norm{q_s e^{ikx} p_s}_{\mathfrak{S}_{\varepsilon}^{1,1}}  \,  ds .
\end{align}
Due to \eqref{eq:semiclassical structure with kinetic weights estimate for q A p} we have
\begin{align}
\norm{ q_s   \boldsymbol{P}_{\alpha_s}  p_s }_{\mathfrak{S}_{\varepsilon}^{1,1}}
&\leq C  \left( 1 + 
 \norm{\alpha_s}_{\mathfrak{h}_1} \right)
\Xi_2[N,p_s] .
\end{align}
Together with $q_s = 1 - p_s$ and 
\begin{align}
\norm{ D_{\varepsilon} p_s e^{ikx} D_{\varepsilon}^{-1} }_{\mathfrak{S}^{\infty}}
&\leq C \left( 1 + \varepsilon \abs{k} + 
\norm{i \varepsilon \nabla p_s}_{\mathfrak{S}^{\infty}} \right)
\end{align}
this gives
\begin{align}
\eqref{eq:semiclassical structure with kinetic weight exponential term b}
&\leq  C  \abs{k} \left( 1 + \varepsilon \abs{k} \right)  
\int_0^t 
\left( 1 + \norm{i \varepsilon \nabla p_s}_{\mathfrak{S}^{\infty}} \right)
 \left( 1 + 
 \norm{\alpha_s}_{\mathfrak{h}_1} \right) \,
\Xi_2[N,p_s] \, ds  .
\end{align}
By means of \eqref{eq:estimate for S-infty norm of commutator with regularized fractional Laplacian} and \eqref{eq:semiclassical structure with kinetic weights commuator between hamiltonian and regularized fractional Laplacian}  we estimate
\begin{align}
\eqref{eq:semiclassical structure with kinetic weight exponential term c}
&\leq C \left( 1 + \varepsilon \abs{k} \right)
\int_0^t 
\left( 1 + \norm{ \alpha_s}_{\dot{\mathfrak{h}}_{1/2}}^2
+  \norm{i \varepsilon \nabla p_s}_{\mathfrak{S}^{\infty}}  \right)
 \norm{ q_s e^{ikx} p_s  }_{\mathfrak{S}_{\varepsilon}^{1,1}}  \, ds 
 \nonumber \\
&\leq C  \left( 1 + \abs{k} \right) \left( 1 + \varepsilon \abs{k} \right)
\int_0^t 
\left( 1 + \norm{ \alpha_s}_{\dot{\mathfrak{h}}_{1/2}}^2
+  \norm{i \varepsilon \nabla p_s}_{\mathfrak{S}^{\infty}}  \right)
\Xi_2[N,p_s] \, ds  .
\end{align}
Altogether we obtain
\begin{align}
&\sup_{k \in \mathbb{R}^3}
\left\{
(1 + \abs{k})^{-1} (1 + \varepsilon \abs{k})^{-1} \norm{q_t e^{ikx}  p_t}_{\mathfrak{S}_{\varepsilon}^{1,1}} \right\}
- \sup_{k \in \mathbb{R}^3}
\left\{
(1 + \abs{k})^{-1}  \norm{q_0 e^{ikx}  p_0}_{\mathfrak{S}_{\varepsilon}^{1,1}} \right\}
\nonumber \\
&\quad \leq 
 C  \int_0^t 
\left( 1 + \norm{i \varepsilon \nabla p_s}_{\mathfrak{S}^{\infty}}^2 \right)
 \left( 1 + 
 \norm{\alpha_s}_{\mathfrak{h}_1}^2 \right)
\Xi_2[N,p_s] \, ds .
\end{align}
Note that
\begin{align}
\sup_{k \in \mathbb{R}^3}
\left\{
(1 + \abs{k})^{-1} \norm{q_t e^{ikx}  p_t}_{\mathfrak{S}_{\varepsilon}^{1,1}} \right\}
&\leq 2
\sup_{k \in \mathbb{R}^3, \abs{k} \leq \varepsilon^{-1}}
\left\{
(1 + \abs{k})^{-1} (1 + \varepsilon \abs{k})^{-1} \norm{q_t e^{ikx}  p_t}_{\mathfrak{S}_{\varepsilon}^{1,1}} \right\}
\nonumber \\
&\quad +
\sup_{k \in \mathbb{R}^3, \abs{k} \geq \varepsilon^{-1}}
\left\{
(1 + \abs{k})^{-1} \norm{q_t e^{ikx}  p_t}_{\mathfrak{S}_{\varepsilon}^{1,1}} \right\}
\end{align}
because $1 \leq 2 (1 + \varepsilon \abs{k})^{-1}$ for all $k \in \mathbb{R}^3$ such that $\abs{k} \leq \varepsilon^{-1}$.  By means of \eqref{eq:estimate for kinetic energy}
we have
\begin{align}
\norm{q_t e^{ikx}  p_t}_{\mathfrak{S}_{\varepsilon}^{1,1}}
&\leq
\norm{e^{ikx}  p_t}_{\mathfrak{S}_{\varepsilon}^{1,1}}
+  \norm{p_t e^{ikx}  p_t}_{\mathfrak{S}_{\varepsilon}^{1,1}}
\nonumber \\
&\leq \norm{D_{\varepsilon} e^{ikx} D_{\varepsilon}^{-1}}_{\mathfrak{S}^{\infty}}
\norm{ p_t}_{\mathfrak{S}_{\varepsilon}^{1,1}}
+ \norm{D_{\varepsilon} p_t}_{\mathfrak{S}^{2}}^2
\nonumber \\
&\leq C \left( 1 + \varepsilon \abs{k} \right) \tr \left( (1 - \varepsilon^2 \Delta) p_t \right)
\nonumber \\
&\leq C  \left( 1 + \varepsilon \abs{k} \right)
\left(  \mathcal{E}^{\rm{MS}}[p_s,\alpha_t] + N \right) .
\end{align}
Thus if we use the conservation of energy and that $1 \leq \varepsilon \abs{k}$ for all $k \in \mathbb{R}^3$ such that $\abs{k} \geq \varepsilon^{-1}$ we get
\begin{align}
\sup_{k \in \mathbb{R}^3, \abs{k} \geq \varepsilon^{-1}}
\left\{
(1 + \abs{k})^{-1} \norm{q_t e^{ikx}  p_t}_{\mathfrak{S}_{\varepsilon}^{1,1}} \right\}
&\leq   C  
\left(  \mathcal{E}^{\rm{MS}}[p_0,\alpha_0] + N \right)
\nonumber \\
&\quad \times
\sup_{k \in \mathbb{R}^3, \abs{k} \geq \varepsilon^{-1}}
\left\{
(1 + \abs{k})^{-1}  \left( 1 + \varepsilon \abs{k} \right)  \right\}
\nonumber \\
&\leq   C    \varepsilon
\left(  \mathcal{E}^{\rm{MS}}[p_0,\alpha_0] + N \right) 
\nonumber \\
&\leq   C    \varepsilon N.
\end{align}
Combining the estimate  from above gives
\begin{align}
\sup_{k \in \mathbb{R}^3}
\left\{
(1 + \abs{k})^{-1} \norm{q_t e^{ikx}  p_t}_{\mathfrak{S}_{\varepsilon}^{1,1}} \right\}
&\leq C  \sup_{k \in \mathbb{R}^3}
\left\{
(1 + \abs{k})^{-1} \norm{q_0 e^{ikx}  p_0}_{\mathfrak{S}_{\varepsilon}^{1,1}} \right\}
+  CN  \varepsilon
\nonumber \\
&\quad +
 C  \int_0^t 
\big< \norm{i \varepsilon \nabla p_s}_{\mathfrak{S}^{\infty}} \big>^2 
\big< \norm{\alpha_s}_{\mathfrak{h}_1} \big>^2 \,  \Xi_2[N,p_s] \,  ds .
\end{align}
If we use the  \eqref{eq:bound for the h-1 norm of alpha},  \eqref{eq:main result condition on the energy and velocity}, and  \eqref{eq:momentum estimate} the estimates lead to
\begin{align}
\Xi_2[N,p_t]
&\leq C  \Xi_2[N,p_0]
+  CN   \varepsilon
+
 C  \int_0^t 
\big< \norm{i \varepsilon \nabla p_s}_{\mathfrak{S}^{\infty}} \big>^2 
\big< \norm{\alpha_s}_{\mathfrak{h}_1} \big>^2 \,  \Xi_2[N,p_s] \,  ds 
\nonumber\\
&\quad \leq 
 C    \Xi_2[N,p_0]
+ C N \varepsilon 
+ \int_0^t 
e^{C \left< s \right>} \Xi_2[N,p_s] \,  ds  .
\end{align}
Applying Gr\"onwall inequality proves the claim.
\end{proof}

\subsection{Proof of Proposition~\ref{proposition:Regularization result for Maxwell-Schroedinger}}

\label{subsection:proof of proposiion about regularization for Maxwell-Schroedinger}

\begin{proof}
Let $(p_t, \alpha_t)$ be the unique solution of the Maxwell--Schr\"odinger equations~\eqref{eq:Maxwell-Schroedinger equations} with initial datum $(p,\alpha)$ and
$(\widetilde{p}_t, \widetilde{\alpha}_t)$ be the unique solution of the Maxwell--Schr\"odinger equations~\eqref{eq:Maxwell-Schroedinger equations} with regularized initial datum $(\mathcal{W}^{-1}[ m_{p,\sigma}], \id_{\abs{\cdot} \leq \Lambda} \alpha)$.  Note that $D_{\leq n}$ is defined as in \eqref{eq:regularized fractional Laplacian with semiclassical parameter},  while $\tr(\widetilde{p}_0) = N$,  $0 \leq \widetilde{p}_0 \leq 1$,  and $\widetilde{p}_0 \in \mathfrak{S}^{2,1}(L^2(\mathbb{R}^3))$ follow from Lemma~\ref{lemma:Vlasov-Maxwell distance between initial data and its regularization updated}.
We define
\begin{align}
\label{eq:definition mean-field Hamiltonian without exchange term}
h_{p_t,\alpha_t} &= \left( - i \varepsilon \nabla - \kappa * \vA_{\alpha_t} \right)^2 + K * \rho_{p_t}  
\end{align}
and the two parameter group given by
\begin{align}
i \varepsilon \partial_t \widetilde{U}(t,s) =  h_{\widetilde{p}_t,\widetilde{\alpha}_t} \widetilde{U}(t,s)
\quad \text{with} \quad \widetilde{U}(s;s) = 1 .
\end{align}
Then
\begin{align}
i \varepsilon \partial_t   \widetilde{U}(0,t) D_{\leq n} \left(  p_t - \widetilde{p}_t \right) D_{\leq n} \widetilde{U}(t,0)
&= 
\widetilde{U}(0,t) D_{\leq n} \left[  h_{p_t,\alpha_t} - h_{\widetilde{p}_t, \widetilde{\alpha}_t}  - X_{p_t} ,  p_t  \right] D_{\leq n} \widetilde{U}(t,0)
\nonumber \\
&\quad +
\widetilde{U}(0,t) D_{\leq n} \left[   X_{\widetilde{p}_t} ,  \widetilde{p}_t  \right] D_{\leq n} \widetilde{U}(t,0)
\nonumber \\
&\quad -
\widetilde{U}(0,t)  \left[   h_{\widetilde{p}_t ,\widetilde{\alpha}_t}   , D_{\leq n}  \right]  \left(  p_t - \widetilde{p}_t \right)  D_{\leq n} \widetilde{U}(t,0)
\nonumber \\
&\quad -
\widetilde{U}(0,t)  D_{\leq n}  \left(  p_t  - \widetilde{p}_t \right)
 \left[   h_{\widetilde{p}_t , \widetilde{\alpha}_t}   , D_{\leq n}  \right]    \widetilde{U}(t,0) .
\end{align}
Note that 
\begin{align}
 h_{p_t,\alpha_t} -  h_{\widetilde{p}_t,\widetilde{\alpha}_t}
 &= 2 \kappa * \vA_{\alpha_t- \widetilde{\alpha}_t} \cdot i \varepsilon \nabla
 +  \left( \kappa * \vA_{\alpha_t} \right)^2 -  \left( \kappa * \vA_{\widetilde{\alpha}_t} \right)^2
 + K * \rho_{p_t - \widetilde{p}_t}  .
\end{align}
Thus, if we apply Duhamel's formula, take the trace norm first and the limit $n \rightarrow \infty$ afterwards, we get
\begin{subequations}
\begin{align}
&\norm{ p_t - \widetilde{p}_t}_{\mathfrak{S}^{1,1}_{\varepsilon}}
- \norm{ p_0 - \widetilde{p}_0}_{\mathfrak{S}^{1,1}_{\varepsilon}}
\nonumber \\
&\quad \leq 
\label{eq:Regularization MS Duhamel a}
\varepsilon^{-1} \int_0^t \norm{\left[  \left( \kappa * \vA_{\alpha_s} \right)^2 -  \left( \kappa * \vA_{\widetilde{\alpha}_s} \right)^2  , p_s \right]}_{\mathfrak{S}^{1,1}_{\varepsilon}} \, ds 
\\
\label{eq:Regularization MS Duhamel b}
&\qquad + 2
\varepsilon^{-1} \int_0^t \norm{\left[ \kappa * \vA_{\alpha_s - \widetilde{\alpha}_s} \cdot i \varepsilon \nabla , p_s \right]}_{\mathfrak{S}^{1,1}_{\varepsilon}}  \, ds 
\\
\label{eq:Regularization MS Duhamel c}
&\qquad +  
\varepsilon^{-1} \int_0^t \norm{\left[ K * \rho_{p_s - \widetilde{p}_s} , p_s  \right]}_{\mathfrak{S}^{1,1}_{\varepsilon}}  \, ds 
\\
\label{eq:Regularization MS Duhamel d}
&\qquad +  
\varepsilon^{-1} \int_0^t \norm{\left[ X_{p_s}  ,  p_s \right]}_{\mathfrak{S}^{1,1}_{\varepsilon}}  \, ds 
\\
\label{eq:Regularization MS Duhamel e}
&\qquad +  
\varepsilon^{-1} \int_0^t \norm{\left[ X_{\widetilde{p}_s } ,  \widetilde{p}_s \right]}_{\mathfrak{S}^{1,1}_{\varepsilon}} \, ds 
\\
\label{eq:Regularization MS Duhamel f}
&\qquad +   \varepsilon^{-1} 
\int_0^t \norm{\left[   h_{\widetilde{p}_s,\widetilde{\alpha}_s}  + \varepsilon^2 \Delta  , D_{\varepsilon}  \right]  \left(  p_s - \widetilde{p}_s \right)  D_{\varepsilon}}_{\mathfrak{S}^1}   \, ds 
\\
\label{eq:Regularization MS Duhamel g}
&\qquad +    \varepsilon^{-1} 
\int_0^t \norm{D_{\varepsilon}  \left(  p_s  - \widetilde{p}_s \right)
 \left[   h_{\widetilde{p}_s,\widetilde{\alpha}_s}  + \varepsilon^2 \Delta  ,  D_{\varepsilon}  \right]}_{\mathfrak{S}^1}   \, ds    .
\end{align}
\end{subequations}
In the following we estimate each term separately. 
\paragraph{The term \eqref{eq:Regularization MS Duhamel a}}
Using
\begin{align}
\kappa * \vA^2_{\alpha_s} - \kappa * \vA^2_{\widetilde{\alpha}_s}
= \kappa * \vA_{\alpha_s + \widetilde{\alpha}_s}
\cdot  \kappa * \vA_{\alpha_s - \widetilde{\alpha}_s} ,
\end{align}
expanding the commutator,  and inserting the identity $1 = D_{\varepsilon}^{-1} D_{\varepsilon}$,  we estimate
\begin{align}
\eqref{eq:Regularization MS Duhamel a}
&\leq 
\varepsilon^{-1} \int_0^t
\bigg(
\norm{D_{\varepsilon} \kappa * \vA_{\alpha_s + \widetilde{\alpha}_s} D_{\varepsilon}^{-1}}_{\mathfrak{S}^{\infty}}
\norm{D_{\varepsilon} \left[ \kappa * \vA_{\alpha_s - \widetilde{\alpha}_s} ,  
 p_s \right] D_{\varepsilon}}_{\mathfrak{S}^1}
\nonumber \\
&\qquad \qquad +
\norm{D_{\varepsilon} \kappa * \vA_{\alpha_s - \widetilde{\alpha}_s} D_{\varepsilon}^{-1}}_{\mathfrak{S}^{\infty}}
\norm{D_{\varepsilon} \left[ \kappa * \vA_{\alpha_s + \widetilde{\alpha}_s} ,  
 p_s \right] D_{\varepsilon}}_{\mathfrak{S}^1}
\bigg)
 \, ds .
\end{align}
Due to \eqref{eq:estimate vector potential W-1-infty norm} we have
\begin{align}
\label{eq:regularized Maxwell Duhamel estimate for vector potential}
\norm{D_{\varepsilon} \kappa * \vA_{\alpha_s \pm \widetilde{\alpha}_s} D_{\varepsilon}^{-1}}_{\mathfrak{S}^{\infty}}
&\leq C \norm{\kappa * \vA_{\alpha_s \pm \widetilde{\alpha}_s}}_{W^{1,\infty}_0(\mathbb{R}^3)}
\leq C 
\norm{\alpha_s \pm \widetilde{\alpha}_s}_{\dot{\mathfrak{h}}_{1/2}} .
\end{align}
Moreover, 
\begin{align}
&\norm{D_{\varepsilon} \left[ \kappa * \vA_{\alpha_s \pm \widetilde{\alpha}_s} ,  
 p_s \right] D_{\varepsilon}}_{\mathfrak{S}^1}
\nonumber \\
&\quad \leq 
2 \sum_{\lambda = 1,2} \int_{\mathbb{R}^3}
\abs{k}^{-1/2}  \abs{\mathcal{F}[\kappa](k)} \abs{\left(\alpha_s \pm \widetilde{\alpha}_s \right)(k,\lambda)}
\norm{\left[ e^{ikx} ,  p_s \right]}_{\mathfrak{S}^{1,1}_{\varepsilon}}
\nonumber \\
&\quad \leq  C \sup_{k \in \mathbb{R}^3}
\left\{ (1 + \abs{k})^{-1} \norm{\left[ e^{ikx} ,  p_s  \right]}_{\mathfrak{S}^{1,1}_{\varepsilon}}
\right\} 
\norm{\alpha_s \pm \widetilde{\alpha}_s}_{\dot{\mathfrak{h}}_{1/2}} 
\end{align}
by the Cauchy--Schwarz inequality.  Hence,
\begin{align}
\eqref{eq:Regularization MS Duhamel a}
&\leq C \varepsilon^{-1} 
\int_0^t \left(
\norm{\alpha_s}_{\dot{\mathfrak{h}}_{1/2}} 
+ \norm{\widetilde{\alpha}_s}_{\dot{\mathfrak{h}}_{1/2}} 
\right)
\norm{\alpha_s - \widetilde{\alpha}_s}_{\dot{\mathfrak{h}}_{1/2}} 
\sup_{k \in \mathbb{R}^3}
\left\{ (1 + \abs{k})^{-1} \norm{\left[ e^{ikx} ,  p_s  \right]}_{\mathfrak{S}^{1,1}_{\varepsilon}}
\right\}  \, ds 
\nonumber \\
&\leq C \varepsilon^{-1} 
\int_0^t \left(
\norm{\alpha_s}_{\dot{\mathfrak{h}}_{1/2}} 
+ \norm{\widetilde{\alpha}_s}_{\dot{\mathfrak{h}}_{1/2}} 
\right)
\norm{\alpha_s - \widetilde{\alpha}_s}_{\dot{\mathfrak{h}}_{1/2}} 
\sup_{k \in \mathbb{R}^3}
\left\{ (1 + \abs{k})^{-1} \norm{ q_s e^{ikx}   p_s }_{\mathfrak{S}^{1,1}_{\varepsilon}}
\right\}  \, ds 
\nonumber \\
&\leq C \varepsilon^{-1} 
\int_0^t \left(
\norm{\alpha_s}_{\dot{\mathfrak{h}}_{1/2}} 
+ \norm{\widetilde{\alpha}_s}_{\dot{\mathfrak{h}}_{1/2}} 
\right) \Xi_2[N,p_s]
\norm{\alpha_s - \widetilde{\alpha}_s}_{\dot{\mathfrak{h}}_{1/2}}  \, ds 
\end{align}
with $q_s =  1 - p_s$ and $\Xi_2[N,p_s]$ being defined as in \eqref{eq:definition function for semiclassical structure with kinetic weights}.
\paragraph{The term \eqref{eq:Regularization MS Duhamel b}}
By means of $\left[ O,  p_s \right] = q_s O p_s - p_s O q_s$,  $1 = p_s + q_s$, and $1 = D_{\varepsilon} D_{\varepsilon}^{-1}$ we get
\begin{align}
\eqref{eq:Regularization MS Duhamel b}
&\leq 
4 \varepsilon^{-1} \int_0^t \norm{q_s  \kappa * \vA_{\alpha_s - \widetilde{\alpha}_s} \cdot i \varepsilon \nabla  p_s }_{\mathfrak{S}^{1,1}_{\varepsilon}}  \, ds 
\nonumber \\
&\leq 
4 \varepsilon^{-1} \int_0^t 
\bigg( 
\norm{q_s  \kappa * \vA_{\alpha_s - \widetilde{\alpha}_s} p_s \cdot i \varepsilon \nabla  p_s }_{\mathfrak{S}^{1,1}_{\varepsilon}}  
\nonumber \\
&\qquad \qquad \quad 
+
\norm{q_s  \kappa * \vA_{\alpha_s - \widetilde{\alpha}_s} q_s \cdot i \varepsilon \nabla  p_s }_{\mathfrak{S}^{1,1}_{\varepsilon}} 
\bigg)
\, ds 
\nonumber \\
&\leq 
4 \varepsilon^{-1} \int_0^t 
\bigg( \norm{D_{\varepsilon} p_s}_{\mathfrak{S}^{\infty}}
\norm{q_s  \kappa * \vA_{\alpha_s - \widetilde{\alpha}_s} p_s }_{\mathfrak{S}^{1,1}_{\varepsilon}}  
\nonumber \\
&\qquad \qquad \quad 
+
\norm{D_{\varepsilon} q_s  \kappa * \vA_{\alpha_s - \widetilde{\alpha}_s} D_{\varepsilon}^{-1}}_{\mathfrak{S}^{\infty}}
\norm{ q_s  i \varepsilon \nabla  p_s }_{\mathfrak{S}^{1,1}_{\varepsilon}} 
\bigg) \, ds 
\nonumber \\
&\leq 
C \varepsilon^{-1} \int_0^t 
\left( 1 + \norm{D_{\varepsilon} p_s}_{\mathfrak{S}^{\infty}} \right)
\bigg( 
\norm{q_s  \kappa * \vA_{\alpha_s - \widetilde{\alpha}_s} p_s }_{\mathfrak{S}^{1,1}_{\varepsilon}}  
\nonumber \\
&\qquad \qquad \quad 
+
\norm{D_{\varepsilon} \kappa * \vA_{\alpha_s - \widetilde{\alpha}_s} D_{\varepsilon}^{-1}}_{\mathfrak{S}^{\infty}}
\norm{ q_s  i \varepsilon \nabla  p_s }_{\mathfrak{S}^{1,1}_{\varepsilon}} 
\bigg) \, ds  .
\end{align}
Using \eqref{eq:semiclassical structure with kinetic weights estimate for q A p}
and  \eqref{eq:regularized Maxwell Duhamel estimate for vector potential} we obtain
\begin{align}
\eqref{eq:Regularization MS Duhamel b}
&\leq C \varepsilon^{-1} \int_0^t 
\left( 1 + \norm{i \varepsilon \nabla p_s}_{\mathfrak{S}^{\infty}} \right)
\Xi_2[N,p_s]
\norm{\alpha_s - \widetilde{\alpha}_s}_{\dot{\mathfrak{h}}_{1/2}}  \, ds .
\end{align}
\paragraph{The term \eqref{eq:Regularization MS Duhamel c}}
By means of 
\eqref{eq:Fourier transform of the potential},  \eqref{eq:regularized MS estimate for the Fourier transform of the density},  and $\left[ e^{ikx},  p_s \right] = q_s e^{ikx} p_s - p_s e^{ikx} q_s$
we estimate
\begin{align}
\eqref{eq:Regularization MS Duhamel c}
&\leq C  \varepsilon^{-1}
\int_0^t   \abs{\mathcal{F}[\kappa](k)}^2 \abs{k}^{-2}
\abs{\mathcal{F}[\rho_{p_s - \widetilde{p}_s}](k)} 
\norm{  q_s e^{ikx} p_s  }_{\mathfrak{S}^{1,1}_{\varepsilon}}
\, ds 
\nonumber \\
&\leq C \varepsilon^{-1}  N^{-1}  
\int_0^t  \Xi_2[N,p_s]
\norm{p_s - \widetilde{p}_s}_{\mathfrak{S}^1}
 \,  ds .
\end{align}
\paragraph{The terms \eqref{eq:Regularization MS Duhamel d} and \eqref{eq:Regularization MS Duhamel e}}
In complete analogy to the estimates for \cite[(IV.19e)]{LS2023} we get
\begin{align}
\eqref{eq:Regularization MS Duhamel d}
+
\eqref{eq:Regularization MS Duhamel e}
&\leq C \varepsilon  \left( \mathcal{E}^{\rm{MS}} \big[p_0, \alpha_0\big] +
\mathcal{E}^{\rm{MS}} \big[ \widetilde{p}_0, \widetilde{\alpha}_0 \big]  + N  \right) t    .
\end{align}
Using \eqref{eq:main result condition on the energy and velocity}, $\mathcal{E}^{\rm{MS}} \big[ \widetilde{p}_0, \widetilde{\alpha}_0 \big] = \mathcal{E}^{\rm{MS}} \big[ \mathcal{W}^{-1}[ m_{p,\sigma}], \id_{\abs{\cdot} \leq \Lambda} \alpha \big]$, and  \eqref{eq:derivation of VM bound for the energy of regularized initial data} we get
\begin{align}
\eqref{eq:Regularization MS Duhamel d}
+
\eqref{eq:Regularization MS Duhamel e}
&\leq  C N \varepsilon   t   .
\end{align}
\paragraph{The terms \eqref{eq:Regularization MS Duhamel f} and \eqref{eq:Regularization MS Duhamel g}}
Applying the same estimate as for the term  \eqref{eq:improved semiclassical structure gradient term a} leads to
\begin{align}
\norm{
 \left[   h_{\widetilde{p}_s  ,\widetilde{\alpha}_s} + \varepsilon^2 \Delta   , D_{\varepsilon}  \right] D_{\varepsilon}^{-1} }_{\mathfrak{S}^{\infty}} 
&\leq  C \varepsilon    \left( 1 +  \norm{ \widetilde{\alpha}_s}_{\dot{\mathfrak{h}}_{1/2}}^2
 \right)
\end{align}
and
\begin{align}
 \eqref{eq:Regularization MS Duhamel f} + \eqref{eq:Regularization MS Duhamel g}
 &\leq 2   \varepsilon^{-1} 
\int_0^t \norm{ p_s - \widetilde{p}_s }_{\mathfrak{S}^{1,1}_{\varepsilon}}
\norm{
 \left[   h_{\widetilde{p}_s ,\widetilde{\alpha}_s} + \varepsilon^2 \Delta    , D_{\varepsilon}  \right] D_{\varepsilon}^{-1} }_{\mathfrak{S}^{\infty}}   \, ds  
 \nonumber \\
&\leq C  \int_0^t  \left( 1 +  \norm{ \widetilde{\alpha}_s}_{\dot{\mathfrak{h}}_{1/2}}^2
\right) \norm{ p_s - \widetilde{p}_s }_{\mathfrak{S}^{1,1}_{\varepsilon}}
\,  ds .
\end{align}
Collecting the estimates,  dividing by $N$ and using  \eqref{eq:estimate for field energy}) leads to 
\begin{align}
&N^{-1} \norm{ p_t - \widetilde{p}_t}_{\mathfrak{S}^{1,1}_{\varepsilon}}
- N^{-1} \norm{p_0 - \widetilde{p}_0}_{\mathfrak{S}^{1,1}_{\varepsilon}}
\nonumber \\
&\quad \leq 
C \varepsilon t  
+ C \int_0^t \left<  \norm{i \varepsilon \nabla p_s}_{\mathfrak{S}^{\infty}} \right>
\left< N^{-1} \varepsilon^{-1}   \Xi_2[N,p_s] \right>
\left( N^{-1} \norm{ p_s - \widetilde{p}_s}_{\mathfrak{S}^{1,1}_{\varepsilon}}
+\norm{\alpha_s - \widetilde{\alpha}_s}_{\dot{\mathfrak{h}}_{1/2}} \right)  \, ds .
\end{align}
Using Duhamel's formula we obtain
\begin{align}
\alpha_t(k,\lambda) - \widetilde{\alpha}_t(k,\lambda)
&= e^{- i \abs{k} t} \left( \alpha_0(k,\lambda) - \widetilde{\alpha}_0(k,\lambda) \right)
\nonumber \\
&\quad +
i \int_0^t e^{- i \abs{k} (t-s)}
\sqrt{\frac{4 \pi^3}{\abs{k}}} \mathcal{F}[\kappa](k)
\vep_{\lambda}(k) \cdot 
\left( \mathcal{F}[\vJ_{p_s,\alpha_s}]
-  \mathcal{F}[\vJ_{\widetilde{p}_s,\widetilde{\alpha}_s}]
\right)  \, ds .
\end{align}
By means of \eqref{eq:assumption cutoff for self-adjointess of PF Hamiltonian} and (see \cite[Lemma III.3]{LS2023} for a proof)
\begin{align}
&\norm{\mathcal{F}[\vJ_{p_s,\alpha_s}] - \mathcal{F}[\vJ_{\widetilde{p}_s,\widetilde{\alpha}_s}]}_{L^{\infty}(\mathbb{R}^3)}
\nonumber \\
&\quad \leq C N^{-1} \Big[ \big( 1 + \norm{\alpha_s}_{\dot{\mathfrak{h}}_{1/2}}  \big) \norm{ p_s - \widetilde{p}_s}_{\mathfrak{S}^{1,1}_{\varepsilon}}
+  \norm{\widetilde{p}_s}_{\mathfrak{S}^1}
\norm{\alpha_s - \widetilde{\alpha}_s}_{\dot{\mathfrak{h}}_{-1/2} \cap \dot{\mathfrak{h}}_{1/2}} \Big]
\end{align}
we estimate
\begin{align}
&\norm{\alpha_t - \widetilde{\alpha}_t}_{\dot{\mathfrak{h}}_{-1/2} \cap \dot{\mathfrak{h}}_{1/2}}
- \norm{\alpha_0 - \widetilde{\alpha}_0}_{\dot{\mathfrak{h}}_{-1/2} \cap \dot{\mathfrak{h}}_{1/2}}
\nonumber \\
&\quad \leq 
C \int_0^t  \norm{
\big( \abs{\cdot}^{-1} + 1 \big)  \mathcal{F}[\kappa](k)
\left( \mathcal{F}[\vJ_{p_s,\alpha_s}] - \mathcal{F}[\vJ_{\widetilde{p}_s},\widetilde{\alpha}_s] \right)}_{L^{2}(\mathbb{R}^3)} \,   ds
\nonumber \\
&\quad \leq 
C \int_0^t  \norm{
\big( \abs{\cdot}^{-1} + 1 \big)  \mathcal{F}[\kappa](k) }_{L^{2}(\mathbb{R}^3)} 
\norm{ \mathcal{F}[\vJ_{p_s,\alpha_s}] - \mathcal{F}[\vJ_{\widetilde{p}_s,\widetilde{\alpha}_s}] }_{L^{\infty}(\mathbb{R}^3)} \,   ds
\nonumber \\
&\quad \leq 
C  \int_0^t   \left( \big( 1 + \norm{\alpha_s}_{\dot{\mathfrak{h}}_{1/2}}  \big) 
N^{-1}
\norm{ p_s - \widetilde{p}_s}_{\mathfrak{S}^{1,1}_{\varepsilon}}
+  \norm{\alpha_s - \widetilde{\alpha}_s}_{\dot{\mathfrak{h}}_{-1/2} \cap \dot{\mathfrak{h}}_{1/2}} \right)  \,  ds.
\end{align}
Using  \eqref{eq:estimate for field energy} and  \eqref{eq:main result condition on the energy and velocity} we get
\begin{align}
&\norm{\alpha_t - \widetilde{\alpha}_t}_{\dot{\mathfrak{h}}_{-1/2} \cap \dot{\mathfrak{h}}_{1/2}}
- \norm{\alpha_0 - \widetilde{\alpha}_0}_{\dot{\mathfrak{h}}_{-1/2} \cap \dot{\mathfrak{h}}_{1/2}}
\nonumber \\
&\quad \leq C 
\int_0^t  \left( N^{-1} \norm{ p_s - \widetilde{p}_s}_{\mathfrak{S}^{1,1}_{\varepsilon}}
+\norm{\alpha_s - \widetilde{\alpha}_s}_{\dot{\mathfrak{h}}_{-1/2} \cap \dot{\mathfrak{h}}_{1/2}} \right)
\, ds .
\end{align}
Thus if we collect the estimates and apply Lemma~\eqref{lemma: semiclassical structure with kinetic weights}, \eqref{eq:momentum estimate},  \eqref{eq:main result condition on the energy and velocity},  and \eqref{eq:Vlasov-Maxwell result more restricted initial semiclassical structure} we obtain 
\begin{align}
&N^{-1} \norm{ p_t - \widetilde{p}_t}_{\mathfrak{S}^{1,1}_{\varepsilon}}
+  \norm{\alpha_t - \widetilde{\alpha}_t}_{\dot{\mathfrak{h}}_{-1/2} \cap \dot{\mathfrak{h}}_{1/2}}
\nonumber \\
&\quad \leq 
N^{-1} \norm{ p_0 - \widetilde{p}_0}_{\mathfrak{S}^{1,1}_{\varepsilon}}
+  \norm{\alpha_0 - \widetilde{\alpha}_0}_{\dot{\mathfrak{h}}_{-1/2} \cap \dot{\mathfrak{h}}_{1/2}}
+ 
C \varepsilon t
\nonumber \\
&\qquad + C
\left< N^{-1} \varepsilon^{-1} \Xi_2[N,p_0] \right> \big< \norm{i \varepsilon \nabla p_0}_{\mathfrak{S}^{\infty}} \big>
\int_0^t
\exp \left[   
 \exp \big[ C \left< s \right> \big] \right]   
\nonumber \\
&\qquad \qquad \qquad \qquad \qquad \qquad  \times
\left( N^{-1} \norm{ p_s - \widetilde{p}_s}_{\mathfrak{S}^{1,1}_{\varepsilon}}
+\norm{\alpha_s - \widetilde{\alpha}_s}_{\dot{\mathfrak{h}}_{1/2}} \right) \,  ds
\nonumber \\
&\quad \leq 
N^{-1} \norm{ p_0 - \widetilde{p}_0}_{\mathfrak{S}^{1,1}_{\varepsilon}}
+  \norm{\alpha_0 - \widetilde{\alpha}_0}_{\dot{\mathfrak{h}}_{-1/2} \cap \dot{\mathfrak{h}}_{1/2}}
+ 
C \varepsilon t
\nonumber \\
&\qquad + 
\int_0^t
\exp \left[  \exp \big[ C \left< s \right> \big] \right]   
\left( N^{-1} \norm{ p_s - \widetilde{p}_s}_{\mathfrak{S}^{1,1}_{\varepsilon}}
+\norm{\alpha_s - \widetilde{\alpha}_s}_{\dot{\mathfrak{h}}_{1/2}} \right) \,  ds .
\end{align}
Gr\"onwall's inequality leads to
\begin{align}
&N^{-1} \norm{ p_t - \widetilde{p}_t}_{\mathfrak{S}^{1,1}_{\varepsilon}}
+  \norm{\alpha_t - \widetilde{\alpha}_t}_{\dot{\mathfrak{h}}_{-1/2} \cap \dot{\mathfrak{h}}_{1/2}}
\nonumber \\
&\quad \leq 
\left( N^{-1} \norm{ p_0 - \widetilde{p}_0}_{\mathfrak{S}^{1,1}_{\varepsilon}}
+  \norm{\alpha_0 - \widetilde{\alpha}_0}_{\dot{\mathfrak{h}}_{-1/2} \cap \dot{\mathfrak{h}}_{1/2}}
+ \varepsilon
\right)
\exp \left[ \exp \left[  \exp \big[ C \left< t \right> \big] \right]  \right]  .
\end{align}
Estimating the norm differences between the initial data by means of Lemma~\ref{lemma:Vlasov-Maxwell distance between initial data and its regularization updated} completes the proof.
\end{proof}


\noindent{\it Acknowledgments.}
N.L. gratefully acknowledges funding from the Swiss National Science Foundation through the SNSF Eccellenza project PCEFP2 181153 and partial support by the NCCR SwissMAP, the Swiss State Secretariat for Research and Innovation through the project P.530.1016 (AEQUA), and the European Union's Horizon 2020 research and innovation programme through the Marie Sk{\l}odowska-Curie Action EFFECT (grant agreement No. 101024712).

\appendix


{}


\end{document}